\documentclass{amsart}

\usepackage{a4wide}
\usepackage{hyperref}

\theoremstyle{plain}
\newtheorem{thm}{Theorem}
\newtheorem{prop}[thm]{Proposition}
\newtheorem{lemma}[thm]{Lemma}
\newtheorem{cor}[thm]{Corollary}
\theoremstyle{definition}
\newtheorem{definition}[thm]{Definition}
\newtheorem{remark}[thm]{Remark}

\usepackage{amssymb,amsmath,amsthm,mathrsfs}
\usepackage{graphicx}        
\usepackage{upgreek}        

\newcommand{\tn}[1]{\ensuremath{\mathbb{T}^{#1}}}

\newcommand{\rn}[1]{\ensuremath{\mathbb{R}^{#1}}}
\newcommand{\zn}[1]{\ensuremath{\mathbb{Z}^{#1}}}
\newcommand{\sn}[1]{\ensuremath{\mathbb{S}^{#1}}}

\newcommand{\Spe}{\mathrm{Sp}}
\newcommand{\whsu}{\mathrm{wh}}
\newcommand{\que}{\mathrm{qu}}

\newcommand{\rosob}{\mathrm{Sob}}
\newcommand{\rohy}{\mathrm{hy}}

\newcommand{\rond}{\mathrm{nd}}

\newcommand{\rowo}{\mathrm{wo}}
\newcommand{\mrod}{\mathrm{od}}

\newcommand{\mrI}{\mathrm{I}}
\newcommand{\mrII}{\mathrm{II}}
\newcommand{\mrIII}{\mathrm{III}}
\newcommand{\mrIV}{\mathrm{IV}}
\newcommand{\mrV}{\mathrm{V}}

\newcommand{\bRic}{\overline{\mathrm{Ric}}}

\newcommand{\zo}{\mathbb{Z}}
\newcommand{\ro}{\mathbb{R}}

\newcommand{\rovar}{\mathrm{var}}

\newcommand{\rodiv}{\mathrm{div}}

\newcommand{\rorel}{\mathrm{rel}}

\newcommand{\so}{\mathbb{S}^{1}}

\newcommand{\sfK}{\mathsf{K}}

\newcommand{\chsfX}{\check{\mathsf{X}}}
\newcommand{\ONSF}{\mathsf{X}}
\newcommand{\ONCS}{\mathsf{\bar{\mathsf{\Gamma}}}}

\newcommand{\msL}{\mathscr{L}}
\newcommand{\msK}{\mathscr{K}}
\newcommand{\msA}{\mathscr{A}}

\newcommand{\mfP}{\mathfrak{P}}

\newcommand{\refer}{\mathrm{ref}}

\newcommand{\rotot}{\mathrm{tot}}

\newcommand{\robas}{\mathrm{bas}}

\newcommand{\rocon}{\mathrm{con}}

\newcommand{\rodiff}{\mathrm{diff}}

\newcommand{\tr}{\mathrm{tr}}

\newcommand{\pre}{\mathrm{pre}}

\newcommand{\Id}{\mathrm{Id}}

\newcommand{\be}{\bar{e}}

\newcommand{\ba}{\bar{a}}

\newcommand{\tM}{\tilde{M}}

\newcommand{\tx}{\tilde{x}}

\newcommand{\tg}{\tilde{g}}

\newcommand{\bk}{\bar{k}}
\newcommand{\bK}{\bar{K}}

\newcommand{\hE}{\hat{E}}

\newcommand{\chth}{\check{\theta}}

\newcommand{\hmu}{\hat{\mu}}
\newcommand{\hml}{\hat{\ml}}

\newcommand{\dotchi}{\dot{\chi}}

\newcommand{\bge}{\bar{g}}
\newcommand{\bS}{\bar{S}}

\newcommand{\bD}{\bar{D}}

\newcommand{\bR}{\bar{R}}
\newcommand{\bM}{\bar{M}}

\newcommand{\bmu}{\bar{\mu}}

\newcommand{\bnabla}{\overline{\nabla}}

\newcommand{\bfI}{\mathbf{I}}

\newcommand{\bfJ}{\mathbf{J}}
\newcommand{\bfK}{\mathbf{K}}
\newcommand{\bfL}{\mathbf{L}}

\newcommand{\bbE}{\mathbb{E}}

\newcommand{\bfE}{\mathbf{E}}

\newcommand{\bfz}{\mathbf{0}}

\newcommand{\bga}{\bar{\gamma}}
\newcommand{\tsigma}{\tilde{\sigma}}
\newcommand{\tSigma}{\tilde{\Sigma}}
\newcommand{\bG}{\bar{\Gamma}}
\renewcommand{\a}{\alpha}
\newcommand{\e}{\epsilon}
\newcommand{\vare}{\varepsilon}

\newcommand{\de}{\delta}

\renewcommand{\b}{\beta}
\newcommand{\teta}{\tilde{\eta}}
\newcommand{\g}{\gamma}

\renewcommand{\d}{\partial}

\newcommand{\me}{\mathcal{E}}
\newcommand{\mK}{\mathcal{K}}

\newcommand{\eSpe}{\varepsilon_{\Spe}}
\newcommand{\emK}{\varepsilon_{\mK}}
\newcommand{\mKsup}{C_{\mK}}
\newcommand{\bDlnhNsup}{C_{\rorel}}

\newcommand{\ma}{\mathcal{A}}

\newcommand{\ml}{\mathcal{L}}

\newcommand{\mt}{\mathcal{T}}

\newcommand{\mU}{\mathcal{U}}

\newcommand{\mS}{\mathcal{S}}
\newcommand{\mSb}{\bar{\mathfrak{S}}}
\newcommand{\mN}{\mathcal{N}}

\newcommand{\mW}{\mathcal{W}}

\newcommand{\mfv}{\mathfrak{v}}

\newcommand{\Weight}{\mathfrak{V}}

\newcommand{\Index}{\mathfrak{I}}
\newcommand{\bfl}{\mathbf{l}}

\newcommand{\weight}{\mathfrak{v}}

\newcommand{\cweight}{\mathfrak{u}}
\newcommand{\bcweight}{\bar{\mathfrak{u}}}

\newcommand{\mH}{\mathcal{H}}

\newcommand{\mP}{\mathcal{P}}
\newcommand{\mR}{\mathcal{R}}

\newcommand{\bmR}{\bar{\mathcal{R}}}
\newcommand{\bmS}{\bar{\mathcal{S}}}
\newcommand{\bmN}{\bar{\mathcal{N}}}
\newcommand{\mQ}{\mathcal{Q}}
\newcommand{\mc}{\mathcal{C}}

\newcommand{\mcY}{\mathcal{Y}}

\newcommand{\hmcY}{\hat{\mcY}}
\newcommand{\md}{\mathcal{D}}

\newcommand{\mfp}{\mathfrak{p}}
\newcommand{\mcP}{\mathcal{P}}
\newcommand{\mfg}{\mathfrak{g}}

\newcommand{\hN}{\hat{N}}

\newcommand{\hGe}{\hat{G}}
\newcommand{\hU}{\hat{U}}

\newcommand{\bGe}{\bar{G}}

\newcommand{\hg}{\hat{g}}
\newcommand{\chg}{\check{g}}
\newcommand{\che}{\check{e}}
\newcommand{\chga}{\check{\gamma}}

\newcommand{\chk}{\check{k}}
\newcommand{\chK}{\check{K}}

\newcommand{\cha}{\check{a}}

\newcommand{\bx}{\bar{x}}

\newcommand{\bp}{\bar{p}}

\newcommand{\tvarphi}{\tilde{\varphi}}

\newcommand{\ldr}[1]{\langle #1\rangle}

\newcommand{\mutgc}{\mu_{\tilde{g};c}}

\begin{document}

\author{Hans Ringstr\"{o}m}
\title{On the geometry of silent and anisotropic big bang singularities}
\begin{abstract}
  The purpose of this and a companion article is to develop a framework for analysing silent and anisotropic big bang singularities. In the companion
  article, we derive asymptotics of solutions to linear systems of wave equations. In the present article, we focus on the geometry. The setting we
  are interested in is that of crushing singularities,
  and a central object of interest is the \textit{expansion normalised Weingarten map} $\mK$, defined as the Weingarten map of the leaves of the
  foliation, divided by the corresponding mean curvature. A remarkable fact concerning big bang singularities is that, as noted and justified in the
  companion article, $\mK$ and its expansion normalised normal derivative are quite generally bounded, and this is our main assumption (though we impose
  several additional conditions). The bounds can be expressed with respect to both weighted $C^{k}$ and Sobolev spaces, but the perspective in this article
  is to think of the bounds as part of a bootstrap assumption and to see what can be deduced from it.

  We are here interested in highly anisotropic settings, where isotropy is understood to mean that the eigenvalues of $\mK$ coincide. In such situations,
  it can be expected to be very important to keep track of the directions in which the contraction is largest/smallest. In short, it can be expected to be
  essential to have a clear geometric
  picture of the asymptotics. This is achieved in the framework we develop. Moreover, we are able to improve some of the assumptions of the framework
  by combining them with Einstein's equations: appealing to the Hamiltonian constraint, the sum of the squares of the
  eigenvalues of $\mK$ is asymptotically bounded from above by $1$, even though no such
  bound was initially assumed; and under certain circumstances, the expansion normalised normal derivative of $\mK$ can be demonstrated to
  decay exponentially and $\mK$ can be demonstrated to converge exponentially, even though we initially only imposed bounds on these quantities. In fact,
  in spite of being general enough that many of the existing results are consistent with the assumptions, the framework is strong enough to reproduce
  the Kasner map, conjectured in the physics literature to constitute the essence of the asymptotic dynamics. Interestingly, the framework suggests a
  way to understand the asymptotics of solutions quite generally. We illustrate this here by discussing some examples in which the asymptotics are known.
  However, an even clearer picture of the quiescent setting is obtained when
  combining the perspective developed in this article with a recent result by Fournodavlos and Luk. This combination leads to a notion of initial data
  on big bang singularities which can be used to give a unified framework for understanding previous results and, potentially, be used as a starting point
  for understanding spatially inhomogeneous and oscillatory big bang singularities. We discuss this topic in a separate article. 

  Finally, the hope is that the results of the present article can lead to the resolution of outstanding problems by providing the bootstrap assumptions
  and arguments needed to obtain a clear picture of the asymptotic geometry. The current state of the art concerning quiescent singularities is the recent
  result by Fournodavlos, Rodnianski and Speck \cite{GIJ} (see also \cite{GPR} and references stated therein). In $3+1$-dimensions, this result states
  that Bianchi type I solutions to the Einstein-scalar field equations with positive definite $\mK$ are past globally non-linearly stable. In the
  present article, we identify a geometric regime in which quiescent behaviour is to be expected but which does not entail proximity to a background
  solution. The article \cite{GPR} (based, in part, on ideas from the present article, \cite{GIJ} and \cite{RinQC}) justifies this expectation. However,
  both \cite{GIJ} and \cite{GPR} are incomplete in the sense that
  they do not provide sufficient information concerning the asymptotics to yield data on the singularity in the sense of \cite{RinQC}. The conclusions
  of the present article are better suited when it comes to achieving this goal. Looking further ahead, the goal is to understand singularities
  which are not necessarily quiescent, a topic on which there is, as yet, only results in the spatially homogeneous setting. Nevertheless, since the
  improvements of the assumptions we obtain here are associated with a loss of derivatives, it is clear that the framework will have to be combined with
  an appropriate gauge choice and energy estimates in order to obtain rough estimates without a loss of derivatives. In order to facilitate this step,
  our framework is general enough that it should be consistent with large classes of gauge choices. 
\end{abstract}

\maketitle

\section{Introduction}

The subject of this article is the asymptotic behaviour of solutions to Einstein's equations in the direction of big bang singularities. The solutions we
have in mind here are, to begin with, maximal globally hyperbolic developments of initial data to Einstein's equations, where the initial manifold is
assumed to be closed (the assumption of compactness is not essential in the context of interest here, but we impose it for the sake of convenience).
Moreover, we assume the existence of a foliation such that the mean curvature of the leaves diverges uniformly to plus or minus infinity to the
future or the past; i.e., there is a crushing singularity. The asymptotic regimes of interest then correspond to the directions of diverging mean
curvature. In addition to these requirements, we here restrict our attention to silent and anisotropic big bang singularities. By silence, we,
heuristically, mean the property that observers going into the singularity typically lose the ability to communicate asymptotically (see
Definition~\ref{def:silenceandnondegeneracy} below for a formal definition). By anisotropy, we mean that the eigenvalues of the expansion
normalised Weingarten map $\mK$ are distinct. The main purpose of this article is twofold. First, we isolate, in a geometric way, the
essential mechanisms causing chaotic dynamics in, e.g., the $3+1$-dimensional vacuum setting. Second, we formulate partial bootstrap
assumptions and deduce improvements of these bootstrap assumptions. The second result should be thought of as a first step to a complete non-linear
result which we hope to obtain in the future. 

\subsection{Background}

In order to put the topic of this article into context, it is useful to recall some of the fundamental questions in the mathematical study of
solutions to Einstein's equations. The most fundamental problem is the \textit{strong cosmic censorship conjecture}, stating that for
generic initial data (in the asymptotically flat or spatially compact setting), the maximal globally hyperbolic development thereof is inextendible.
This conjecture corresponds to the expectation that, generically, Einstein's equations are deterministic. Another fundamental problem is that of
\textit{global non-linear stability}; the standard models of the universe are, e.g., highly symmetric, and it is of interest to know if perturbing
the corresponding initial data yields developments that are globally similar. In particular, the standard models have big bang singularities, and it
is natural to ask if these persist upon perturbations of the initial data. In order to answer questions of this type, it is necessary to analyse the
asymptotic behaviour. Due to Hawking's singularity theorem, see e.g. \cite[Theorem~55A, p.~431]{oneill}, static and stationary solutions with closed Cauchy
hypersurfaces are unnatural,
and the asymptotics are typically in either an ``expanding'' or in a ``contracting'' direction. We shall not attempt to formally define these
notions here, but attributes associated with an expanding/contracting direction are causal geodesic completeness/incompleteness, as well unbounded
growth/decay of the logarithm of the volume of the spatial hypersurfaces of an appropriate foliation. 

Since it is difficult to analyse the asymptotics of solutions in all generality, it is natural to focus on a restricted setting. One way of doing
so is to impose the existence of a group of isometries of the initial data, and this perspective has led to many interesting results. Another
perspective is to study the future/past global non-linear stability of solutions. The topic of interest in this article is singularities, but it is
of interest to put the questions considered here into the bigger context of asymptotic behaviour. Considering the results, it is natural to divide
the different types of asymptotic behaviour into the following categories:

\textit{Solutions that become isotropic asymptotically} (in expanding directions). Solutions undergoing accelerated expansion exemplify the phenomenon of
asymptotic isotropisation; see, e.g., \cite{f} for a proof of global non-linear stability of de Sitter space; \cite{wald2} for an illustration of
how the presence of a positive cosmological constant quite generally causes isotropisation in the spatially homogeneous setting; \cite{stab} for a
future global non-linear stability result of the currently preferred models of the universe; \cite{heinzle,HRPL,LAI} for examples in which the
expansion is of power law type; and \cite{AAR} for an example of isotropisation in a spatially inhomogeneous and anisotropic setting. Other important
examples are provided by the results concerning the future global non-linear stability of the Milne and similar models; see, e.g., \cite{aam,aam2}. See
also, e.g., \cite{garfinkleetal} for numerical work. 

\textit{Quiescent solutions.} Assume the spacetime to be endowed with a foliation by spacelike hypersurfaces and let $\bK$ be
  the Weingarten map of the leaves of the foliation (i.e., the second fundamental form with one index raised). Letting $\theta=\mathrm{tr}\bK$
  denote the mean curvature, the \textit{expansion normalised Weingarten map} is defined by $\mK:=\bK/\theta$, assuming $\theta$ does not vanish.
  For the purposes of the present discussion, we associate quiescence with convergence of the eigenvalues of $\mK$ along causal curves (note that
  this notion is meaningful both to the future and to the past). Note that
  in the case of isotropisation, $\mK$ converges (not only along causal curves, but globally) to $\mathrm{Id}/n$. Asymptotically isotropic solutions
  are therefore, in particular, quiescent. However, quiescence appears much more generally. For example, large classes of symmetric
  solutions have quiescent big bang singularities; see, e.g.,
  \cite{whsu,Wellis,BianchiIXattr,asvelGowdy,SCCGowdy,RadermacherNonStiff,RadermacherStiff}. Another class
  of important results are the ones obtained by specifying the asymptotics and then proving that there are solutions with the corresponding
  asymptotics. There are large numbers of results in this category, starting, to the best of our knowledge, with \cite{kar}. Most of the results
  have been obtained in the presence of symmetries. However, there are also results in the absence of symmetries; see, e.g., \cite{aarendall,daetal,fal}.
  Finally, there are past global non-linear stability results. In \cite{rasql,rasq,specks3}, the authors demonstrate stable big bang formation close
  to isotropy in the Einstein-scalar field and Einstein-stiff fluid settings; see also \cite{bao,bao2,fau}. In \cite{rsh}, the authors demonstrate
  stable big bang formation for a class of moderately anisotropic solutions in the high dimensional (spatial dimension $\geq 38$) vacuum setting.
  In addition, in \cite{GIJ}, the authors prove stability of Bianchi type I solutions in the full expected regime for the Einstein-scalar field system
  (in $n+1$-dimensions for $n\geq 3$) and for the Einstein vacuum equations in $n+1$-dimensions for $n\geq 10$. Finally, in \cite{GPR}, the authors
  identify a stable regime for quiescent big bang formation.

\textit{Oscillatory solutions.} In the direction of a big bang singularity, quiescence is expected under certain symmetry assumptions, for
  suitable matter models (such as stiff fluids and scalar fields) and for high dimensions (spatial dimension at least $10$); see
  Section~\ref{section:resultsformal} and Subsection~\ref{ssection:quiescentregimes} below for a justification. However, in general,
  solutions are expected to be oscillatory. In other words, the eigenvalues of $\mK$ are not expected to converge along causal curves going into
  the singularity (in $3+1$-dimensions, there is a more detailed picture, essentially saying that the eigenvalues of $\mK$ should evolve
  according to a $1$-dimensional chaotic dynamical system called the BKL-map; see Figure~\ref{fig:TheKasnerMap} below). 
  There are much fewer results concerning this situation, presumably partly due to the difficulty of analysing the behaviour in
  this setting. In fact, to the best of our knowledge, all of the results in this setting concern spatially homogeneous solutions in $3+1$-dimensions.
  The two first
  results were obtained in \cite{cbu} (concerning vacuum Bianchi type VIII and IX solutions) and in \cite{wea} (concerning magnetic Bianchi type
  VI${}_{0}$ solutions). These results demonstrate that generic Bianchi type VIII and IX solutions, as well as suitable magnetic Bianchi type
  VI${}_{0}$ solutions have oscillatory singularities. Following this initial progress, there were results demonstrating that generic Bianchi type
  IX orthogonal perfect fluid solutions with a linear equation of state $p=(\g-1)\rho$ (where $p$ is the pressure, $\rho$ is the energy density
  and $\g$ is a constant) with $2/3<\g<2$ converge to an attractor; see \cite{BianchiIXattr}. This is significant due to the fact the dynamics on
  the attractor are described by the so-called Kasner (or BKL) map; going back to the work of Belinski\v{\i}, Khalatnikov and Lifschitz (BKL), the
  Kasner map has been suggested as a model for the asymptotic behaviour. If this expectation is correct, the asymptotics would be described by a
  one-dimensional chaotic dynamical system. In the spatially homogeneous setting, support for this expectation has been obtained in, e.g.,
  \cite{lea,beguin,du}. For numerical work in the oscillatory setting, see, e.g., \cite{Betal1,Betal2,gap} and references cited therein. 

\textit{Silent solutions.} Another important aspect of the asymptotic behaviour is the causal structure. The main distinction of interest
  here is between silent and non-silent asymptotics. There are many ways of imposing silence, but in practice, it is often
  formulated in terms of a foliation. Given that a spacetime $(M,g)$ is foliated according to $M=\bM\times I$, where $I=(t_{-},t_{+})$ and
  $\bM\times\{t\}$ are spacelike hypersurfaces, one version is the requirement that the behaviour localises in the following sense: Fix a subset
  $U$ of $\bM$ diffeomorphic to a ball $B_{2r}(0)$ (with center $0$ and radius $2r>0$) in $\rn{n}$. Let $V\subset U$ correspond to $B_{r}(0)$ under
  the same diffeomorphism. Then silence
  in the future (and past) direction typically means that there is a $T\in I$ such that $J^{+}(V\times \{T\})\subset U\times [T,t_{+})$
  (and $J^{-}(V\times \{T\})\subset U\times (t_{-},T]$) (if $S\subset M$, $J^+(S)$ denotes the set of points to the causal future of $S$ and
  $J^-(S)$ denotes the set of points to the causal past of $S$; see \cite[Chapter~14]{oneill} for a formal definition and basic properties of
  $J^{\pm}(S)$). In particular, the global topology of the spatial hypersurfaces is asymptotically not visible to
  observers. For the above mentioned solutions undergoing accelerated expansion, silence is one feature of the asymptotics; see, e.g.,
  \cite{f,HRINV,stab,HRPL}. Similarly, the spatially homogeneous and inhomogeneous solutions mentioned above with quiescent singularities also
  have asymptotics with this property; see, e.g.,
  \cite{whsu,Wellis,BianchiIXattr,RadermacherNonStiff,RadermacherStiff,aarendall,daetal,fal,rasql,rasq,specks3}. We justify this statement in greater
  detail in \cite[Appendix~C]{RinWave}. On the other hand, in the BKL picture, generic singularities are expected to be oscillatory and silent, and the
  model for the (spatially) local behaviour of solutions is Bianchi type VIII and IX vacuum solutions. For this reason, it is of course of central
  importance to analyse the asymptotic causal structure of Bianchi type VIII and IX vacuum solutions. In spite of the fact that the importance of
  this question has been clear since 1969 (see the work \cite{misner} of Misner), it was only resolved in 2016; see \cite{brehm}. Moreover, even though
  \cite{brehm} ensures silence for solutions corresponding to a set of initial data with full measure, there are also indications; see \cite{brehm};
  that there is a Baire generic set of initial data such that the corresponding solutions do not have silent asymptotics. Finally, note that in the case
  of \cite{aam,aam2}, the causal structure is (as opposed to the spacetimes undergoing accelerated expansion) not silent. 

Considering the above division, it is clear that there is at least a partial hierarchy as far as the level of difficulty is concerned. A situation in
which the solutions isotropise is of course simpler than a setting with general quiescent asymptotics. Moreover, quiescent asymptotics are of course
easier to analyse than oscillatory asymptotics. Turning to the causal structure, silent behaviour simplifies the analysis; in a non-silent situation,
the global spatial topology could potentially come into play. Finally, one could impose symmetry assumptions leading to a separate hierarchy of
difficulty. 

\subsection{A hierarchy of problems in cosmology}

Due to the above observations, it is natural to begin by considering de Sitter space, or, more generally, solutions with accelerated
expansion. The reason for this is that then the solutions isotropise and the causal structure is asymptotically silent. This is the most favourable
situation. It is therefore not surprising that the first global non-linear stability result was obtained in the case of de Sitter; see \cite{f}. 
Next, it is natural to relax either the condition of isotropisation or the condition of silence. Relaxing the condition of silence, one is led to the
study of the future global non-linear stability of the Milne model (i.e., stability in the expanding direction, but not in the direction of the big bang -
due to strong cosmic censorship, the Milne model is expected to be unstable to the past); see \cite{aam,aam2}. Note also that the Milne model is locally
isometric to
Minkowski space (the universal covering space of a Milne model is isometric to the timelike future of a point in Minkowski space), so that \cite{aam}
can, roughly speaking, be thought of as a cosmological companion of the proof of the global non-linear stability of Minkowski space; see
\cite{cak}; and of the proof of future global non-linear stability of a hyperbolic foliation of Minkowski space; see \cite{f}. In fact, the latter
reference is more relevant due to the fact that it is based on a hyperbolic foliation.

The natural next step is to consider a situation which is anisotropic but quiescent. Moreover, it is favourable if the corresponding asymptotics are
silent. This naturally leads to the study of big bang singularities for scalar field matter or for vacuum solutions in case the spatial dimension $n$
satisfies $n\geq 10$. As mentioned above, the first results obtained concerning this setting (in the absence of symmetries) were based on specifying
the asymptotics; see \cite{aarendall,daetal}. A more recent result in this direction
is given by \cite{fal}. However, due to the fact that the asymptotic behaviour is specified a priori in this case, two issues that arise if one starts
with initial data on a Cauchy hypersurface inside the spacetime are avoided. The first issue is the synchronisation of the big bang; in
\cite{aarendall,daetal,fal}, the big bang corresponds to $t=0$ by construction. The second issue is the asymptotic geometry, more specifically the
directions corresponding to the maximal/minimal contraction. Again, in \cite{aarendall,daetal,fal}, what these directions are is known a priori.
When starting with initial data on a Cauchy hypersurface inside the spacetime, what these directions are has to be deduced in the framework of a
bootstrap argument. In \cite{rasql,rasq,rsh,specks3}, the authors address the issue of synchronisation by using a constant mean curvature (CMC)
foliation (a different perspective, which only involves local gauges, is developed in \cite{bao}). However, the conclusions of \cite{rasql,rasq,specks3}
are obtained in the near isotropic setting, and the conclusions of \cite{rsh} are obtained in the moderately anisotropic setting. On the other hand,
there is a natural range in which one expects to obtain stable big bang formation both in the case of the Einstein-scalar field equations and in the
case of Einstein's vacuum equations. This range is not exhausted by \cite{rasql,rasq,rsh,specks3,bao}. However, it is exhausted by the breakthrough
result \cite{GIJ}. One might expect it to be necessary to have a detailed understanding of the geometry in order to control the full quiescent
regime. However, the bootstrap assumptions of \cite{GIJ}, remarkably, only entail quite mild assumptions concerning the geometry. In the oscillatory
setting, on the other hand, a detailed control of the geometry can be expected to be necessary. In the context of our framework, we justify this
statement here by identifying the mechanism causing the oscillations; the mechanism can be expressed in terms of the expansion
normalised Weingarten map. Moreover, in the oscillatory setting, solutions are expected to generically be highly anisotropic most of the time. For
the above reasons, it is clear that it is of interest to develop a geometric framework for understanding highly anisotropic solutions.

\subsection{A framework for understanding anisotropic solutions}\label{ssection:Aframeworkforanisotropicsolns}

For reasons mentioned above, we expect it to be important to have a clear picture of the geometry in order to be able to analyse oscillatory
big bang singularities. In order to be more precise concerning how this is to be achieved, we need to introduce some terminology. We do so
in the present subsection, and at the same time give a rough idea of the main assumptions; see \cite[Chapter~3]{RinWave} of the companion article
\cite{RinWave} or Section~\ref{section:Assumptions} below for a formal statement. 

\textit{The expansion normalised Weingarten map.} Let $(M,g)$ be a spacetime (i.e., a time oriented Lorentz manifold) with a foliation $M=\bM\times I$,
where $I=(t_{-},t_{+})$ is an open interval, such that the mean curvature, say $\theta$, of the leaves $\bM_{t}:=\bM\times \{t\}$ is never zero. Here we
tacitly assume the leaves of the foliation to be closed and spacelike, and $\d_{t}$ to be future pointing timelike. In fact, we are mainly interested in
crushing singularities; i.e., the situation that $\theta\rightarrow\infty$ as $t\rightarrow t_{-}$, see Definition~\ref{def:crushingsingularity} below. A
central object in this article is the \textit{expansion normalised Weingarten map} $\mK:=\bK/\theta$, where $\bK$ is the Weingarten map of the leaves of
the foliation. Clearly, $\mK$ is symmetric with respect to the metric $\bge$ induced on $\bM_{t}$. It therefore has real eigenvalues
$\ell_{A}$, $A=1,\dots,n$, which we here assume to be distinct. In this setting, there are, at least locally, eigenvectors corresponding to the $\ell_{A}$,
say $\{X_{A}\}$, which are orthogonal with respect to $\bge$. For the purposes of the present discussion, we assume the $X_{A}$ to be smooth and globally
defined, we normalise them with respect to a fixed reference metric $\bge_\refer$, and we use the notation $\chg(X_{A},X_{A})=e^{2\mu_{A}}$, where
$\chg:=\theta^{2}\bge$. Here $\bge_{\refer}$ is the metric induced by $g$ on $\bM_{t_{0}}$ for some $t_{0}\in I$ (and we think of $\bge_{\refer}$ as being defined on
$\bM$).

\textit{Boundedness and non-degeneracy.} It is a remarkable fact that $\mK$ is bounded with respect to a fixed Riemannian reference metric on $\bM$
for quite a large class of big
bang singularities; see \cite[Appendix~C]{RinWave} for a justification of this statement. For this reason, the starting point of this article is the
assumption that $\mK$ is bounded in suitable weighted $C^{k}$ and Sobolev spaces. As mentioned above, we are interested in the highly anisotropic setting.
For that reason, we here assume that there is an $\e_{\rond}>0$ such that $|\ell_{A}-\ell_{B}|\geq\e_{\rond}$
for $A\neq B$. We refer to this condition as \textit{non-degeneracy}. In the quiescent setting, it is reasonable to make this assumption globally. However,
in an oscillatory setting, this assumption should be thought of as being valid in a restricted region (in both space and time). Nevertheless, the hope is
that the analysis will clarify how one leaves this region (how the eigenvalues $\ell_{A}$ change). From now on, we assume the eigenvalues to be ordered so
that $\ell_{1}<\cdots<\ell_{n}$. 

\textit{Silence.} As mentioned above, we assume the asymptotics to be silent; see \cite{EUW,PastAttr} for the origin of the terminology. The precise
assumption we make here is that the Weingarten map $\chK$ of the conformally rescaled metric $\hg=\theta^{2}g$ is negative definite. In fact we assume the
eigenvalues of $\chK$ to be bounded above by $-\e_{\Spe}$ for an $\e_{\Spe}\in (0,\infty)$. Roughly speaking, the consequence of this assumption is that $\hg$
exhibits exponential expansion in the direction of the singularity with respect to an appropriate time coordinate (namely $\tau$ introduced in
(\ref{eq:taudefinition}) below). This means that we obtain silence in the above sense of the term. In particular, particle horizons form. Justifying these
statements in the degree of generality in which we are interested in here is the subject of \cite[Chapter~7]{RinWave}. However, as an example, in the
case of the Kasner solutions; see (\ref{eq:gKasnergenn}) below; this condition is satisfied for all but the flat Kasner solutions (i.e., the solutions
for which one $p_A$ equals $1$ and all the others equal $0$). In other words, the condition $\chK\leq -\e_{\Spe}$ distinguishes the Kasner solutions that
form particle horizons from those that do not. 

\textit{Logarithmic volume density.} One typical feature of a big bang singularity is that the volume of the leaves of the foliation converges to zero.
For this reason, it is natural to introduce the \textit{volume density} $\varphi$, defined by 
\begin{equation}\label{eq:varphidefitobmubgedp}
\mu_{\bge}=\varphi\mu_{\bge_{\refer}}.
\end{equation}
Here $\mu_{\bge}$ and $\mu_{\bge_{\refer}}$ are the volume forms with respect to $\bge$ and $\bge_{\refer}$ respectively, where we think of $\bge$ and $\bge_{\refer}$
as both being defined on $\bM$. Thus $\varphi(\bx,t_{0})=1$ for all $\bx\in\bM$. It is also convenient to introduce the \textit{logarithmic volume density}:
\begin{equation}\label{eq:varrhodefdp}
\varrho:=\ln\varphi. 
\end{equation}
Since we expect $\varphi$ to converge to zero as $t\rightarrow t_{-}$ (in which case $\varrho\rightarrow -\infty$ as $t\rightarrow t_{-}$), it is
clear that $\varrho$ is a measure of the distance to the singularity, and we use powers of $\ldr{\varrho}:=(\varrho^{2}+1)^{1/2}$ as a weight in
the definition of the weighted $C^{k}$ and Sobolev spaces.

\textit{Mean curvature, lapse and shift.} We also need to impose bounds on the relative spatial variation of the mean curvature, as well as on
the \textit{lapse function} $N$ and \textit{shift vector field} $\chi$, defined by $\d_{t}=NU+\chi$, where $U$ is the future pointing unit normal
of the leaves of the foliation and $\chi$ is tangential to the leaves. We do not impose a specific gauge choice in this article, and therefore we
need to impose conditions on $N$ and $\chi$ directly. In the case of a CMC foliation, the bounds on the relative spatial variation of $\theta$ are
of course automatically satisfied.

\textit{Normal derivatives.} It is not sufficient to only impose conditions on $\mK$, $\ln\theta$, $N$ and $\chi$. We also need to impose conditions
on normal derivatives of these quantities. In the case of $\mK$, we impose conditions on $\hml_{U}\mK$, a quantity defined in
\cite[Appendix~A.2]{RinWave}; see also (\ref{eq:hmlUmtinfixedspatialcoord}) below; and which we refer to as the expansion normalised normal derivative
of $\mK$ (it is essentially an expansion normalised Lie derivative of $\mK$). In the case of $\ln\theta$, we impose conditions on the
\textit{deceleration parameter} $q$, defined by
\begin{equation}\label{eq:hUnlnthetamomqbas}
  \hU(n\ln\theta)=-1-q,
\end{equation}
where $\hU$ is the future pointing unit normal with respect to the conformally rescaled metric $\hg:=\theta^{2}g$.
In case Einstein's equations are satisfied, $\hml_{U}\mK$ and $q$ can be calculated in terms of the stress energy tensor and geometric quantities.
However, we also need to impose conditions on expansion normalised versions of the Lie derivatives of $N$ and $\chi$ with respect to $U$. In practice,
the assumptions concerning the lapse function are imposed on the lapse function $\hN=\theta N$ of the conformally rescaled metric $\hg=\theta^{2}g$.

\textit{Assumptions, rough formulation.} Beyond all the assumptions stated above, we also impose a smallness assumption on the shift vector field and
conditions on some of the components of $\hml_{U}\mK$ with respect to the frame $\{X_{A}\}$. However, the main assumptions are formulated as bounds on
weighted $C^{k}$ and Sobolev norms of $\mK$, $\hml_{U}\mK$, $\ln\theta$, $q$, $\ln\hN$, $\hU(\ln\hN)$, $\chi$ and $\dotchi$. Here $\dotchi$ denotes
the spatial components of the expansion normalised normal Lie derivative of $\chi$. 

\subsection{Results}\label{ssection:results}

Given the assumptions roughly listed in Subsection~\ref{ssection:Aframeworkforanisotropicsolns}, and described in greater detail in
Section~\ref{section:Assumptions} below, we here state some of the conclusions. In \cite{RinWave}, we deduce the asymptotic behaviour of solutions
to linear systems of wave equations on the corresponding backgrounds. However, in the present article, we focus on the geometry. In what follows,
and in contrast to \cite{RinWave}, we assume Einstein's equations
\begin{equation}\label{eq:EE}
  G+\Lambda g=T
\end{equation}
to hold. Here $G=\mathrm{Ric}-Sg/2$, where $\mathrm{Ric}$ and $S$ are the Ricci and scalar curvature of $(M,g)$ respectively; $\Lambda$ is the
cosmological constant (the sign of which is irrelevant in the main results); and $T$ is the stress energy tensor. For most of the results, we
only impose general conditions on the stress energy tensor, but some of the results concern scalar field matter specifically.

\textit{Spatial scalar curvature.} The first observation concerns the spatial scalar curvature $\bS$. Given assumptions as above, including bounds of
weighted $C^{k}$-norms up to a suitable order,  
\begin{equation}\label{eq:renormalisedbSestintro}
  \big|\theta^{-2}\bS+\textstyle{\frac{1}{4}\sum}_{A,C=2}^{n}\sum_{B=1}^{\min\{A,C\}-1}e^{2\mu_{B}-2\mu_{A}-2\mu_{C}}(\g^{B}_{AC})^{2}\big|
  \leq C_{\bS}\ldr{\varrho}^{2(2\cweight+1)}e^{2\e_{\Spe}\varrho}
\end{equation}
for all $t$ in the interval $(t_{-},t_{0}]$; note that if $\mK X_A=\ell_AX_A$ (no summation), then $|X_A|_{\chg}=e^{\mu_A}$, so that, since the $\ell_A$ are
ordered, the relative sizes of $A$, $B$ and $C$ appearing in the sum in (\ref{eq:renormalisedbSestintro}) are important. This estimate follows from
(\ref{eq:renormalisedbSest}) below. Here $\g^{A}_{BC}$ are the structure coefficients associated with the frame $\{X_{A}\}$. In other words,
\begin{equation}\label{eq:structure coeffients}
  [X_{A},X_{B}]=\g^{C}_{AB}X_{C}.
\end{equation}
Moreover, $C_{\bS}$ is a constant and $0\leq\cweight\in\ro$ is a constant appearing in the definition of the weights. Since $\e_{\Spe}>0$ and
$\varrho\rightarrow-\infty$ in the direction of the singularity, the right hand side of (\ref{eq:renormalisedbSestintro}) converges to zero exponentially
in the direction of the singularity. The first, very important, conclusion to draw from (\ref{eq:renormalisedbSestintro}) is that, up to exponentially
decaying terms, $\theta^{-2}\bS$ is negative. Note, however, that (\ref{eq:renormalisedbSestintro}) does not prevent $\theta^{-2}\bS$ from tending to
$-\infty$. On the other hand, a bound on $\theta^{-2}\bS$ arises by combining (\ref{eq:renormalisedbSestintro}) with the Hamiltonian constraint. In the
case of (\ref{eq:EE}), the Hamiltonian constraint reads
\begin{equation}\label{eq:Hamiltonianconstraintbasicversion}
  \bS-\bk^{ij}\bk_{ij}+\theta^{2}=2\rho+2\Lambda,
\end{equation}
where $\rho$ is the energy density, given by $\rho=T(U,U)$, and $\bk$ is the second fundamental form of the
leaves of the foliation. Dividing (\ref{eq:Hamiltonianconstraintbasicversion}) with $\theta^{2}$ and reformulating the terms slightly yields
\begin{equation}\label{eq:reformulatedandrenormalisedHamcon}
1=2\Omega+2\Omega_{\Lambda}+\tr(\mK^{2})-\theta^{-2}\bS,
\end{equation}
where $\Omega:=\theta^{-2}\rho$ and $\Omega_{\Lambda}:=\theta^{-2}\Lambda$. We are mainly interested in matter models such that $\rho\geq 0$, so 
that $\Omega\geq 0$. In the cosmological setting, $\Lambda\geq 0$ is also a natural assumption. When $\rho\geq 0$ and $\Lambda\geq 0$, the fact
that $\theta^{-2}\bS$ is asymptotically non-positive; see (\ref{eq:renormalisedbSestintro}); combined with the equality
(\ref{eq:reformulatedandrenormalisedHamcon}) yields asymptotic bounds on all the quantities on the right hand side of
(\ref{eq:reformulatedandrenormalisedHamcon}). However, irrespective of the sign of $\Lambda$, the assumptions of the previous subsection,
including bounds of weighted $C^{k}$-norms up to a suitable order, yield
\begin{equation}\label{eq:OmegaplussumellAsqbdpnegscintro}
  \left|2\Omega+\textstyle{\sum}_{A}\ell_{A}^{2}+\mSb_{b}-1\right| \leq C_{\rocon}\ldr{\varrho}^{2(2\cweight+1)}e^{2\e_{\Spe}\varrho}
\end{equation}
on $M_{-}$; see Proposition~\ref{prop:effectiveHamconstraint}; where $C_{\rocon}$ is a constant and
\begin{equation}\label{eq:mSbbdef}
  \mSb_{b}:=\textstyle{\frac{1}{4}\sum}_{A,C=2}^{n}\sum_{B=1}^{\min\{A,C\}-1}e^{2\mu_{B}-2\mu_{A}-2\mu_{C}}(\g^{B}_{AC})^{2}.
\end{equation}
Due to (\ref{eq:OmegaplussumellAsqbdpnegscintro}), it is clear that $\textstyle{\sum}_{A}\ell_{A}^{2}\leq 1$ asymptotically, and 
$\textstyle{\sum}_{A}\ell_{A}^{2}\approx 1$ unless $\mSb_{b}$ or $\Omega$ are non-negligible. That the vector $\ell=(\ell_{1},\dots,\ell_{n})$
is contained in the closed unit ball asymptotically is remarkable if we keep in mind that we initially only impose the condition that
$\ell$ be contained in some ball. It is additionally of interest to note that for
$B<A$, $B<C$, 
\begin{equation}\label{eq:emuBmmuAmmuCgaBACchgarel}
  e^{\mu_{B}-\mu_{A}-\mu_{C}}\g^{B}_{AC}=\chga^{B}_{AC},
\end{equation}
where $\chsfX_{A}:=e^{-\mu_{A}}X_{A}$ and $[\chsfX_{A},\chsfX_{C}]=\chga^{B}_{AC}\chsfX_{B}$. In particular,
\begin{equation}\label{eq:mSbbchgaversion}
  \mSb_{b}=\textstyle{\frac{1}{4}\sum}_{A,C=2}^{n}\sum_{B=1}^{\min\{A,C\}-1}(\chga^{B}_{AC})^{2}.
\end{equation}
This formula is remarkable in the sense that the right hand side can be computed by a geometrically quite natural sequence of steps: First, compute
the eigenvector fields $\{X_{A}\}$ of $\mK$. Next, normalise them so that they are orthonormal with respect to the conformally rescaled metric
$\chg:=\theta^{2}\bge$. Finally, compute the associated structure coefficients $\chga^{A}_{BC}$. Then $\mSb_{b}$ is given by (\ref{eq:mSbbchgaversion}). Moreover,
this formula is particularly informative in the case of $3+1$-dimensions; in that case, there are only two non-zero terms in
the sum, namely the terms corresponding to $(B,A,C)=(1,2,3)$ and $(B,A,C)=(1,3,2)$, and since the terms coincide, 
\begin{equation}\label{eq:mSbbthreeplusone}
  \mSb_{b}=\textstyle{\frac{1}{2}}(\chga^{1}_{23})^{2}.
\end{equation}
In what follows, it will become clear that, in $3+1$-dimensions, the function $(\chga^{1}_{23})^{2}$ is of central importance, e.g. when it comes to
distinguishing between quiescent and oscillatory behaviour. It is therefore convenient to give it a name, and we refer to $(\chga^{1}_{23})^{2}$ as the
\textit{bounce function}; due to the above, this function is uniquely determined by the foliation (in regions such that $\mK$ is well defined and
its eigenvalues are distinct). The function $(\g^{1}_{23})^{2}$ is not uniquely defined. However, it is uniquely defined up to a strictly positive
multiplicative function which is independent of $t$ (what the multiplicative function is depends on the choice of reference metric). We refer to
$(\g^{1}_{23})^{2}$ as \textit{a pre-bounce function}. 

It is also of interest to draw conclusions concerning the deceleration parameter $q$ introduced in (\ref{eq:hUnlnthetamomqbas}), as well as concerning
the eigenvalues of $\chK$. In order to be able to do so, we need to make some assumptions concerning the matter. Letting $\bp:=\bge^{ij}T_{ij}/n$ (where
the Latin indices range from $1$ to $n$ and correspond to directions tangential to the leaves of the foliation), we assume $\theta^{-2}(\rho-\bp)$ to
decay exponentially in the direction of the singularity. In vacuum, this assumption is clearly satisfied. However, under suitable circumstances, it is
also satisfied by other matter models such as scalar fields; see Section~\ref{section:scalarfieldmatter} below for a more detailed discussion. Combining
the requirement that $\theta^{-2}(\rho-\bp)$ decay exponentially with the above assumptions, including bounds of weighted $C^{k}$-norms up to a suitable
order,  yields
\begin{equation}\label{eq:qminusnminusoneetc}
  |q-(n-1)+n\mSb_{b}|+|\lambda_{A}-(\ell_{A}-1)-\mSb_{b}|\leq C_{a}\ldr{\varrho}^{2(2\cweight+1)}e^{2\e_{q}\varrho}
\end{equation}
for some $\e_{q}>0$,
where $C_{a}$ is a constant and $\lambda_{A}$ are the eigenvalues of $\chK$ corresponding to the eigenvectors $X_{A}$; see
Proposition~\ref{prop:qandlambdaAestimates} below. In particular, the only obstruction to $q$ equalling $n-1$ is if $\mSb_{b}$ is non-zero. However, this
estimate also makes it clear that $\ell_{A}=1$ is asymptotically not consistent with the assumptions; the requirement of silence corresponds to the
inequality $\lambda_{A}\leq -\e_{\Spe}$ for all $A\in \{1,\dots,n\}$ (note also that if $\ell_{A}=1$, then (\ref{eq:OmegaplussumellAsqbdpnegscintro})
implies that $\mSb_{b}$ is exponentially small). 

\textbf{$3+1$-dimensions, dynamics.} It is of interest to specialise the above observations to the case of $3+1$-dimensions, and to deduce how the
$\ell_{A}$ evolve. Since $\ell_{1}+\ell_{2}+\ell_{3}=1$, we can summarise the information concerning the $\ell_{A}$ by two functions. In fact, it is
convenient to introduce
\begin{subequations}\label{seq:ellpm}
  \begin{align}
    \ell_{+} := & \tfrac{3}{2}\big(\ell_{2}+\ell_{3}-\tfrac{2}{3}\big)=\tfrac{3}{2}\left(\tfrac{1}{3}-\ell_{1}\right),\label{eq:ellplus}\\
    \ell_{-} := & \tfrac{\sqrt{3}}{2}(\ell_{2}-\ell_{3}).\label{eq:ellminus}
  \end{align}
\end{subequations}
Then
\begin{equation}\label{eq:ellpellminuscircconv}
  \ell_{+}^{2}+\ell_{-}^{2}=\tfrac{3}{2}(\ell_{1}^{2}+\ell_{2}^{2}+\ell_{3}^{2})-\tfrac{1}{2}
\end{equation}
and (\ref{eq:OmegaplussumellAsqbdpnegscintro}) reads
\[
\left|\ell_{+}^{2}+\ell_{-}^{2}+3\Omega+\tfrac{3}{4}(\chga^{1}_{23})^{2}-1\right| \leq C_{a}\ldr{\varrho}^{2(2\cweight+1)}e^{2\e_{\Spe}\varrho}.
\]
Next, it is of interest to analyse the evolution of $\ell_{\pm}$. In order to obtain conclusions, we need to make additional assumptions concerning
the matter, in particular concerning
\begin{equation}\label{eq:mfpmcPdef}
  \mfp^{i}_{\phantom{i}j}:=\bge^{il}T_{lj},\ \ \
  \mcP^{i}_{\phantom{i}j}:=\mfp^{i}_{\phantom{i}j}-\bp\de^{i}_{j}.
\end{equation}
With this notation, we assume $\theta^{-2}(\rho-\bp)$ and $\mcP^{A}_{\phantom{A}A}$ (no summation) to decay exponentially. Again, this condition is
satisfied in the case of vacuum, but also, under suitable circumstances, for scalar fields. Then we obtain the conclusion that
\begin{subequations}\label{seq:hUellpmintro}
  \begin{align}
    \big|\hU(\ell_{+})+\tfrac{1}{2}(\chga^{1}_{23})^{2}(\ell_{+}-2)\big| \leq & K_{+}\ldr{\varrho}^{2(2\cweight+1)}e^{2\e_{q}\varrho},\label{eq:hUellplusintro}\\
    \big|\hU(\ell_{-})+\tfrac{1}{2}(\chga^{1}_{23})^{2}\ell_{-}\big| \leq & K_{-}\ldr{\varrho}^{2(2\cweight+1)}e^{2\e_{q}\varrho}\label{eq:hUellminusintro}
  \end{align}
\end{subequations}
on $M_{-}$, where $K_{\pm}$ are constants and $0<\e_{q}\in\ro$. In order to obtain this conclusion, we appeal to (\ref{eq:renormalisedbSestintro}),
(\ref{eq:mSbbthreeplusone}) and (\ref{subeq:hUellpm}). The estimates (\ref{seq:hUellpmintro}) are quite remarkable for several reasons. First, up to
normalisation, the main ingredients of the left hand side are defined solely in terms of $\mK$. Second, the time evolution for $\ell_{\pm}$ implied by
(\ref{seq:hUellpmintro}) is that if the $\ell_\pm$ evolve, they do so according to the Kasner map; see Remark~\ref{remark:BKLmap} below for a more detailed
justification of this statement. Third, in the $3+1$-dimensional vacuum setting, if the $\ell_\pm$ converge along an integral curve of $\hU$
(in the direction of the singularity), then $\g^1_{23}$ has to converge to zero exponentially along this integral curve; see Remark~\ref{remark:g123 to zero}.
Finally, if $(\chga^{1}_{23})^{2}$ vanishes or decays to zero exponentially, then convergence of $\ell_{\pm}$ follows.
The last three observations constitute the justification for calling $(\chga^{1}_{23})^{2}$ the bounce function; evolution by means of the Kasner map is
often referred to as a bounce. The last observation is also
of interest when considering symmetry classes of solutions to Einstein's equations. In fact, the vanishing of $(\chga^{1}_{23})^{2}$
can be a consequence of special symmetry assumptions. In particular, it is then consistent for $\ell_{\pm}$ to converge, even though this is not to be
expected more generally. We refer the interested reader to Subsections~\ref{ssection:NGQuiescentResults} and
\ref{ssection:revisitsphom}--\ref{ssection:revisitnosymm} below for more details.  

\textit{Quiescent singularities.} Considering estimates and formulae such as (\ref{eq:renormalisedbSestintro}), (\ref{eq:emuBmmuAmmuCgaBACchgarel})
and (\ref{seq:hUellpmintro}), it is clear that terms including a factor of the form $e^{2\mu_{A}-2\mu_{B}-2\mu_{C}}$, where $B\neq C$,
are of particular importance. The reason for this is that such terms could, potentially, grow exponentially in the direction of the singularity. Moreover,
they could prevent the convergence of the $\ell_{A}$. Due to the ordering of the eigenvalues of $\mK$, the worst (in the sense that it is the largest in the
direction of the singularity) combination that can occur is $e^{2\mu_{1}-2\mu_{n-1}-2\mu_{n}}$; see (\ref{eq:muAmmuBmmuCottest}) below. However,
$\mu_{1}-\mu_{n-1}-\mu_{n}$ can be expressed in terms of the $\ell_{A}$ and $q$. In fact, if $\g$ is an integral curve of $\hU$ with
$\g(0)\in\bM\times\{t_{0}\}$, then there is a constant $K$ such that
\[
\big|(\mu_{1}-\mu_{n-1}-\mu_{n})\circ\g(s)+\textstyle{\int}_{s}^{0}\left[\ell_{1}-\ell_{n-1}-\ell_{n}+\tfrac{1}{n}(1+q)\right]\circ\g(u)du\big|\leq K
\]
for all $s\leq 0$ in the domain of definition of $\g$ (note that since $\hU$ is future oriented, $s\leq 0$ corresponds to the direction of the
singularity); see Remark~\ref{remark:muAalonggamma} below. 

Consider now the quiescent setting. This means, in particular, that the $\ell_{A}$ converge. For this to be consistent under generic
circumstances, the expression $e^{2\mu_{1}-2\mu_{n-1}-2\mu_{n}}$ has to converge to zero exponentially. If this happens, it can be argued that $\mSb_{b}$ converges
to zero exponentially. Combining this observation with (\ref{eq:qminusnminusoneetc}) leads to the conclusion that $q-(n-1)$ converges to zero exponentially.
The natural (bootstrap) assumption to make is thus that $q-(n-1)$ converges to zero exponentially and that $\ell_{1}-\ell_{n-1}-\ell_{n}+1$ is strictly
positive. To be more precise, assume, therefore, that there are constants $K_{q}$, $C_{\rho}$, $\e_{\rho}$ and $0<\e_{\que},\e_{p}<1$ such that
\begin{equation}\label{eq:qminus nmo ell alg rho m p}
  |q-(n-1)|\leq K_{q}e^{\e_{\que}\varrho},\ \ \
  \ell_{1}-\ell_{n-1}-\ell_{n}+1\geq\e_{p},\ \ \
  \theta^{-2}|\rho-\bp|\leq C_{\rho}e^{2\e_{\rho}}.
\end{equation}
Given these assumptions, it follows that $q-(n-1)$ decays exponentially; see (\ref{eq:qexpconvimprove}) below. Moreover, the rate of decay is determined by
$\e_{\Spe}$, $\e_{\rho}$ and $\e_{p}$. In particular, if $\e_{\que}$ is small enough, the estimate (\ref{eq:qexpconvimprove}) represents an improvement of the
initial assumption concerning $q$.

Combining the above assumptions with bounds on weighted $C^{k}$ and Sobolev norms of $\mK$, $\hml_{U}\mK$, $\ln\theta$, $q$, $\ln\hN$,
$\hU(\ln\hN)$, $\chi$ and $\dotchi$ leads to an improvement of the assumptions. In fact, in the vacuum setting, it can be demonstrated that
$\hml_{U}\mK$ converges to zero exponentially and that $\mK$ converges exponentially, even though we initially only assume these quantities
to be bounded with respect to weighted norms. On the other hand, the conclusion involves a loss of derivatives. Since the precise statements
are somewhat technical, we refer the reader to Section~\ref{section:resultsformal} below for further details. There are similar results if the
matter consists of a scalar field; see Section~\ref{section:resultsformal} below.

It is of interest to ask under what situations the assumption that $\ell_{1}-\ell_{n-1}-\ell_{n}+1>0$ is consistent. In case $n=3$, this condition
reads $\ell_{1}>0$. In other words, it corresponds to the requirement that $\mK$ be positive definite. This requirement is not consistent with
vacuum (since, in vacuum, the sum of the $\ell_{A}$ and the sum of their squares have to equal $1$ in the limit). However, it is consistent if the
matter content consists of a scalar field; see, e.g., \cite{aarendall}. In the case of vacuum, $n$ has to satisfy $n\geq 10$ in order for quiescent
behaviour to be allowed. This observation goes back to \cite{Henneauxetal}, and the consistency is illustrated by \cite{daetal,GIJ}.

\textit{Improvement of the assumptions.} One interesting consequence of the arguments is that, both in the general and in the quiescent setting,
the assumptions imply improvements of some of the assumptions. In that sense, there is reason to hope that the assumptions formulated here can
constitute part of a bootstrap argument. A precise statement of the results requires a precise statement of the assumptions, and we therefore
postpone a more detailed discussion of this topic to Section~\ref{section:resultsformal} below.

\textit{Unified perspective/data on the singularity.} Finally, it is of interest to note that the framework developed here gives a unified perspective on
existing results. This is important for two reasons. First of all, knowledge of existing results are of crucial importance when formulating conjectures.
However, it is not easy, especially for someone new to the field, to get an overview due to the details such as gauge choices, Lie algebra classification,
frame choices etc. To have a general guiding principle, depending only on the foliation and the geometry, is therefore very valuable. The framework
developed here yields such a guiding principle, based on the expansion normalised Weingarten map $\mK$ and the bounce function. Under the assumptions of
his article, these quantities can be used to explain the dynamics in $3+1$-dimensions. However, it is natural to conjecture that they  determine the
dynamics more generally. This leads to the second reason for why it is important to revisit existing results; if the framework is supposed to be a useful
guiding principle when formulating conjectures, it of course needs to reproduce existing results. For this reason, we devote
Subsections~\ref{ssection:revisitsphom}--\ref{ssection:revisitnosymm} and Appendix~\ref{section:Bianchi class B} below to verifying this. 

We wish to emphasise here that formulating conjectures concerning the asymptotics is particularly important when trying to understand the global dynamics
(potentially in a given symmetry class), rather than when proving stability results (since stability results are based on starting close to a specific solution).

Another consequence of the above focus on $\mK$ and the bounce function is that it naturally leads to a notion of initial data on the singularity; see
\cite[Definition~2]{RinQC} for the $3+1$-dimensional vacuum setting and \cite[Definition~10]{RinQC} for the $n+1$-dimensional Einstein-scalar field
setting. It is natural to conjecture that if a solution is convergent, then
it should give rise to data on the singularity in the sense of \cite{RinQC}; see \cite{RinQC} for a conditional justification of this statement. Again, it
is useful to compare this expectation with results such as \cite[Theorem~1.1]{fal}, where the authors introduce a certain type of initial data on the
singularity and prove that there is a corresponding solution. This is an important result. However, the conditions defining the initial data are very
technical. Due to the unified perspective developed here and in \cite{RinQC}, one would, however, expect the conditions in \cite[Theorem~1.1]{fal} to be a
special case of \cite[Definition~2]{RinQC}. This turns out to be the case; see \cite[Propositions~5 and 6]{RinQC}. 

The notion of initial data on the singularity can also be used to provide conjectures concerning the global dynamics. By a straightforward algebraic
calculation, one can verify that there are no non-degenerate vacuum initial data on the singularity in the case of Bianchi types VIII and IX; see
\cite{RinID}. One is then naturally led to the conjecture that Bianchi type VIII and IX vacuum solutions are oscillatory, which turns out to be true.
Similarly, for Bianchi VIII and IX solutions to the Einstein-scalar field equations, there are initial data on the singularity. This leads to the
conjecture that these data on the singularity should describe the asymptotics. This turns out to be true (in spite of the fact that the dynamics in the
case of Bianchi types VIII and IX can be very complicated - starting close enough to vacuum, solutions can shadow a chaotic vacuum solution along any
number of BKL type bounces). However, these are just two simple examples of conjectures that arise from the framework developed here. We hope that this
convinces the reader that the perspective developed here is useful both in succinctly summarising known results and in formulating new conjectures.

\subsection{Previous work, outline}

Several frameworks for understanding big bang singularities have been developed over the years. To begin with, there is the very influential work of
Belinski\v{\i}, Khalatnikov and Lifschitz (BKL); see, e.g., \cite{bkl1} and \cite{bkl2}. This work has been refined since then, and alternative approaches
have been developed; see, e.g., \cite{dhan,dah,huar,HUL}. When comparing with the present article, the scale invariant perspective developed in \cite{huar}
is of particular importance; it is closest to the framework developed here. The above results are due to physicists, and the conclusions are far reaching.
However, the assumptions rarely take the form of estimates, but rather involve heuristic arguments justifying why certain terms in the equations can be
neglected. A more mathematical approach is taken in \cite{Lott1,Lott2}. In these articles, the author makes a priori assumptions (in the
form of, e.g., scale invariant estimates for the curvature) in the direction of the singularity, and draws conclusions concerning the asymptotics. The
results are very interesting. However, it is not so clear to us how the arguments would fit into a derivation of asymptotics without prior assumptions.
The purpose of the present article and its predecessor is to develop methods that will hopefully allow us to construct solutions with oscillatory
singularities. In the end, we wish to prove the existence of initial data leading to semi-global existence and detailed asymptotics. In order to be able
to do so, we need to understand the dynamics (in particular, to identify the main mechanisms) and to deduce how solutions to systems of wave equations
behave (in particular, with regards to asymptotics and regularity). We achieve some of these goals in this
article and its predecessor. However, there is one piece missing: our arguments yield a clear picture of the geometry and to some extent of the dynamics,
but they are associated with a loss of derivatives. In order to complete the argument, it is necessary to make an appropriate gauge choice and to use
energy methods to derive rough estimates without a loss of derivatives. Combining such arguments with our framework will hopefully result in the
desired conclusions.

The outline of the article is the following. We begin, in Section~\ref{section:Assumptions}, by formulating the assumptions more precisely. This is
followed, in Section~\ref{section:resultsformal}, by a description of the results. In particular, we discuss conditions that yield quiescent asymptotics
(both generically and in the presence of symmetries). In the context of the Einstein-scalar field equations, we  also derive asymptotics for the scalar
field in the quiescent setting. This discussion is followed by a description of partial improvements of the bootstrap assumptions; see
Subsection~\ref{ssection:partialbootstraparg}. In Subsections~\ref{ssection:revisitsphom}--\ref{ssection:revisitnosymm}, we demonstrate that many
previous results fit into the framework developed here. See also Appendix~\ref{section:Bianchi class B}. In Section~\ref{section:basiccurvatureconcl},
we estimate the expansion normalised spatial scalar curvature. This estimate is of central importance and has several consequences,
especially when combined with the Hamiltonian constraint. We record some of these consequences. In Section~\ref{section:dynamics}, we illustrate how
the Kasner map arises in our setting. We also derive improvements of some of the assumptions. The remainder of the article is devoted to the
quiescent setting. As is clear from (\ref{eq:renormalisedbSestintro}), expressions such as $2\mu_{B}-2\mu_{A}-2\mu_{C}$ play an important role in
the analysis. In order to make progress, we need to have some understanding for how this quantity evolves. In particular, the most favourable
situation is that this quantity tends to $-\infty$ linearly. We begin by recording such conditions in Subsections~\ref{ssection:asbehmuAetc}
and \ref{ssection:quiescentregimes}. We then prove that the expansion normalised normal derivative of $\mK$ decays exponentially and that $\mK$
converges exponentially in the case of Einstein's vacuum equations, improving the original assumptions (but with a loss of derivatives). In
Section~\ref{section:quiescent}, we derive conclusions of this nature, given $C^{k}$-assumptions, and assuming Einstein's vacuum equations to
hold. In Section~\ref{section:Sobestquiescentsetting}, we derive Sobolev estimates based on Sobolev assumptions. Finally, in
Section~\ref{section:scalarfieldmatter}, we consider quiescent behaviour in the context of the Einstein-scalar field equations. We again prove
that the expansion normalised normal derivative of $\mK$ decays exponentially and that $\mK$ converges exponentially. We also derive asymptotics
for the scalar field. In the appendices, we record material from \cite{RinWave} we need here, derive formulae we use in the main text, and prove
technical statements we need but whose proofs are not very illuminating. 

\subsection{Auxiliary material}

We use many of the conclusions of \cite{RinWave} in this article. Since the results in \cite{RinWave} we need are spread over more than 200 pages, it
is useful to collect the material we use in one place. We here also need several expressions for curvature components with respect to special frames. Deriving
the corresponding expressions is straightforward but lengthy. In other words, this article depends on extensive material which either exists in the literature
or can be derived in a straightforward manner. Including this material in a published article is inappropriate. Not including it at all is also
inappropriate. For this reason, we have chosen to include this material as appendices to the present arXiv version of the article.

\subsection{Notation}

For the benefit of the reader, we here summarise the notation used in the article:

\textit{Notation from Definition~\ref{def:basicnotions}}: The manifolds $M$ and $\bM$; the interval $I$; $t_0$, $t_-$ and $t_+$; the \textit{metrics} $g$,
$\bge$, $\hg$, $\chg$, $\bge_\refer$; the \textit{second fundamental forms/Weingarten maps} associated with the Lorentz
metrics (including traces and rescalings), $\bk$, $\theta$, $\chk$, $\chth$, $\bK$, $\mK$, $\chK$; the \textit{unit normals} $U$ and $\hU$; the
\textit{lapse functions} $N$ and $\hN$; the \textit{shift vector field} $\chi$; the \textit{Levi-Civita connection} $\bD$ associated with $\bge_\refer$;
and the \textit{volume and logarithmic volume densisties} $\varphi$ and $\varrho$ are introduced in Definition~\ref{def:basicnotions}.
The \textit{deceleration parameter} $q$ is defined by (\ref{eq:hUnlnthetamomqbas}).

Next, $\{E_i\}$ is a global orthonormal frame of $(\bM,\bge_\refer)$ with dual frame $\{\omega^i\}$; see Remark~\ref{remark:globalframe}.
The expansion normalised normal derivative of $\mK$ is given by $\hml_{U}\mK:=\theta^{-1}\ml_{U}\mK$; see (\ref{eq:hmlUmtinfixedspatialcoord}) and
\cite[Appendix~A.2]{RinWave}. Moreover, $\dotchi$ denotes the spatial part of the expansion normalised normal Lie
derivative of $\chi$. The \textit{time coordinate} $\tau$ is defined by (\ref{eq:taudefinition}). The interval $I_-$ is introduced in (\ref{eq:Iminusdef}).
The manifold $M_-$ is defined by $M_-:=\bM\times I_-$. Moreover, $\bM_t:=\bM\times\{t\}$.

$G$, $\mathrm{Ric}$, $T$, $S$, $\bS$ and $\Lambda$ denote the \textit{Einstein tensor} of $g$, the \textit{Ricci tensor of} $g$, the
\textit{stress energy tensor}, the \textit{scalar curvature} of $g$, the
scalar curvature of $\bge$ and the \textit{cosmological constant} respectively. The \textit{energy density} is denoted by $\rho$; $\Omega:=\rho/\theta^2$;
and $\Omega_\Lambda:=\Lambda/\theta^2$. The average pressure is denoted by $\bp:=\bge^{ij}T_{ij}/n$. Next, $\mSb_b$ is defined by (\ref{eq:mSbbdef}).
The \textit{eigenvalues} $\ell_A$ of $\mK$ and the frames $\{X_A\}$ and $\{Y^A\}$ are introduced in Definition~\ref{def:XAellA}; $\mu_A$ and $\bmu_A$ are
introduced in Remark~\ref{remark:globalframe}. The eigenvalues of $\chK$ are denoted $\lambda_A$. The \textit{structure coefficients} associated with
$\{X_A\}$ are denoted $\g_{BC}^A$; see (\ref{eq:structure coeffients}).
We also use the notation $a_A:=\g_{AB}^{B}/2$. Next, $\chsfX_{A}:=e^{-\mu_{A}}X_{A}$ and $\chga^{B}_{AC}$ denotes the associated structure coefficients.
The functions $\ell_\pm$ are introduced in (\ref{seq:ellpm}); $\mfp^{i}_{\phantom{i}j}$ and $\mcP^{i}_{\phantom{i}j}$ are introduced in (\ref{eq:mfpmcPdef}).
The family $\bmN$ of $(1,1)$-tensor fields is introduced in (\ref{eq:bmNAA def}). When we write $\bmN^{A}_{\phantom{A}B}$, this means the components of
$\bmN$ with respect to the frame $\{X_A\}$ and co-frame $\{Y^A\}$. Similarly, $\bge^{AB}$ are the components of the inverse of $\bge$ with respect
to the frame $\{X_A\}$ and co-frame $\{Y^A\}$.

The \textit{parameters} $\e_\Spe$ and $\e_\rond$ are introduced in Definition~\ref{def:silenceandnondegeneracy}. The parameters $C_{\mK,\mrod}$,
$G_{\mK,\mrod}$, $M_{\mK,\mrod}$ and $\e_\mK$ are introduced in
Definition~\ref{def:offdiagonalexpdec}. The parameters $\cweight$, $\weight_0$, $K_\cweight$, $\bDlnhNsup$, $c_\robas$ and $\mKsup$ are introduced
in Definition~\ref{def:basicassumptions}. The parameters $\theta_{0,\pm}$ are introduced in (\ref{eq:thetazdef}). Moreover, $\bfl$, $\bfl_0$,
$\bfl_1$, $\weight$ and $s_{\cweight,l}$ are introduced in Definition~\ref{def:sobklassumptions}; and $c_{\cweight,l}$ is introduced in
Definition~\ref{def:supmfulassumptions}. The constant $K_{\rovar}$ is introduced in Lemma~\ref{lemma:smallnessshiftconsequences} below. 

The \textit{weight space} $\Weight$ and \textit{index space} $\Index$ are introduced in (\ref{eq:Weight Index def}). For $\xi\in\rn{n}$, the Japanese
bracket is defined by $\ldr{\xi}:=(1+|\xi|^2)^{1/2}$. If one of the \textit{parameters} $K_{q}$, $C_{\rho}$, $\e_{\rho}$ and $0<\e_{\que},\e_{p}<1$ appear, it
means that we are making the corresponding assumption in (\ref{eq:qminus nmo ell alg rho m p}). The parameters $\e_{R}$, $\vare_{R}$ and $\vare_{\Spe}$ are
introduced in Theorem~\ref{thm:SobestimatesQuiescentVacuum}. The parameters $\kappa_1$, $\mfv_1$ and $\vare_\phi$ are introduced in
Theorem~\ref{thm:SobEstimatesQuiescentMatter}. The parameter $d_q$ is introduced in Theorem~\ref{thm:SobEstimatesQuiescentMatterNG}.
The parameters $C_\rho$ and $C_{\rho,0}$ are introduced in Proposition~\ref{prop:qandlambdaAestimates}. 
The parameters $C_\mP$, $C_{\mP,0}$ and $\e_q$ are introduced in Proposition~\ref{prop:hUellA}. 

The \textit{norms} $\|\cdot\|_{C^\bfl_\weight(\bM)}$, $\|\cdot\|_{H^\bfl_\weight(\bM)}$, $\|\cdot\|_{H^{\bfl,\weight}_\rohy(\bM)}$ and $\|\cdot\|_{C^{\bfl,\weight}_\rohy(\bM)}$
are introduced in (\ref{eq:mtClbS}), (\ref{eq:mtHlbS}), (\ref{eq:Hlrohydef}) and (\ref{eq:Clrohydef}) and respectively.

\subsection*{Acknowledgments}

This research was funded by the Swedish Research Council (Vetenskapsr\aa det), dnr. 2017-03863 and 2022-03053. It was also supported by the Swedish
Research Council under grant no. 2016-06596 while the author was in residence at Institut Mittag-Leffler in Djursholm, Sweden during the fall of 2019.
Finally, the author would like to thank the referees for their careful reading of the article and their numerous suggestions for improvements.

\section{Assumptions}\label{section:Assumptions}

The framework of this article is the same as that of \cite{RinWave}. We here therefore briefly recall the terminology and assumptions of
\cite{RinWave}, starting with the notion of crushing singularities as introduced in \cite[Definition~2.1]{RinWave}:

\begin{definition}\label{def:crushingsingularity}
  A spacetime $(M,g)$ is said to have a \textit{crushing singularity} if the following conditions are satisfied. First, $(M,g)$ can be foliated by spacelike
  Cauchy hypersurfaces in the sense that $M=\bM\times I$, where $\bM$ is an $n$-dimensional manifold, $I=(t_{-},t_{+})$ is an interval, the metric $\bge$
  induced on the leaves $\bM_{t}:=\bM\times\{t\}$ of the foliation is Riemannian, and $\bM_{t}$ is a Cauchy hypersurface in $(M,g)$ for all $t\in I$
  (i.e., $\bM_t$ is intersected exactly once by every inextendible timelike curve). Second,
  the mean curvature, say $\theta$, of the leaves of the foliation tends to infinity as $t\rightarrow t_{-}+$.
\end{definition}
\begin{remark}
  A spacetime is a time oriented Lorentz manifold. And given a foliation as in the statement of the definition, $\d_{t}$ is always assumed to be future
  pointing. 
\end{remark}
More generally, we use the following terminology and consider foliations of the following type (\cite[Definition~3.1]{RinWave}): 
\begin{definition}\label{def:basicnotions}
  Let $(M,g)$ be a spacetime. A \textit{partial pointed foliation} of $(M,g)$ is a triple $\bM$, $I$ and $t_{0}\in I$, where $\bM$
  is a closed $n$-dimensional manifold; $I$ is an interval with left end point $t_{-}$ and right end point $t_{+}$; and there is an open interval $J$
  containing $I$ and a diffeomorphism from $\bM\times J$ to an open subset of $M$. Moreover, the hypersurfaces $\bM_{t}:=\bM\times \{t\}$ are required
  to be spacelike Cauchy hypersurfaces (in $(\bM\times J,g)$) and $\d_{t}$ is required to be timelike with respect to $g$ (where $\d_{t}$ represents
  differentiation with respect to the variable on $I$). Given a partial pointed foliation,
  the \textit{associated induced metric, second fundamental form, mean curvature and future pointing unit normal} are denoted $\bge$, $\bk$, $\theta$ and
  $U$ respectively;
  the \textit{associated Weingarten map} $\bK$ is the family of $(1,1)$ tensor fields on $\bM$ obtained by raising one of the indices of $\bk$ with
  $\bge$;
  the \textit{associated reference metric} is the metric induced on $\bM_{t_{0}}$ by $g$ (it is denoted by $\bge_{\refer}$ with associated Levi-Civita
  connection $\bD$);
  and the \textit{volume density $\varphi$ and logarithmic volume density $\varrho$ associated with the pointed foliation} are defined by
  (\ref{eq:varphidefitobmubgedp}) and (\ref{eq:varrhodefdp}) respectively.

  An \textit{expanding partial pointed foliation} is a partial pointed foliation such that the mean curvature $\theta$ of the leaves of the foliation is
  always strictly positive. Given an expanding partial pointed foliation,
  the \textit{associated expansion normalised Weingarten map} $\mK$ is the family of $(1,1)$ tensor fields on $\bM$ given by $\mK:=\bK/\theta$;
  the \textit{associated conformal metric} is $\hg:=\theta^{2}g$;
  the \textit{associated induced conformal metric, second fundamental form, mean curvature and future pointing unit normal} are denoted $\chg$,
  $\chk$, $\chth$ and $\hU$ respectively, and they are the objects induced on the hypersurfaces $\bM_{t}$ by the conformal metric $\hg$;
  and the \textit{associated conformal Weingarten map} $\chK$ is the family of $(1,1)$ tensor fields on $\bM$ obtained by raising one of the indices of
  $\chk$ with $\chg$.
\end{definition}
\begin{remark}
  As is clear from the definition, the reference metric $\bge_{\refer}$ is determined by specifying the time $t_0\in I$. Moreover, due to this choice,
  the logarithmic volume density $\varrho$ introduced in (\ref{eq:varrhodefdp}) has the property that $\varrho(\bx,t_0)=0$ for all $\bx\in\bM$. This is
  a convenient property and it means that one can think of $t_0$ as a reference time in the definition of the logarithmic volume density. 
\end{remark}
\begin{remark}
  Note that for each $t\in I$, there is a corresponding $\bge$, $\bk$, $\bK$, $\mK$ etc. induced on $\bM_t$. In this sense, it is natural to think
  of $\bge$ as a family of Riemannian metrics on $\bM$ for $t\in I$ etc.
\end{remark}
\begin{remark}\label{remark:chitoregvar}
  Here $\chg=\theta^{2}\bge$, $\mK=\chK+\hU(\ln\theta)\mathrm{Id}$ and $\chth=-q$, where $q$ is defined by (\ref{eq:hUnlnthetamomqbas}).
\end{remark}

In what follows, we also speak of the \textit{deceleration parameter} $q$, defined by (\ref{eq:hUnlnthetamomqbas}); note that $q$ converging to
$n-1$ is naturally associated with the quiescent setting; see (\ref{eq:ninvoneplusqfirstestimate}) below; and that $q-(n-1)$ is the perhaps most
important term when it comes to determining the asymptotic behaviour of solutions to the wave equation $\Box_g\phi=0$; see (\ref{eq:waveequforphi})
below. Moreover, the \textit{lapse function}
$N$ and \textit{shift vector field} $\chi$ are defined by $\d_{t}=NU+\chi$. The lapse function associated with the conformally rescaled metric $\hg$ is
denoted $\hN=\theta N$. 

Non-degeneracy and silence are two important notions in what follows; see \cite[Definition~3.10]{RinWave}:
\begin{definition}\label{def:silenceandnondegeneracy}
  Let $(M,g)$ be a spacetime. Assume that it has an expanding partial pointed foliation. If there is a constant $\e_{\Spe}>0$
  such that
  \begin{equation}\label{eq:chKeSpe}
    \chK\leq -\e_{\Spe}\Id
  \end{equation}
  (i.e., if $\chK$ is negative definite) on $\bM\times I$, then $\chK$ is said to have a \textit{silent upper bound} on $I$. In what follows, $\e_{\Spe}$
  is assumed to satisfy $\e_{\Spe}\leq 2$. If the eigenvalues of $\mK$ are distinct and there is an $\e_{\rond}>0$ such that the distance between different
  eigenvalues is bounded from below by $\e_{\rond}$ on $I$, then $\mK$ is said to be \textit{non-degenerate} on $I$. 
\end{definition}
It is of interest to note that if $\lambda_A$ are the eigenvalues of $\chK$, then, due to \cite[(7.32)]{RinWave}, 
\begin{equation}\label{eq:lambdaArelqellA}
\lambda_{A}=\ell_{A}+\hU(\ln\theta)=\ell_{A}-(q+1)/n.
\end{equation}
Combining (\ref{eq:lambdaArelqellA}) with (\ref{eq:chKeSpe}) yields the conclusion that $q\geq n\e_{\Spe}$, so that $\hU(n\ln\theta)\leq -1-n\e_{\Spe}$ due to 
(\ref{eq:hUnlnthetamomqbas}). The latter estimate can be used to deduce that the singularity is crushing (if $\varrho$ diverges uniformly to
$-\infty$); see (\ref{eq:lnthetalowbd}), the justification of which can be found in the proof of \cite[Corollary~7.7]{RinWave}.

Next, due to the non-degeneracy, it is possible to define a frame; see \cite[Definition~3.13]{RinWave}:
\begin{definition}\label{def:XAellA}
  Let $(M,g)$ be a spacetime. Assume it to have an expanding partial pointed foliation and $\mK$ to be non-degenerate on
  $I$. By assumption, the eigenvalues, say $\ell_{1}<\cdots<\ell_{n}$, of $\mK$ are distinct. Locally, there is, for each $A\in \{1,\dots,n\}$
  an eigenvector $X_{A}$ of $\mK$ corresponding to $\ell_{A}$ such that
  \begin{equation}\label{eq:XAbgenormcond}
    |X_{A}|_{\bge_{\refer}}=1.
  \end{equation}
  If there is a global smooth frame with this property, say $\{X_{A}\}$, then $\mK$ is said to have a \textit{global frame} and $\{Y^{A}\}$ denotes the
  frame dual to $\{X_{A}\}$.  
\end{definition}
\begin{remark}
  Given the assumptions we make in this article, $|Y^{A}|_{\bge_{\refer}}$ is bounded; see Lemma~\ref{lemma:frameinvest} and
  Definition~\ref{def:basicassumptions} below.
\end{remark}
\begin{remark}\label{remark:globalframe}
  Due to Lemma~\ref{lemma:Lemma A.1}, we can (under the assumptions of Definition~\ref{def:XAellA}), by taking a finite covering space of $\bM$, if
  necessary, assume $\mK$ to have a global frame. For this reason, we assume $\mK$ to have a global frame in what follows. In this setting, we can
  introduce $\bmu_{A}$ and $\mu_{A}=\bmu_{A}+\ln\theta$ by demanding that (note that $\{X_{A}\}$ is orthogonal with respect to $\bge$ and $\chg$)
  \begin{equation}\label{eq:muAbmuAdefinition}
    |X_A|_{\bge}=e^{\bmu_A},\ \ \ |X_A|_{\chg}=e^{\mu_A}.
  \end{equation}
  Moreover, it is convenient to introduce a frame $\{E_{i}\}$, which is orthonormal with respect to $\bge_{\refer}$, with dual frame $\{\omega^{i}\}$.
\end{remark}

In what follows, we need to control the expansion normalised normal derivative of $\mK$: $\hml_{U}\mK:=\theta^{-1}\ml_{U}\mK$, 
given by (\ref{eq:hmlUmtinfixedspatialcoord}); see also \cite[Appendix~A.2]{RinWave}. Einstein's equations can be used to calculate $\hml_{U}\mK$, and
we will use this fact to estimate
$\hml_{U}\mK$ in some situations. However, in order to even get started, we need to make some assumptions concerning this quantity; see
\cite[Definition~3.19]{RinWave}:
\begin{definition}\label{def:offdiagonalexpdec}
  Let $(M,g)$ be a spacetime. Assume it to have an expanding partial pointed foliation and $\mK$ to be non-degenerate on
  $I$ and to have a global frame. Then $\hml_{U}\mK$ is said to satisfy an \textit{off-diagonal exponential bound} if there are constants
  $C_{\mK,\mrod}\geq 0$, $G_{\mK,\mrod}\geq 0$, $M_{\mK,\mrod}\geq 0$ and $0<\e_{\mK}\leq 2$ such that
  \begin{equation}\label{eq:hmlUhmlUsqmKoffdiagonalexpbd}
    |(\hml_{U}\mK)(Y^{A},X_{B})|\leq C_{\mK,\mrod}e^{\e_{\mK}\varrho}+G_{\mK,\mrod}e^{-\e_{\mK}\varrho}
  \end{equation}
  on $\bM\times I$ for all $A\neq B$, where
  \begin{equation}\label{eq:expgrowthwithupperbound}
    G_{\mK,\mrod}e^{-\e_{\mK}\varrho}\leq M_{\mK,\mrod}
  \end{equation}
  on $\bM\times I$. If there are constants $C_{\mK,\mrod}\geq 0$, $G_{\mK,\mrod}\geq 0$, $M_{\mK,\mrod}\geq 0$ and $0<\e_{\mK}\leq 2$ such that
  (\ref{eq:hmlUhmlUsqmKoffdiagonalexpbd}) and (\ref{eq:expgrowthwithupperbound}) hold on $\bM\times I$ for all $A,B$ such that $A\neq B$ and $B>1$,
  then $\hml_{U}\mK$ is said to satisfy a \textit{weak off-diagonal exponential bound}. 
\end{definition}
\begin{remark}
  The requirement that $\hml_{U}\mK$ satisfy a weak off-diagonal exponential bound is perhaps the most unnatural assumption of the framework. However,
  it is consistent with the presence of oscillations; see the remarks following \cite[Definition~3.19]{RinWave}. More specifically, this condition
  is satisfied by Bianchi type VIII and IX solutions, for which the dynamics are asymptotically well approximated by the BKL map (on the other hand,
  $(\hml_{U}\mK)(Y^{A},X_{A})$ does not converge to zero in this setting). Moreover, under
  suitable circumstances, the framework leads to an improvement of this assumption; see, e.g., Corollary~\ref{cor:improvedoffdiagexpbd} below and the
  results in the quiescent setting. In such a situation, we can think of the weak off-diagonal exponential bound as part of a bootstrap argument. 
\end{remark}

Next, we introduce the norms with respect to which we express the main assumptions. Let, to begin with, 
\begin{equation}\label{eq:Weight Index def}
  \Weight:=\{(\weight_{a},\weight_{b})\in\rn{2}:\weight_{a}\geq 0, \weight_{b}\geq 0\},\ \ \
  \Index:=\{(l_{0},l_{1})\in\zn{2}:0\leq l_{0}\leq l_{1}\}.
\end{equation}
Then, if $\weight\in\Weight$, $(l_{0},l_{1})=\bfl\in\Index$ and $\mt$ is a family of tensor fields (of arbitrary type) on $\bM$ for $t\in I$, 
\begin{align}
\|\mt(\cdot,t)\|_{C^{\bfl}_{\weight}(\bM)} := & 
\textstyle{\sup}_{\bx\in\bM}\big(\textstyle{\sum}_{j=l_{0}}^{l_{1}}\ldr{\varrho(\bx,t)}^{-2\weight_{a}-2j\weight_{b}}
|\bD^{j}\mt(\bx,t)|_{\bge_{\refer}}^{2}\big)^{1/2},\label{eq:mtClbS}\\
\|\mt(\cdot,t)\|_{H^{\bfl}_{\weight}(\bM)} := & 
\big(\textstyle{\int}_{\bM}
\textstyle{\sum}_{j=l_{0}}^{l_{1}}\ldr{\varrho(\cdot,t)}^{-2\weight_{a}-2j\weight_{b}}|\bD^{j}\mt(\cdot,t)|_{\bge_{\refer}}^{2}\mu_{\bge_{\refer}}\big)^{1/2}.
\label{eq:mtHlbS}
\end{align}
Here $\ldr{\xi}:=(1+|\xi|^{2})^{1/2}$ and $\bD$ is the Levi-Civita connection associated with $\bge_\refer$. In case $\weight=0$, we write $C^{\bfl}(\bM)$ and
$H^{\bfl}(\bM)$ for the spaces and correspondingly for the norms. When
$\bfl=(0,l)$, we replace $\bfl$ with $l$. Given a vector field $\xi$ on $M$ which is tangential to the leaves of the foliation; $\weight\in\Weight$; and
$(l_{0},l_{1})=\bfl\in\Index$, we also use the notation
\begin{align}
  |\bD^{k}\xi|_{\rohy} := & N^{-1}\big(\bge_{\refer}^{i_{1}j_{1}}\cdots \bge_{\refer}^{i_{k}j_{k}}\bge_{lm}\bD_{i_{1}}\cdots\bD_{i_{k}}\xi^{l}
  \bD_{j_{1}}\cdots\bD_{j_{k}}\xi^{m}\big)^{1/2},\label{eq:bDkchirohydef}\\
  \|\xi(\cdot,t)\|_{H^{\bfl,\weight}_{\rohy}(\bM)} := & \big(\textstyle{\int}_{\bM}\textstyle{\sum}_{k=l_{0}}^{l_{1}}
\ldr{\varrho(\cdot,t)}^{-2\weight_{a}-2k\weight_{b}}|\bD^{k}\xi(\cdot,t)|_{\rohy}^{2}\mu_{\bge_{\refer}}\big)^{1/2},
\label{eq:Hlrohydef}\\
\|\xi(\cdot,t)\|_{C^{\bfl,\weight}_{\rohy}(\bM)} := & \textstyle{\sup}_{\bx\in\bM}\textstyle{\sum}_{k=l_{0}}^{l_{1}}
\ldr{\varrho(\cdot,t)}^{-\weight_{a}-k\weight_{b}}|\bD^{k}\xi(\bx,t)|_{\rohy}. \label{eq:Clrohydef}
\end{align}
Here, again, $\bD$ is the Levi-Civita connection associated with $\bge_\refer$. Note also that there is one $\bge_{lm}$ factor appearing in (\ref{eq:bDkchirohydef}),
whereas the remaining metric factors are of the form $\bge_{\refer}^{ij}$. The reason for this is that we wish to measure the length of $\chi$ with respect to the
actual metric, but the length associated with the derivatives with respect to the background metric.

\subsection{Basic assumptions}

Given the above terminology, we are in a position to formulate the basic assumptions; see \cite[Definition~3.27]{RinWave}:
\begin{definition}\label{def:basicassumptions}
  Let $(M,g)$ be a spacetime. Assume it to have an expanding partial pointed foliation, $\mK$ to be non-degenerate on $I$, $\mK$ to have a global
  frame and $\chK$ to have a silent upper bound on $I$; see Definition~\ref{def:silenceandnondegeneracy}. Assume, moreover, $\hml_{U}\mK$ to satisfy a
  weak off-diagonal exponential bound; see Definition~\ref{def:offdiagonalexpdec}. Next, let $\weight_{0}=(0,\cweight)\in\Weight$ and assume
  that there are constants $K_{\cweight}$ and $\bDlnhNsup$ such that
  \begin{equation}\label{eq:mKbDlnNchicombest}
    \|\mK(\cdot,t)\|_{C^{1}_{\weight_{0}}(\bM)}\leq K_{\cweight},\ \ \
    \|(\bD \ln \hN)(\cdot,t)\|_{C^{0}(\bM)}\leq \bDlnhNsup,\ \ \
    \|\chi(\cdot,t)\|_{C^{0}_{\rohy}(\bM)}\leq \tfrac{1}{2}
  \end{equation}
  for all $t\in I_{-}$, where
  \begin{equation}\label{eq:Iminusdef}
    I_{-}:=\{t\in I:t\leq t_{0}\}.
  \end{equation}
  Then the \textit{basic assumptions} are said to be fulfilled. The associated constants are denoted by
  \[
  c_{\robas}:=(n,\e_{\Spe},\e_{\mK},\e_{\rond},C_{\mK,\mrod},M_{\mK,\mrod},\cweight,K_{\cweight},\bDlnhNsup).
  \]
\end{definition}
Note, in particular, that if the basic assumptions; see Definition~\ref{def:basicassumptions}; are satisfied, there is a constant $\mKsup$ such that
\begin{equation}\label{eq:mKsupbasest}
  \|\mK(\cdot,t)\|_{C^{0}(\bM)}\leq \mKsup
\end{equation}
holds for all $t\in I_{-}$.

In practice, we also need to impose a smallness assumption on $\chi$. The particular form of this assumption is that the conditions of
Lemma~\ref{lemma:smallnessshiftconsequences} are satisfied. This leads to the following terminology.
\begin{definition}\label{def:standardassumptions}
  Assume that the conditions of Definition~\ref{def:basicassumptions} and of Lemma~\ref{lemma:smallnessshiftconsequences} are satisfied. Then the
  \textit{standard assumptions} are said to be satisfied. 
\end{definition}
\begin{remark}
  Due to these assumptions, it is meaningful to introduce a new time coordinate
  \begin{equation}\label{eq:taudefinition}
    \tau(t):=\varrho(\bx_{0},t),
  \end{equation}
  where $\bx_{0}\in\bM$ is a fixed reference point. Some of the basic properties of this time coordinate are given by
  Lemmas~\ref{lemma:smallnessshiftconsequences} and \ref{lemma:epsilonlowdefEi}. In particular, there is a constant $c_0>1$ such that
  \begin{equation}\label{eq:ldrrho ldr tau equiv}
    c_0^{-1}\ldr{\varrho(\bx,t)}\leq \ldr{\tau(t)}\leq c_0\ldr{\varrho(\bx,t)}
  \end{equation}
  for all $\bx\in\bM$ and all $t\in I_-$, where $\varrho$ is the logarithmic volume densisty, our main measure of distance to the singularity.
  Note also that $\varrho(\bx,t_0)=0$ by the definition of $\bge_\refer$ and $\varrho$, so that, in particular, $\tau(t_0)=0$.
  Since the singularity corresponds to $\varrho\rightarrow-\infty$, we also have $\tau(t)\rightarrow-\infty$ in the direction of the singularity.
  Using $\tau$ instead of $\varrho$ is convenient in many contexts, for instance when appealing to Gagliardo-Nirenberg estimates for weighted
  spaces and when quantifying the growth/decay of a norm.
\end{remark}

\subsection{Higher order Sobolev assumptions}\label{subsection:higherordersobolevassumptions}

In the formulation of the Sobolev assumptions (see \cite[Definition~3.28]{RinWave}), it is convenient to use the following notation
\begin{equation}\label{eq:thetazdef}
  \theta_{0,-}:=\textstyle{\inf}_{\bx\in \bM}\theta(\bx,t_{0}),\ \ \
  \theta_{0,+}:=\textstyle{\sup}_{\bx\in \bM}\theta(\bx,t_{0}). 
\end{equation}

\begin{definition}\label{def:sobklassumptions}
  Given that the basic assumptions; see Definition~\ref{def:basicassumptions}; are satisfied, let $1\leq l\in\zo$, $\bfl_{0}:=(1,1)$,
  $\bfl:=(1,l)$ and $\bfl_{1}:=(1,l+1)$. Let $\cweight$ and $\weight_{0}$ be defined as in the statement of Definition~\ref{def:basicassumptions}.
  Let, moreover, $\weight:=(\cweight,\cweight)$. Then the $(\cweight,l)$-\textit{Sobolev assumptions} are said to be satisfied if there are constants
  $S_{\rorel,\chi,l}$, $S_{\mK,\theta,l}$, $C_{\rorel,\chi,1}$ and $C_{\mK,\theta,1}$ such that
  \begin{align*}
    \|\ln\hN\|_{H^{\bfl_{1}}_{\weight_{0}}(\bM)}+\|\hU(\ln\hN)\|_{H^{l}_{\weight}(\bM)}
    +\theta_{0,-}^{-1}\|\chi\|_{H^{l+2,\weight_{0}}_{\rohy}(\bM)}+\theta_{0,-}^{-1}\|\dotchi\|_{H^{l,\weight}_{\rohy}(\bM)}\leq & S_{\rorel,\chi,l},\\
    \|\mK\|_{H^{l+1}_{\weight_{0}}(\bM)}+\|\hml_{U}\mK\|_{H^{l+1}_{\weight}(\bM)}+\|\ln\theta\|_{H^{\bfl_{1}}_{\weight_{0}}(\bM)}+\|q\|_{H^{l}_{\weight_{0}}(\bM)} \leq & S_{\mK,\theta,l}
  \end{align*}
  for all $t\in I_{-}$, where $I_{-}$ is defined by (\ref{eq:Iminusdef}), and
  \begin{align*}
    \|\ln\hN\|_{C^{\bfl_{0}}_{\weight_{0}}(\bM)}+\|\hU(\ln\hN)\|_{C^{0}_{\weight}(\bM)}
    +\theta_{0,-}^{-1}\|\chi\|_{C^{2,\weight_{0}}_{\rohy}(\bM)}+\theta_{0,-}^{-1}\|\dotchi\|_{C^{1,\weight}_{\rohy}(\bM)}\leq & C_{\rorel,\chi,1},\\
    \|\mK\|_{C^{1}_{\weight_{0}}(\bM)}+\|\hml_{U}\mK\|_{C^{0}_{\weight}(\bM)} +\|\ln\theta\|_{C^{\bfl_{0}}_{\weight_{0}}(\bM)}+\|q\|_{C^{0}_{\weight_{0}}(\bM)}\leq & C_{\mK,\theta,1}
  \end{align*}
  for all $t\in I_{-}$. Given that the $(\cweight,l)$-Sobolev assumptions hold, let
  \[
  s_{\cweight,l}:=(c_{\robas},l,S_{\rorel,\chi,l},S_{\mK,\theta,l},C_{\rorel,\chi,1},C_{\mK,\theta,1}).
  \]
\end{definition}

\subsection{Higher order $C^{k}$-assumptions}\label{subsection:higherorderckassumptions}

Next, we introduce the $C^{k}$-terminology analogous to Definition~\ref{def:sobklassumptions}; see \cite[Definition~3.31]{RinWave}: 

\begin{definition}\label{def:supmfulassumptions}
  Given that the basic assumptions; see Definition~\ref{def:basicassumptions}; are satisfied, let $0\leq l\in\zo$ and $\bfl_{1}:=(1,l+1)$. Let $\cweight$
  and $\weight_{0}$ be defined as in the statement of Definition~\ref{def:basicassumptions}.
  Let, moreover, $\weight:=(\cweight,\cweight)$. Then the $(\cweight,l)$-\textit{supremum assumptions} are said to be satisfied if there are constants
  $C_{\rorel,\chi,l}$ and $C_{\mK,\theta,l}$ such that
  \begin{align*}
    \|\ln\hN\|_{C^{\bfl_{1}}_{\weight_{0}}(\bM)}+\|\hU(\ln\hN)\|_{C^{l}_{\weight}(\bM)}
    +\theta_{0,-}^{-1}\|\chi\|_{C^{l+2,\weight_{0}}_{\rohy}(\bM)}+\theta_{0,-}^{-1}\|\dotchi\|_{C^{l,\weight}_{\rohy}(\bM)}\leq & C_{\rorel,\chi,l},\\
    \|\mK\|_{C^{l+1}_{\weight_{0}}(\bM)}+\|\hml_{U}\mK\|_{C^{l+1}_{\weight}(\bM)}+\|\ln\theta\|_{C^{\bfl_{1}}_{\weight_{0}}(\bM)}+\|q\|_{C^{l}_{\weight_{0}}(\bM)} \leq & C_{\mK,\theta,l}
  \end{align*}
  for all $t\in I_{-}$. Given that the $(\cweight,l)$-supremum assumptions hold, let
  \[
  c_{\cweight,l}:=(c_{\robas},l,C_{\rorel,\chi,l},C_{\mK,\theta,l}).
  \]
\end{definition}

\subsection{Summary, discussion and examples}\label{ssection:summary discussion ex}

The basic assumptions; see Definition~\ref{def:basicassumptions}; can roughly speaking be summarised as saying that the spacetime is \textit{expanding}
(the mean curvature is strictly positive);
the geometry is \textit{anisotropic} ($\mK$ has a global frame); the causal structure is \textit{silent} (that silence, as introduced in
Definition~\ref{def:silenceandnondegeneracy}, has the consequences for the causal structure discussed in the introduction, is clarified in detail in
\cite[Subsection~15.1]{RinWave}); and the geometry is \textit{bounded} (in the sense that the various estimates listed in
Definition~\ref{def:basicassumptions} are satisfied). In order for the standard assumptions; see Definition~\ref{def:standardassumptions};
to hold, we additionally need to impose a smallness
condition on the shift vector field. Since the size of the shift vector field can in part be controlled by the gauge choice, ensuring
this condition is partly a matter of choice. The Sobolev and $C^{k}$-assumptions in Definitions~\ref{def:sobklassumptions} and
\ref{def:supmfulassumptions} simply consitute more refined boundedness assumptions.

In order to deduce results that only depend on assumptions concerning initial data (and are not conditional on assumptions concerning the spacetime),
it is expected to be necessary to fix the gauge. The choice of gauge has a significant impact on the lapse function and
the shift vector field; e.g., in a Gaussian foliation, $N=1$ and $\chi=0$, so that the conditions on $\chi$ are automatically
satisfied. Another gauge which has been used frequently in recent stability proofs is constant mean curvature (CMC) and vanishing shift. In this
setting, one expects $N$ to converge to $1$. By these choices, some of the conditions on $N$, $\theta$ and $\chi$ are automatically satisfied, and
some are expected to be satisfied asymptotically. Moreover, in the CMC setting, $N$ satisfies a linear elliptic equation, given the remaining
fields. This means that bootstrap assumptions concerning the remaining fields typically yield conclusions concerning $N$. To summarise: in the context
of a bootstrap argument, the conditions on $N$ and $\chi$ will depend strongly on the choice of gauge. For this reason we will not discuss the
improvement of bootstrap assumptions for these quantities.

In the context of a bootstrap argument, it would be natural to include the conditions listed in Definitions~\ref{def:standardassumptions},
\ref{def:sobklassumptions} and \ref{def:supmfulassumptions} in the bootstrap assumptions. However, in the general quiescent setting, we also impose
the bounds
\begin{equation}\label{eq:q n m o ell bds}
  |q-(n-1)|\leq K_{q}e^{\e_{\que}\varrho},\ \ \
  \ell_{1}-\ell_{n-1}-\ell_{n}+1\geq\e_{p}.
\end{equation}
That the solution satisfies these bounds means that it is in a quiescent regime. In order to justify this statement in the $3+1$-dimensional context,
recall that quiescence means
that the $\ell_{\pm}$ converge. Combining this observation with (\ref{seq:hUellpmintro}), it is clear that we want $\chga^{1}_{23}$ to be small,
preferably exponentially decaying, so that it is integrable in the direction of the singularity. Combining this observation with
(\ref{eq:mSbbthreeplusone}) and (\ref{eq:qminusnminusoneetc}), it is natural to make the assumption corresponding to the first estimate in
(\ref{eq:q n m o ell bds}). The second estimate in (\ref{eq:q n m o ell bds}) is then there to ensure the consistency of the decay of 
$\chga^{1}_{23}$. Turning to corresponding improvements, we refer the reader to Subsection~\ref{ssection:partialbootstraparg} below. Note also
that even though we do not make any bootstrap assumptions concerning the scalar field, we obtain detailed conclusions. Generically, you would
expect $\chga^1_{23}$ to be non-vanishing. However, if you impose the condition that it vanish, then quiescence occurs without imposing
(\ref{eq:q n m o ell bds}).

In this article, we do not make a specific choice of gauge, since we expect the discussion to be
relevant for large classes of gauges. In \cite[Chapter~C]{RinWave} we justify this expectation by verifying
that the assumptions we make here are satisfied for solutions with, e.g.,
CMC foliations, Gaussian foliations and foliations which are neither CMC nor Gaussian (the areal foliation in the $\tn{3}$-Gowdy symmetric setting). 
Finally, it is of interest to address the pieces that are missing in order to complete the picture. First, it is necessary to
make a gauge choice. Second, it is likely to be necessary to make separate bootstrap assumptions concerning
a low number of derivatives and concerning a high number of derivatives. Concerning the low number of derivatives, it is likely necessary to make
detailed assumptions, similar to the ones made in the present article. However, the assumptions concerning the high number of derivatives can be
expected to be much rougher. The idea would then be to obtain conclusions concerning lower derivatives from the higher derivatives by successively
integrating the equations for a high number of derivatives. This is a line of reasoning that appears in many contexts. One example is provided by
the proof of \cite[Theorem~16.1, p.~174]{RinWave}.

\section{Results}\label{section:resultsformal}

We already formulated some of the conclusions of the present article in Subsection~\ref{ssection:results}. However, we here provide more precise statements
of the results in the quiescent setting. Moreover, we discuss the extent to which the assumptions, combined with Einstein's equations, imply improvements
of the assumptions. Finally, we demonstrate that many previous results can be interpreted in terms of our framework. We begin with a discussion of the
generic quiescent setting.

\subsection{Generic quiescent setting}\label{ssection:QuiescentResults}

Quiescent behaviour can occur in several different settings. It can be due to the spacetime dimension being high enough, it can be due to the presence
of a particular matter model, and it can be due to certain symmetry assumptions being fulfilled. In the present subsection, we consider quiescent
solutions to Einstein's vacuum equations in higher dimensions and to the Einstein-scalar field equations. The relevant statement in the case of
Einstein's vacuum equations is the following.

\begin{thm}\label{thm:SobestimatesQuiescentVacuum}
  Let $0\leq\cweight\in\ro$, $1\leq l\in\zo$, $n\geq 10$ and assume that the standard assumptions; see Definition~\ref{def:standardassumptions}; the
  $(\cweight,l)$-Sobolev assumptions; and the $(\cweight,1)$-supremum assumptions are fulfilled. In particular, the spacetime is $n+1$-dimensional.
  Assume, moreover, that $\tau(t)\rightarrow -\infty$ as $t\rightarrow t_{-}$ and that there are constants $K_{q}$, $0<\e_{\que}<1$ and $0<\e_{p}<1$ such
  that
  \begin{equation}\label{eq:qexpconvergencequiescentregimeintro}
    |q-(n-1)|\leq K_{q}e^{\e_{\que}\varrho},\ \ \
    \ell_{1}-\ell_{n-1}-\ell_{n}+1\geq\e_{p}
  \end{equation}
  hold on $M_{-}:=\bM\times I_{-}$. Assume, finally, that Einstein's vacuum equations with a cosmological constant $\Lambda$ are satisfied; i.e.,
  (\ref{eq:EE}) holds with $T=0$. Then there is a constant $K_{\mK,l}$ such that
  \begin{align}
    \|\hml_{U}\mK(\cdot,\tau)\|_{H^{l-1}_{\weight_{1}}(\bM)}+\|[q-(n-1)](\cdot,\tau)\|_{H^{l-1}_{\weight_{1}}(\bM)}
    \leq & K_{\mK,l}\ldr{\tau}^{(l+1)(\cweight+1)}e^{2\vare_{R}\tau}\label{eq:SobestofhmlUmKqugeointro}
  \end{align}
  for all $\tau\leq 0$, where $\weight_{1}:=(2\cweight,\cweight)$, $K_{\mK,l}=K_{\mK,l,0}\theta_{0,-}^{-2}$ and $K_{\mK,l,0}$ only depends on $s_{\cweight,l}$,
  $c_{\cweight,1}$, $\e_{\que}$, $K_{q}$, $\Lambda$ and $(\bM,\bge_{\refer})$. Here $\e_{R}:=\min\{\e_{p},\e_{\Spe}\}$ and $\vare_{R}:=\e_{R}/(3K_{\rovar})$, where
  $K_{\rovar}$ is the constant appearing in Lemma~\ref{lemma:smallnessshiftconsequences}.

  Finally, there is a continuous $(1,1)$-tensor field $\mK_{\infty}$ and continuous functions $\ell_{A,\infty}$ on
  $\bM$, which also belong to $H^{l-1}(\bM)$, and constants $K_{\infty,i,l}$, $i=0,1$, such that
  \begin{equation}\label{eq:mKconvtomKinfHlintro}
    \begin{split}
      & \|\mK(\cdot,t)-\mK_{\infty}\|_{H^{l-1}(\bM)}+\textstyle{\sum}_{A}\|\ell_{A}(\cdot,t)-\ell_{A,\infty}\|_{H^{l-1}(\bM)}\\
      \leq & K_{\infty,0,l}\theta_{0,-}^{-2}\ldr{\tau}^{(l+1)(2\cweight+1)}e^{2\vare_{R}\tau}+K_{\infty,1,l}\ldr{\tau}^{l\cweight}e^{\vare_{\Spe}\tau}
    \end{split}
  \end{equation}
  for all $\tau\leq 0$, where $K_{\infty,0,l}$ only depends on $s_{\cweight,l}$, $c_{\cweight,1}$, $\e_{p}$, $\e_{\que}$, $K_{q}$, $\Lambda$ and
  $(\bM,\bge_{\refer})$; and $K_{\infty,1,l}$ only depends on $s_{\cweight,l}$, $c_{\cweight,1}$ and $(\bM,\bge_{\refer})$. Here $\vare_{\Spe}:=\e_{\Spe}/(3K_{\rovar})$.
\end{thm}
\begin{remark}
  Note that if $\e_{\que}<2\vare_{R}$ and $l>n/2+1$, then (\ref{eq:SobestofhmlUmKqugeointro}) represents an improvement of the first estimate in
  (\ref{eq:qexpconvergencequiescentregimeintro}). 
\end{remark}
\begin{remark}\label{remark:pointwise vacuum Kasner}
  Since $\Omega=0$ and $\mSb_{b}$ converges to zero in this setting (see the proof of Theorem~\ref{thm:SobhmlUmKestimates} below), we conclude, by
  appealing to (\ref{eq:OmegaplussumellAsqbdpnegscintro}), that
  \begin{equation}\label{eq:ellAinf cond}
  \ell_{1,\infty}<\dots<\ell_{n,\infty},\ \ \
  \textstyle{\sum}_{A}\ell_{A,\infty}=1,\ \ \
  \textstyle{\sum}_{A}\ell_{A,\infty}^{2}=1,\ \ \
  \ell_{1,\infty}-\ell_{n-1,\infty}-\ell_{n,\infty}+1\geq\e_{p}
  \end{equation}
  (in order to obtain the last estimate, we use the fact that the second estimate in (\ref{eq:qexpconvergencequiescentregimeintro}) is satisfied).
  Only for $n\geq 10$ is there a solution satisfying these relations with $\e_{p}>0$; see \cite{Henneauxetal} or Lemma~\ref{lemma:vacuumquiescent} below.
  If, on the other hand, $p_{A}$ are constants satisfying (\ref{eq:ellAinf cond}) with $l_{A,\infty}=p_A$,
  \begin{equation}\label{eq:gKasnergenn}
    g=-dt\otimes dt+\textstyle{\sum}_{A}t^{2p_{A}}dx^{A}\otimes dx^{A}
  \end{equation}
  is a solution on $\tn{n}\times (0,\infty)$ (or $\rn{n}\times (0,\infty)$) to Einstein's vacuum equation, satisfying $\ell_{A}=p_{A}$ (solutions with
  metrics of this form will in what follows be referred to as \textit{Kasner solutions}). It is reasonable to expect this solution to be stable in the
  direction of the singularity. In fact, this is demonstrated in \cite{GIJ}. Results demonstrating that one can specify data on the singularity in this
  setting have been obtained in \cite{daetal}.
\end{remark}
\begin{remark}\label{remark:AsInfoID}
  Assuming that $\chi=0$ and combining (\ref{eq:ellAinf cond}), \cite[Theorem~49, p.~16]{RinQC} and the assumptions and conclusions of the
  theorem, it follows that the solution induces initial data on the singularity; see \cite[Subsection~1.7]{RinQC} for details. There is a natural
  connection between this conclusion and the idea, contained in the BKL proposal, that the spatial derivatives should become negligible in the
  direction of the singularity; i.e., that actual solutions should be well approximated by solutions to the \textit{velocity term dominated system} obtained
  by dropping the spatial derivatives. In the context of the BKL proposal, this idea is somewhat vague. However, it is made precise in \cite{isanmo},
  a paper in which it is also demonstrated that polarized $\tn{3}$-Gowdy symmetric solutions asymptote to solutions to the velocity term dominated system.
  Moreover, in \cite{aarendall,daetal}, the authors use the ideas of \cite{isanmo} to give a precise definition of the velocity term dominated system in the
  context of Gaussian foliations; it is obtained by dropping most
  (\textit{but not all}) spatial derivatives. Moreover, given a real analytic solution to the velocity term dominated
  system, say $x$, there is a unique real analytic solution to the actual equations asymptoting to $x$; see \cite{aarendall,daetal}. On the other
  hand, due to the results in \cite[Subsection~1.5]{RinQC}, initial data on the singularity constitute information equivalent to a solution to the
  velocity term dominated system. In this sense, initial data on the singularity can be considered to be a generalised version (they are meaningful in a
  more general context than Gaussian foliations) of the notion of a solution to the velocity term dominated system. Convergence to initial data on the
  singularity can thus be thought of as a precise, general and geometric reformulation of the idea, contained in the BKL proposal, that spatial
  derivatives become negligible asymptotically. 
\end{remark}
\begin{proof}
  The statement is an immediate consequence of Theorems~\ref{thm:SobhmlUmKestimates} and \ref{thm:SobconvofmK} below, keeping in mind that
  (\ref{eq:DeltavarrhorelvariationEiintro}) holds. 
\end{proof}

The next result concerns the Einstein-scalar field setting. 

\begin{thm}\label{thm:SobEstimatesQuiescentMatter}
  Let $0\leq\cweight\in\ro$, $n\geq 3$, $\kappa_{1}$ be the smallest integer strictly larger than $n/2+1$, $\kappa_{1}\leq l\in\zo$ and assume the
  standard assumptions; see Definition~\ref{def:standardassumptions}; the $(\cweight,l)$-Sobolev assumptions; and the $(\cweight,\kappa_{1})$-supremum
  assumptions to be fulfilled. In particular, the spacetime is $n+1$-dimensional. Assume, moreover, that $\tau(t)\rightarrow -\infty$ as
  $t\rightarrow t_{-}$ and that there are constants $K_{q}$, $0<\e_{\que}<1$ and $0<\e_{p}<1$ such that 
  \begin{equation}\label{eq:qexpconvergencequiescentregimeintroMatter}
    |q-(n-1)|\leq K_{q}e^{\e_{\que}\varrho},\ \ \
    \ell_{1}-\ell_{n-1}-\ell_{n}+1\geq\e_{p}
  \end{equation}
  on $M_{-}$. Assume, finally, that $(M,g,\phi)$ satisfy the Einstein-scalar field equations with a cosmological constant $\Lambda$; i.e.
  (\ref{eq:EE}) and $\Box_g\phi=0$ hold, where $T$ is given by (\ref{eq:setnlsf}). Then there is a continuous
  $(1,1)$-tensor field $\mK_{\infty}$ and continuous functions $\Psi_{\infty}$ and $\ell_{A,\infty}$ on $\bM$, which also belong to $H^{l-1}(\bM)$, and
  constants $a_{l,n}$, $b_{l,n}$, $\sfK_{\mK,l}$, $\sfK_{\infty,l}$ and $\sfK_{\phi,1,l}$ such that
  \begin{align}
    \|[q-(n-1)](\cdot,\tau)\|_{H^{l-1}_{\weight_{1}}(\bM)}+\|\hml_{U}\mK(\cdot,\tau)\|_{H^{l-1}_{\weight_{1}}(\bM)}
    \leq & \sfK_{\mK,l}\ldr{\tau}^{a_{l,n}\cweight+b_{l,n}}e^{2\vare_{R}\tau},\label{eq:SobestofhmlUmKqugeomatterintro}\\    
    \|\mK(\cdot,\tau)-\mK_{\infty}\|_{H^{l-1}(\bM)}+\textstyle{\sum}_{A}\|\ell_{A}(\cdot,\tau)-\ell_{A,\infty}\|_{H^{l-1}(\bM)}
    \leq & \sfK_{\infty,l}\ldr{\tau}^{a_{l,n}\cweight+b_{l,n}}e^{\vare_{\phi}\tau},\label{eq:mKconvtomKinfHlmatterintro}\\
    \|[\hU(\phi)](\cdot,\tau)-\Psi_{\infty}\|_{H^{l-1}(\bM)}\leq & \sfK_{\phi,1,l}\ldr{\tau}^{a_{l,n}\cweight+b_{l,n}}e^{\vare_{\phi}\tau}\label{eq:asofhUphiintro}
  \end{align}
  for all $\tau\leq 0$, where $\weight_{1}:=(2\cweight,\cweight)$ and $\sfK_{\mK,l}$, $\sfK_{\infty,l}$ and $\sfK_{\phi,1,l}$ only depend on $s_{\cweight,l}$,
  $c_{\cweight,\kappa_{1}}$, $\e_{p}$, $\e_{\que}$, $K_{q}$, $\Lambda$, $\hGe_{l}(0)$, $(\bM,\bge_{\refer})$ and a lower bound on $\theta_{0,-}$. Moreover,
  $\vare_{\phi}:=\min\{2\vare_{R},\vare_{\Spe}\}$ and $\vare_{R}$ and $\vare_{\Spe}$ are defined as in the statement of
  Theorem~\ref{thm:SobestimatesQuiescentVacuum}. Finally, $a_{l,n}$ and $b_{l,n}$ only depend on $n$ and $l$. If, in addition to the above,
  $l\geq \kappa_{1}+1$, then there is a function $\Phi_{\infty}\in C^{0}(\bM)\cap H^{l-2}(\bM)$ and constants $a_{l,n}$, $b_{l,n}$ and $\sfK_{\phi,0,l}$ such that
  \begin{equation}\label{eq:asofhphiintro}
    \|\phi(\cdot,\tau)-\Psi_{\infty}\varrho(\cdot,\tau)-\Phi_{\infty}\|_{H^{l-2}(\bM)}\leq \sfK_{\phi,0,l}\ldr{\tau}^{a_{l,n}\cweight+b_{l,n}}e^{\vare_{\phi}\tau}
  \end{equation}  
  for all $\tau\leq 0$, where $a_{l,n}$, $b_{l,n}$ and $\sfK_{\phi,0,l}$ have the same dependence as $a_{l,n}$, $b_{l,n}$ and $\sfK_{\phi,1,l}$ appearing in
  (\ref{eq:asofhUphiintro}).
\end{thm}
\begin{remark}\label{remark:hGelzdefintro}
  The expression $\hGe_{l}(0)$ denotes the energy of $\phi$ at $\tau=0$; see (\ref{eq:mekdef}) and (\ref{eq:hGeldef}) below for a precise definition.
\end{remark}
\begin{remark}\label{remark:pointwise non-vacuum Kasner}
  Due to Remark~\ref{remark:asHamConSF} below,
  \begin{equation}\label{eq:asHamConSFintro}
    \Psi_{\infty}^{2}+\textstyle{\sum}_{A}\ell_{A,\infty}^{2}=1.
  \end{equation}
  Note, moreover, that if $p_{A}$, $A=1,\dots,n$, is an ordered set of constants such that $p_{1}-p_{n-1}-p_{n}+1>0$, $p_{1}+\dots+p_{n}=1$ and
  $p_{1}^{2}+\dots+p_{n}^{2}<1$, then
  \begin{equation}\label{eq:Kasnermetricandphithreed}
    g=-dt\otimes dt+\textstyle{\sum}_{A=1}^{n}t^{2p_{A}}dx^{i}\otimes dx^{i},\ \ \
    \phi=p_{\phi}\ln t+\phi_{0}
  \end{equation}  
  is a solution to the Einstein-scalar field equations with $\ell_{A}=p_{A}$, where $\phi_{0}$ is any constant and
  \[
  p_{\phi}=\left(1-\textstyle{\sum}_{A}p_{A}^{2}\right)^{1/2}.
  \]
  That this solution is stable in the direction of the singularity is demonstrated in \cite{GIJ}. It is also possible to specify data on the
  singularity in the setting discussed in the theorem; see \cite{aarendall,daetal}. 
\end{remark}
\begin{remark}
  Remark~\ref{remark:AsInfoID} is equally relevant in the present setting, keeping (\ref{eq:asHamConSFintro}) and
  (\ref{eq:hUsqphiHlminusonestimate}) in mind.
\end{remark}
\begin{proof}
  The statement is an immediate consequence of Theorems~\ref{thm:SobconvofmKdecofhmlUmKmatter} and \ref{thm:Sobasymptoticsscalarfieldmatter},
  keeping in mind that (\ref{eq:DeltavarrhorelvariationEiintro}) holds. 
\end{proof}

\subsection{Non-generic quiescent setting}\label{ssection:NGQuiescentResults}

In the previous subsection, we record results that are not dependent on any symmetry assumptions. However, in order to obtain them, we have
to impose conditions such as $\ell_{1}-\ell_{n}-\ell_{n-1}+1>0$ (and the dimension and matter model have to be consistent with such a condition being
satisfied). Considering the arguments in greater detail, it is clear that an important obstruction to quiescence occurring more generally is the
non-vanishing of $\g^{A}_{BC}$ for $A<\min\{B,C\}$. On the other hand, in some situations, these structure coefficients could vanish due to symmetry
assumptions. Alternately, they could converge to zero (corresponding to a, presumably positive co-dimension, submanifold of the state space). This
leads to quiescence in a setting which can be expected to be non-generic. In particular, we have the following results.

\begin{thm}\label{thm:SobestimatesQuiescentVacuumNG}
  Let $0\leq\cweight\in\ro$, $1\leq l\in\zo$, $n\geq 3$ and assume that the standard assumptions; see Definition~\ref{def:standardassumptions}; the
  $(\cweight,l)$-Sobolev assumptions; and the $(\cweight,1)$-supremum assumptions are fulfilled. In particular, the spacetime is $n+1$-dimensional.
  Assume, moreover, that $\g^{A}_{BC}=0$ for $A<\min\{B,C\}$. Assume, finally, that Einstein's vacuum equations with a cosmological constant $\Lambda$
  are satisfied; i.e., (\ref{eq:EE}) holds with $T=0$; and that $\tau(t)\rightarrow -\infty$ as $t\rightarrow t_{-}$. Then there is a constant $K_{\mK,l}$
  such that
  \begin{align}
    \|[q-(n-1)](\cdot,\tau)\|_{H^{l-1}_{\weight_{1}}(\bM)}+\|\hml_{U}\mK(\cdot,\tau)\|_{H^{l-1}_{\weight_{1}}(\bM)}
    \leq & K_{\mK,l}\ldr{\tau}^{(l+1)(\cweight+1)}e^{2\vare_{\Spe}\tau}\label{eq:SobestofhmlUmKqugeointroNG}
  \end{align}
  for all $\tau\leq 0$, where $\weight_{1}:=(2\cweight,\cweight)$, $K_{\mK,l}=K_{\mK,l,0}\theta_{0,-}^{-2}$ and $K_{\mK,l,0}$ only depends on $s_{\cweight,l}$,
  $c_{\cweight,1}$, $\Lambda$ and $(\bM,\bge_{\refer})$. Here $\vare_{\Spe}:=\e_{\Spe}/(3K_{\rovar})$, where $K_{\rovar}$ is introduced in the statement of
  Lemma~\ref{lemma:smallnessshiftconsequences}. 

  Finally, there is a continuous $(1,1)$-tensor field $\mK_{\infty}$ and continuous functions $\ell_{A,\infty}$ on
  $\bM$, which also belong to $H^{l-1}(\bM)$, and constants $K_{\infty,i,l}$, $i=0,1$, such that
  \begin{equation}\label{eq:mKconvtomKinfHlintroNG}
    \begin{split}
      & \|\mK(\cdot,\tau)-\mK_{\infty}\|_{H^{l-1}(\bM)}+\textstyle{\sum}_{A}\|\ell_{A}(\cdot,\tau)-\ell_{A,\infty}\|_{H^{l-1}(\bM)}\\
      \leq & K_{\infty,0,l}\theta_{0,-}^{-2}\ldr{\tau}^{(l+1)(2\cweight+1)}e^{2\vare_{\Spe}\tau}+K_{\infty,1,l}\ldr{\tau}^{l\cweight}e^{\vare_{\Spe}\tau}
    \end{split}
  \end{equation}
  for all $\tau\leq 0$, where $K_{\infty,0,l}$ only depends on $s_{\cweight,l}$, $c_{\cweight,1}$, $\Lambda$ and $(\bM,\bge_{\refer})$; and $K_{\infty,1,l}$ only
  depends on $s_{\cweight,l}$, $c_{\cweight,1}$ and $(\bM,\bge_{\refer})$.
\end{thm}
\begin{remark}\label{remark:eigenvaluesQuiescentVacuumNG}
  In this setting, $\Omega$ and $\mSb_{b}$ vanish. Combining these observations with (\ref{eq:OmegaplussumellAsqbdpnegscintro}) yields
  \[
  \ell_{1,\infty}<\dots<\ell_{n,\infty},\ \ \
  \textstyle{\sum}_{A}\ell_{A,\infty}=1,\ \ \
  \textstyle{\sum}_{A}\ell_{A,\infty}^{2}=1.
  \]
  On the other hand, fixing constants $p_{A}$ satisfying these relations yields a vacuum solution to Einstein's equations of the form
  (\ref{eq:gKasnergenn}). We expect this solution to be stable in symmetry classes such that $\g^{A}_{BC}=0$ for $A<\min\{B,C\}$.
\end{remark}
\begin{remark}
  In the case of $3+1$-dimensions, it is sufficient to assume that $\g^{1}_{23}=0$. More optimistically, one could hope that it would be sufficient
  to assume that $\g^{1}_{23}$ converges to zero. 
\end{remark}
\begin{remark}
  Remark~\ref{remark:AsInfoID} is equally relevant in the present setting.
\end{remark}
\begin{proof}
  The statement is an immediate consequence of Remarks~\ref{remark:SobestofhmlUmKsymmetry} and \ref{remark:SobconvofmKsymmetry},
  keeping in mind that (\ref{eq:DeltavarrhorelvariationEiintro}) holds. 
\end{proof}

The next result concerns the Einstein-scalar field equations.

\begin{thm}\label{thm:SobEstimatesQuiescentMatterNG}
  Let $0\leq\cweight\in\ro$, $n\geq 3$, $\kappa_{1}$ be the smallest integer strictly larger than $n/2+1$, $\kappa_{1}\leq l\in\zo$ and assume the
  standard assumptions; see Definition~\ref{def:standardassumptions}; the $(\cweight,l)$-Sobolev assumptions; and the $(\cweight,\kappa_{1})$-supremum
  assumptions to be fulfilled. In particular, the spacetime is $n+1$-dimensional. Assume, moreover, that $\g^{A}_{BC}=0$ for $A<\min\{B,C\}$ and that
  there is a constant $d_{q}$ such that
  \begin{equation}\label{eq:qminusnminusoneestintro}
    \|\ldr{\varrho(\cdot,t)}^{3/2}[q(\cdot,t)-(n-1)]\|_{C^{0}(\bM)}\leq d_{q}
  \end{equation}
  for all $t\leq t_{0}$. Assume, finally, that the Einstein-scalar field equations with a cosmological constant $\Lambda$ are satisfied;
  i.e. (\ref{eq:EE}) and $\Box_g\phi=0$ hold, where $T$ is given by (\ref{eq:setnlsf}); and that
  $\tau(t)\rightarrow -\infty$ as $t\rightarrow t_{-}$. Then there is a continuous $(1,1)$-tensor field $\mK_{\infty}$ and continuous functions
  $\Psi_{\infty}$ and $\ell_{A,\infty}$ on $\bM$, which also belong to $H^{l-1}(\bM)$, and constants $a_{l,n}$, $b_{l,n}$, $\sfK_{\mK,l}$, $\sfK_{\infty,l}$ and
  $\sfK_{\phi,1,l}$ such that
  \begin{align}
    \|[q-(n-1)](\cdot,\tau)\|_{H^{l-1}_{\weight_{1}}(\bM)}+\|\hml_{U}\mK(\cdot,\tau)\|_{H^{l-1}_{\weight_{1}}(\bM)}
    \leq & \sfK_{\mK,l}\ldr{\tau}^{a_{l,n}\cweight+b_{l,n}}e^{2\vare_{\Spe}\tau},\label{eq:SobestofhmlUmKqugeomatterintroNG}\\    
    \|\mK(\cdot,\tau)-\mK_{\infty}\|_{H^{l-1}(\bM)}+\textstyle{\sum}_{A}\|\ell_{A}(\cdot,\tau)-\ell_{A,\infty}\|_{H^{l-1}(\bM)}
    \leq & \sfK_{\infty,l}\ldr{\tau}^{a_{l,n}\cweight+b_{l,n}}e^{\vare_{\Spe}\tau},\label{eq:mKconvtomKinfHlmatterintroNG}\\
    \|[\hU(\phi)](\cdot,\tau)-\Psi_{\infty}\|_{H^{l-1}(\bM)}\leq \sfK_{\phi,1,l}\ldr{\tau}^{a_{l,n}\cweight+b_{l,n}}e^{\vare_{\Spe}\tau}\label{eq:asofhUphiintroNG}
  \end{align}
  for all $\tau\leq 0$, where $\sfK_{\mK,l}$, $\sfK_{\infty,l}$ and $\sfK_{\phi,1,l}$ only depend on $s_{\cweight,l}$, $c_{\cweight,\kappa_{1}}$, $d_{q}$, $\Lambda$,
  $\hGe_{l}(0)$, $(\bM,\bge_{\refer})$ and a lower bound on $\theta_{0,-}$. Finally, $a_{l,n}$ and $b_{l,n}$ only depend on $n$ and $l$. If, in addition to
  the above, $l\geq \kappa_{1}+1$, then there is a function $\Phi_{\infty}\in C^{0}(\bM)\cap H^{l-2}(\bM)$ and constants $a_{l,n}$, $b_{l,n}$ and
  $\sfK_{\phi,0,l}$ such that
  \begin{equation}\label{eq:asofhphiintroNG}
    \|\phi(\cdot,\tau)-\Psi_{\infty}\varrho(\cdot,\tau)-\Phi_{\infty}\|_{H^{l-2}(\bM)}\leq \sfK_{\phi,0,l}\ldr{\tau}^{a_{l,n}\cweight+b_{l,n}}e^{\vare_{\Spe}\tau}
  \end{equation}  
  for all $\tau\leq 0$, where $a_{l,n}$, $b_{l,n}$ and $\sfK_{\phi,0,l}$ have the same dependence as the corresponding constants in (\ref{eq:asofhUphiintroNG}).
\end{thm}
\begin{remark}
  In Theorem~\ref{thm:SobEstimatesQuiescentMatter} we impose conditions on the eigenvalues $\ell_A$, but not on the structure coefficients; this
  situation corresponds to behaviour that is expected to be robust. In Theorem~\ref{thm:SobEstimatesQuiescentMatterNG}, we impose conditions on the
  structure coefficients, but not on the $\ell_A$; this situation can occur, e.g., in situations with symmetry. 
\end{remark}
\begin{remark}
  Remark~\ref{remark:hGelzdefintro} is equally relevant here. 
\end{remark}
\begin{remark}
  Remark~\ref{remark:AsInfoID} is equally relevant in the present setting, keeping (\ref{eq:asHamConSFintro}) and
  (\ref{eq:hUsqphiHlminusonestimate}) in mind.
\end{remark}
\begin{remark}\label{remark:NG Kasner SF}
  Due to Remarks~\ref{remark:asHamConSF} and \ref{remark:Sobasymptoticsscalarfieldmattersymmetry}, (\ref{eq:asHamConSFintro}) holds in this case as
  well. Note, moreover, that if $p_{A}$, $A=1,\dots,n$, is an ordered set of constants and $p_{\phi}$ is a constant such that and $p_{1}+\dots+p_{n}=1$ and
  \[
  p_{\phi}^{2}+\textstyle{\sum}_{A}p_{A}^{2}=1,
  \]
  then $(g,\phi)$ given by (\ref{eq:Kasnermetricandphithreed}) is a solution to the Einstein-scalar field equations with $\ell_{A}=p_{A}$. We expect
  this solution to be stable in symmetry classes such that $\g^{A}_{BC}=0$ for $A<\min\{B,C\}$. 
\end{remark}
\begin{proof}
  The statement is an immediate consequence of Remarks~\ref{remark:SobconvofmKdecofhmlUmKmattersymmetry} and
  \ref{remark:Sobasymptoticsscalarfieldmattersymmetry}, keeping in mind that (\ref{eq:DeltavarrhorelvariationEiintro}) holds. 
\end{proof}

\subsection{Partial bootstrap arguments}\label{ssection:partialbootstraparg}

The purpose of the present subsection is to compare the conclusions with the assumptions. In particular, we wish to point out some partial bootstrap
arguments, yielding improvements of the initial assumptions. To begin with, even though we only assume $\mK$ to be bounded, the assumptions combined
with the Hamiltonian constraint yield the conclusion that the eigenvalues of $\mK$, say $\ell=(\ell_{1},\dots,\ell_{n})$, have to belong to the unit ball
asymptotically. In the case of $3+1$-dimensions, this means that $\ell_{+}^{2}+\ell_{-}^{2}\leq 1$ asymptotically; see
(\ref{seq:ellpm})--(\ref{eq:ellpellminuscircconv}). This set is referred to as the Kasner disc, and the arguments make it clear that this set is of
importance under quite general circumstances. Another important conclusion is that $\theta^{-2}\bS$ is negative, up to exponentially decaying terms.
Combining this observation with the Hamiltonian constraint yields the conclusion that $\theta^{-2}\bS$ is bounded, even though we did not assume this
to be the case initially. Moreover, we conclude that $e^{2\mu_{B}-2\mu_{A}-2\mu_{C}}(\g^{B}_{AC})^{2}$ is bounded for $B<A,C$ (for $B\geq\min\{A,C\}$, this holds
automatically). This is remarkable, since we would, in general, expect this expression to tend to infinity exponentially. At a generic spacetime point,
we would also expect $(\g^{B}_{AC})^{2}$ to be bounded away from zero. This would then imply the boundedness of $e^{2\mu_{B}-2\mu_{A}-2\mu_{C}}$. This boundedness
can then, in its turn, be combined with formulae for the expansion normalised Ricci tensor; see 
(\ref{eq:rescaledRiccicurvatureformwoderoflntheta}) below; in order to conclude that $\theta^{-2}\bR^{i}_{\phantom{i}j}$ does not grow faster than polynomially.
Comparing this observation with
\begin{equation}\label{eq:mlUmKwithEinstein}
  \begin{split}
    (\hml_{U}\mK)^{i}_{\phantom{i}j} =  & -\left(\tfrac{\Delta_{\bge}N}{\theta^{2}N}+\tfrac{n}{n-1}\tfrac{\rho-\bp}{\theta^{2}}
    +\tfrac{2n}{n-1}\tfrac{\Lambda}{\theta^{2}}-\tfrac{\bS}{\theta^{2}}\right)\mK^{i}_{\phantom{i}j}+\tfrac{1}{n-1}\tfrac{\rho-\bp}{\theta^{2}}\de^{i}_{j}\\
    & +\tfrac{2}{n-1}\tfrac{\Lambda}{\theta^{2}}\de^{i}_{j}
    +\tfrac{1}{N\theta^{2}}\bnabla^{i}\bnabla_{j}N+\tfrac{\mcP^{i}_{\phantom{i}j}}{\theta^{2}}-\tfrac{\bR^{i}_{\phantom{i}j}}{\theta^{2}};
  \end{split}
\end{equation}
see (\ref{eq:mlUmKwithEinstein appendix}); it then follows that polynomial growth of the expansion normalised normal derivative of
$\mK$ is consistent. The latter arguments are of course heuristic in nature. However, they point to the importance of distinguishing between the points
(both in spacetime and in phase space) where $\g^{1}_{23}$ vanishes (in the case of $3+1$-dimensions) and the points where it does not; see
also Theorem~\ref{thm:SobEstimatesQuiescentMatterNG} above. 

Turning more specifically to the assumptions formulated in Section~\ref{section:Assumptions}, the perhaps least natural of them is the requirement that
$\hml_{U}\mK$ satisfy a weak off-diagonal exponential bound. However, under certain circumstances, the assumptions actually imply an improvement of
this assumption. In fact, in the case of Einstein's vacuum equations in $3+1$-dimensions, Corollary~\ref{cor:improvedoffdiagexpbd} implies an improvement
of the weak off-diagonal exponential bounds, assuming $\e_{\mK}<2\e_{\Spe}$; see Definition~\ref{def:offdiagonalexpdec}. Making suitable assumptions concerning
the matter; see Corollary~\ref{cor:improvedoffdiagexpbd}; a similar conclusion holds for the Einstein-matter equations in $3+1$-dimensions. Turning to
the assumed bounds concerning the norms, a more detailed analysis is necessary. This can, in general, be expected to be quite difficult. However, in
the quiescent setting, discussed in Subsections~\ref{ssection:QuiescentResults} and \ref{ssection:NGQuiescentResults}, we do obtain improvements; by
quiescent setting we here mean a situation in which either (\ref{eq:qexpconvergencequiescentregimeintro}) is satisfied or $\g^{A}_{BC}=0$ for
$A<\min\{B,C\}$ is satisfied (see Subsection~\ref{ssection:summary discussion ex} for a discussion of these conditions).

\textit{Quiescent setting.} Due to the results of Subsections~\ref{ssection:QuiescentResults} and \ref{ssection:NGQuiescentResults}, it is clear that
several of the bounds required to hold in Section~\ref{section:Assumptions} can be improved in the quiescent setting, albeit with a loss of regularity.
More specifically, the assumptions concerning $q$ and $\mK$ can
be improved. Turning to the lapse function and the shift vector field (and, to some extent, also the mean curvature), it is necessary to impose gauge
conditions in order to obtain improvements. However, considering, e.g., CMC transported spatial coordinates; see \cite[Section~3]{rasq}; $\chi=0$,
$\theta$ is independent of the spatial coordinates and $N$ satisfies an elliptic equation. In that setting, the conditions concerning $\chi$ and $\theta$
are thus automatically satisfied. Moreover, under the circumstances considered here, we expect $N$ to converge to $1$ exponentially (though the
corresponding estimate might involve a loss of derivatives) and the expansion normalised normal derivative of $N$ to decay exponentially (again, with a
loss of derivatives). However, we do not write down the corresponding arguments here.

In the quiescent setting, it is also of interest to note that the $\ell_{A}$ converge exponentially and that $q$ converges to $n-1$, so that, starting
close enough to the asymptotic regime, conditions concerning $\ell_{A}$ and the eigenvalues of $\chK$ satisfied initially (with a slight margin) should be
satisfied to the past. Note also that in the quiescent setting, the condition of weak off-diagonal exponential bounds is automatically satisfied, since
$\hml_{U}\mK$ converges to zero exponentially in that setting. 

\subsection{Revisiting results obtained in the spatially homogeneous setting}\label{ssection:revisitsphom}

With the above observations in mind, it is of interest to revisit the results that have been obtained concerning Einstein's equations under additional
symmetry assumptions; see the end of Subsection~\ref{ssection:results} for an explanation of the importance in doing so. Let us begin by considering
spatially homogeneous solutions to the Einstein orthogonal perfect fluid equations in $3+1$-dimensions; see \cite{Bclass} for a description of the
origin of the Bianchi classification and \cite{Wellis} and references therein for a description of the state of the art in this setting in 1997. To
clarify what this means, recall that the stress energy tensor associated with a perfect fluid is of the form
\begin{equation}\label{eq:T perfect fluid}
  T=(\rho+p)u^{\flat}\otimes u^{\flat}+pg,
\end{equation}
where $u$ is a future directed unit timelike vector field, $u^{\flat}(X)=g(u,X)$, $\rho$ denotes the energy density and $p$ the pressure. In the spatially
homogeneous setting, an orthogonal perfect fluid is one for which $u$ is perpendicular to the hypersurfaces of spatial homogeneity. In our setting, $M=G\times I$,
where $G$ is a Lie group and $I$ is an interval. Moreover, $\d_t$ is a future directed unit vector field perpendicular to the hypersurfaces of spatial homogeneity.
This means that $u^{\flat}=dt$. In order to obtain a closed system of equations, we need to impose an equation of state, giving $p$ in terms of $\rho$.
There are many possible choices, leading to many different types of behaviour. For that reason, we will here only focus on a linear equation of
state: $p=(\g-1)\rho$ with $\g\in (2/3,2]$ (even in this class there are many types of different behaviour). In this setting, Einstein's equations
coupled to the equation for $\rho$, resulting from the fact that $T$ is divergence free, yield the equations in the spatially homogeneous Einstein
orthogonal perfect fluid setting. 

\textbf{Spatially homogeneous solutions.} In $3+1$-dimensions, there are two types of spatially homogeneous solutions: maximal globally hyperbolic
developments of left invariant initial data on a $3$-dimensional Lie group $G$ (the corresponding solutions are referred to as \textit{Bianchi spacetimes});
and maximal globally hyperbolic developments of initial data on $\sn{2}\times\ro$ invariant under the isometry group of the standard Riemannian metric on
$\sn{2}\times\ro$ (the corresponding spacetimes are said to be of \textit{Kantowski-Sachs} type). The Bianchi spacetimes are, additionally, divided into
\textit{Bianchi class A} and \textit{Bianchi class B}, according to whether the relevant Lie group is unimodular or non-unimodular respectively.

\textbf{Bianchi class A.} In the case of Bianchi class A orthogonal perfect fluids, the metric can be written
\[
g=-dt\otimes dt+\textstyle{\sum}_{i=1}^{3}a_{i}^{2}(t)\xi^{i}\otimes\xi^{i}
\]
on $M:=G\times I$, where $I$ is an open interval, $G$ is a $3$-dimensional unimodular Lie group, $\{e_{i}\}$ is a basis of the Lie algebra $\mfg$ of $G$
and $\{\xi^{i}\}$ is the dual basis. Moreover, the structure constants $\g^{i}_{jk}$ associated with $\{e_{i}\}$, defined by $[e_{j},e_{k}]=\g^{i}_{jk}e_{i}$,
satisfy $\g^{i}_{jk}=\e_{jkl}\nu^{li}$, where $\e_{ijk}$ is totally antisymmetric and satisfies
$\e_{123}=1$. Finally, $\nu$ is a diagonal matrix whose diagonal elements belong to $\{-1,0,1\}$. We refer the interested reader to \cite{BianchiIXattr}, in
particular \cite[Appendix~21]{BianchiIXattr} for a justification of these statements. Moreover, the $3$-dimensional unimodular Lie groups can be classified
according to the diagonal elements of $\nu$: if $\nu=0$, the group is said to be of Bianchi type I; if one of the diagonal elements of $\nu$ is non-zero, the
group is said to be of Bianchi type II; if two of the diagonal elements are non-zero, and they have different sign, the group is said to be of Bianchi type
VI${}_{0}$; if two of the diagonal elements are non-zero, and they have the same sign, the group is said to be of Bianchi type VII${}_{0}$; if all of the
diagonal elements are non-zero, and they do not have the same sign, the group is said to be of Bianchi type VIII; and if all of the diagonal elements are
non-zero, and they have the same sign, the group is said to be of Bianchi type IX. The spacetimes are labeled accordingly. In the present setting,
$T$ takes the form (\ref{eq:T perfect fluid}) with $u^{\flat}=dt$, and we here assume $p=(\g-1)\rho$, where $\g\in (2/3,2]$. 

In this setting, there is a natural set of variables, introduced by Wainwright and Hsu in \cite{whsu}. In order to describe them, note first that $\mK$ is
diagonal with respect to the basis $\{e_{i}\}$. We can therefore think of the basis $\{e_{i}\}$ as corresponding to the basis $\{X_{A}\}$ in the present article.
However, here we do not insist that the eigenvectors are ordered according to the eigenvalues. In other words, if we let $l_{i}$ denote the eigenvalue of
$\mK$ corresponding to $e_{i}$, \textit{we do not assume} $l_{1}<l_{2}<l_{3}$. Two of the variables used in \cite{whsu} are $\Sigma_{\pm}$, which are
defined as $\ell_{\pm}$ in (\ref{seq:ellpm}), but with $\ell_{i}$ replaced by $l_{i}$. Next, let $\be_{i}:=a_{i}^{-1}e_{i}$ (no
summation). Then $\d_{t}$ combined with $\{\be_{i}\}$ is an orthonormal frame. Moreover, the structure constants, say $\bga^{i}_{jk}$ associated with
$\{\be_{i}\}$ satisfy $\bga^{i}_{jk}=\e_{jkl}n^{li}$, where $n$ is a diagonal matrix. Denoting the diagonal components of $n$ by $n_{i}$, $i=1,2,3$, Wainwright
and Hsu introduce the variables $N_{i}:=n_{i}/\theta$ and $\Omega_{\whsu}:=3\rho/\theta^{2}$. In addition, they change the time coordinate to $\tau$,
which is essentially the logarithmic volume density (so that $\d_\tau$ and $\hU$ are essentially the same). 

The variables introduced by Wainwright and Hsu are partially motivated by the algebraic structure of unimodular Lie algebras; see the definition of the
$N_{i}$, the vanishing and signs of which correspond to a classification of the corresponding Lie algebras. However, if we let $\che_{i}:=\theta^{-1}\be_{i}$
and $\chga^{i}_{jk}$ be the associated structure constants, then
\begin{equation}\label{eq:Nirelations}
  N_{1}=\chga^{1}_{23},\ \ \
  N_{2}=\chga^{2}_{31},\ \ \
  N_{3}=\chga^{3}_{12}.
\end{equation}
Note that all of these quantities are of importance here. The reason for this is that we have here (as opposed to in the remainder of this article) not ordered
the eigenvectors according to the eigenvalues; which of the $l_{i}$ is smallest depends on the position in the phase space. Note also that the objects
(\ref{eq:Nirelations}) are geometric; see the discussion following (\ref{eq:mSbbchgaversion}). To conclude: the Wainwright-Hsu variables in the Bianchi
class A orthogonal perfect fluid setting are particularly favourable in that they, on the one hand, consist of the geometric objects $\Sigma_{\pm}$,
$\chga^{1}_{23}$, $\chga^{2}_{31}$ and $\chga^{3}_{12}$; and, on the other hand, the vanishing and non-vanishing of the last three quantities is tied to the
Lie algebra and is therefore invariant under the dynamics. Finally, the Einstein-orthogonal perfect fluid equations can be formulated as a dynamical system
in terms of $\Sigma_{\pm}$, the $N_{i}$ and $\Omega_{\whsu}$. This dynamical system is referred to as the \textit{Wainwright-Hsu equations}; see, e.g.,
\cite{whsu} or \cite{BianchiIXattr} for details. 

It is of interest to express the estimates contained in Theorems~\ref{thm:SobestimatesQuiescentVacuum}, \ref{thm:SobEstimatesQuiescentMatter},
\ref{thm:SobestimatesQuiescentVacuumNG} and \ref{thm:SobEstimatesQuiescentMatterNG} in terms of the Wainwright-Hsu variables. Note, first of all,
that the generic vacuum setting is only consistent for $11$ spacetime dimensions and above; see Remark~\ref{remark:pointwise vacuum Kasner}. In
other words, Theorem~\ref{thm:SobestimatesQuiescentVacuum} is not relevant in the $3+1$-dimensional setting. This is also the reason one does not
expect quiescence in the case of vacuum Bianchi type VIII and IX; $\chga^{1}_{23}$, $\chga^{2}_{31}$ and $\chga^{3}_{12}$ are all non-zero in the
case of Bianchi types VIII and IX, so that one is in the generic $3+1$-dimensional setting. That quiescence does not hold in these cases follows
from \cite{cbu}. On the other hand, Theorem~\ref{thm:SobestimatesQuiescentVacuumNG} is relevant in the Bianchi class A vacuum setting, if the Bianchi
type is different from VIII and IX. Similarly, Theorem~\ref{thm:SobEstimatesQuiescentMatter} is relevant for all Bianchi class A types and
Theorem~\ref{thm:SobEstimatesQuiescentMatterNG} is relevant for all Bianchi class A types different from VIII and IX. The estimates of the
second terms on the left hand sides of (\ref{eq:SobestofhmlUmKqugeomatterintro}), (\ref{eq:SobestofhmlUmKqugeointroNG}) and
(\ref{eq:SobestofhmlUmKqugeomatterintroNG}) correspond to the conclusion that the $\d_\tau\Sigma_{\pm}$ decay to zero exponentially in the direction
of the singularity. The estimates of the first terms on the left hand sides of these equations correspond to the conclusion that the expansion
normalised spatial scalar curvature (which is an explicit polynomial in the $N_i$) converges to zero exponentially. The estimates
(\ref{eq:mKconvtomKinfHlmatterintro}), (\ref{eq:mKconvtomKinfHlintroNG}) and (\ref{eq:mKconvtomKinfHlmatterintroNG}) correspond to the
convergence of the $\Sigma_{\pm}$. Finally, (\ref{eq:asofhUphiintro}) and (\ref{eq:asofhUphiintroNG}) correspond to the convergence of
$\Omega_{\whsu}$.

\begin{figure}
  \begin{center}
    \includegraphics{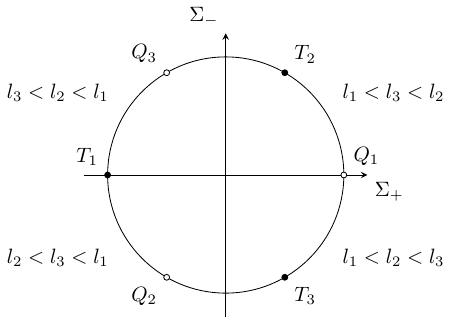}
  \end{center}
  \caption{The Kasner circle with the special points $T_{i}$, antipodal points $Q_{i}$ and segments $\mathscr{K}_i$, $i=1,2,3$, indicated. At each of the points
    $T_{i}$, $Q_{i}$, $i=1,2,3$, two of the eigenvalues coincide. In between these points, the eigenvalues are all distinct. Moreover, on $\mathscr{K}_i$,
    $l_i$ is the smallest eigenvalue, $i=1,2,3$. We have also indicated that between $Q_{1}$ and $T_{2}$, $l_{1}<l_{3}<l_{2}$ etc. Note also that between $T_{2}$ and
    $Q_{3}$, $l_{3}<l_{1}<l_{2}$ and that between $Q_{2}$ and $T_{3}$, $l_{2}<l_{1}<l_{3}$.}\label{fig:Kasnercircle}
\end{figure}
\textit{Bianchi type I.} Considering the results concerning quiescent singularities in Subsections~\ref{ssection:QuiescentResults} and
\ref{ssection:NGQuiescentResults}, it is clear that the Bianchi type I solutions play a special role; pointwise, they are the natural asymptotic limits
(this is a
consequence of Remarks~\ref{remark:pointwise vacuum Kasner}, \ref{remark:pointwise non-vacuum Kasner}, \ref{remark:eigenvaluesQuiescentVacuumNG}
and \ref{remark:NG Kasner SF}, keeping in mind that Bianchi type I solutions take the form (\ref{eq:gKasnergenn}) and (\ref{eq:Kasnermetricandphithreed})
in the vacuum and scalar field settings respectively, and that orthogonal stiff fluids in the spatially homogeneous setting are equivalent to scalar
fields; see below). 
In the case of vacuum, the Bianchi type I solutions are all fixed points to the Wainwright-Hsu equations, defined by $\Sigma_{+}^{2}+\Sigma_{-}^{2}=1$;
see Figure~\ref{fig:Kasnercircle}. The corresponding set is referred to as the \textit{Kasner circle}. Since $\g^{A}_{BC}=0$ for all $A$, $B$ and $C$ 
in the case of Bianchi type I vacuum solutions, these solutions fall into the framework of Theorem~\ref{thm:SobestimatesQuiescentVacuumNG}. On the Kasner
circle, there are three so called \textit{special points}: $T_{1}=(-1,0)$, $T_{2}=(1/2,\sqrt{3}/2)$ and $T_{3}=(1/2,-\sqrt{3}/2)$. They correspond to the
\textit{flat Kasner solutions} which are subsets (or quotients of subsets) of Minkowski space. As is clear from the discussions of the previous sections
and subsections, the order of the eigenvalues in the different segments of the Kasner circle are important. We have therefore indicated the order in
Figure~\ref{fig:Kasnercircle}. In what follows, we also use the following notation: if $\{i,j,k\}$ is a permutation of $\{1,2,3\}$, then the segment
between $T_{j}$ and $T_{k}$ is denoted $\msK_{i}$. 

\begin{figure}
  \begin{center}
    \includegraphics{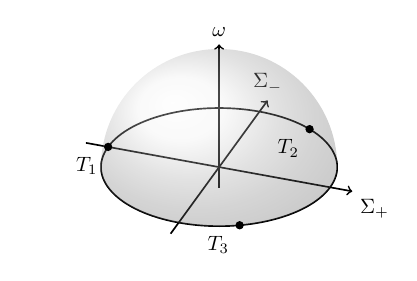}
    \includegraphics{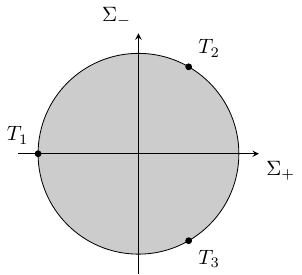}    
  \end{center}
  \caption{The state space of Bianchi type I orthogonal stiff fluid solutions (i.e., orthogonal perfect fluids with equation of state $p=\rho$) 
    is depicted on the left, where $\omega^{2}:=\Omega_{\whsu}$, $\omega\geq 0$. Since the Hamiltonian constraint in this case reads
    $\omega^{2}+\Sigma_{+}^{2}+\Sigma_{-}^{2}=1$, where $\omega\geq 0$, it is sometimes more convenient to represent the state space by the
    \textit{Kasner disc}; see the image on the right.}\label{fig:Bianchi-I-Fluid-full}
\end{figure}
Next, let us consider orthogonal stiff fluids. In this case, the equation of state is given by $p=\rho$; see, e.g., \cite[Subsection~2.1.2]{RinWave} for
a discussion of the non-stiff case. Note that solutions to the Einstein-scalar field equations can be thought of as orthogonal stiff fluid solutions with
$\Omega_{\whsu}=3(\hU\phi)^{2}/2$. Similarly, orthogonal stiff fluid solutions can be thought of as solutions to the Einstein-scalar field equations. 
The results in the stiff fluid setting discussed below therefore belong to the framework of
Theorems~\ref{thm:SobEstimatesQuiescentMatter} and \ref{thm:SobEstimatesQuiescentMatterNG}. In the case of Bianchi type I stiff fluids, the state space
is given by $\Omega_{\whsu}+\Sigma_{+}^{2}+\Sigma_{-}^{2}=1$; see Figure~\ref{fig:Bianchi-I-Fluid-full}. Moreover, the state space consists solely of
fixed points. It is convenient to depict the state space in terms of the Kasner disc; see Figure~\ref{fig:Bianchi-I-Fluid-full}. Again, the order of
the eigenvalues is important. However, this can be deduced from Figure~\ref{fig:Kasnercircle}. For example, the interior of the sector of the Kasner disc
defined by the origin, $Q_{1}$ and $T_{2}$ is such that $l_{1}<l_{3}<l_{2}$. Keeping Theorem~\ref{thm:SobEstimatesQuiescentMatter} in mind, it is also of
interest to determine where $\mK$ is positive definite. The relevant region is depicted on the right in Figure~\ref{fig:QuiescentConvergentRegime}.

Given the above observations, we are in a position to predict the asymptotics of Bianchi class A solutions in the vacuum and orthogonal stiff fluid
settings. Moreover, these predictions are based solely on the order of the eigenvalues $l_{i}$ in different parts of the phase space, and whether the
corresponding $\g^{i}_{jk}$ vanishes or not (where $\{i,j,k\}$ is the appropriate permutation of $\{1,2,3\}$). As a preliminary observation, note
that there are vacuum solutions such that the $\Sigma_{+}\Sigma_{-}$-coordinates converge to one of the special points (this can not happen in the stiff
fluid setting, assuming that there is matter present). However, the corresponding solutions are very special in that they are locally rotationally symmetric
(in other words, they have an additional degree of symmetry); see \cite[Proposition~1, p.~721]{cbu}. In what follows, we therefore exclude them from the
discussion.

\textit{Bianchi type II.} In the case of Bianchi type II, we can, without loss of generality, assume that $N_{1}\neq 0$ and that $N_{2}=N_{3}=0$. Considering
Figure~\ref{fig:Kasnercircle}, keeping (\ref{eq:Nirelations}) in mind, we see that convergence to $\msK_{1}$ should be prevented by the fact that
$\g^{1}_{23}\neq 0$. On the other hand, convergence to $\msK_{2}\cup\msK_{3}$ is consistent with Theorem~\ref{thm:SobestimatesQuiescentVacuumNG}. In the
vacuum setting, this is exactly what happens: generic Bianchi type II vacuum solutions with $N_{1}\neq 0$ converge to a point in $\msK_{2}\cup\msK_{3}$; see,
e.g., \cite[Proposition~22.15, p.~238]{minbok}.

Turning to the orthogonal stiff fluid setting, the sector defined by $T_{1}$, $T_{2}$ and the center, as well as the sector defined by $T_{1}$, $T_{3}$ and
the center are consistent with Theorem~\ref{thm:SobEstimatesQuiescentMatterNG}. Moreover, the region in which $\mK$ is positive definite (i.e., the shaded
region depicted on the right in Figure~\ref{fig:QuiescentConvergentRegime}) is consistent with Theorem~\ref{thm:SobEstimatesQuiescentMatter}. We thus expect
the union of these regions to be consistent with quiescent behaviour, but in the complement of this region (in the Kasner disc), we do not expect quiescence.
This is exactly what happens. In fact, Bianchi type II orthogonal stiff fluid solutions with $N_{1}\neq 0$ converge to a point in the shaded region in the
figure on the left in Figure~\ref{fig:QuiescentConvergentRegime}; see, e.g., \cite[Theorem~19.1, p.~478]{BianchiIXattr}.

\begin{figure}
  \begin{center}
    \includegraphics{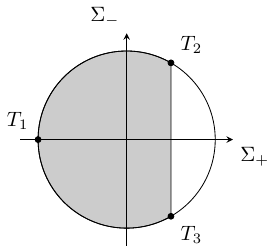}
    \includegraphics{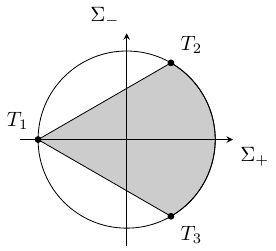}
    \includegraphics{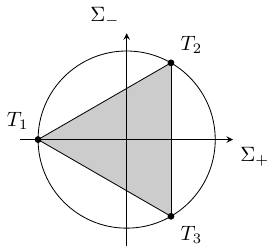}
  \end{center}
  \caption{The gray area in the figure on the right indicates the subset of the Kasner disc in which $\mK$ is positive definite.
    The gray area in the figure on the left indicates the subset of the Kasner disc to which Bianchi type II solutions converge in
    case $N_{1}\neq 0$. As expected, the complement of this area corresponds to the part of the phase space where $l_{1}\leq 0$
    (since $\g^{1}_{23}\neq 0$, $\g^{2}_{31}=0$ and $\g^{3}_{12}=0$). The gray area in the figure in the center indicates the subset
    of the Kasner disc to which Bianchi type VI${}_{0}$ and VII${}_{0}$ solutions converge in case $N_{2}\neq 0$ and $N_{3}\neq 0$.
    As expected, the complement of this area corresponds to the part of the phase space where either $l_{2}\leq 0$ or $l_{3}\leq 0$
    (since $\g^{1}_{23}=0$, $\g^{2}_{31}\neq 0$ and $\g^{3}_{12}\neq 0$). In the case of Bianchi type VIII and IX, $\g^{1}_{23}\neq 0$,
    $\g^{2}_{31}\neq 0$ and $\g^{3}_{12}\neq 0$. Moreover, Bianchi type VIII and IX orthogonal stiff fluid solutions converge, as
    expected, to a point inside the triangle on the right.} 
  \label{fig:QuiescentConvergentRegime}
\end{figure}

\textit{Bianchi type VI${}_{0}$ and VII${}_{0}$.} In the case of Bianchi type VI${}_{0}$ and VII${}_{0}$, we can, without loss of generality, assume that
$N_{1}=0$, that $N_{2}\neq 0$ and that $N_{3}\neq 0$. Considering Figure~\ref{fig:Kasnercircle}, keeping (\ref{eq:Nirelations}) in mind, we see that
convergence to $\msK_{2}\cup\msK_{3}$ should be prevented by the fact that $\g^{2}_{31}\neq 0$ and $\g^{3}_{12}\neq 0$. On the other hand, convergence to
$\msK_{1}$ is consistent with Theorem~\ref{thm:SobestimatesQuiescentVacuumNG}. In the vacuum setting, this is exactly what happens: generic Bianchi type
VI${}_{0}$ and VII${}_{0}$ vacuum solutions with $N_{1}=0$ converge to a point in $\msK_{1}$; see, e.g., \cite[Proposition~22.16, p.~239]{minbok} and
\cite[Proposition~22.18, p.~240]{minbok}. For further information concerning the asymptotics of Bianchi type VII${}_{0}$, we refer the interested reader
to \cite[Section~3.3, pp.~9--13]{hau}.

Turning to the orthogonal stiff fluid setting, the sector defined by $T_{2}$, $T_{3}$ and the center is consistent with
Theorem~\ref{thm:SobEstimatesQuiescentMatterNG}. Moreover, the region in which $\mK$ is positive definite is consistent with
Theorem~\ref{thm:SobEstimatesQuiescentMatter}. We thus expect the
union of these regions to be consistent with quiescent behaviour, but in the complement of this region (in the Kasner disc), we do not expect quiescence.
This is exactly what happens. In fact, Bianchi type VI${}_{0}$ and VII${}_{0}$ orthogonal stiff fluid solutions with $N_{1}=0$ converge to a point in the
shaded region in the figure in the center in Figure~\ref{fig:QuiescentConvergentRegime}; see \cite[Theorem~19.1, p.~478]{BianchiIXattr}.

\textit{Bianchi type VIII and IX.} In the case of Bianchi type VIII and IX, $\g^{1}_{23}$, $\g^{2}_{31}$ and $\g^{3}_{12}$ are all non-zero. For this reason,
we do not expect quiescent behaviour in the vacuum setting. This is, in fact, what happens: for generic Bianchi type VIII and IX vacuum solutions, the
eigenvalues of $\mK$ do not converge in the direction of the singularity; see \cite[Theorem~5, p.~727]{cbu}. More optimistically, one could hope to prove
that the only variables that are relevant asymptotically are $\Sigma_{\pm}$ and one of $\chga^{1}_{23}$, $\chga^{2}_{31}$ and $\chga^{3}_{12}$ at a time. In
other words, one would expect the solution to converge to the subset of the state space defined by the union of the type I and II subspaces (we refer to
this union as the attractor, denoted $\msA$). Again, this is what happens; see \cite[Theorem~19.2, p.~479]{BianchiIXattr} and \cite[Theorem~5, p.~46]{brehm}. 
Finally, one could hope to prove that the asymptotic dynamics are governed by the evolution equations obtained by equating the right hand sides of
(\ref{seq:hUellpmintro}) to zero; see Figure~\ref{fig:TheKasnerMap} for an illustration. In fact, there are results
of this type; see, e.g., \cite{lea,beguin,du}. The idea of these results is to take a heteroclinic orbit of the Kasner map (obtained by equating the right
hand sides of (\ref{seq:hUellpmintro}) to zero) and to prove that there is a stable manifold of solutions converging to
this orbit. In \cite{lea,beguin}, this is done for non-generic heteroclinic orbits. However, the union of the stable manifolds corresponding to the orbits
considered in \cite{du} has positive Lebesgue measure in the full state space.

In the Bianchi type VIII and IX orthogonal stiff fluid setting, the region in which $\mK$ is positive definite is consistent with
Theorem~\ref{thm:SobEstimatesQuiescentMatter}. However, there is no region in which Theorem~\ref{thm:SobEstimatesQuiescentMatterNG} applies.
We thus expect solutions to converge to points in the shaded region depicted on the right in Figure~\ref{fig:QuiescentConvergentRegime}. 
This is what happens; see \cite[Theorem~19.1, p.~478]{BianchiIXattr}.

\begin{figure}
  \begin{center}
    \includegraphics{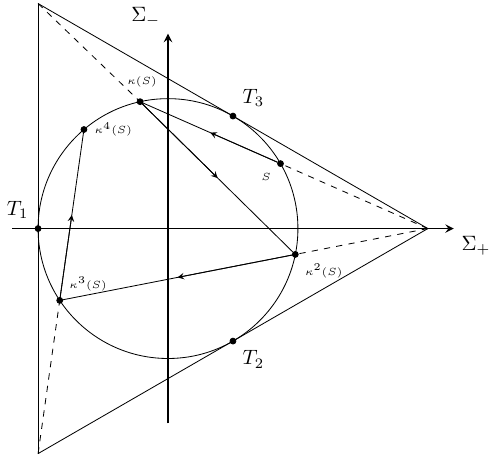}    
  \end{center}
  \caption{Equating the right hand sides of (\ref{seq:hUellpmintro}) with zero yields a map from the Kasner circle
    to itself, called the Kasner map (the straight line connecting $S$ and $\kappa(S)$ is the projection of a Bianchi type II solution to the Einstein
    vacuum equations). Given a point $S$ on the circle, $\kappa(S)$ is obtained by taking the nearest corner of the triangle, drawing
    a straight line from this corner to $S$, and then continuing this straight line to the next intersection with the circle. This next intersection
    defines $\kappa(S)$.}\label{fig:TheKasnerMap}
\end{figure}

\subsection{Revisiting results obtained in the $\tn{3}$-Gowdy symmetric setting}\label{ssection: t3 Gowdy}

Going beyond spatial homogeneity, it is natural to consider $\tn{3}$-Gowdy symmetric solutions \cite{Gowdy,Chr}. In this case, the metric takes the form
\begin{equation}\label{eq:Gowdymetricareal}
  g=t^{-1/2}e^{\lambda/2}(-dt^{2}+d\vartheta^{2})+te^{P}(dx+Qdy)^{2}+te^{-P}dy^{2}
\end{equation}
on $\tn{3}\times (0,\infty)$. Here the functions $P$, $Q$ and $\lambda$ only depend on $t$ and $\vartheta$, so that the metric is invariant under the
action of $\tn{2}$ corresponding to translations in $x$ and $y$. In what follows, it is also convenient to use the time coordinate $\tau=-\ln t$. With
this choice, the big bang singularity corresponds to $\tau\rightarrow\infty$. Moreover, it is convenient to divide the analysis into the
\textit{polarised} ($Q=0$) and \textit{non-polarised} ($Q\neq 0$) settings. 

\textbf{Polarised $\tn{3}$-Gowdy symmetric setting.} In the polarised setting, $\mK$ is diagonal in the $(\vartheta,x,y)$-coordinates;
see \cite[Lemma~C.2]{RinWave}. In particular, $\d_{\vartheta}$, $\d_{x}$ and $\d_{y}$ are all eigenvector fields of $\mK$ and $\g^{i}_{jk}=0$ for
$\{i,j,k\}=\{1,2,3\}$, regardless of the choice of
eigenvector fields $\{e_{i}\}$ for $\mK$. In other words, we would expect the polarised $\tn{3}$-Gowdy symmetric vacuum solutions to always be quiescent. 
This is exactly what happens. In the polarised setting, $P_{\tau\tau}-e^{-2\tau}P_{\vartheta\vartheta}=0$. Appealing to standard energy
estimates and Sobolev embedding, it can be demonstrated that $P_{\tau}$ is bounded in any $C^{k}$ norm to the future. This means that $P$ does not grow
faster than linearly in any $C^{k}$-norm. Combining this observation with the wave equation for $P$ yields the conclusion that $P_{\tau\tau}$ decays
exponentially in any
$C^{k}$ norm. In particular, $P_{\tau}$ converges exponentially to a smooth limit function, say $v$ on $\so$. Integrating this estimate yields a smooth
function $\psi$ on $\so$. To summarise,
\[
P_{\tau}(\cdot,\tau)-v,\ \ \
P(\cdot,\tau)-v\tau-\psi
\]
converge to zero exponentially in any $C^{k}$ norm. Combining this information with \cite[Lemma~C.2]{RinWave} and \cite[(C.11)]{RinWave} yields the
conclusion that $\mK$ converges exponentially to a smooth tensor field. The reader interested in results concerning polarised Gowdy for other topologies
is referred to \cite{chisamo}.

\textbf{Non-polarised $\tn{3}$-Gowdy symmetric setting.} In the non-polarised setting, $\d_{\vartheta}$ is still an eigenvector of $\mK$; see
\cite[Lemma~C.2]{RinWave}. However, $\d_{x}$ and $\d_{y}$ are typically not. In fact, even with respect to the orthonormal frame $\{e_{i}\}$ introduced
in \cite[(1.7), p.~1571]{AAR}, $\mK$ is not diagonal. On the other hand, it is clear that $E_{1}:=\d_{\vartheta}$ is an eigenvector field and that if $E_{A}$,
$A=2,3$, are appropriately chosen eigenvector fields of $\mK$ perpendicular to $E_{1}$, then $[E_{2},E_{3}]=0$. In particular, if $\g^{i}_{jk}$ arises from
the structure coefficients
of the $\{E_{i}\}$, then $\g^{1}_{23}=0$. On the other hand, $\g^{2}_{31}$ and $\g^{3}_{12}$ cannot in general be expected to vanish. From this perspective,
it is clear that the most favourable situation (when it comes to obtaining quiescence) is if the eigenvalue corresponding to $E_{1}$ (i.e. corresponding
to $\d_{\vartheta}$) is the smallest one. On the other hand, since we, in the quiescent setting, expect the eigenvalues to asymptotically satisfy the
Kasner relations, see Remark~\ref{remark:eigenvaluesQuiescentVacuumNG}, we expect the smallest eigenvalue to asymptotically be negative. Combining this
observation with \cite[Remark~C.4]{RinWave}, in particular \cite[(C.8)]{RinWave}, leads to the conclusion that $t\lambda_{t}<1$ is a
natural asymptotic regime. In the end, it turns out that generic non-polarised solutions are such that, except for a finite number of points on the
singularity, the limit of $t\lambda_{t}$ belongs to $(0,1)$. Moreover, the behaviour is quiescent. For a justification of this statement, the interested
reader is referred to \cite{SCCGowdy}; see also \cite[Section~C.4]{RinWave}. 

\subsection{Revisiting results obtained in the absence of symmetries}\label{ssection:revisitnosymm}

Next, consider the results obtained in \cite{aarendall,daetal,fal}. These results are concerned with settings in which there are no symmetries. We revisit
them in \cite{RinQC}; combining the results of the present article with those of \cite{RinQC} yields even more detailed information concerning the
asymptotics and a more complete framework for interpreting previous results in the quiescent setting. In \cite{RinID}, we also relate the results in the
Bianchi class A setting and in the $\tn{3}$-Gowdy symmetric setting to the notion of initial data on the singularity.

\subsection{Conclusions}

As is clear from the above, many of the known results concerning big bang formation in general relativity can be interpreted in terms of the framework
developed in this article. Moreover, the results provide partial bootstrap arguments that could be useful in the context of stability proofs or more
general studies of singularity formation. However, one of the most remarkable conclusions of the framework is that, in the case of $3+1$-dimensions, the
essential information is reduced to three scalar functions: $\ell_{\pm}$ and $\chga^{1}_{23}$. Moreover, these functions are foliation dependent but
otherwise geometric in nature, and they are all essentially derived from the expansion normalised Weingarten map. In addition, we naturally obtain criteria
that ensure quiescence. The criteria are associated with the eigenvalues $\ell_{A}$, the particular matter fields present and the symmetry assumptions
(or, more generally, the asymptotic vanishing of $\g^{1}_{23}$, where $\g^{i}_{jk}$ is calculated with respect to a frame normalised with respect to a fixed
metric). Finally, we see that oscillations are associated with $\chga^{1}_{23}$ not converging to zero and that when the eigenvalues $\ell_{A}$ change,
they do so by means of the Kasner map.

In spite of the advantages of the framework, it should be mentioned that it certainly does not cover all the situations of interest. In particular,
there are situations when the matter fields can give rise to oscillatory behaviour; see, e.g., \cite{wea,cahgeneral,cahIX}. The latter results are
concerned with situations in which the matter is anisotropic. This is an important setting which has not yet received as much attention as it deserves.
Concerning the quiescent setting, the article \cite{luw} provides many additional examples of possible behaviour. 

\section{Estimating the spatial scalar curvature}\label{section:basiccurvatureconcl} 

The spatial scalar curvature and the diagonal components of the spatial Ricci curvature (with respect to the frame $\{X_A\}$ and co-frame $\{Y^A\}$) play
a central role in the analysis. We therefore estimate them in the present section. Since the formulae for these expressions (and their derivation) are
lengthy, and since the derivations consist of standard calculations, the corresponding material is to be found in Appendices~\ref{section:sp sc curv} and
\ref{section:sp ricci cu} below. The most important constituents of the leading order terms are the $\mu_A$ introduced in
Remark~\ref{remark:globalframe} and the structure coefficients $\g_{AB}^{C}$, defined by (\ref{eq:structure coeffients}).

\begin{prop}\label{prop:normalisedbSbd}
  Let $0\leq\cweight\in\ro$ and assume that the standard assumptions; see Definition~\ref{def:standardassumptions}; as well as the $(\cweight,1)$-supremum
  assumptions are fulfilled. Then there is a constant $C_{\bS}$ such that
  \begin{equation}\label{eq:renormalisedbSest}
    \left|\theta^{-2}\bS+\textstyle{\frac{1}{4}\sum}_{A,C=2}^{n}\sum_{B=1}^{\min\{A,C\}-1}e^{2\mu_{B}-2\mu_{A}-2\mu_{C}}(\g^{B}_{AC})^{2}\right|
    \leq C_{\bS}\ldr{\varrho}^{2(2\cweight+1)}e^{2\e_{\Spe}\varrho}
  \end{equation}
  on $M_{-}$, where $C_{\bS}=C_{\bS,0}\theta_{0,-}^{-2}$ and $C_{\bS,0}$ only depends on $c_{\cweight,1}$ and $(\bM,\bge_{\refer})$. Moreover,
  \begin{equation}\label{eq:renormalisedRicAAest}
    \begin{split}
      & \left|\theta^{-2}\bR^{A}_{\phantom{A}A}+\textstyle{\frac{1}{2}\sum}_{C,D}e^{2\mu_{D}-2\mu_{A}-2\mu_{C}}(\g^{D}_{AC})^{2}
      -\textstyle{\frac{1}{4}\sum}_{C,D}e^{2\mu_{A}-2\mu_{C}-2\mu_{D}}(\g^{A}_{CD})^{2}\right|\\
      \leq &  C_{\bR}\ldr{\varrho}^{2(2\cweight+1)}e^{2\e_{\Spe}\varrho}
    \end{split}
  \end{equation}
  (no summation on $A$) on $M_{-}$, where $C_{\bR}=C_{\bR,0}\theta_{0,-}^{-2}$ and $C_{\bR,0}$ only depends on $c_{\cweight,1}$ and $(\bM,\bge_{\refer})$.
\end{prop}
\begin{remark}
  The main reason for proving this proposition is that it leads to (\ref{subeq:hUellpm}). Note also that the conceptual consequences of
  (\ref{eq:renormalisedbSest}) and (\ref{subeq:hUellpm}) are discussed in Subsection~\ref{ssection:results}. Let us,
  however, here point out that (\ref{eq:renormalisedbSest}) and (\ref{eq:renormalisedRicAAest}) illustrate that the dominant contribution to
  $\theta^{-2}\bS$ and $\theta^{-2}\bR^{A}_{\phantom{A}A}$ comes from the expressions that contribute in the spatially homogeneous setting; only terms
  involving the squares of structure coefficients contribute. There is no contribution from terms involving either the spatial derivatives of structure
  coefficients or the spatial derivatives of the quantities $\mu_A$.
  Moreover, one particular consequence of the proposition is that the positive part of $\theta^{-2}\bS$ converges to zero exponentially as
  $\varrho\rightarrow-\infty$. When analysing the asymptotics of solutions to the Klein-Gordon equation in the spatially homogeneous setting, this
  property turns out to be of central importance; see \cite{KGCos}.
\end{remark}
\begin{proof}
  Consider (\ref{eq:bSexpression}). First, we wish to prove that all the terms on the right hand side but the first one decay exponentially. Note, to this
  end, that
  \begin{equation}\label{eq:emtwomuAscsqest}
    \begin{split}
      e^{-2\mu_{A}}|\g^{B}_{AC}\g^{C}_{AB}| \leq & C\theta_{0,-}^{-2}\ldr{\varrho}^{2\cweight}e^{2\e_{\Spe}\varrho}
    \end{split}
  \end{equation}
  on $M_{-}$, where we appealed to (\ref{eq:mKbDlnNchicombest}), (\ref{eq:gaCABbasest}) and (\ref{eq:muminmainlowerbound}), and the constant $C$ only
  depends on $c_{\robas}$. Next, note that
  \begin{equation}\label{eq:emmuAXAbmuBest}
    e^{-2\mu_{A}}|X_{A}(\bmu_{B})|^{2}\leq C\theta_{0,-}^{-2}\ldr{\varrho}^{2(2\cweight+1)}e^{2\e_{\Spe}\varrho}
  \end{equation}
  on $M_{-}$, where we appealed to (\ref{eq:muminmainlowerbound}); (\ref{eq:bmuAmclmfwestEi}); (\ref{eq:ldrrho ldr tau equiv}); and the fact that
  $|X_{A}|_{\bge_{\refer}}=1$. Here $C$ only depends on $c_{\cweight,1}$ and $(\bM,\bge_{\refer})$.
  Most of the remaining terms appearing on the right hand side of (\ref{eq:bSexpression}) can be estimated similarly. However, there is one exception:
  \begin{equation}\label{eq:exceptionaltermscalarcurvature}
    \begin{split}
      e^{-2\mu_{A}}|X_{A}(\cha_{A})| \leq & e^{-2\mu_{A}}|X_{A}^{2}(\hmu_{A})|+e^{-2\mu_{A}}|X_{A}(\g^{B}_{AB})|.
    \end{split}
  \end{equation}
  In order to estimate the terms on the right hand side, note, to begin with, that
  \begin{equation}\label{eq:gammaCABLeviCivita}
    \g^{C}_{AB}=Y^{C}(\bD_{X_{A}}X_{B}-\bD_{X_{B}}X_{A})=\omega^{i}(X_{A})Y^{C}(\bD_{E_{i}}X_{B})-\omega^{i}(X_{B})Y^{C}(\bD_{E_{i}}X_{A}),
  \end{equation}
  where $\{E_{i}\}$ and $\{\omega^{i}\}$ are the frames introduced in Remark~\ref{remark:globalframe}. Applying, say, $X_{D}$ to the first term on
  the right hand side yields
  \begin{equation}\label{eq:XDgammaCABreformulatedterm}
    \begin{split}
      & (\bD_{X_{D}}\omega^{i})(X_{A})Y^{C}(\bD_{E_{i}}X_{B})+\omega^{i}(\bD_{X_{D}}X_{A})Y^{C}(\bD_{E_{i}}X_{B})\\
      & +\omega^{i}(X_{A})(\bD_{X_{D}}Y^{C})(\bD_{E_{i}}X_{B})+\omega^{i}(X_{A})Y^{C}(\bD_{X_{D}}\bD_{E_{i}}X_{B}).
    \end{split}
  \end{equation}
  Combining this observation with a similar observation concerning the second term on the right hand side of (\ref{eq:gammaCABLeviCivita}), the
  assumptions
  (see Definition~\ref{def:supmfulassumptions}), properties of the frames (such as $|X_{A}|_{\bge_{\refer}}=1$ and the fact that $|Y^{A}|_{\bge_{\refer}}$ is
  bounded; see Lemma~\ref{lemma:frameinvest}) and (\ref{eq:bDbfAellAetcpteststmtEi}), it follows that
  \[
  |X_{D}(\g^{C}_{AB})|\leq C(|\bD\mK|_{\bge_{\refer}}+|\bD\mK|_{\bge_{\refer}}^{2}+|\bD^{2}\mK|_{\bge_{\refer}})\leq C\ldr{\varrho}^{2\cweight}
  \]
  on $M_{-}$, where $C$ only depends on $c_{\cweight,1}$ and $(\bM,\bge_{\refer})$. Combining this estimate with (\ref{eq:muminmainlowerbound})
  and the assumptions yields the conclusion that
  \begin{equation}\label{eq:XDgaCABest}
    e^{-2\mu_{E}}|X_{D}(\g^{C}_{AB})|\leq C_{a}\theta_{0,-}^{-2}\ldr{\varrho}^{2\cweight}e^{2\e_{\Spe}\varrho}
  \end{equation}
  on $M_{-}$, where $C_{a}$ only depends on $c_{\cweight,1}$ and $(\bM,\bge_{\refer})$. Next, consider
  \begin{equation}\label{eq:XAsquaredexpansion}
    X_{A}^{2}(\bmu_{B})=(\bD_{X_{A}}\omega^{i})(X_{A})E_{i}(\bmu_{B})+\omega^{i}(\bD_{X_{A}}X_{A})E_{i}(\bmu_{B})+\omega^{i}(X_{A})X_{A}E_{i}(\bmu_{B}).
  \end{equation}
  Combining this identity with the assumptions (see Definition~\ref{def:supmfulassumptions}); properties of the frames;
  Lemma~\ref{lemma:bDbfabDlmjchKest}; (\ref{eq:muminmainlowerbound}); (\ref{eq:bmuAmclmfwestEi}); and (\ref{eq:ldrrho ldr tau equiv})
  yields
  \begin{equation}\label{eq:XAsqbmuBsccurvest}
    e^{-2\mu_{C}}|X_{A}^{2}(\bmu_{B})|\leq C\theta_{0,-}^{-2}\ldr{\varrho}^{3\cweight+1}e^{2\e_{\Spe}\varrho}    
  \end{equation}
  on $M_{-}$, 
  where $C$ only depends on $c_{\cweight,1}$ and $(\bM,\bge_{\refer})$. Combining (\ref{eq:exceptionaltermscalarcurvature}), (\ref{eq:XDgaCABest})
  and (\ref{eq:XAsqbmuBsccurvest}) with the previous observations, in particular (\ref{eq:emtwomuAscsqest}) and (\ref{eq:emmuAXAbmuBest}), yields
  the conclusion that
  \begin{equation}\label{eq:renormalisedbSestinterm}
    \left|\theta^{-2}\bS+\textstyle{\frac{1}{4}\sum}_{A,B,C}e^{2\mu_{B}-2\mu_{A}-2\mu_{C}}(\g^{B}_{AC})^{2}\right|
    \leq C_{\bS,\pre}\ldr{\varrho}^{2(2\cweight+1)}e^{2\e_{\Spe}\varrho}
  \end{equation}
  holds on $M_{-}$, where $C_{\bS,\pre}=C_{\bS,\pre,0}\theta_{0,-}^{-2}$ and $C_{\bS,\pre,0}$ only depends on $c_{\cweight,1}$ and $(\bM,\bge_{\refer})$. What remains
  is to justify the elimination of the terms in the sum on the left hand side with $B\geq A$ and $B\geq C$. Say, for the sake of argument, that
  $B\geq A$. Then
  \[
    2\mu_{B}-2\mu_{A}-2\mu_{C}\leq 2(B-A)\e_{\rond}\varrho+2M_{\rodiff}-2\mu_{C},
  \]
  where we appealed to (\ref{eq:bmuAmbmuBlowbd}) and $M_{\rodiff}$ only depends on $c_{\robas}$. In particular, if $B\geq A$, then
  \[
  e^{2\mu_{B}-2\mu_{A}-2\mu_{C}}(\g^{B}_{AC})^{2}\leq C_{a}e^{-2\mu_{C}}(\g^{B}_{AC})^{2}
  \]
  on $M_{-}$. Since the right hand side can be bounded by the right hand side of (\ref{eq:emtwomuAscsqest}), the corresponding term can be eliminated.
  The argument in case $B\geq C$ is identical. 

  Keeping Corollary~\ref{cor:rescaledRiccicurvatureformwoderoflntheta} in mind, the argument to prove (\ref{eq:renormalisedRicAAest}) is similar. 
\end{proof}

\subsection{Consequences of the scalar curvature estimate}

The estimate (\ref{eq:renormalisedbSest}) is of particular interest when combined with the reformulated version
(\ref{eq:reformulatedandrenormalisedHamcon}) of the Hamiltonian constraint. This leads to the following result. 

\begin{prop}\label{prop:effectiveHamconstraint}
  Let $0\leq\cweight\in\ro$ and assume that the standard assumptions; see Definition~\ref{def:standardassumptions}; as well as the $(\cweight,1)$-supremum
  assumptions are fulfilled. Assume, in addition, that (\ref{eq:reformulatedandrenormalisedHamcon}) holds. Then
  \begin{equation}\label{eq:OmegaplussumellAsqbdpnegsc}
    \begin{split}
      \left|2\Omega+\textstyle{\sum}_{A}\ell_{A}^{2}+\mSb_{b}-1\right| \leq C_{\Lambda}e^{2\e_{\Spe}\varrho}+C_{\bS}\ldr{\varrho}^{2(2\cweight+1)}e^{2\e_{\Spe}\varrho}
    \end{split}
  \end{equation}
  on $M_{-}$, where $C_{\bS}$ is the constant appearing in Proposition~\ref{prop:normalisedbSbd}, $C_{\Lambda}=C_{\Lambda,0}\theta_{0,-}^{-2}$ and $C_{\Lambda,0}$
  only depends on $\Lambda$. Moreover, $\mSb_{b}$ is defined by (\ref{eq:mSbbdef}). 
\end{prop}
\begin{remark}
  In particular, up to exponentially decaying terms, the sum of the $\ell_{A}^{2}$ is bounded from above by $1$ and the sum of the
  $\ell_{A}$ equals $1$. 
\end{remark}
\begin{proof}
  Note that $\tr(\mK^2)=\sum_A\ell_A^2$ and that (\ref{eq:lnthetalowbd}) yields
  \begin{equation}\label{eq:OmegaLambdaestimate}
    |\Omega_{\Lambda}|\leq C_{\Lambda}e^{2\e_{\Spe}\varrho}
  \end{equation}
  on $M_{-}$, where $C_{\Lambda}=C_{\Lambda,0}\theta_{0,-}^{-2}$ and $C_{\Lambda,0}$ only depends on $\Lambda$. Combining these observations with
  (\ref{eq:renormalisedbSest}) and (\ref{eq:reformulatedandrenormalisedHamcon}) yields the conclusion of the proposition.
\end{proof}

Next, we wish to use the estimate (\ref{eq:OmegaplussumellAsqbdpnegsc}) to draw conclusions concerning the deceleration parameter introduced in 
(\ref{eq:hUnlnthetamomqbas}). Given that Einstein's equations are satisfied; see (\ref{eq:EE}); the deceleration parameter is given by 
\begin{equation}\label{eq:qmainformula}
  q=n-1-n\tfrac{\Delta_{\bge}N}{\theta^{2}N}-\tfrac{n^{2}}{n-1}\tfrac{\rho-\bp}{\theta^{2}}-\tfrac{2n^{2}}{n-1}\tfrac{\Lambda}{\theta^{2}}
  +n\tfrac{\bS}{\theta^{2}},
\end{equation}
see (\ref{eq:qmainformula appendix}). For this reason, it is clearly of interest to estimate $\theta^{-2}\Delta_{\bge}N/N$. Note that this expression is 
calculated in (\ref{eq:normalisedDeltaNthroughN}) below. This yields the following conclusion. 

\begin{lemma}\label{lemma:normalisedDeltabgelnN}
  Let $0\leq\cweight\in\ro$ and assume that the standard assumptions; see Definition~\ref{def:standardassumptions}; as well as the $(\cweight,1)$-supremum
  assumptions are fulfilled. Then
  \begin{equation}\label{eq:normalisedDeltabgelnN}
    \theta^{-2}N^{-1}|\Delta_{\bge}N|\leq C_{N}\ldr{\varrho}^{3\cweight+1}e^{2\e_{\Spe}\varrho}
  \end{equation}
  on $M_{-}$, where $C_{N}=C_{N,0}\theta_{0,-}^{-2}$ and $C_{N,0}$ only depends on $c_{\cweight,1}$ and $(\bM,\bge_{\refer})$. 
\end{lemma}
\begin{proof}
  Consider (\ref{eq:normalisedDeltaNthroughN}). In this expression, some terms contain factors of the form $e^{-\mu_{C}}a_{C}$ and $e^{-\mu_{A}}X_{A}(\bmu_{B})$.
  These factors can be estimated as in (\ref{eq:emtwomuAscsqest}) and (\ref{eq:emmuAXAbmuBest}). We also need to estimate
  \begin{equation}\label{eq:emonemuAXalnN}
    e^{-\mu_{A}}|X_{A}(\ln N)|\leq C\theta_{0,-}^{-1}\ldr{\varrho}^{\cweight}e^{\e_{\Spe}\varrho},
  \end{equation}
  where we appealed to the assumptions and (\ref{eq:muminmainlowerbound}) and $C$ only depends on $c_{\cweight,0}$. Consider $e^{-2\mu_{A}}X_{A}^{2}(\ln N)$.
  Rewriting $X_{A}^{2}(\ln N)$ in analogy with (\ref{eq:XAsquaredexpansion}) and arguing as in the proof of Proposition~\ref{prop:normalisedbSbd},
  \begin{equation}\label{eq:emtwomuAXasqlnN}
    e^{-2\mu_{A}}|X_{A}^{2}\ln N|\leq C\theta_{0,-}^{-2}\ldr{\varrho}^{2\cweight}e^{2\e_{\Spe}\varrho},
  \end{equation}
  where $C$ only depends on $c_{\cweight,1}$ and $(\bM,\bge_{\refer})$. Combining the above estimates yields the conclusion of the lemma. 
\end{proof}

\begin{cor}
  Let $0\leq\cweight\in\ro$ and assume that the standard assumptions; see Definition~\ref{def:standardassumptions}; as well as the $(\cweight,1)$-supremum
  assumptions are fulfilled. Assume, moreover, that Einstein's equations with a cosmological constant $\Lambda$ (\ref{eq:EE}) are satisfied and that $T$
  satisfies the dominant energy condition, so that $\bp\leq\rho$. Then
  \begin{equation}\label{eq:qupperbound}
    q+n\mSb_{b}\leq n-1+C_{q}\ldr{\varrho}^{2(2\cweight+1)}e^{2\e_{\Spe}\varrho}
  \end{equation}
  on $M_{-}$, where $C_{q}=C_{q,0}\theta_{0,-}^{-2}$ and $C_{q,0}$ only depends on $c_{\cweight,1}$, $\Lambda$ and $(\bM,\bge_{\refer})$. In particular,
  if $\lambda_{A}$, $A=1,\dots,n$, are the eigenvalues of $\chK$,
  \begin{equation}\label{eq:lambdaAlowerbound}
    \lambda_{A}\geq\ell_{A}-1+\mSb_{b}-n^{-1}C_{q}\ldr{\varrho}^{3(2\cweight+1)}e^{2\e_{\Spe}\varrho}
  \end{equation}
  on $M_{-}$. 
\end{cor}
\begin{remark}\label{remark:decparameterandsilence}
  Due to (\ref{eq:OmegaplussumellAsqbdpnegsc}), it is clear that, up to exponentially decaying terms, the sum of the $\ell_{A}^{2}$ is bounded from above
  by $1$. However, the situation that $\ell_{n}=1$ and $\ell_{i}=0$ for $i<n$ is not excluded by this estimate (though it is excluded by the non-degeneracy
  condition). Due to (\ref{eq:lambdaAlowerbound}) it follows that in that
  case, $\lambda_{n}\geq 0$ (up to exponentially decaying terms). However, since $\ell_{n}=1$, the estimate (\ref{eq:OmegaplussumellAsqbdpnegsc}) implies
  that $\Omega$ and $\mSb_{b}$ are exponentially small (so that, if $\Omega$ is known to have a positive lower bound, $\ell_{n}=1$ is asymptotically
  prohibited). On the other hand, in order to have silence, we need to have $\lambda_{n}\leq -\e_{\Spe}$, an
  estimate which implies that $\ell_{A}\leq 1-\e_{\Spe}$, up to exponentially decaying terms. From this point of view, the case that $\ell_{n}=1$ and
  $\ell_{i}=0$ for $i<n$ is excluded. Note, finally, that there are solutions to Einstein's vacuum equations satisfying this condition,
  namely the flat Kasner solutions; see, e.g., \cite[Example~2.8]{RinWave}. 
\end{remark}
\begin{proof}
  First note that the dominant energy condition states that $T(v,w)\geq 0$ for all future directed causal vectors $v,w$ based at the same point.
  Given a $p=(\bx,t)\in M$, let
  $\{\zeta_A\}$, $A=1,\dots,n$, be an orthonormal basis for $T_{\bx}\bM$ with respect to $\bge_{\bx}$ such that $T(\zeta_A,\zeta_B)=0$ for $A\neq B$.
  Then $U\pm\zeta_A$ are future directed causal vectors. This means that
  \[
    0\leq T(U_p+\zeta_A,U_p-\zeta_A)=\rho(p)-T(\zeta_A,\zeta_A).
  \]
  Summing over $A$ yields $\rho\geq\bp$. 

  By assumption, both the conditions yielding (\ref{eq:renormalisedbSest}) and (\ref{eq:normalisedDeltabgelnN}) are satisfied.
  Combining the conclusions with (\ref{eq:OmegaLambdaestimate}); the fact that (\ref{eq:EE}) is satisfied; the fact that $\bp\leq\rho$; and the
  fact that (\ref{eq:qmainformula}) holds yields (\ref{eq:qupperbound}). Combining (\ref{eq:qupperbound}) and (\ref{eq:lambdaArelqellA}) yields
  (\ref{eq:lambdaAlowerbound}). 
\end{proof}

The inequality (\ref{eq:lambdaAlowerbound}) only gives a lower bound on $\lambda_{A}$. However, it is of interest to estimate $\lambda_{A}$ in
terms of $\ell_{A}$ and the spatial scalar curvature. In order to obtain such an estimate, we need to make assumptions concerning the matter.
Considering (\ref{eq:qmainformula}), it is clear that the contribution from the matter is of the form $\theta^{-2}(\rho-\bp)$. In the case of
Einstein's vacuum equations, this expression vanishes. The same is true in the case of a stiff fluid (defined by the condition that the pressure
equals the energy density). Next, consider spatially homogeneous solutions with orthogonal perfect fluid matter with a linear equation of state:
$p=(\g-1)\rho$. If $\g<2$, the expectation is then that for a generic solution, $\Omega$ should decay exponentially in the direction of the singularity;
see, e.g., \cite{BianchiIXattr} for examples of results.

Here, we do not wish to specialise to a particular matter model. In order to obtain conclusion, we therefore need to impose a general
assumption. Due to the above observations, we assume $\theta^{-2}(\rho-\bp)$ to decay exponentially. 

\begin{prop}\label{prop:qandlambdaAestimates}
  Let $0\leq\cweight\in\ro$ and assume that the standard assumptions; see Definition~\ref{def:standardassumptions}; as well as the $(\cweight,1)$-supremum
  assumptions are fulfilled. Assume, moreover, that Einstein's equations with a cosmological constant $\Lambda$ (\ref{eq:EE}) are satisfied. Finally,
  assume that
  there are constants $C_{\rho}$ and $\e_{\rho}>0$ such that
  \begin{equation}\label{eq:closetostiff}
    \theta^{-2}|\rho-\bp|\leq C_{\rho}e^{2\e_{\rho}\varrho}
  \end{equation}
  on $M_{-}$, where $C_{\rho}=C_{\rho,0}\theta_{0,-}^{-2}$. Then
  \begin{equation}\label{eq:ninvoneplusqfirstestimate}
    |q-(n-1)+n\mSb_{b}|\leq C_{q}\ldr{\varrho}^{2(2\cweight+1)}e^{2\e_{q}\varrho}
  \end{equation}
  on $M_{-}$, where $C_{q}=C_{q,0}\theta_{0,-}^{-2}$ and $\e_{q}:=\min\{\e_{\Spe},\e_{\rho}\}$. Moreover, $C_{q,0}$ only depends on $c_{\cweight,1}$, $C_{\rho,0}$,
  $\Lambda$ and $(\bM,\bge_{\refer})$. Finally,
  \begin{equation}\label{eq:lambdaAellAmSbest}
    |\lambda_{A}-(\ell_{A}-1)-\mSb_{b}|\leq n^{-1}C_{q}\ldr{\varrho}^{2(2\cweight+1)}e^{2\e_{q}\varrho}
  \end{equation}
  on $M_{-}$. 
\end{prop}
\begin{remark}
  Due to (\ref{subeq:hUellpm}) below, we know that in order for $\ell_\pm$ to converge, $\mSb_b$ has to be integrable. Due to
  (\ref{eq:ninvoneplusqfirstestimate}), this is equivalent to $q-(n-1)$ being integrable. In other words, the size of $q-(n-1)$
  can be considered a measure of the proximity to the quiescent regime. Turning to the eigenvalues of $\chK$, i.e. the $\lambda_A$,
  (\ref{eq:lambdaAellAmSbest}) illustrates that in the quiescent regime, i.e. when $\mSb_b$ is small, the $\lambda_A$ are determined
  by the eigenvalues of $\mK$. Note also that whether the causal structure is silent or not is determined by the $\lambda_A$. In
  the quiescent setting, whether the causal structure is silent or not is thus determined by whether the eigenvalues of $\mK$ are
  strictly less than $1$ or not. 
\end{remark}
\begin{remark}
  Due to (\ref{eq:renormalisedbSest}), $\mSb_{b}$ can be replaced by $-\theta^{-2}\bS$ in (\ref{eq:ninvoneplusqfirstestimate}) and
  (\ref{eq:lambdaAellAmSbest}).
\end{remark}
\begin{proof}
  Combining (\ref{eq:closetostiff}) with (\ref{eq:renormalisedbSest}), (\ref{eq:OmegaLambdaestimate}), (\ref{eq:normalisedDeltabgelnN}) and
  (\ref{eq:qmainformula}) yields (\ref{eq:ninvoneplusqfirstestimate}). Combining this estimate with (\ref{eq:lambdaArelqellA}) yields
  (\ref{eq:lambdaAellAmSbest}). 
\end{proof}

\section{Dynamics}\label{section:dynamics}

A basic assumption in this article is that the eigenvalues of $\mK$ are distinct. However, it is of interest to analyse the extent to which this
assumption is consistent with the dynamics and to know how the $\ell_{A}$ change. To begin with, note that (\ref{eq:hU ellA}) reads (no summation)
\begin{equation}\label{eq:hUellAnongeo}
  \hU(\ell_{A})=(\hml_{U}\mK)^{A}_{\phantom{A}A}.
\end{equation}
Appealing to (\ref{eq:mlUmKwithEinstein}), it is thus clear that we need to estimate (no summation)
\begin{equation}\label{eq:hmlUmKAAform}
  \begin{split}
    (\hml_{U}\mK)^{A}_{\phantom{A}A} =  & -\left(\tfrac{\Delta_{\bge}N}{\theta^{2}N}+\tfrac{n}{n-1}\tfrac{\rho-\bp}{\theta^{2}}
    +\tfrac{2n}{n-1}\tfrac{\Lambda}{\theta^{2}}-\tfrac{\bS}{\theta^{2}}\right)\ell_{A}+\tfrac{1}{n-1}\tfrac{\rho-\bp}{\theta^{2}}\\
    & +\tfrac{2}{n-1}\tfrac{\Lambda}{\theta^{2}}
    +\tfrac{1}{N\theta^{2}}\bnabla^{A}\bnabla_{A}N+\tfrac{\mcP^{A}_{\phantom{A}A}}{\theta^{2}}-\tfrac{\bR^{A}_{\phantom{A}A}}{\theta^{2}}.
  \end{split}
\end{equation}

\begin{prop}\label{prop:hUellA}
  Let $0\leq\cweight\in\ro$ and assume that the standard assumptions; see Definition~\ref{def:standardassumptions}; as well as the $(\cweight,1)$-supremum
  assumptions are fulfilled. Assume, moreover, that Einstein's equations with a cosmological constant $\Lambda$ (\ref{eq:EE}) are satisfied. Finally,
  assume that (\ref{eq:closetostiff}) is satisfied and that there is a constant $C_{\mcP}$ such that
  \begin{equation}\label{eq:anisotropicpressurebd}
    \theta^{-2}|\mcP^{A}_{\phantom{A}A}|\leq C_{\mcP}e^{2\e_{\rho}\varrho}
  \end{equation}
  on $M_{-}$ for all $A$ (no summation), where $\mcP$ is introduced in (\ref{eq:mfpmcPdef}) and $C_{\mcP}=C_{\mcP,0}\theta_{0,-}^{-2}$. Then
  \begin{equation}\label{eq:hUelldominantterms}
    |\hU(\ell_{A})-\theta^{-2}\bS\ell_{A}+\theta^{-2}\bR^{A}_{\phantom{A}A}|
    \leq C_{\ell}\ldr{\varrho}^{3\cweight+1}e^{2\e_{q}\varrho}
  \end{equation}
  on $M_{-}$, where $\e_{q}:=\min\{\e_{\Spe},\e_{\rho}\}$ and $C_{\ell}=C_{\ell,0}\theta_{0,-}^{-2}$. Moreover, $C_{\ell,0}$ only depends on $c_{\cweight,1}$,
  $C_{\rho,0}$, $C_{\mP,0}$, $\Lambda$ and $(\bM,\bge_{\refer})$.
\end{prop}
\begin{remark}
  The main reason for proving this proposition is that it leads to (\ref{subeq:hUellpm}). Note, however, that by combining (\ref{eq:hUelldominantterms})
  with (\ref{eq:renormalisedbSest}) and (\ref{eq:renormalisedRicAAest}), it follows that the dominant contribution to $\hU(\ell_A)$ comes from the
  expressions that contribute in the spatially homogeneous setting; only terms involving the squares of structure coefficients contribute. There is no
  contribution from terms involving either the spatial derivatives of the structure coefficients or the spatial derivatives of the quantities $\mu_A$.
\end{remark}
\begin{remark}
Since the sum of the $\ell_{A}$ equals $1$, 
\[
\textstyle{\sum}_{A}(\hU(\ell_{A})-\theta^{-2}\bS\ell_{A}+\theta^{-2}\bR^{A}_{\phantom{A}A})
=\hU\left(\textstyle{\sum}_{A}\ell_{A}\right)-\theta^{-2}\bS\textstyle{\sum}_{A}\ell_{A}+\theta^{-2}\bS=0.
\]
Estimating $n-1$ of the expressions on the left hand side of (\ref{eq:hUelldominantterms}) thus yields an estimate for the remaining expression. 
\end{remark}
\begin{proof}
  Due to (\ref{eq:hUellAnongeo}), it is sufficient to estimate some of the terms on the right hand side of (\ref{eq:hmlUmKAAform}). Moreover,
  most of them can be estimated by appealing to (\ref{eq:OmegaLambdaestimate}), (\ref{eq:normalisedDeltabgelnN}), (\ref{eq:closetostiff}) and
  (\ref{eq:anisotropicpressurebd}). The only exception is the fourth term on the right hand side of  (\ref{eq:hmlUmKAAform}). However, considering
  Lemma~\ref{lemma:contributionsfromlapse}, it is clear that the corresponding term can be estimated as in the proof of
  Lemma~\ref{lemma:normalisedDeltabgelnN}. Combining these observations yields (\ref{eq:hUelldominantterms}).
\end{proof}

Combining Proposition~\ref{prop:normalisedbSbd} with Proposition~\ref{prop:hUellA} yields more detailed information 
concerning the evolution of the $\ell_{A}$. Since we are mainly interested in $4$-dimensional spacetimes, we here assume $n=3$. Before 
stating the result, it is convenient to introduce $\ell_{\pm}$ according to (\ref{seq:ellpm});
since the sum of the $\ell_{A}$ equals $1$, it is sufficient to know the values of $\ell_{+}$ and $\ell_{-}$. Moreover, 
(\ref{eq:OmegaplussumellAsqbdpnegsc}) can be reformulated to 
\begin{equation}\label{eq:Hamconellpmmainparts}
  \begin{split}
    \left|\ell_{+}^{2}+\ell_{-}^{2}+3\Omega+\textstyle{\frac{3}{2}}\mSb_{b}-1\right|
    \leq & \textstyle{\frac{3}{2}}C_{\Lambda}e^{2\e_{\Spe}\varrho}+\tfrac{3}{2}C_{\bS}\ldr{\varrho}^{2(2\cweight+1)}e^{2\e_{\Spe}\varrho}
  \end{split}
\end{equation}
on $M_{-}$, where $C_{\Lambda}$ and $C_{\bS}$ are the constants appearing in (\ref{eq:OmegaplussumellAsqbdpnegsc}).
In particular, $\ell_{+}^{2}+\ell_{-}^{2}\leq 1$, up to exponentially decaying terms. 

\begin{thm}\label{thm:main bounce result}
  Given that the assumptions of Proposition~\ref{prop:hUellA} are satisfied and that $n=3$,
  \begin{subequations}\label{subeq:hUellpm}
    \begin{align}
      |\hU(\ell_{+})-\theta^{-2}\bS(\ell_{+}-2)| \leq & C_{+}\ldr{\varrho}^{2(2\cweight+1)}e^{2\e_{q}\varrho},\label{eq:hUellplus}\\
      |\hU(\ell_{-})-\theta^{-2}\bS\ell_{-}| \leq & C_{-}\ldr{\varrho}^{2(2\cweight+1)}e^{2\e_{q}\varrho}\label{eq:hUellminus}
    \end{align}
  \end{subequations}
  on $M_{-}$. Here $C_{\pm}=C_{\pm,0}\theta_{0,-}^{-2}$, where $C_{\pm,0}$ only depends on $c_{\cweight,1}$, $C_{\rho,0}$, $C_{\mP,0}$, $\Lambda$ and $(\bM,\bge_{\refer})$.
\end{thm}
\begin{remark}\label{remark:BKLmap}
  Due to (\ref{eq:Hamconellpmmainparts}), (\ref{eq:hUellplus}) and the fact that (up to exponentially decaying terms) $\theta^{-2}\bS$ is negative, it is
  clear that $\hU(\ell_{+})$ is positive, up to exponentially decaying terms. This means that $\ell_{+}$ decreases to the past. Combining this observation
  with (\ref{eq:ellplus}) yields the conclusion that $\ell_{1}$ increases to the past.
  
  Next, due to the fact that
  $\ell_{+}^{2}+\ell_{-}^{2}\leq 1$, up to exponentially decaying terms, the function $2-\ell_{+}$ should, asymptotically, be bounded strictly away from
  zero. It should therefore be meaningful to consider the quantity $\ell_{-}/(2-\ell_{+})$. Moreover, (\ref{subeq:hUellpm}) yields the conclusion that
  $\hU[\ell_{-}/(2-\ell_{+})]$ decays exponentially. This corresponds exactly to the behaviour of Bianchi type II vacuum solutions; see, e.g.,
  \cite[(22.20), p. 239]{minbok}. Moreover, it is clear that as long as the integral of $\theta^{-2}\bS$ (along integral curves of $\hU$) is small, the
  eigenvalues $\ell_{A}$ do not change much. It is only when the integral of $\theta^{-2}\bS$ becomes
  non-negligible that they change. Moreover, when the integral becomes non-negligible, this means that $(\ell_{+},\ell_{-})$ moves roughly speaking
  along a straight line described by the condition that $\ell_{-}/(2-\ell_{+})$ remain constant. In addition, the initial value of $(\ell_{+},\ell_{-})$
  determines the line. Finally, since $\ell_{+}^{2}+\ell_{-}^{2}\leq 1$, up to exponentially decaying terms, it is not possible for the solution to
  continue indefinitely along the line. At some point, the fact that (\ref{eq:Hamconellpmmainparts}) holds is going to force $\theta^{-2}\bS$ to
  tend to zero. This will fix the ``end value'' of $(\ell_{+},\ell_{-})$. Note that this process reproduces the BKL-map, taking the initial values
  of $(\ell_{+},\ell_{-})$ to the end values. At that stage $e^{\mu_{1}-\mu_{2}-\mu_{3}}$ should then generically decay to zero exponentially; $\ell_{1}$
  is no longer the smallest eigenvalue.

  It should of course be noted that in the transition from the initial state to the end state, the smallest and the second smallest of the eigenvalues
  will trade places, so that the assumption of non-degeneracy is violated. On the other hand, the transition can be expected to be rapid, so that other
  methods could potentially be used to control the evolution in this regime. Moreover, once the transition has been accomplished, one would generically
  expect non-degeneracy to be reinstated. 
\end{remark}
\begin{remark}\label{remark:g123 to zero}
  Fix an integral curve, say $\g$, of $\hU$. Note that we can parametrise $\g$ so that $s-1/2\leq\varrho\circ\g(s)\leq s+1/2$ and $\g(0)\in \bM_{t_0}$, see
  Lemma~\ref{lemma:lowerbdonmumin}. Consider the case that the integral of $(\theta^{-2}\bS)\circ\g$ over some interval, say $[s_a,s_b]$, is small.
  In what follows, we assume $s_b$ to be close to $-\infty$, since we are interested in the asymptotic behaviour. Due to
  (\ref{eq:mSbbdef}) and (\ref{eq:renormalisedbSest}), it follows that
  \[
  |(\theta^{-2}\bS)\circ\g(s)+\mSb_b\circ\g(s)|\leq C\ldr{s}^{2(2\cweight+1)}e^{2\e_{\Spe} s}
  \]
  for all $s\leq 0$. This means that the integral of $\mSb_b\circ\g$ over $[s_a,s_b]$ is small, assuming $s_b$ to be close enough to $-\infty$. Moreover,
  since $\mSb_b\circ\g$ has a sign,
  the integral of $\mSb_b\circ\g$ (and of $(\theta^{-2}\bS)\circ\g$) over any subinterval of $[s_a,s_b]$ is small. Due to (\ref{subeq:hUellpm}),
  this means that the $\ell_{\pm}$ do not change much in $[s_a,s_b]$, assuming $s_b$ is close enough to $-\infty$; note that $(\ell_+^2+\ell_-^2)\circ\g$
  is essentially bounded from above by $1$ for $s_b$ close enough to $-\infty$ due to
  (\ref{eq:Hamconellpmmainparts}). Moreover, since $\mSb_b\circ\g$ has a sign and the integral is small, $\mSb_b\circ\g$ is small on average. Combining
  this observation with (\ref{eq:OmegaplussumellAsqbdpnegsc}) yields the conclusion that $3\Omega\circ\g+\ell_+^2\circ\g+\ell_-^2\circ\g=1$ on
  average; recall (\ref{eq:ellpellminuscircconv}). Since $\ell_{\pm}\circ\g$ do not change much, this means that $\Omega\circ\g$ is constant, on average.
  Next, due to (\ref{eq:mu123 est})
  \begin{equation}\label{eq:mu123 est q two}
    \begin{split}
      & \big|\mu_{\rowo}\circ\g(s_b)-\mu_{\rowo}\circ\g(s_a)-\textstyle{\int}_{s_{a}}^{s_b}[\ell_{\rowo}\circ\g(s)+1]ds\big|\\
      \leq & C+\textstyle{\int}_{s_{a}}^{s_b}|q\circ\g(s)-2|ds/3
    \end{split}
  \end{equation}
  where we use the notation $\ell_{\rowo}:=\ell_{1}-\ell_2-\ell_3$ and $\mu_{\rowo}:=\mu_1-\mu_2-\mu_3$, and $C$ only depends on $c_\robas$. Next, due to
  (\ref{eq:ninvoneplusqfirstestimate}) and the fact that the integral of $\mSb_b\circ\g$ over $[s_a,s_b]$ is small, it is clear that the integral
  appearing on the right hand side of (\ref{eq:mu123 est q two}) is bounded (and small if $s_b$ is close to $-\infty$). Since $\ell_{\rowo}\circ\g$ does not
  change much in the interval $[s_a,s_b]$, this means that $e^{\mu_{\rowo}\circ\g}$ is either essentially exponentially growing or exponentially decaying,
  depending on the value of $\ell_{\rowo}\circ\g+1=2\ell_{1}\circ\g$. If $\ell_1<0$, we obtain exponential growth backwards in time. If $\ell_1>0$, we
  obtain exponential decay. Since $\mSb_b\circ\g$ is bounded, we conclude that $\g^{1}_{23}\circ\g$ has to decay exponentially if $\ell_1\circ\g<0$,
  but that there are no significant bounds on $\g^{1}_{23}\circ\g$ in case $\ell_1\circ\g>0$.

  Finally, let us consider the vacuum setting. Due to (\ref{eq:chKeSpe}), (\ref{eq:OmegaplussumellAsqbdpnegsc}), (\ref{eq:lambdaAellAmSbest}) and the
  fact that $\mSb_b\geq 0$, there are constants $C>0$ and $\e>0$ such that
  \begin{equation}\label{eq:intermediate estimates g123 to zero}
    \textstyle{\sum}_A \ell_A^{2}\circ\g(s)-1\geq -\mSb_b\circ\g(s)-Ce^{\e s},\ \ \
    \ell_3\circ\g(s)\leq 1-\e_{\Spe}+Ce^{\e s}
  \end{equation}
  for all $s\leq 0$; note that $\lambda_3\leq-\e_{\Spe}$ is a consequence of (\ref{eq:chKeSpe}). Combining these estimates with the fact that the $\ell_A$ sum up to $1$ yields
  \begin{equation}\label{eq:ell one lb}
    \begin{split}
      -\ell_1\circ\g(s)\cdot\ell_2\circ\g(s) = & \ell_3\circ\g(s)[1-\ell_3\circ\g(s)]+\tfrac{1}{2}(\textstyle{\sum}_A\ell_A^2\circ\g(s)-1)\\
      \geq &  \e_{\Spe}\ell_3\circ\g(s)-\mSb_b\circ\g(s)/2-Ce^{\e s}
    \end{split}    
  \end{equation}
  for all $s\leq 0$ (where the value of $C$ changes from inequality to inequality). Since $\ell_3\geq 1/3$, the first term on the far right hand side of (\ref{eq:ell one lb})
  is bounded from below by $\e_\Spe/3$. Assuming $s_b$ to be close enough to $-\infty$, it can thus be ensured that the sum of the first and the last term on the far right hand side of
  (\ref{eq:ell one lb}) is bounded from below by $\e_\Spe/4$. Divide $[s_a,s_b]$ into $\ma_+$, consisting of the $s\in [s_a,s_b]$ such that $\mSb_b\circ\g(s)\geq \e_\Spe/10$ and
  $\ma_-=[s_a,s_b]-\ma_+$. Since the integral of $\mSb_b\circ\g$ over $[s_a,s_b]$ is small and $\mSb_b$ is non-negative, it is clear that $|\ma_+|$ (the Lebesgue measure of $\ma_+$)
  should be small. On $\ma_-$, the far right hand side of (\ref{eq:ell one lb}) is bounded from below by $\e_\Spe/5$. This means that $\ell_1\circ\g<0$ and $\ell_2\circ\g>0$ on $\ma_-$.
  Assuming $s_b$ to be close enough to $-\infty$, we can assume $\ell_2\circ\g\leq 1$; see the second inequality in (\ref{eq:intermediate estimates g123 to zero}). This means that on
  $\ma_-$, $\ell_1\circ\g\leq -\e_\Spe/5$. Combining (\ref{eq:mu123 est q two}) with the above observations yields
  \[
    \mu_{\rowo}\circ\g(s_a)\geq\mu_{\rowo}\circ\g(s_b)-\textstyle{\int}_{s_{a}}^{s_b}2\ell_1\circ\g(s)ds-C
    \geq \mu_{\rowo}\circ\g(s_b)+\e_\Spe(s_b-s_a)/5-C.
  \]
  This means that $e^{\mu_{\rowo}\circ\g}$ is exponentially growing in the direction of the singularity. Due to the fact that $\mSb_b\circ\g$ is bounded,
  we conclude that $\g_{23}^1\circ\g$ has to decrease exponentially. 
\end{remark}
\begin{proof}[Proof of Theorem~\ref{thm:main bounce result}]
  Consider (\ref{eq:renormalisedbSest}) and (\ref{eq:renormalisedRicAAest}). We need to estimate terms of the form
  \begin{equation}\label{eq:termsinbSandbRtobestimated}
    e^{2\mu_{D}-2\mu_{A}-2\mu_{C}}(\g^{D}_{AC})^{2}.
  \end{equation}
  When doing so, it is convenient to recall (\ref{eq:bmuAmbmuBlowbd}). If $D>1$, then, since $A\neq C$, one of $A,C$ has to equal either $1$ or $D$.
  If one of them equals $D$, we have exponential decay as in (\ref{eq:emtwomuAscsqest}). If one of them, say $A$, equals $1$, we need to estimate
  \[
  e^{2\mu_{D}-2\mu_{A}-2\mu_{C}}(\g^{D}_{AC})^{2}=e^{2\mu_{D}-2\mu_{1}}\cdot[e^{-2\mu_{C}}(\g^{D}_{AC})^{2}].
  \]
  The first factor we can estimate by appealing to (\ref{eq:bmuAmbmuBlowbd}) and the second factor we can estimate by appealing to estimates of
  the form (\ref{eq:emtwomuAscsqest}). What remains to be considered is thus the case that $D=1$. If one of $A,C$ equals $1$, we again obtain exponential
  decay. Thus $\{A,C\}=\{2,3\}$ (since $A$ and $C$ have to be different from $1$ and different from each other), so that the only case of interest is
  $D=1$ and $\{A,C\}=\{2,3\}$. Combining observations of this type with (\ref{eq:renormalisedbSest}) and (\ref{eq:renormalisedRicAAest}) yields
  \begin{align}
    \left|\theta^{-2}\bS+\textstyle{\frac{1}{2}}e^{2\mu_{1}-2\mu_{2}-2\mu_{3}}(\g^{1}_{23})^{2}\right|
    \leq C_{\bS,3}\ldr{\varrho}^{2(2\cweight+1)}e^{2\e_{\Spe}\varrho},\label{eq:bSonetwothree}\\
    \left|\theta^{-2}\bR^{1}_{\phantom{1}1}-\textstyle{\frac{1}{2}}e^{2\mu_{1}-2\mu_{2}-2\mu_{3}}(\g^{1}_{23})^{2}\right|
    \leq C_{\bR,3}\ldr{\varrho}^{2(2\cweight+1)}e^{2\e_{\Spe}\varrho},\label{eq:bRoneone}\\
    \left|\theta^{-2}\bR^{2}_{\phantom{2}2}+\textstyle{\frac{1}{2}}e^{2\mu_{1}-2\mu_{2}-2\mu_{3}}(\g^{1}_{23})^{2}\right|
    \leq C_{\bR,3}\ldr{\varrho}^{2(2\cweight+1)}e^{2\e_{\Spe}\varrho},\label{eq:bRtwotwo}\\
    \left|\theta^{-2}\bR^{3}_{\phantom{3}3}+\textstyle{\frac{1}{2}}e^{2\mu_{1}-2\mu_{2}-2\mu_{3}}(\g^{1}_{23})^{2}\right|
    \leq C_{\bR,3}\ldr{\varrho}^{2(2\weight+1)}e^{2\e_{\Spe}\varrho}\label{eq:bRthreethree}
  \end{align}
  on $M_{-}$, where $C_{\bS,3}=C_{\bS,3,0}\theta_{0,-}^{-2}$, $C_{\bR,3}=C_{\bR,3,0}\theta_{0,-}^{-2}$ and $C_{\bS,3,0}$ and $C_{\bR,3,0}$ only depend on $c_{\cweight,1}$
  and $(\bM,\bge_{\refer})$. Summing up,
  \begin{equation*}
    \theta^{-2}|\bR^{1}_{\phantom{1}1}+\bS|+\theta^{-2}|\bR^{2}_{\phantom{2}2}-\bS|+\theta^{-2}|\bR^{3}_{\phantom{3}3}-\bS|
    \leq C_{\bR,\rotot}\ldr{\varrho}^{2(2\cweight+1)}e^{2\e_{\Spe}\varrho}
  \end{equation*}
  on $M_{-}$, where $C_{\bR,\rotot}=C_{\bR,\rotot,0}\theta_{0,-}^{-2}$ and $C_{\bR,\rotot,0}$ only depends on $c_{\cweight,1}$ and $(\bM,\bge_{\refer})$.
  Combining this estimate with (\ref{seq:ellpm}) and (\ref{eq:hUelldominantterms}) yields the conclusion of the theorem.
\end{proof}

\subsection{Refined analysis}

It is of interest to estimate all the components of the rescaled Ricci tensor in the case that $n=3$. Since we already have the estimates 
(\ref{eq:bRoneone})--(\ref{eq:bRthreethree}), we focus on the off-diagonal components. The main purpose of these estimates is to demonstrate
that the perhaps most unnatural assumption of our framework, namely the condition that $\hml_{U}\mK$ satisfy a weak off-diagonal exponential bound;
see Definition~\ref{def:offdiagonalexpdec}; can at this stage be improved. 

\begin{lemma}\label{lemma:goodoffdiagtermsbRABbnablaAB}
  Let $0\leq\cweight\in\ro$ and assume that the standard assumptions; see Definition~\ref{def:standardassumptions}; as well as the $(\cweight,1)$-supremum
  assumptions are fulfilled. Assume, moreover, that $n=3$ and let $A,B\in \{1,2,3\}$ be such that $A\neq B$ and $B\neq 1$. Then 
  \begin{equation}\label{eq:goodoffdiagtermsbRAB}
    \theta^{-2}|\bR^{A}_{\phantom{A}B}|\leq C_{\bR,\mrod}\ldr{\varrho}^{2(2\cweight+1)}e^{2\e_{\Spe}\varrho}
  \end{equation}
  on $M_{-}$, where $C_{\bR,\mrod}=C_{\bR,\mrod,0}\theta_{0,-}^{-2}$ and $C_{\bR,\mrod,0}$ only depends on $c_{\cweight,1}$ and $(\bM,\bge_{\refer})$. Moreover,
  there is a constant $C_{N,\mrod}$ such that
  \begin{equation}\label{eq:goodoffdiagtermsNsd}
    \theta^{-2}N^{-1}|\bnabla^{A}\bnabla_{B}N|\leq C_{N,\mrod}\ldr{\varrho}^{3\cweight+1}e^{2\e_{\Spe}\varrho}
  \end{equation}
  on $M_{-}$, where $C_{N,\mrod}=C_{N,\mrod,0}\theta_{0,-}^{-2}$ and $C_{N,\mrod,0}$ only depends on $c_{\cweight,1}$ and $(\bM,\bge_{\refer})$.
\end{lemma}
\begin{proof}
  In order to prove (\ref{eq:goodoffdiagtermsbRAB}), let us return to Corollary~\ref{cor:rescaledRiccicurvatureformwoderoflntheta} below, in particular
  (\ref{eq:mSIABdef})--(\ref{eq:mSVABdef}). Most of the terms can be estimated as in the proof of Proposition~\ref{prop:normalisedbSbd}. However, there are
  some exceptions. For that reason, we, from now on, focus on the terms that we do not a priori know to decay exponentially. Note, to this end, that
  \[
  2\mS_{\mrI,B}^{A} \approx -\textstyle{\sum}_{C}e^{2\mu_{B}-2\mu_{A}-2\mu_{C}}X_{C}(\g^{B}_{AC}),
  \]
  where the sign $\approx$ is used to indicate that the equality holds up to terms that tend to zero exponentially. In particular, the only possibility for
  $\mS_{\mrI,B}^{A}$ not to decay to zero exponentially is $B=1$, $A\in \{2,3\}$. In case $B=1$, potentially all of the terms appearing on the right hand
  side of (\ref{eq:mSIIABdef}) are problematic. However, if $B\in\{2,3\}$, the first and third terms on the right hand side are clearly exponentially
  decaying. In the case of the first term, this is clear since if $B\notin \{A,C\}$, then, since $A\neq C$, $1\in \{A,C\}$. The argument in the case of the
  third term is similar. Consider the second term on the right hand side of (\ref{eq:mSIIABdef}). In order for this term to be non-zero, $C$ has to be
  different from $B$. In order for it not to be exponentially decaying, $C$ has to be different from $1$. Since $B\in\{2,3\}$, we conclude that
  $\{B,C\}=\{2,3\}$. Similarly, $A$ has to be different from $1$ and different from $C$. This means that $A=B$, contradicting the assumption that $A\neq B$.
  Turning to (\ref{eq:mSIIIABdef})--(\ref{eq:mSVABdef}), all the terms appearing on the right hand side decay to zero exponentially. The estimate
  (\ref{eq:goodoffdiagtermsbRAB}) follows. 

  The proof of (\ref{eq:goodoffdiagtermsNsd}) is similar to the proof of (\ref{eq:normalisedDeltabgelnN}), keeping Lemma~\ref{lemma:contributionsfromlapse}
  in mind. 
\end{proof}

Due to this lemma, we can estimate the corresponding components of $\hml_{U}\mK$. 

\begin{cor}\label{cor:improvedoffdiagexpbd}
  Let $0\leq\cweight\in\ro$ and assume that the standard assumptions, see Definition~\ref{def:standardassumptions}, as well as the $(\cweight,1)$-supremum
  assumptions are fulfilled. Assume, moreover, that Einstein's equations with a cosmological constant $\Lambda$ (\ref{eq:EE}) are satisfied, that $n=3$,
  and let $A,B\in \{1,2,3\}$ be such that $A\neq B$ and $B\neq 1$.
  Assume, finally, that there are constants $C_{\mfp,\mrod}$, $\e_{\mfp}>0$ such that
  \begin{equation}\label{normmfpoffdiagest}
    \theta^{-2}|\mfp^{A}_{\phantom{A}B}|\leq C_{\mfp,\mrod}e^{2\e_{\mfp}\varrho}
  \end{equation}
  on $M_{-}$, where $C_{\mfp,\mrod}=C_{\mfp,\mrod,0}\theta_{0,-}^{-2}$ and $\mfp$ is introduced in (\ref{eq:mfpmcPdef}). Then
  \begin{equation}\label{lemma:goodoffdiagtermshmlUmKABbnablaAB}
    |(\hml_{U}\mK)^{A}_{\phantom{A}B}|\leq C_{D,\mrod}\ldr{\varrho}^{2(2\cweight+1)}e^{2\e_{D}\varrho}
  \end{equation}
  on $M_{-}$, where $\e_{D}:=\min\{\e_{\Spe},\e_{\mfp}\}$, $C_{D,\mrod}=C_{D,\mrod,0}\theta_{0,-}^{-2}$ and $C_{D,\mrod,0}$ only depends on $c_{\cweight,1}$,
  $C_{\mfp,\mrod,0}$ and $(\bM,\bge_{\refer})$. 
\end{cor}
\begin{remark}
The estimate (\ref{lemma:goodoffdiagtermshmlUmKABbnablaAB}) improves the assumption that $\hml_{U}\mK$ satisfy a weak off-diagonal exponential bound.
\end{remark}
\begin{proof}
The estimate (\ref{lemma:goodoffdiagtermshmlUmKABbnablaAB}) follows from (\ref{eq:goodoffdiagtermsbRAB}), (\ref{eq:goodoffdiagtermsNsd}),
(\ref{normmfpoffdiagest}) and (\ref{eq:mlUmKwithEinstein}). 
\end{proof}

\section{Quiescent setting}\label{section:quiescent}

One case which is of particular interest is when $\mK$ converges in the direction of the singularity. This is referred to as the \textit{convergent} or
\textit{quiescent} setting. 
In order for this to happen, we expect the expansion normalised normal derivative of $\mK$ to have to converge to zero. Considering
(\ref{eq:mlUmKwithEinstein}), we thus expect it to be necessary for the rescaled spatial curvature quantities, such as $\theta^{-2}\bS$, to converge to
zero. Considering (\ref{eq:renormalisedbSest}), we would thus, generically, expect it to be necessary for $\mu_{A}-\mu_{B}-\mu_{C}$, $B\neq C$, to tend to
$-\infty$ as $\varrho$ tends to $-\infty$. For this reason, we begin by deriving conditions that ensure that $\mu_{A}-\mu_{B}-\mu_{C}$, $B\neq C$, tends to
$-\infty$ linearly in $\varrho$. 

\subsection{Asymptotic behaviour of $\mu_{A}-\mu_{B}-\mu_{C}$, $B\neq C$}\label{ssection:asbehmuAetc}

Consider (\ref{eq:ninvoneplusqfirstestimate}). If $\mu_{A}-\mu_{B}-\mu_{C}$, $B\neq C$, tends to $-\infty$ linearly in $\varrho$, then $\mSb_{b}$ tends to
zero exponentially; note that $\g^{A}_{BC}$ grows at worst as a power of $\ldr{\varrho}$ - see the proof of (\ref{eq:emtwomuAscsqest}). Combining this
observation with (\ref{eq:ninvoneplusqfirstestimate}) yields the conclusion that $q-(n-1)$ converges to zero
exponentially. Similarly, if we make particular assumptions concerning the eigenvalues $\ell_{A}$ and assume that $q-(n-1)$ converges to zero exponentially,
then we can conclude that $\mu_{A}-\mu_{B}-\mu_{C}$, $B\neq C$, tends to $-\infty$ linearly in $\varrho$ (this is the purpose of
Lemma~\ref{lemma:muAmmuBmmuCrelquiescent} below). On the other hand, it is not so clear how to deduce conclusions without making any of these assumptions. 
In practice, when deriving asymptotics in the context of a proof of global non-linear stability, it is necessary to make bootstrap assumptions that are then
strengthened. The following lemma should be understood as being part of such a bootstrap argument: We assume $q$ to converge, exponentially, to $n-1$, and
find a condition on the $\ell_{A}$ that ensures that $\mu_{A}-\mu_{B}-\mu_{C}$, $B\neq C$, tends to $-\infty$ linearly in $\varrho$. Once this conclusion
has been obtained, we strengthen the estimates of $q-(n-1)$ and the conclusions concerning $\ell_{A}$. Two examples of settings in which solutions
exhibit this type of behaviour are provided by \cite{aarendall,daetal}; see also Subsection~\ref{ssection:summary discussion ex}. 

\begin{lemma}\label{lemma:muAmmuBmmuCrelquiescent}
  Let $0\leq\cweight\in\ro$ and assume that the standard assumptions are fulfilled; see Definition~\ref{def:standardassumptions}. Assume, moreover, that
  there are constants $K_{q}$ and $0<\e_{\que}<1$ such that
  \begin{equation}\label{eq:qexpconvergence}
    |q-(n-1)|\leq K_{q}e^{\e_{\que}\varrho}
  \end{equation}
  on $M_{-}$. Let $\g$ be an integral curve of $\hU$ with $\g(0)\in\bM\times \{t_{0}\}$ and let $J_{-}=\g^{-1}(M_{-})$. Then there is a constant $K_{\lambda}$
  such that
  \begin{equation}\label{eq:muAalonggammaintegralestimate}
    \big|\mu_{A}\circ\g(s)-\ln\theta\circ\g(0)+\textstyle{\int}_{s}^{0}[\ell_{A}\circ\g(u)-1]du\big|\leq K_{\lambda}
  \end{equation}
  for all $A$ and all $s\in J_{-}$, where $K_{\lambda}$ only depends on $c_{\robas}$, $\e_{\que}$ and $K_{q}$. Assume, in addition, that there is an
  $\e_{p}\in (0,1)$ such that
  \begin{equation}\label{eq:quiescentregime}
    \ell_{1}-\ell_{n-1}-\ell_{n}+1\geq\e_{p}.
  \end{equation}
  Then, assuming $A,B,C\in\{1,\dots,n\}$ and $B\neq C$,
  \begin{equation}\label{eq:muABCdiffest}
    \mu_{A}-\mu_{B}-\mu_{C}\leq \e_{p}\varrho-\ln\theta_{0,-}+K
  \end{equation}
  on $M_{-}$, where $K$ only depends on $c_{\robas}$, $\e_{\que}$ and $K_{q}$. 
\end{lemma}
\begin{remark}
  In the case that $n=3$, the condition (\ref{eq:quiescentregime}) is equivalent to $2\ell_{1}\geq\e_{p}$; i.e., to the condition that $\mK$ is
  positive definite with a uniform positive lower bound. 
\end{remark}
\begin{remark}
  Depending on dimension and presence/absence of matter fields, it is not always possible to fulfill (\ref{eq:quiescentregime}). We discuss this topic
  further in Subsection~\ref{ssection:quiescentregimes} below.
\end{remark}
\begin{remark}\label{remark:muAalonggamma}
  Due to (\ref{eq:hUnlnthetamomqbas}) and (\ref{eq:bmuAestalonggammagenA}) below, it follows that if $\g$ is an integral curve of $\hU$,
  then
  \[
  \big|\mu_{A}\circ\g(s)-\ln\theta\circ\g(0)+\textstyle{\int}_{s}^{0}\left[\ell_{A}-\tfrac{1}{n}(1+q)\right]\circ\g(u)du\big|\leq K
  \]
  for $s\in J_{-}$, where $K$ only depends on $c_{\robas}$. 
\end{remark}
\begin{proof}
  Due to (\ref{eq:bmuAeqintellA}) and (\ref{eq:bmuoneeqintellone}), we know that
  \begin{equation}\label{eq:bmuAestalonggammagenA}
    \big|\bmu_{A}\circ\g(s)+\textstyle{\int}_{s}^{0}\ell_{A}\circ\g(u)du\big|\leq K
  \end{equation}
  for all $A$ and all $s\in J_{-}$, where $K$ only depends on $c_{\robas}$. On the other hand, due to the assumption concerning $q$ and
  (\ref{eq:hUnlnthetamomqbas}), it follows that
  \[
  \big|\tfrac{d}{ds}\ln\theta\circ\g(s)+1\big|=n^{-1}|q\circ\g(s)-(n-1)|\leq n^{-1}K_{q}e^{\e_{\que}\varrho\circ\g(s)}\leq 2n^{-1}K_{q}e^{\e_{\que}s}
  \]
  for all $s\in J_{-}$, where we appealed to (\ref{eq:varrhosequivalencestmt}). Integrating this estimate and combining the result with
  (\ref{eq:bmuAestalonggammagenA}) yields (\ref{eq:muAalonggammaintegralestimate}). Next, let $A,B,C\in\{1,\dots,n\}$ and $B\neq C$. Then
  one of $B,C$ has to be $\leq n-1$. Assume, without loss of generality, that $B\leq n-1$. Then
  \begin{equation}\label{eq:muAmmuBmmuCottest}
    \begin{split}
      \mu_{A}-\mu_{B}-\mu_{C} = & \mu_{1}-\mu_{n-1}-\mu_{n}+(\mu_{A}-\mu_{1})+(\mu_{n-1}-\mu_{B})+(\mu_{n}-\mu_{C})\\
      \leq & \mu_{1}-\mu_{n-1}-\mu_{n}+3M_{\rodiff}
    \end{split}
  \end{equation}
  on $M_{-}$, where we appealed to (\ref{eq:bmuAmbmuBlowbd}) and $M_\rodiff$ is the constant appearing in (\ref{eq:bmuAmbmuBlowbd}).
  On the other hand, (\ref{eq:muAalonggammaintegralestimate}) yields
  \begin{equation}\label{eq:canmudifference}
    \begin{split}
      (\mu_{1}-\mu_{n-1}-\mu_{n})\circ\g(s)  \leq & -\textstyle{\int}_{s}^{0}[(\ell_{1}-\ell_{n-1}-\ell_{n})\circ\g(u)+1]du-\ln\theta\circ\g(0)+3K_{\lambda}\\
      \leq & \e_{p}s-\ln\theta\circ\g(0)+3K_{\lambda}
    \end{split}
  \end{equation}
  for all $s\in J_{-}$. Combining the fact that this estimate holds for all integral curves of $\hU$; (\ref{eq:varrhosequivalencestmt}); and
  (\ref{eq:muAmmuBmmuCottest}) yields the conclusion that (\ref{eq:muABCdiffest}) holds.
\end{proof}

\subsection{Quiescent regimes}\label{ssection:quiescentregimes}

It is of interest to determine the regime in which the expression 
\[
F(\ell):=\ell_{1}-\ell_{n-1}-\ell_{n}+1
\]
has a positive lower bound, where $\ell:=(\ell_{1},\dots,\ell_{n})$. This is due to the fact that if $F$ is bounded from below by a positive constant,
then Lemma~\ref{lemma:muAmmuBmmuCrelquiescent} implies that the left hand side of (\ref{eq:muABCdiffest}) diverges linearly to $-\infty$ in
the direction of the singularity, a condition expected to be necessary in order for $\mK$ to converge (i.e., in order to have quiescence);
see the discussion at the beginning of Section~\ref{section:quiescent}. Here it is understood that $\ell_{1}\leq\ell_{2}\leq\cdots\leq\ell_{n}$ 
(even though we are mainly interested in the case of strict inequalities) and that 
\[
\textstyle{\sum}_{A}\ell_{A}=1,\ \ \
\textstyle{\sum}_{A}\ell_{A}^{2}\leq 1,
\]
where the second assumption is a consequence of (\ref{eq:OmegaplussumellAsqbdpnegsc}); due to (\ref{eq:OmegaplussumellAsqbdpnegsc}), this inequality has
to hold asymptotically. In general, the choice $\ell_{i}=1/n$, $i=1,\dots,n$, ensures that these conditions are satisfied and that $F(\ell)>0$,
assuming $n\geq 2$. However, depending on the matter model, this choice may not be consistent with the Hamiltonian constraint
(\ref{eq:reformulatedandrenormalisedHamcon}). Note, in particular, that when $F(\ell)>0$, then $\theta^{-2}\bS$ converges to zero exponentially. Combining
this observation with (\ref{eq:reformulatedandrenormalisedHamcon}) yields the conclusion that 
\[
2\Omega+\textstyle{\sum}_{A}\ell_{A}^{2}-1
\]
converges to zero exponentially. In particular, when $\Omega=0$,
\begin{equation}\label{eq:vacuumcondition}
\textstyle{\sum}_{A}\ell_{A}^{2}-1
\end{equation}
converges to zero exponentially. This is not consistent with $\ell_{i}=1/n$, $i=1,\dots,n$. 

\textbf{The vacuum setting.} Due to the above, fulfilling the condition $F(\ell)>0$ is not a problem if there is a scalar field present. 
In order to illustrate that there is a problem in the vacuum setting, assume that $n=3$. In that case, $F(\ell)=2\ell_{1}$. In other words, in the case
of $n=3$, $F(\ell)$ is strictly positive if and only if $\ell_{1}>0$. However, $\ell_{1}>0$ is equivalent to $\mK$ being positive definite. Moreover,
$\ell_{1}>0$ implies that $0<\ell_{i}<1$ for $i=1,\dots,n$. Thus
\[
\textstyle{\sum}_{A}\ell_{A}^{2}<\textstyle{\sum}_{A}\ell_{A}=1.
\]
In particular, the condition that $\mK$ be positive definite is not consistent with vacuum. On the other hand, in higher dimensions,
it turns out that $F(\ell)>0$ is consistent with vacuum. In order to see this, let $K_{o}$ be the set of $\ell\in\rn{n}$ such that 
\begin{equation}\label{eq:bothKasnerrelations}
\textstyle{\sum}_{A}\ell_{A}=1,\ \ \
\textstyle{\sum}_{A}\ell_{A}^{2}=1
\end{equation}
and
\begin{equation}\label{eq:ellorder}
\ell_{1}\leq \cdots\leq \ell_{n}
\end{equation}
are satisfied. Note that $K_{o}$ is
non-empty and compact. By the above observations, there is, for $n=3$, no $\ell\in K_{o}$ such that $F(\ell)>0$. On the other hand, regardless
of the value of $n$, $\ell=(0,\dots,0,1)$ belongs to $K_{o}$ and is such that $F(\ell)=0$. It is of interest to find the minimum $n$ such that 
there is an $\ell\in K_{o}$ satisfying $F(\ell)>0$. This was done in \cite{Henneauxetal}, and the answer is $n=10$. The precise statement is the
following (and since the justification given in \cite{Henneauxetal} is somewhat terse, we prove the below lemma in
Appendix~\ref{section:condquiescentregimes}). 

\begin{lemma}\label{lemma:vacuumquiescent}
  For each $3\leq n\in\zo$, let $K_{o}$ be defined as above. If $3\leq n\leq 9$, there is no $\ell\in K_{o}$ such that $F(\ell)>0$. If
  $n\geq 10$, there is an $\ell\in K_{o}$ such that $F(\ell)>0$. 
\end{lemma}

\subsection{Estimating the curvature in $C^{k}$}

Combining the conclusions of Lemma~\ref{lemma:muAmmuBmmuCrelquiescent} with the estimates derived in \cite{RinWave} yields estimates for the rescaled
spatial curvature quantities. In particular, we are interested in the family $\bmR$ of $(1,1)$-tensor fields on $\bM$ and the rescaled spatial scalar
curvature, defined by 
\[
\bmR^{i}_{\phantom{i}j}:=\theta^{-2}\bR^{i}_{\phantom{i}j},\ \ \
\bmS:=\theta^{-2}\bS
\]
respectively, where $\bR$ denotes the Ricci curvature of the metric $\bge$. In order to derive conclusions concerning these quantities, let us begin by
making $C^{k}$-assumptions. 

\begin{lemma}\label{lemma:CkRicciestimates}
  Let $0\leq\cweight\in\ro$ and $1\leq l\in\zo$ and assume that the standard assumptions; see Definition~\ref{def:standardassumptions}; as well as the
  $(\cweight,l)$-supremum assumptions are fulfilled. Assume, moreover, that there are constants $K_{q}$ and $0<\e_{\que}<1$ such that
  (\ref{eq:qexpconvergence}) holds and an $\e_{p}\in (0,1)$ such that (\ref{eq:quiescentregime}) holds. Then there is a constant $K_{\mR,l}$ such that
  \begin{equation}\label{eq:estofRicciqugeo}
    \textstyle{\sum}_{k\leq l-1}\ldr{\varrho}^{-(k+3)\cweight}|\bD^{k}\bmR|_{\bge_{\refer}} \leq K_{\mR,l}\ldr{\varrho}^{l\cweight+l+1}e^{2\e_{R}\varrho}
  \end{equation}
  on $M_{-}$, where $\e_{R}:=\min\{\e_{p},\e_{\Spe}\}$, $K_{\mR,l}=K_{\mR,l,0}\theta_{0,-}^{-2}$ and $K_{\mR,l,0}$ only depends on $c_{\cweight,l}$, $\e_{\que}$,
  $K_{q}$ and $(\bM,\bge_{\refer})$. 
\end{lemma}
\begin{remark}\label{remark:Ckscalarcurvatureestimate}
  The rescaled spatial curvature, $\bmS$, as well as $\mSb_{b}$ satisfy estimates analogous to (\ref{eq:estofRicciqugeo}). 
\end{remark}
\begin{remark}\label{remark:Ckscalarcurvatureestimatesymmetry}
  A similar conclusion holds in case we replace the requirement that (\ref{eq:qexpconvergence}) and (\ref{eq:quiescentregime}) hold with the requirement
  that $\g^{A}_{BC}=0$ for all $A<\min\{B,C\}$ (this condition can be fulfilled for symmetry classes of solutions to Einstein's equations; see, e.g.,
  Subsections~\ref{ssection:revisitsphom} and \ref{ssection: t3 Gowdy}). In fact, in that case,
  \begin{equation}\label{eq:estofRicciqugeosymmetry}
    \textstyle{\sum}_{k\leq l-1}\ldr{\varrho}^{-(k+3)\cweight}|\bD^{k}\bmR|_{\bge_{\refer}} \leq K_{\mR,l}\ldr{\varrho}^{l\cweight+l+1}e^{2\e_{\Spe}\varrho}
  \end{equation}
  on $M_{-}$, where $K_{\mR,l}=K_{\mR,l,0}\theta_{0,-}^{-2}$ and $K_{\mR,l,0}$ only depends on $c_{\cweight,l}$ and $(\bM,\bge_{\refer})$. The proof of this statement
  is similar to, but simpler than, the proof below (in that some of the terms we need to estimate below vanish in this case). Moreover, similar estimates
  hold for $\bmS$. 
\end{remark}
\begin{proof}
  The proof is based on the formulae derived in Corollary~\ref{cor:rescaledRiccicurvatureformwoderoflntheta}, in particular
  (\ref{eq:mSIABdef})--(\ref{eq:mSVABdef}). The terms appearing on the right hand sides of (\ref{eq:mSIABdef})--(\ref{eq:mSVABdef})
  are of two different types. One type takes the form
  \[
  e^{h}X_{A}(f),
  \]
  where $h$ either equals $-2\mu_{D}$ for some choice of $D$ or equals $2\mu_{B}-2\mu_{C}-2\mu_{D}$ for some choice of $B$, $C$, $D$ such that $C\neq D$;
  and $f$ either equals $X_{B}(\bmu_{C})$ for some choice of $B$, $C$, or equals $\g^{B}_{CD}$ for some choice of $B$, $C$, $D$. We refer to terms of this type
  as terms of type one. The second type takes the form
  \[
  e^{h}f_{1}f_{2},
  \]
  where $h$ either equals $-2\mu_{A}$ for some choice of $A$, or equals $2\mu_{B}-2\mu_{C}-2\mu_{D}$ for some choice of $B$, $C$, $D$ such that $C\neq D$;
  and the $f_{i}$ either equal $X_{A}(\bmu_{B})$ for some choice of $A$, $B$, or equal $\g^{A}_{BC}$ for some choice of $A$, $B$, $C$. We refer to terms of
  this type as terms of type two. Before estimating terms of these different types, we need to estimate the constituents. 

  \textbf{Estimating the constituents.} Note, to begin with, that (\ref{eq:bmuAmclmfwestEi}) and (\ref{eq:muAClestimatermk app}) yield
  \begin{equation}\label{eq:muAClestimatermk}
    \|\mu_{A}(\cdot,\tau)\|_{\mc^{\bfl_{1}}_{\bbE,\weight}(\bM)}+\|\bmu_{A}(\cdot,\tau)\|_{\mc^{\bfl_{1}}_{\bbE,\weight}(\bM)}\leq C_{\mu,l}\ldr{\tau}
  \end{equation}
  on $I_{-}$ for all $A$, where $\bfl_{1}:=(1,l+1)$, $C_{\mu,l}$ only depends on $c_{\cweight,l}$ and $(\bM,\bge_{\refer})$, and we use the notation
  introduced in (\ref{eq:mtmClbbEbS}). In particular, this estimate implies that
  \begin{equation}\label{eq:EbfIbmuAmuAest}
    \textstyle{\sum}_{j=1}^{l+1}\sum_{|\bfI|=j}\ldr{\varrho}^{-(j+1)\cweight}(|E_{\bfI}\bmu_{A}|+|E_{\bfI}\mu_{A}|)\leq C_{\mu,l}\ldr{\varrho}
  \end{equation}
  on $M_{-}$, where we used the fact that $\ldr{\tau}$ and $\ldr{\varrho}$ are equivalent; see (\ref{eq:ldrrho ldr tau equiv}).
  Here $E_{\bfI}:=E_{I_1}\cdots E_{I_k}$ and $|\bfI|=k$ if $\bfI=(I_1,\dots,I_k)$. However, we also need to estimate $E_{\bfI}X_{A}(\bmu_{B})$
  and $E_{\bfI}X_{A}X_{B}(\bmu_{C})$. Note, to this end, that
  \begin{align}
    X_{A}(\bmu_{B}) = & \omega^{i}(X_{A})E_{i}(\bmu_{B}),\label{eq:XAmuBexpansion}\\
    X_{A}X_{B}(\bmu_{C}) = & \omega^{j}(X_{A})(\bD_{E_{j}}\omega^{i})(X_{B})E_{i}(\bmu_{C})
    +\omega^{j}(X_{A})\omega^{i}(\bD_{E_{j}}X_{B})E_{i}(\bmu_{C})\label{eq:XAXBmuCexpansion}\\
    & +\omega^{i}(X_{B})\omega^{j}(X_{A})E_{j}E_{i}(\bmu_{C}).\nonumber
  \end{align}
  Combining (\ref{eq:XAmuBexpansion}) with (\ref{eq:bDbfAellAetcpteststmtEi}), it follows that $|E_{\bfI}X_{A}(\bmu_{B})|$ can (up to constants depending only
  on $c_{\robas}$, $(\bM,\bge_{\refer})$ and $|\bfI|$) be estimated by a linear combination of terms of the form
  \[
  |\bD^{m_{1}}\mK|_{\bge_{\refer}}\cdots |\bD^{m_{k}}\mK|_{\bge_{\refer}}|E_{\bfJ}E_{i}\bmu_{B}|,
  \]
  where $m_{1}+\dots+m_{k}+|\bfJ|\leq |\bfI|$ and $m_{i}\neq 0$. Similarly, $|E_{\bfI}X_{A}X_{B}(\bmu_{C})|$ can (up to constants depending only on $c_{\robas}$,
  $(\bM,\bge_{\refer})$ and $|\bfI|$) be estimated by a linear combination of terms of the form
  \begin{equation}\label{eq:EbfIXAXBbmuCdecompterms}
    |\bD^{m_{1}}\mK|_{\bge_{\refer}}\cdots |\bD^{m_{k}}\mK|_{\bge_{\refer}}|E_{\bfJ}E_{i}\bmu_{B}|,
  \end{equation}
  where $m_{1}+\dots+m_{k}+|\bfJ|\leq |\bfI|+1$ and $m_{i}\neq 0$. Combining these observations with (\ref{eq:EbfIbmuAmuAest}) and the assumptions
  yields the conclusion that for $|\bfI|\leq l$ and $|\bfJ|+1\leq l$, 
  \begin{equation}\label{eq:EbfIXAbmuBXAXBbmuC}
    \ldr{\varrho}^{-(|\bfI|+2)\cweight}|E_{\bfI}X_{A}(\bmu_{B})|\leq C_{a}\ldr{\varrho},\ \ \
    \ldr{\varrho}^{-(|\bfJ|+3)\cweight}|E_{\bfJ}X_{A}X_{B}(\bmu_{C})|\leq C_{a}\ldr{\varrho}
  \end{equation}
  on $M_{-}$, where $C_{a}$ only depends on $c_{\cweight,l}$ and $(\bM,\bge_{\refer})$. 

  Next, we need to estimate derivatives of the structure constants. However, combining (\ref{eq:gammaCABLeviCivita}) and
  (\ref{eq:bDbfAellAetcpteststmtEi}), it can be deduced that
  \[
  |E_{\bfI}(\g^{A}_{BC})|\leq C_{\g}\textstyle{\sum}_{m=1}^{j+1}\mfP_{\mK,m},
  \]
  where the notation $\mfP_{\mK,m}$ is introduced in Definition~\ref{def:mfPmKhN}; $j:=|\bfI|$; and $C_{\g}$ only depends on $c_{\robas}$, $j$ and
  $(\bM,\bge_{\refer})$. Due to reformulations such as the one made
  in connection with (\ref{eq:XDgammaCABreformulatedterm}), it can similarly be deduced that 
  \begin{equation}\label{eq:EbfIXDgABCitomK}
    |E_{\bfI}[X_{D}(\g^{A}_{BC})]|\leq C_{\g}\textstyle{\sum}_{m=1}^{j+2}\mfP_{\mK,m},
  \end{equation}
  where $j:=|\bfI|$; and $C_{\g}$ only depends on $c_{\robas}$, $j$ and $(\bM,\bge_{\refer})$. Combining these estimates with the assumptions yields
  \begin{equation}\label{eq:EbfIgABCEbfIXDgABC}
    \ldr{\varrho}^{-(|\bfI|+1)\cweight}|E_{\bfI}(\g^{A}_{BC})|\leq C_{a},\ \ \
    \ldr{\varrho}^{-(|\bfJ|+2)\cweight}|E_{\bfJ}[X_{D}(\g^{A}_{BC})]|\leq C_{b}
  \end{equation}
  on $M_{-}$ for $|\bfI|\leq l$ and $|\bfJ|+1\leq l$, where $C_{a}$ and $C_{b}$ only depend on $c_{\cweight,l}$ and $(\bM,\bge_{\refer})$. 

  \textbf{Estimating terms of type one.} Applying $E_{\bfI}$ to a term of type one yields a linear combination of terms of the form
  \[
  e^{h}E_{\bfI_{1}}(h)\cdots E_{\bfI_{m}}(h)E_{\bfJ}X_{A}(f),
  \]
  where $\bfI_{i}\neq 0$, $i=1,\dots,m$, and $|\bfI_{1}|+\dots+|\bfI_{m}|+|\bfJ|=|\bfI|$. On the other hand, a term of this form can be estimated by
  \begin{equation}\label{eq:termstobeestimatedtypeoneterm}
    \begin{split}
      & \ldr{\varrho}^{-(|\bfI|+3)\cweight}e^{h}|E_{\bfI_{1}}(h)\cdots E_{\bfI_{m}}(h) E_{\bfJ}X_{A}(f)|\\
      \leq & e^{h}\ldr{\varrho}^{-|\bfI_{1}|\cweight}|E_{\bfI_{1}}(h)|\cdots \ldr{\varrho}^{-|\bfI_{m}|\cweight}|E_{\bfI_{m}}(h)|
      \ldr{\varrho}^{-(|\bfJ|+3)\cweight}|E_{\bfJ}X_{A}(f)|.
    \end{split}
  \end{equation}
  In case $f=X_{B}(\bmu_{C})$ and $|\bfI|\leq l-1$, then the right hand side can be estimated by
  \[
  C_{a}e^{h}\ldr{\varrho}^{(l-1)\cweight+l}\leq C_{b}\theta_{0,-}^{-2}\ldr{\varrho}^{(l-1)\cweight+l}e^{2\e_{R}\varrho}
  \]
  on $M_{-}$, where $\e_{R}$ is introduced in the statement of the lemma; $C_{a}$ only depends on $c_{\cweight,l}$ and $(\bM,\bge_{\refer})$; and
  $C_{b}$ only depends on $c_{\cweight,l}$, $\e_{\que}$, $K_{q}$ and $(\bM,\bge_{\refer})$. In order to arrive at this conclusion, we appealed to
  (\ref{eq:muABCdiffest}), (\ref{eq:EbfIbmuAmuAest}), (\ref{eq:EbfIXAbmuBXAXBbmuC}) and (\ref{eq:muminmainlowerbound}). In case $f=\g^{B}_{CD}$ and
  $|\bfI|\leq l-1$, the right hand side of (\ref{eq:termstobeestimatedtypeoneterm}) satisfies a slightly better estimate due to
  (\ref{eq:EbfIgABCEbfIXDgABC}); the power of $\ldr{\varrho}$ can be decreased. 

  \textbf{Estimating terms of type two.} Applying $E_{\bfI}$ to a term of type two yields a linear combination of terms of the form
  \[
  e^{h}E_{\bfI_{1}}(h)\cdots E_{\bfI_{m}}(h) E_{\bfJ}(f_{1})E_{\bfK}(f_{2}),
  \]
  where $\bfI_{i}\neq 0$, $i=1,\dots,m$, and $|\bfI_{1}|+\dots+|\bfI_{m}|+|\bfJ|+|\bfK|=|\bfI|$. On the other hand,
  \begin{equation}\label{eq:termstobeestimatedtypetwoterm}
    \begin{split}
      & \ldr{\varrho}^{-(|\bfI|+3)\cweight}e^{h}|E_{\bfI_{1}}(h)\cdots E_{\bfI_{m}}(h) E_{\bfJ}(f_{1})E_{\bfK}(f_{2})|\\
      \leq & e^{h}\ldr{\varrho}^{-|\bfI_{1}|\cweight}|E_{\bfI_{1}}(h)|\cdots \ldr{\varrho}^{-|\bfI_{m}|\cweight}|E_{\bfI_{m}}(h)|
      \ldr{\varrho}^{-(|\bfJ|+1)\cweight}|E_{\bfJ}(f_{1})|\ldr{\varrho}^{-(|\bfK|+2)\cweight}|E_{\bfK}(f_{2})|.
    \end{split}
  \end{equation}
  Regardless of what the $f_{i}$ are, an argument similar to the above yields the conclusion that if $|\bfI|\leq l-1$, the right hand side can be
  estimated by
  \[
  C_{a}\theta_{0,-}^{-2}\ldr{\varrho}^{l\cweight+l+1}e^{2\e_{R}\varrho}
  \]
  on $M_{-}$, where $\e_{R}$ is introduced in the statement of the lemma and $C_{a}$ only depends on $c_{\cweight,l}$, $\e_{\que}$, $K_{q}$ and
  $(\bM,\bge_{\refer})$. 

  \textbf{Estimating the Ricci curvature.} Summarising the above,
  \begin{equation}\label{eq:estofRiccicoeff}
    \textstyle{\sum}_{|\bfI|\leq l-1}\ldr{\varrho}^{-(|\bfI|+3)\cweight}|E_{\bfI}(\theta^{-2}\bR^{B}_{\phantom{B}C})|
    \leq C_{a}\theta_{0,-}^{-2}\ldr{\varrho}^{l\cweight+l+1}e^{2\e_{R}\varrho}
  \end{equation}
  on $M_{-}$, where $C_{a}$ only depends on $c_{\cweight,l}$, $\e_{\que}$, $K_{q}$ and $(\bM,\bge_{\refer})$. On the other hand,
  \[
  \bmR=\textstyle{\sum}_{B,C}\theta^{-2}\bR^{B}_{\phantom{B}C}X_{B}\otimes Y^{C}.
  \]
  Thus $\bD_{\bfI}\bmR$ can be written as a sum of terms of the form
  \[
  E_{\bfJ}(\theta^{-2}\bR^{B}_{\phantom{B}C})(\bD_{\bfK}X_{B})\otimes (\bD_{\bfL}Y^{C}),
  \]
  where $|\bfJ|+|\bfK|+|\bfL|=|\bfI|$. Combining this estimate with (\ref{eq:bDbfAellAetcpteststmtEi}), (\ref{eq:estofRiccicoeff}) and the assumptions yields
  \begin{equation}\label{eq:estofRicciframe}
    \textstyle{\sum}_{|\bfI|\leq l-1}\ldr{\varrho}^{-(|\bfI|+3)\cweight}|D_{\bfI}\bmR|_{\bge_{\refer}}
    \leq C_{a}\theta_{0,-}^{-2}\ldr{\varrho}^{l\cweight+l+1}e^{2\e_{R}\varrho}
  \end{equation}
  on $M_{-}$, where $C_{a}$ only depends on $c_{\cweight,l}$, $\e_{\que}$, $K_{q}$ and $(\bM,\bge_{\refer})$. Combining this estimate with
  Lemma~\ref{lemma:bDbfAbDkequiv} yields the conclusion that (\ref{eq:estofRicciqugeo}) holds.
\end{proof}

\subsection{Revisiting the assumptions}

At this stage, we can revisit the assumptions of Lemma~\ref{lemma:muAmmuBmmuCrelquiescent}. In particular, we obtain improvements of some of the
assumptions, so that this subsection can be thought of as giving a partial improvement of the bootstrap assumptions. 

\begin{lemma}
  Let $0\leq\cweight\in\ro$ and assume that the standard assumptions; see Definition~\ref{def:standardassumptions}; as well as the
  $(\cweight,1)$-supremum assumptions are fulfilled. Assume that Einstein's equations with a cosmological constant $\Lambda$ (\ref{eq:EE}) are
  satisfied. Assume, moreover, that there are constants $K_{q}$, $C_{\rho,0}$, $\e_{\rho}>0$ and $0<\e_{\que}<1$ such that (\ref{eq:closetostiff}) and
  (\ref{eq:qexpconvergence}) hold, and an $\e_{p}\in (0,1)$ such that (\ref{eq:quiescentregime}) holds.  Then there is a constant $L_{q}$ such that
  \begin{equation}\label{eq:qexpconvimprove}
    |q-(n-1)|\leq L_{q}\ldr{\varrho}^{2(2\cweight+1)}e^{2\de_{q}\varrho}
  \end{equation}
  on $M_{-}$, where $L_{q}=L_{q,0}\theta_{0,-}^{-2}$ and $L_{q,0}$ only depends on $c_{\cweight,1}$, $C_{\rho,0}$, $K_{q}$, $\Lambda$, $\e_{\que}$
  and $(\bM,\bge_{\refer})$. Moreover, $\de_{q}:=\min\{\e_{\Spe},\e_{\rho},\e_{p}\}$. Finally, there is a constant $K_{H}$ such that
  \begin{equation}\label{eq:Hamconstrasympt}
    \left|2\Omega+\textstyle{\sum}_{A}\ell_{A}^{2}-1\right|\leq K_{H}\ldr{\varrho}^{2(2\cweight+1)}e^{2\e_{R}\varrho}
  \end{equation}
  on $M_{-}$, where $K_{H}=K_{H,0}\theta_{0,-}^{-2}$ and $K_{H,0}$ only depends on $c_{\cweight,1}$, $K_{q}$, $\Lambda$, $\e_{\que}$ and
  $(\bM,\bge_{\refer})$. Moreover, $\e_{R}=\min\{\e_{\Spe},\e_{p}\}$.
\end{lemma}
\begin{remark}
  In Lemmas~\ref{lemma:muAmmuBmmuCrelquiescent} and \ref{lemma:CkRicciestimates}, it is sufficient to assume that (\ref{eq:qexpconvergence}) holds
  for some $\e_{\que}>0$. In particular, $\e_{\que}$ could be strictly smaller than $\de_{q}$. In this sense, (\ref{eq:qexpconvimprove})
  represents an improvement relative to the original assumption. 
\end{remark}
\begin{remark}
  A similar conclusion holds in case we replace the requirement that (\ref{eq:qexpconvergence}) and (\ref{eq:quiescentregime}) hold with the requirement
  that $\g^{A}_{BC}=0$ for all $A<\min\{B,C\}$. This follows by appealing to Remark~\ref{remark:Ckscalarcurvatureestimatesymmetry} and arguments similar
  to the proof below, keeping in mind that $\mSb_{b}=0$ in this setting. However, in this case, $\de_{q}$ can be replaced by $\min\{\e_{\Spe},\e_{\rho}\}$
  in (\ref{eq:qexpconvimprove}); and $\e_{R}$ can be replaced by $\e_{\Spe}$ in (\ref{eq:Hamconstrasympt}). Moreover, the constants depend on neither
  $K_{q}$ nor $\e_{\que}$. 
\end{remark}
\begin{proof}
  Due to the assumptions, (\ref{eq:ninvoneplusqfirstestimate}) holds. Combining this estimate with Remark~\ref{remark:Ckscalarcurvatureestimate}
  yields (\ref{eq:qexpconvimprove}). Combining Remark~\ref{remark:Ckscalarcurvatureestimate} with (\ref{eq:reformulatedandrenormalisedHamcon}) and
  (\ref{eq:OmegaLambdaestimate}) yields (\ref{eq:Hamconstrasympt}). 
\end{proof}

\subsection{Estimating the contribution from the lapse function}

We wish to estimate the right hand side of (\ref{eq:mlUmKwithEinstein}). Due to Lemma~\ref{lemma:CkRicciestimates} and
Remark~\ref{remark:Ckscalarcurvatureestimate}, we can already estimate the terms arising from the spatial curvature. Nevertheless, we also need to
estimate the contribution from the lapse function. Define, to this end, the family $\bmN$ of $(1,1)$-tensor fields on $\bM$ by 
\begin{equation}\label{eq:bmNAA def}
  \bmN^{A}_{\phantom{A}B}:=\theta^{-2}N^{-1}\bnabla^{A}\bnabla_{B}N.
\end{equation}
\begin{lemma}\label{lemma:Cklapseestimates}
  Let $0\leq\cweight\in\ro$ and $1\leq l\in\zo$ and assume that the standard assumptions; see Definition~\ref{def:standardassumptions}; as well as the
  $(\cweight,l)$-supremum assumptions are fulfilled. Assume, moreover, that there are constants $K_{q}$ and $0<\e_{\que}<1$ such that
  (\ref{eq:qexpconvergence}) holds and an $\e_{p}\in (0,1)$ such that (\ref{eq:quiescentregime}) holds. Then there is a constant $K_{\mN,l}$ such that
  \begin{align}
    \textstyle{\sum}_{k\leq l-1}\ldr{\varrho}^{-(k+3)\cweight}|\bD^{k}\bmN|_{\bge_{\refer}} \leq &
    K_{\mN,l}\ldr{\varrho}^{l\cweight+l+1}e^{2\e_{R}\varrho}\label{eq:estoflapsequgeo}
  \end{align}
  on $M_{-}$, where $\e_{R}:=\min\{\e_{p},\e_{\Spe}\}$, $K_{\mN,l}=K_{\mN,l,0}\theta_{0,-}^{-2}$ and $K_{\mN,l,0}$ only depends on $c_{\cweight,l}$, $\e_{\que}$,
  $K_{q}$ and $(\bM,\bge_{\refer})$. 
\end{lemma}
\begin{remark}\label{remark:Cklapseestimatessymmetry}
  A similar conclusion holds in case we replace the requirement that (\ref{eq:qexpconvergence}) and (\ref{eq:quiescentregime}) hold with the requirement
  that $\g^{A}_{BC}=0$ for all $A<\min\{B,C\}$. Moreover, in that case we can replace $\e_{R}$ with $\e_{\Spe}$ in (\ref{eq:estoflapsequgeo}), and the
  constant depends on neither $K_{q}$ nor $\e_{\que}$. 
\end{remark}
\begin{proof}
  Since the proof of the statement is quite similar to the proof of Lemma~\ref{lemma:CkRicciestimates} (keeping
  (\ref{eq:bnablaABnormN})--(\ref{eq:normalisedDeltaNthroughN}) in mind), we leave the details to the reader. 
\end{proof}

\subsection{Estimating the Lie derivative of $\mK$}

Combining Lemma~\ref{lemma:CkRicciestimates} and Lemma~\ref{lemma:Cklapseestimates} with (\ref{eq:mlUmKwithEinstein}) and assumptions concerning
the matter yields an estimate for $\hml_{U}\mK$. In practice, we are interested in two different situations: vacuum and scalar field matter. 
Let us begin by considering the vacuum setting. Due to the comments made in Subsection~\ref{ssection:quiescentregimes}, we must then have $n\geq 10$. 

\begin{thm}\label{thm:CkestofhmlUmK}
  Let $0\leq\cweight\in\ro$, $n\geq 10$, $1\leq l\in\zo$ and assume that the standard assumptions; see Definition~\ref{def:standardassumptions}; as
  well as the $(\cweight,l)$-supremum assumptions are fulfilled. In particular, the spacetime is $n+1$-dimensional. Assume, moreover, that there are
  constants $K_{q}$ and $0<\e_{\que}<1$ such that (\ref{eq:qexpconvergence}) holds and an $\e_{p}\in (0,1)$ such that (\ref{eq:quiescentregime}) holds.
  Assume, finally, that Einstein's vacuum equations with a cosmological constant $\Lambda$ are satisfied; i.e., (\ref{eq:EE}) holds with $T=0$. Then there
  is a constant $K_{\mK,l}$ such that
  \begin{align}
    \textstyle{\sum}_{k\leq l-1}\ldr{\varrho}^{-(k+3)\cweight}|\bD^{k}\hml_{U}\mK|_{\bge_{\refer}} \leq &
    K_{\mK,l}\ldr{\varrho}^{l\cweight+l+1}e^{2\e_{R}\varrho}\label{eq:estofhmlUmKequgeo}
  \end{align}
  on $M_{-}$, where $\e_{R}:=\min\{\e_{p},\e_{\Spe}\}$, $K_{\mK,l}=K_{\mK,l,0}\theta_{0,-}^{-2}$ and $K_{\mK,l,0}$ only depends on $c_{\cweight,l}$, $\e_{\que}$,
  $K_{q}$, $\Lambda$ and $(\bM,\bge_{\refer})$. Moreover, there is a constant $K_{q,l}$ such that
  \begin{align}
    \textstyle{\sum}_{k\leq l-1}\ldr{\varrho}^{-(k+3)\cweight}|\bD^{k}[q-(n-1)]|_{\bge_{\refer}} \leq &
    K_{q,l}\ldr{\varrho}^{l\cweight+l+1}e^{2\e_{R}\varrho}\label{eq:estofqminusnminusoneequgeo}
  \end{align}
  on $M_{-}$, where $\e_{R}:=\min\{\e_{p},\e_{\Spe}\}$, $K_{q,l}=K_{q,l,0}\theta_{0,-}^{-2}$ and $K_{q,l,0}$ only depends on $c_{\cweight,l}$, $\e_{\que}$,
  $K_{q}$, $\Lambda$ and $(\bM,\bge_{\refer})$.
\end{thm}
\begin{remark}\label{remark:improvementhmlUmKest}
  Note that $\hml_{U}\mK$ and the derivatives of $q$ growing polynomially is consistent with the assumptions. However, due to (\ref{eq:estofhmlUmKequgeo})
  and (\ref{eq:estofqminusnminusoneequgeo}), we conclude that these objects decay exponentially. In that sense, the estimates (\ref{eq:estofhmlUmKequgeo})
  and (\ref{eq:estofqminusnminusoneequgeo}) represent an improvement of the assumptions. Needless to say, this observation is of interest in the context of
  a bootstrap argument. Unfortunately, there is a loss of two derivatives in the process. This necessitates separate arguments to close the bootstrap;
  see Subsection~\ref{ssection:summary discussion ex} for a discussion. 
\end{remark}
\begin{remark}\label{remark:CkestofhmlUmKsymmetry}
  A similar conclusion holds in case we replace the requirement that (\ref{eq:qexpconvergence}) and (\ref{eq:quiescentregime}) hold with the requirement
  that $\g^{A}_{BC}=0$ for all $A<\min\{B,C\}$. This follows by appealing to Remarks~\ref{remark:Ckscalarcurvatureestimatesymmetry} and
  \ref{remark:Cklapseestimatessymmetry} and an argument similar to the proof below. Moreover, in this case, $\e_{R}$ can be replaced by $\e_{\Spe}$ in
  (\ref{eq:estofhmlUmKequgeo}) and (\ref{eq:estofqminusnminusoneequgeo}). Finally, the constants depend on neither $K_{q}$ nor $\e_{\que}$. 
\end{remark}
\begin{proof}
  In the vacuum setting, (\ref{eq:mlUmKwithEinstein}) can be written
  \begin{equation}\label{eq:hmlUvacuumsetting}
    \begin{split}
      \hml_{U}\mK =  & -\left(\tfrac{\Delta_{\bge}N}{\theta^{2}N}+\tfrac{2n}{n-1}\tfrac{\Lambda}{\theta^{2}}
      -\tfrac{\bS}{\theta^{2}}\right)\mK+\tfrac{2}{n-1}\tfrac{\Lambda}{\theta^{2}}\mathrm{Id}+\bmN-\bmR.
    \end{split}
  \end{equation}
  The last two terms on the right hand side can be estimated by appealing to Lemmas~\ref{lemma:CkRicciestimates} and \ref{lemma:Cklapseestimates}.
  In order to estimate the third last term, note that $\bD_{\bfI}(\theta^{-2}\Lambda\mathrm{Id})$ can be written as a linear combination of terms of the form
  \[
  \theta^{-2}E_{\bfI_{1}}\ln\theta\cdots E_{\bfI_{k}}\ln\theta\cdot\Lambda\mathrm{Id},
  \]
  where $|\bfI_{1}|+\dots+|\bfI_{k}|=|\bfI|$ and $\bfI_{i}\neq 0$. Due to (\ref{eq:OmegaLambdaestimate}) and the assumptions, it follows that for
  $|\bfI|\leq l$, 
  \[
  \ldr{\varrho}^{-|\bfI|\cweight}|\bD_{\bfI}(\theta^{-2}\Lambda\mathrm{Id})|_{\bge_{\refer}}\leq C\theta_{0,-}^{-2}e^{2\e_{\Spe}\varrho}
  \]
  on $M_{-}$, where $C$ only depends on $c_{\cweight,l}$, $\Lambda$ and $(\bM,\bge_{\refer})$. The terms appearing in the first factor of the first term on the
  right hand side of (\ref{eq:hmlUvacuumsetting}) can be estimated similarly to the above; note that
  \begin{equation}\label{eq:DeltabgeN bmR}
    \tfrac{\Delta_{\bge}N}{\theta^{2}N}=\mathrm{tr}\bmN,\ \ \
  \theta^{-2}\bS=\mathrm{tr}\bmR.
  \end{equation}
  Combining the above observations with the assumptions and (\ref{eq:Clmclequiv}) yields (\ref{eq:estofhmlUmKequgeo}). Due to (\ref{eq:qmainformula})
  the proof of (\ref{eq:estofqminusnminusoneequgeo}) is similar but simpler. 
\end{proof}

\subsection{Convergence of $\mK$}

As noted in Remark~\ref{remark:improvementhmlUmKest}, the estimate (\ref{eq:estofhmlUmKequgeo}) in some respects represents an improvement of the 
assumptions. By integrating this estimate, it can also be demonstrated that $\mK$ converges, which represents an improvement of the assumption that 
the derivatives of $\mK$ grow at most polynomially. However, the improvement is, again, associated with a loss of derivatives. 

\begin{thm}\label{thm:mKconvCk}
  Let $0\leq\cweight\in\ro$, $n\geq 10$, $1\leq l\in\zo$ and assume that the standard assumptions; see Definition~\ref{def:standardassumptions}; as
  well as the $(\cweight,l)$-supremum assumptions are fulfilled. In particular, the spacetime is $n+1$-dimensional. Assume, moreover, that there are
  constants $K_{q}$ and $0<\e_{\que}<1$ such that (\ref{eq:qexpconvergence}) holds and an $\e_{p}\in (0,1)$ such that (\ref{eq:quiescentregime}) holds.
  Assume, finally, that Einstein's vacuum equations with a cosmological constant $\Lambda$ are satisfied; i.e., (\ref{eq:EE}) holds with $T=0$; and
  that $\varrho$ diverges uniformly to $-\infty$ in the direction of the singularity. Then there
  is a $(1,1)$-tensor field $\mK_{\infty}$ and functions $\ell_{A,\infty}$ on $\bM$ which are $C^{l-1}$, and constants $K_{\infty,i,l}$, $i=0,1$, such that
  \begin{equation}\label{eq:mKconvtomKinf}
    \begin{split}
      & \|\mK(\cdot,\tau)-\mK_{\infty}\|_{C^{l-1}(\bM)}+\textstyle{\sum}_{A}\|\ell_{A}(\cdot,\tau)-\ell_{A,\infty}\|_{C^{l-1}(\bM)}\\
      \leq & K_{\infty,0,l}\theta_{0,-}^{-2}\ldr{\tau}^{(l+1)(2\cweight+1)}e^{2\vare_{R}\tau}+K_{\infty,1,l}\ldr{\tau}^{l\cweight}e^{\vare_{\Spe}\tau}
    \end{split}
  \end{equation}
  for all $\tau\leq 0$, where $\tau$ is the time coordinate introduced in (\ref{eq:taudefinition}); $\vare_{R}:=\e_{R}/(3K_{\rovar})$;
  $\vare_{\Spe}:=\e_{\Spe}/(3K_{\rovar})$; $K_{\rovar}$ is introduced in the statement of Lemma~\ref{lemma:smallnessshiftconsequences};
  $K_{\infty,0,l}$ only depends on $c_{\cweight,l}$, $\e_{\que}$,
  $K_{q}$, $\e_{p}$, $\Lambda$ and $(\bM,\bge_{\refer})$; and $K_{\infty,1,l}$ only depends on $c_{\cweight,l}$ and $(\bM,\bge_{\refer})$.
\end{thm}
\begin{remark}
  A similar conclusion holds in case we replace the requirement that (\ref{eq:qexpconvergence}) and (\ref{eq:quiescentregime}) hold with the requirement
  that $\g^{A}_{BC}=0$ for all $A<\min\{B,C\}$. This follows by appealing to Remark~\ref{remark:CkestofhmlUmKsymmetry} and an argument similar to the proof
  below. However, in this case, $\vare_{R}$ can be replaced by $\vare_{\Spe}$ in (\ref{eq:mKconvtomKinf}).
  Moreover, the constants depend on neither $K_{q}$, $\e_{\que}$ nor $\e_{p}$. 
\end{remark}
\begin{proof}
  Note, to begin with, that with respect to frames $\{E_{i}\}$ and $\{\omega^{i}\}$ as in Remark~\ref{remark:globalframe},
  \[
  (\hml_{U}\mK)^{i}_{\phantom{i}j}=\hN^{-1}\d_{t}(\mK^{i}_{\phantom{i}j})-\hN^{-1}(\ml_{\chi}\mK)^{i}_{\phantom{i}j},
  \]
  where we appealed to (\ref{eq:hmlUmtinfixedspatialcoord}). Introducing a time coordinate according to (\ref{eq:taudefinition}), it follows that
  \begin{equation}\label{eq:dtaumKijform}
    \d_{\tau}(\mK^{i}_{\phantom{i}j})=\xi (\hml_{U}\mK)^{i}_{\phantom{i}j}+\xi\hN^{-1}(\ml_{\chi}\mK)^{i}_{\phantom{i}j},
  \end{equation}
  where $\xi:=(\d_{t}\tau)^{-1}\hN$. Next, we apply $E_{\bfI}$ to the constituents of the right hand side. Applying $E_{\bfI}$ to $\xi$
  yields a linear combination of terms of the form
  \begin{equation}\label{eq:EbfIphiterms}
    \xi\cdot E_{\bfI_{1}}(\ln\hN)\cdots E_{\bfI_{k}}(\ln\hN),
  \end{equation}
  where $|\bfI_{1}|+\dots+|\bfI_{k}|=|\bfI|$ and $\bfI_{j}\neq 0$. However, $\xi$ is bounded due to (\ref{eq:hNtaudotequivEi}). Combining this
  observation with the assumptions and (\ref{eq:Clmclequiv}) yields the conclusion that if $|\bfI|\leq l$, then
  \begin{equation}\label{eq:EbfIphiweightedCkestimate}
    \ldr{\varrho}^{-|\bfI|\cweight}|E_{\bfI}(\xi)|\leq K
  \end{equation}
  on $M_{-}$, where $K$ only depends on $c_{\cweight,l}$ and $(\bM,\bge_{\refer})$. Next, note that
  \[
  \ldr{\varrho}^{-(|\bfI|+3)\cweight}|E_{\bfI}[(\hml_{U}\mK)^{i}_{\phantom{i}j}]|
  \leq C\textstyle{\sum}_{k\leq |\bfI|}\ldr{\varrho}^{-(k+3)\cweight}|\bD^{k}\hml_{U}\mK|_{\bge_{\refer}}
  \]
  where $C$ only depends on $|\bfI|$ and $(\bM,\bge_{\refer})$. Combining this estimate with (\ref{eq:estofhmlUmKequgeo}) and the previously derived
  estimate for $\xi$ yields the conclusion that if $|\bfI|\leq l-1$, then
  \[
  \ldr{\varrho}^{-(|\bfI|+3)\cweight}|E_{\bfI}[\xi(\hml_{U}\mK)^{i}_{\phantom{i}j}]|
  \leq K_{l}\ldr{\varrho}^{l\cweight+l+1}e^{2\e_{R}\varrho}
  \]
  on $M_{-}$, where $K_{l}=K_{l,0}\theta_{0,-}^{-2}$, and $K_{l,0}$ only depends on $c_{\cweight,l}$, $\e_{\que}$, $K_{q}$, $\Lambda$ and $(\bM,\bge_{\refer})$.
  Next, note that
  \begin{equation}\label{eq:mlchimKformforest}
    (\ml_{\chi}\mK)^{i}_{\phantom{i}j}=(\bD_{\chi}\mK)^{i}_{\phantom{i}j}-\mK^{m}_{\phantom{m}j}\omega^{i}(\bD_{E_{m}}\chi)+\mK^{i}_{\phantom{i}m}\omega^{m}(\bD_{E_{j}}\chi).
  \end{equation}
  Thus
  \[
  |E_{\bfI}[(\ml_{\chi}\mK)^{i}_{\phantom{i}j}]|\leq C\textstyle{\sum}_{l_{a}+|\bfI_{b}|\leq |\bfI|+1}|\bD^{l_{a}}\mK|_{\bge_{\refer}}|\bD_{\bfI_{b}}\chi|_{\bge_{\refer}},
  \]
  where $C$ only depends on $|\bfI|$ and $(\bM,\bge_{\refer})$. Combining this estimate with Remark~\ref{remark:chiclvarrhodecay} and the assumed bounds on
  $\chi$ and $\mK$ yields the conclusion that if $|\bfI|\leq l$, then
  \[
  \hN^{-1}\ldr{\varrho}^{-(|\bfI|+1)\cweight}|E_{\bfI}[(\ml_{\chi}\mK)^{i}_{\phantom{i}j}]|\leq Ce^{\e_{\Spe}\varrho}
  \]
  on $M_{-}$, where $C$ only depends on $c_{\cweight,l}$ and $(\bM,\bge_{\refer})$. Since $\xi\hN^{-1}$ is independent of the spatial variable, combining
  this estimate with arguments similar to the above yields
  \[
  \ldr{\varrho}^{-(|\bfI|+1)\cweight}|E_{\bfI}[\xi\hN^{-1}(\ml_{\chi}\mK)^{i}_{\phantom{i}j}]|\leq Ce^{\e_{\Spe}\varrho}
  \]
  on $M_{-}$, where $C$ only depends on $c_{\cweight,l}$ and $(\bM,\bge_{\refer})$. Summing up the above yields the conclusion that if $|\bfI|\leq l-1$,
  \[
  |E_{\bfI}[\d_{\tau}(\mK^{i}_{\phantom{i}j})]|\leq K_{l}\ldr{\varrho}^{(l+1)(2\cweight+1)}e^{2\e_{R}\varrho}+C\ldr{\varrho}^{l\cweight}e^{\e_{\Spe}\varrho}
  \]
  on $M_{-}$. Since $E_{\bfI}$ and $\d_{\tau}$ commute, and since (\ref{eq:DeltavarrhorelvariationEiintro}) holds, this estimate can be integrated in order
  to conclude that
  $\mK$ converges to a limit $\mK_{\infty}$ in $C^{l-1}(\bM)$ and that $\mK-\mK_{\infty}$ can be estimated by the right hand side of (\ref{eq:mKconvtomKinf})
  in $C^{l-1}(\bM)$.

  Turning to the $\ell_{A}$, note that by arguments similar to the ones presented at the beginning of the proof,
  \begin{equation}\label{eq:dtauellA}
    \d_{\tau}\ell_{A}=\xi\hU(\ell_{A})+\xi\hN^{-1}\chi(\ell_{A})=\xi(\hml_{U}\mK)(Y^{A},X_{A})+\xi\hN^{-1}\chi[\mK(Y^{A},X_{A})],
  \end{equation}
  where there is no summation on the right hand side and we appealed to (\ref{eq:hU ellA}). Applying $E_{\bfI}$ to the first term on the far right
  hand side of (\ref{eq:dtauellA}) yields a linear combination of terms of the form
  \begin{equation}\label{eq:dtauellAfirstterm}
    E_{\bfI_{1}}(\xi)(\bD_{\bfI_{2}}\hml_{U}\mK)(\bD_{\bfI_{3}}Y^{A},\bD_{\bfI_{4}}X_{A}),
  \end{equation}
  where $|\bfI|=|\bfI_{1}|+\dots+|\bfI_{4}|$. However, combining Lemma~\ref{lemma:bDbfAbDkequiv} and (\ref{eq:bDbfAellAetcpteststmtEi}),
  \[
  |(\bD_{\bfI_{2}}\hml_{U}\mK)(\bD_{\bfI_{3}}Y^{A},\bD_{\bfI_{4}}X_{A})|
  \]
  can be estimated by a linear combination of terms of the form
  \[
  |\bD^{m_{1}}\mK|_{\bge_{\refer}}\cdots |\bD^{m_{j}}\mK|_{\bge_{\refer}}|\bD^{k}\hml_{U}\mK|_{\bge_{\refer}},
  \]
  where the implicit constants only depend on $c_{\robas}$, $|\bfI|$ and $(\bM,\bge_{\refer})$; $m_{p}\neq 0$; and
  $m_{1}+\dots+m_{j}+k\leq |\bfI_{2}|+|\bfI_{3}|+|\bfI_{4}|$. Combining this observation with the assumptions, (\ref{eq:estofhmlUmKequgeo}) and
  (\ref{eq:EbfIphiweightedCkestimate}) yields the conclusion that expressions of the form (\ref{eq:dtauellAfirstterm}) can be estimated by
  \[
  C_{a}\theta_{0,-}^{-2}\ldr{\varrho}^{(l+1)(2\cweight+1)}e^{2\e_{R}\varrho},
  \]
  assuming $|\bfI|\leq l-1$, where $C_{a}$ only depends on $c_{\cweight,l}$, $\e_{\que}$, $K_{q}$, $\Lambda$ and $(\bM,\bge_{\refer})$. Next, consider
  the second term on the right hand side of (\ref{eq:dtauellA}). Applying $E_{\bfI}$ to this expression and noting that $\xi\hN^{-1}$ is independent
  of the spatial coordinates, it is clear (by arguments similar to the above) that it is sufficient to estimate
  \[
  \xi\hN^{-1}|\bD_{\bfI_{1}}\chi|_{\bge_{\refer}}|\bD^{m_{1}}\mK|_{\bge_{\refer}}\cdots |\bD^{m_{j}}\mK|_{\bge_{\refer}}
  \]
  where $|\bfI_{1}|+m_{1}+\dots+m_{j}\leq |\bfI|+1$. Combining this observation with arguments similar to the above yields the conclusion that there
  are $\ell_{A,\infty}\in C^{l-1}(\bM)$ such that (\ref{eq:mKconvtomKinf}) holds. 
\end{proof}

\section{Sobolev estimates in the quiescent setting}\label{section:Sobestquiescentsetting}

In what follows, we wish to derive Sobolev estimates for $\hml_{U}\mK$. To this end, we need to estimate all the terms on the right hand side of
(\ref{eq:hmlUvacuumsetting}). However, due to (\ref{eq:hmlUvacuumsetting}), (\ref{eq:DeltabgeN bmR})  and assumptions as in
Theorems~\ref{thm:SobestimatesQuiescentVacuum} and \ref{thm:SobEstimatesQuiescentMatter}, it follows that the
main difficulty is to  estimate $\bmR$ and $\bmN$. We begin by estimating the spatial curvature. 

\begin{lemma}\label{lemma:SobRicciestimates}
  Let $0\leq\cweight\in\ro$, $1\leq l\in\zo$ and assume that the standard assumptions; see Definition~\ref{def:standardassumptions}; the
  $(\cweight,l)$-Sobolev assumptions; and the $(\cweight,1)$-supremum assumptions are fulfilled. Assume, moreover, that there are constants $K_{q}$
  and $0<\e_{\que}<1$ such that (\ref{eq:qexpconvergence}) holds and an $\e_{p}\in (0,1)$ such that (\ref{eq:quiescentregime}) holds.
  Then there is a constant $K_{\mR,l}$ such that
  \begin{align}
    \|\bmR(\cdot,\tau)\|_{H^{l-1}_{\weight_{1}}(\bM)} \leq & K_{\mR,l}\ldr{\tau}^{(l+1)(\cweight+1)}e^{2\vare_{R}\tau}\label{eq:SobestofRicciqugeo}
  \end{align}
  for all $\tau\leq 0$, where $\weight_{1}:=(2\cweight,\cweight)$, $K_{\mR,l}=K_{\mR,l,0}\theta_{0,-}^{-2}$ and $K_{\mR,l,0}$ only depends on $s_{\cweight,l}$,
  $c_{\cweight,1}$, $\e_{\que}$, $K_{q}$ and $(\bM,\bge_{\refer})$. Here $\e_{R}:=\min\{\e_{p},\e_{\Spe}\}$ and $\vare_{R}:=\e_{R}/(3K_{\rovar})$, where $K_{\rovar}$
  is introduced in the statement of Lemma~\ref{lemma:smallnessshiftconsequences}.
\end{lemma}
\begin{remark}\label{remark:Sobscalarcurvatureestimatesymmetry}
  A similar conclusion holds in case we replace the requirement that (\ref{eq:qexpconvergence}) and (\ref{eq:quiescentregime}) hold with the requirement
  that $\g^{A}_{BC}=0$ for all $A<\min\{B,C\}$. Moreover, in that case $\vare_{R}$ can be replaced by $\vare_{\Spe}$ and the constants depend neither on
  $\e_{\que}$ nor on $K_{q}$. The proof of this statement is similar to, but simpler than, the proof below (in that some of the terms we need to estimate
  below vanish in this case).
\end{remark}
\begin{proof}
  The terms that need to be estimated are the same as in the proof of Lemma~\ref{lemma:CkRicciestimates}. To begin with, we thus need to estimate
  terms of type one. In practice, this means that we need to estimate the right hand side of (\ref{eq:termstobeestimatedtypeoneterm}) in $L^{2}$.
  Note that there are two possibilities for $f$; it either equals $X_{B}(\bmu_{C})$ or $\g^{D}_{BC}$. In case $f=X_{B}(\bmu_{C})$, the last factor on
  the right hand side of (\ref{eq:termstobeestimatedtypeoneterm}) can be estimated by a linear combination of terms of the form
  \[
  |\bD^{m_{1}}\mK|_{\bge_{\refer}}\cdots |\bD^{m_{k}}\mK|_{\bge_{\refer}}|E_{\bfK}E_{i}\bmu_{C}|,
  \]
  where $m_{1}+\dots+m_{k}+|\bfK|\leq |\bfJ|+1$ and $m_{i}\neq 0$; see (\ref{eq:EbfIXAXBbmuCdecompterms}) and the adjacent text. Combining this observation
  with the fact that $\ldr{\varrho}$ and $\ldr{\tau}$ are equivalent; see (\ref{eq:ldrrho ldr tau equiv}); as well as the fact that
  $e^{2\e_{R}\varrho}\leq e^{2\varepsilon_{R}\tau}$; see (\ref{eq:eSpevarrhoeelowtaurelEi}); what we need to estimate in $L^{2}$ is
  \begin{equation*}
    \begin{split}
      & e^{2\varepsilon_{R}\tau}\ldr{\tau}^{-\cweight}\textstyle{\prod}_{r=1}^{m}\ldr{\tau}^{-|\bfJ_{r}|\cweight}|E_{\bfJ_{r}}(\ldr{\tau}^{-\cweight}E_{i_{r}}h)|\cdot
      \textstyle{\prod}_{s=1}^{k}\ldr{\tau}^{-m_{s}\cweight}|\bD^{m_{s}}\mK|_{\bge_{\refer}}\\
      & \cdot \ldr{\tau}^{-|\bfK|\cweight}|E_{\bfK}(\ldr{\tau}^{-\cweight}E_{i}\bmu_{C})|.
    \end{split}
  \end{equation*}  
  Here $\bfJ_{r}$ and $i_{r}$ are chosen so that $E_{\bfI_{r}}=E_{\bfJ_{r}}E_{i_{r}}$. In order to estimate this expression in $L^{2}$, we can appeal to
  Corollary~\ref{cor:mixedmoserestweight}. In the notation of Corollary~\ref{cor:mixedmoserestweight}, $\mt_{j}=\mK$, $u_{j}=1$, $g_{j}=\ldr{\tau}^{-\cweight}$;
  the $\mU_{m}$ are either $\ldr{\tau}^{-\cweight}E_{i_{r}}h$ or $\ldr{\tau}^{-\cweight}E_{i}\bmu_{C}$, $v_{m}=1$, $h_{m}=\ldr{\tau}^{-\cweight}$. Note also that
  \[
  |\bfJ_{1}|+\dots+|\bfJ_{m}|+m_{1}+\dots+m_{k}+|\bfK|\leq |\bfI|+1-m.
  \]
  If $|\bfI|\leq l-1$, then, due to Corollary~\ref{cor:mixedmoserestweight}, the type of terms that we need to estimate are thus
  \begin{align*}
    & e^{2\varepsilon_{R}\tau}\ldr{\tau}^{-\cweight}\|\mK\|_{H^{l}_{\weight_{0}}(\bM)}\|h\|_{\mc^{\bfl_{0}}_{\weight_{0}}(\bM)}^{m}\|\bmu_{C}\|_{\mc^{\bfl_{0}}_{\weight_{0}}(\bM)},\\
    & e^{2\varepsilon_{R}\tau}\ldr{\tau}^{-\cweight}\|h\|_{\mH^{\bfl_{1}}_{\bbE,\weight_{0}}(\bM)}\|h\|_{\mc^{\bfl_{0}}_{\weight_{0}}(\bM)}^{m-1}\|\bmu_{C}\|_{\mc^{\bfl_{0}}_{\weight_{0}}(\bM)},\\
    & e^{2\varepsilon_{R}\tau}\ldr{\tau}^{-\cweight}\|\bmu_{C}\|_{\mH^{\bfl_{1}}_{\bbE,\weight_{0}}(\bM)}\|h\|_{\mc^{\bfl_{0}}_{\weight_{0}}(\bM)}^{m}.
  \end{align*}
  Here $\bfl_{0}=(1,1)$ and $\bfl_{1}=(1,l+1)$. In deriving this estimate, we replaced $\ldr{\varrho}$ with $\ldr{\tau}$ whenever convenient (and
  vice versa), and we used the notation introduced in (\ref{eq:mtmHlbbEbS}). Combining this observation with (\ref{eq:bmuAmHlestimate}),
  (\ref{eq:muAmHlestimatermk}), (\ref{eq:bmuAmclmfwestEi}), (\ref{eq:muAClestimatermk app}) and the assumptions yields the conclusion that 
  \begin{equation}\label{eq:SobTypeoneterms}
    \ldr{\tau}^{-(|\bfI|+3)\cweight}\|E_{\bfI}[e^{h}X_{A}(f)]\|_{2}\leq C\theta_{0,-}^{-2}\ldr{\tau}^{m\cweight+m+1}e^{2\vare_{R}\tau}
  \end{equation}
  for all $\tau\leq 0$, where $C$ only depends on $c_{\cweight,1}$, $s_{\cweight,l}$, $\e_{\que}$, $K_{q}$ and $(\bM,\bge_{\refer})$. Since $m\leq l$, we can
  replace $m$ by $l$ in (\ref{eq:SobTypeoneterms}). The argument in case $f=\g^{D}_{BC}$ is similar and the result is slightly better; see
  (\ref{eq:EbfIXDgABCitomK}). 

  Next, consider terms of type two, where we use the terminology introduced in the proof of Lemma~\ref{lemma:CkRicciestimates}. In this case we need to
  estimate the right hand side of (\ref{eq:termstobeestimatedtypetwoterm}) in $L^{2}$. The argument is quite similar to the above and leads to the
  conclusion that
  \[
  \ldr{\tau}^{-(|\bfI|+3)\cweight}\|E_{\bfI}[e^{h}f_{1}f_{2}]\|_{2}\leq C\theta_{0,-}^{-2}\ldr{\tau}^{l\cweight+l+1}e^{2\vare_{R}\tau}
  \]
  for all $\tau\leq 0$, where $C$ only depends on $c_{\cweight,1}$, $s_{\cweight,l}$, $\e_{\que}$, $K_{q}$ and $(\bM,\bge_{\refer})$. This yields estimates
  for the components of $\bmR$. In order to prove (\ref{eq:SobestofRicciqugeo}), it is then sufficient to combine the argument at the end of the proof
  of Lemma~\ref{lemma:CkRicciestimates} with (\ref{eq:bDbfAellAetcpteststmtEi}), a $C^{0}$-estimate for $\bmR$ (that results when applying
  Lemma~\ref{lemma:CkRicciestimates} with $l=1$) and arguments similar to the above.
\end{proof}

Next, we need to estimate $\bmN$. 

\begin{lemma}\label{lemma:SobbmNestimates}
  Let $0\leq\cweight\in\ro$, $1\leq l\in\zo$ and assume that the standard assumptions; see Definition~\ref{def:standardassumptions}; the
  $(\cweight,l)$-Sobolev assumptions; and the $(\cweight,1)$-supremum assumptions are fulfilled. Assume, moreover, that there are constants $K_{q}$
  and $0<\e_{\que}<1$ such that (\ref{eq:qexpconvergence}) holds and an $\e_{p}\in (0,1)$ such that (\ref{eq:quiescentregime}) holds.
  Then there is a constant $K_{\mN,l}$ such that
  \begin{align}
    \|\bmN(\cdot,\tau)\|_{H^{l-1}_{\weight_{1}}(\bM)} \leq & K_{\mN,l}\ldr{\tau}^{(l+1)(\cweight+1)}e^{2\vare_{R}\tau}\label{eq:SobestofbmNqugeo}
  \end{align}
  for all $\tau\leq 0$, where $\weight_{1}:=(2\cweight,\cweight)$, $K_{\mN,l}=K_{\mN,l,0}\theta_{0,-}^{-2}$ and $K_{\mN,l,0}$ only depends on $s_{\cweight,l}$,
  $c_{\cweight,1}$, $\e_{\que}$, $K_{q}$ and $(\bM,\bge_{\refer})$. Here $\e_{R}:=\min\{\e_{p},\e_{\Spe}\}$ and $\vare_{R}:=\e_{R}/(3K_{\rovar})$, where $K_{\rovar}$
  is introduced in the statement of Lemma~\ref{lemma:smallnessshiftconsequences}.
\end{lemma}
\begin{remark}\label{remark:SobbmNestimatesymmetry}
  A similar conclusion holds in case we replace the requirement that (\ref{eq:qexpconvergence}) and (\ref{eq:quiescentregime}) hold with the requirement
  that $\g^{A}_{BC}=0$ for all $A<\min\{B,C\}$. Moreover, in that case $\vare_{R}$ can be replaced by $\vare_{\Spe}$ and the constants depend neither on
  $\e_{\que}$ nor on $K_{q}$. The proof of this statement is similar to, but simpler than, the proof below (in that some of the terms we need to estimate
  below vanish in this case).
\end{remark}
\begin{proof}
  The argument is similar to the proof of Lemma~\ref{lemma:SobRicciestimates}, keeping
  (\ref{eq:bnablaABnormN})--(\ref{eq:normalisedDeltaNthroughN}) in mind. The only difference is that the $f$ appearing in terms of type
  one (using the terminology introduced in the proof of Lemma~\ref{lemma:CkRicciestimates}) equals $X_{B}(\ln N)$ and the $h$ equals $-2\mu_{A}$.
  Concerning terms of type two, the $h$'s are the same and the $f_{i}$'s can, in addition to the previous possibilities, also equal $X_{C}(\ln N)$.
  However, due to the assumptions, $\ln N$ satisfies better bounds than $\bmu_{A}$ in spaces of the form $\mc^{\bfl_{a}}_{\bbE,\weight_{a}}(\bM)$
  and $\mH^{\bfl_{a}}_{\bbE,\weight_{a}}(\bM)$, assuming $\bfl_{a}$ is of the form $\bfl_{a}=(1,l_{a})$. For these reasons, the result is a consequence of
  an argument similar to the proof of Lemma~\ref{lemma:SobRicciestimates}. 
\end{proof}

Finally, let us turn to $\hml_{U}\mK$ in the vacuum setting. 

\begin{thm}\label{thm:SobhmlUmKestimates}
  Let $0\leq\cweight\in\ro$, $n\geq 3$, $1\leq l\in\zo$ and assume that the standard assumptions; see Definition~\ref{def:standardassumptions}; the
  $(\cweight,l)$-Sobolev assumptions; and the $(\cweight,1)$-supremum assumptions are fulfilled. In particular, the spacetime is $n+1$-dimensional.
  Assume, moreover, that there are constants $K_{q}$ and $0<\e_{\que}<1$ such that (\ref{eq:qexpconvergence}) holds and an $\e_{p}\in (0,1)$ such that
  (\ref{eq:quiescentregime}) holds. Assume, finally, that Einstein's vacuum equations with a cosmological constant $\Lambda$ are satisfied; i.e.,
  (\ref{eq:EE}) holds with $T=0$. Then there is a constant $K_{\mK,l}$ such that
  \begin{align}
    \|\hml_{U}\mK(\cdot,\tau)\|_{H^{l-1}_{\weight_{1}}(\bM)} \leq & K_{\mK,l}\ldr{\tau}^{(l+1)(\cweight+1)}e^{2\vare_{R}\tau}\label{eq:SobestofhmlUmKqugeo}
  \end{align}
  for all $\tau\leq 0$, where $\weight_{1}:=(2\cweight,\cweight)$, $K_{\mK,l}=K_{\mK,l,0}\theta_{0,-}^{-2}$ and $K_{\mK,l,0}$ only depends on $s_{\cweight,l}$,
  $c_{\cweight,1}$, $\e_{\que}$, $K_{q}$, $\Lambda$ and $(\bM,\bge_{\refer})$. Moreover, there is a constant $K_{q,l}$ such that
  \begin{align}
    \|[q-(n-1)](\cdot,\tau)\|_{H^{l-1}_{\weight_{1}}(\bM)} \leq & K_{q,l}\ldr{\tau}^{(l+1)(\cweight+1)}e^{2\vare_{R}\tau}\label{eq:Sobestofqminusnminusonequgeo}
  \end{align}
  for all $\tau\leq 0$, where $K_{q,l}=K_{q,l,0}\theta_{0,-}^{-2}$ and $K_{q,l,0}$ only depends on $s_{\cweight,l}$, $c_{\cweight,1}$, $\e_{\que}$, $K_{q}$,
  $\Lambda$ and $(\bM,\bge_{\refer})$. Here $\e_{R}:=\min\{\e_{p},\e_{\Spe}\}$ and $\vare_{R}:=\e_{R}/(3K_{\rovar})$, where $K_{\rovar}$
  is introduced in the statement of Lemma~\ref{lemma:smallnessshiftconsequences}.
\end{thm}
\begin{remark}\label{remark:SobestofhmlUmKsymmetry}
  A similar conclusion holds in case we replace the requirement that (\ref{eq:qexpconvergence}) and (\ref{eq:quiescentregime}) hold with the requirement
  that $\g^{A}_{BC}=0$ for all $A<\min\{B,C\}$. This follows by appealing to Remarks~\ref{remark:Sobscalarcurvatureestimatesymmetry} and
  \ref{remark:SobbmNestimatesymmetry} and an argument similar to the proof below. However, in this case, $\vare_{R}$ can be replaced by
  $\vare_{\Spe}:=\e_{\Spe}/(3K_{\rovar})$ in (\ref{eq:SobestofhmlUmKqugeo}) and (\ref{eq:Sobestofqminusnminusonequgeo}). Moreover, the constants depend on
  neither $K_{q}$ nor $\e_{\que}$. 
\end{remark}
\begin{proof}
  The expression we need to estimate is the right hand side of (\ref{eq:hmlUvacuumsetting}). In order to estimate the last two terms, it is sufficient
  to appeal to Lemmas~\ref{lemma:SobRicciestimates} and \ref{lemma:SobbmNestimates}. Turning to the third last term, note that
  (\ref{eq:OmegaLambdaestimate}) holds. What we need to estimate is thus
  \[
  \ldr{\varrho}^{-(l_{a}+2)\cweight}e^{2\e_{\Spe}\varrho}\textstyle{\prod}_{m=1}^{k}|E_{\bfI_{m}}E_{i_{m}}\ln\theta|
  \]
  in $L^{2}$, where $l_{a}=|\bfI|$ and $|\bfI_{1}|+\dots+|\bfI_{k}|=l_{a}-k$. In case $k=0$, the product should be interpreted as $1$. At this stage, we
  can use the equivalence between $\ldr{\tau}$ and $\ldr{\varrho}$, the fact that $e^{2\e_{\Spe}\varrho}\leq e^{2\vare_{\Spe}\tau}$ and
  Corollary~\ref{cor:mixedmoserestweight} to estimate this expression in $L^{2}$. In the end, it is thus sufficient to estimate
  \[
  \ldr{\tau}^{-2\cweight}e^{2\vare_{\Spe}\tau}\|\ln\theta\|_{\mc^{\bfl_{0}}_{\bbE,\weight_{0}}(\bM)}^{k-1}\|\ln\theta\|_{\mH^{\bfl}_{\bbE,\weight_{0}}(\bM)}.
  \]
  Here $\bfl_{0}=(1,1)$, $\bfl=(1,l)$ and the last two factors are bounded due to the assumptions and (\ref{eq:Clmclequiv}). The resulting estimate,
  combined with (\ref{eq:Clmclequiv}), yields
  \[
  \|\Lambda\theta^{-2}\mathrm{Id}\|_{H^{l-1}_{\weight_{1}}(\bM)}\leq C_{a}\theta_{0,-}^{-2}\ldr{\tau}^{-2\cweight}e^{2\vare_{\Spe}\tau},
  \]
  where $C_{a}$ only depends on $s_{\cweight,l}$, $\Lambda$ and $(\bM,\bge_{\refer})$. 

  The first term on the right hand side of (\ref{eq:hmlUvacuumsetting}) can be written $f\mK$. In order to estimate this expression in
  $\mH^{l-1}_{\weight_{1}}(\bM)$, we need to estimate expressions of the form 
  \[
  \ldr{\varrho}^{-(|\bfI|+2)\cweight}\bD_{\bfI_{1}}f\bD_{\bfI_{2}}\mK
  =\ldr{\varrho}^{-(|\bfI_{1}|+2)\cweight}\bD_{\bfI_{1}}f\cdot \ldr{\varrho}^{-|\bfI_{2}|\cweight}\bD_{\bfI_{2}}\mK
  \]
  in $L^{2}$, where $|\bfI_{1}|+|\bfI_{2}|\leq l-1$. However, appealing to Corollary~\ref{cor:mixedmoserestweight}, (\ref{eq:Clmclequiv}) and the equivalence
  between $\ldr{\tau}$ and $\ldr{\varrho}$, this expression can be estimated in $L^{2}$ by
  \[
  C\|f\|_{H^{l-1}_{\weight_{1}}(\bM)}+C\|f\|_{C^{0}_{\weight_{1}}(\bM)}\|\mK\|_{H^{l-1}_{\weight_{0}}(\bM)},
  \]
  where $C$ only depends on $c_{\robas}$, $|\bfI|$, $\weight_{1}$ and $(\bM,\bge_{\refer})$. The function $f$ can be estimated in the desired Sobolev
  space the same way that the last three terms on the right hand side of (\ref{eq:hmlUvacuumsetting}) were estimated. The function $f$ can also be
  estimated in $C^{0}_{\weight_{1}}(\bM)$ by appealing to (\ref{eq:OmegaLambdaestimate}), Remark~\ref{remark:Ckscalarcurvatureestimate} and
  Lemma~\ref{lemma:Cklapseestimates}. Finally, $\mK$ can be estimated in the desired Sobolev space due to the assumptions. Appealing to
  (\ref{eq:Clmclequiv}) again yields (\ref{eq:SobestofhmlUmKqugeo}). Keeping (\ref{eq:qmainformula appendix}), (\ref{eq:bSexpression}) and
  (\ref{eq:rescaledRiccicurvatureformwoderoflntheta}) in mind, the proof of (\ref{eq:Sobestofqminusnminusonequgeo}) is similar but simpler. 
\end{proof}

Next, we wish to prove that $\mK$ converges.

\begin{thm}\label{thm:SobconvofmK}
  Let $0\leq\cweight\in\ro$, $n\geq 3$, $1\leq l\in\zo$ and assume that the standard assumptions; see Definition~\ref{def:standardassumptions}; the
  $(\cweight,l)$-Sobolev assumptions; and the $(\cweight,1)$-supremum assumptions are fulfilled. In particular, the spacetime is $n+1$-dimensional.
  Assume, moreover, that there are constants $K_{q}$ and $0<\e_{\que}<1$ such that (\ref{eq:qexpconvergence}) holds and an $\e_{p}\in (0,1)$ such that
  (\ref{eq:quiescentregime}) holds. Assume, finally, that Einstein's vacuum equations with a cosmological constant $\Lambda$ are satisfied; i.e.,
  (\ref{eq:EE}) holds with $T=0$; and that $\varrho$ diverges uniformly to $-\infty$ in the
  direction of the singularity. Then there is a continuous $(1,1)$-tensor field $\mK_{\infty}$ and continuous functions
  $\ell_{A,\infty}$ on $\bM$, which also belong to $H^{l-1}(\bM)$, and constants $K_{\infty,i,l}$, $i=0,1$, such that
  \begin{equation}\label{eq:mKconvtomKinfHl}
    \begin{split}
      & \|\mK(\cdot,\tau)-\mK_{\infty}\|_{H^{l-1}(\bM)}+\textstyle{\sum}_{A}\|\ell_{A}(\cdot,\tau)-\ell_{A,\infty}\|_{H^{l-1}(\bM)}\\
      \leq & K_{\infty,0,l}\theta_{0,-}^{-2}\ldr{\tau}^{(l+1)(2\cweight+1)}e^{2\vare_{R}\tau}+K_{\infty,1,l}\ldr{\tau}^{l\cweight}e^{\vare_{\Spe}\tau}
    \end{split}
  \end{equation}
  for all $\tau\leq 0$, where $K_{\infty,0,l}$ only depends on $s_{\cweight,l}$, $c_{\cweight,1}$, $\e_{p}$, $\e_{\que}$, $K_{q}$, $\Lambda$ and
  $(\bM,\bge_{\refer})$; and $K_{\infty,1,l}$ only depends on $s_{\cweight,l}$, $c_{\cweight,1}$ and $(\bM,\bge_{\refer})$.
\end{thm}
\begin{remark}\label{remark:SobconvofmKsymmetry}
  A similar conclusion holds in case we replace the requirement that (\ref{eq:qexpconvergence}) and (\ref{eq:quiescentregime}) hold with the requirement
  that $\g^{A}_{BC}=0$ for all $A<\min\{B,C\}$. This follows by appealing to Remark~\ref{remark:SobestofhmlUmKsymmetry} and an argument similar to the proof
  below. However, in this case, $\vare_{R}$ can be replaced by $\vare_{\Spe}$ in (\ref{eq:mKconvtomKinfHl}).
  Moreover, the constants depend on neither $K_{q}$, $\e_{\que}$ nor $\e_{p}$. 
\end{remark}
\begin{proof}
  We need to estimate the right hand side of (\ref{eq:dtaumKijform}) in $\mH^{l-1}_{\weight_{1}}(\bM)$. Note, to this end, that
  \[
  \ldr{\varrho}^{-(|\bfI|+2)\cweight}|E_{\bfI}[\xi(\hml_{U}\mK)^{i}_{\phantom{i}j}]|
  \]
  can be estimated by a linear combination of terms of the form
  \[
  |\xi|\textstyle{\prod}_{j=1}^{k}\ldr{\tau}^{-(|\bfI_{j}|+1)\cweight}|E_{\bfI_{j}}E_{i_{j}}\ln\hN|
  \cdot \ldr{\tau}^{-(|\bfJ|+2)\cweight}|\bD_{\bfJ}\hml_{U}\mK|_{\bge_{\refer}},
  \]
  where the relevant constants only depend on $c_{\robas}$, $|\bfI|$ and $(\bM,\bge_{\refer})$; and $|\bfI|=|\bfI_{1}|+\dots+|\bfI_{k}|+k+|\bfJ|$. However,
  this expression can be estimated in $L^{2}$ by appealing to the equivalence of $\ldr{\varrho}$ and $\ldr{\tau}$, (\ref{eq:Clmclequiv})
  and Corollary~\ref{cor:mixedmoserestweight}. This results in expressions of the form
  \[
  C\|\hml_{U}\mK\|_{H^{l-1}_{\weight_{1}}(\bM)}+C\|\hml_{U}\mK\|_{C^{0}_{\weight_{1}}(\bM)}\|\ln\hN\|_{H^{\bfl_{a}}_{\weight_{0}}(\bM)}
  \]
  where $\bfl_{a}=(1,l-1)$ (in case $l=1$, the second term is absent) and $C$ only depends on $c_{\robas}$, $|\bfI|$ and $(\bM,\bge_{\refer})$. At this
  stage, it is sufficient to appeal to Theorems~\ref{thm:CkestofhmlUmK} and \ref{thm:SobhmlUmKestimates} as well as the assumptions.
  In fact,
  \begin{equation}\label{eq:phihmlUmKest}
    \|\xi (\hml_{U}\mK)^{i}_{\phantom{i}j}\|_{\mH^{l-1}_{\weight_{1}}(\bM)}\leq K_{\mK,l}\ldr{\tau}^{(l+1)(\cweight+1)}e^{2\vare_{R}\tau}
  \end{equation}
  for $\tau\leq 0$, where $K_{\mK,l}$ is of the same form as in Theorem~\ref{thm:SobhmlUmKestimates}. 

  In order to estimate the second term on the right hand side of (\ref{eq:dtaumKijform}), note that
  \[
  E_{\bfI}[\xi\hN^{-1}(\ml_{\chi}\mK)^{i}_{\phantom{i}j}]=(\d_{t}\tau)^{-1}E_{\bfI}[(\ml_{\chi}\mK)^{i}_{\phantom{i}j}].
  \]
  Combining this observation with (\ref{eq:mlchimKformforest}) yields the conclusion that it is sufficient to estimate terms of the form
  \[
  (\d_{t}\tau)^{-1}|\bD_{\bfI_{a}}\mK|_{\bge_{\refer}}|\bD_{\bfI_{b}}\chi|_{\bge_{\refer}},
  \]
  where $|\bfI_{a}|+|\bfI_{b}|\leq |\bfI|+1$. Multiplying this expression by $\ldr{\varrho}^{-(|\bfI|+2)\cweight}$ and estimating the result
  in $L^{2}$ yields the conclusion that it is bounded by 
  \[
  C(\d_{t}\tau)^{-1}\|\chi\|_{\mH^{l}_{\bbE,\weight}(\bM)}+C(\d_{t}\tau)^{-1}\|\chi\|_{C^{0}_{\weight}(\bM)}\|\mK\|_{H^{l}_{\weight_{0}}(\bM)},
  \]
  assuming that $|\bfI|\leq l-1$. Here $\weight=(\cweight,\cweight)$ and $C$ only depends on $c_{\robas}$, $l$ and $(\bM,\bge_{\refer})$ and we appealed
  to Corollary~\ref{cor:mixedmoserestweight} and the equivalence of $\ldr{\tau}$ and $\ldr{\varrho}$. On the other hand,
  \[
  (\d_{t}\tau)^{-1}\|\chi\|_{\mH^{l}_{\bbE,\weight}(\bM)}\leq C\|\chi\|_{\mH^{l,\weight}_{\bbE,\rocon}(\bM)}\leq C\ldr{\tau}^{-\cweight}e^{\vare_{\Spe}\tau}
  \]
  for all $\tau\leq 0$, where $C$ only depends on $s_{\cweight,l}$ and $(\bM,\bge_{\refer})$ and we appealed to (\ref{eq:hNtaudotequivEi}),
  (\ref{eq:mHlbbeWnorm}), (\ref{eq:chimHlHlest}) and the assumptions. Similarly,
  \[
  (\d_{t}\tau)^{-1}\|\chi\|_{C^{0}_{\weight}(\bM)}\leq C\|\chi\|_{C^{0,\weight}_{\bbE,\rocon}(\bM)}\leq C\ldr{\tau}^{-\cweight}e^{\vare_{\Spe}\tau}
  \]
  for all $\tau\leq 0$, where $C$ only depends on $c_{\cweight,1}$ and $(\bM,\bge_{\refer})$, and we appealed to (\ref{eq:hNtaudotequivEi}),
  (\ref{eq:mClbbeWnorm}), (\ref{eq:chimClClest}) and the assumptions. Combining the above estimates, we conclude that
  \[
  \|\xi\hN^{-1}(\ml_{\chi}\mK)^{i}_{\phantom{i}j}\|_{\mH^{l-1}_{\weight_{1}}(\bM)}\leq C\ldr{\tau}^{-\cweight}e^{\vare_{\Spe}\tau}
  \]
  for all $\tau\leq 0$, where $C$ only depends on $s_{\cweight,l}$, $c_{\cweight,1}$ and $(\bM,\bge_{\refer})$. Combining this estimate
  with (\ref{eq:phihmlUmKest}) yields the conclusion that
  \[
  \|\d_{\tau}\mK^{i}_{\phantom{i}j}\|_{\mH^{l-1}(\bM)}\leq K_{\mK,l}\ldr{\tau}^{(l+1)(2\cweight+1)}e^{2\vare_{R}\tau}
  +C_{l}\ldr{\tau}^{l\cweight}e^{\vare_{\Spe}\tau}
  \]
  for all $\tau\leq 0$, where $C_{l}$ only depends on $s_{\cweight,l}$, $c_{\cweight,1}$ and $(\bM,\bge_{\refer})$; and $K_{\mK,l}$ is of the same form as in
  Theorem~\ref{thm:SobhmlUmKestimates}. Integrating this estimate yields the desired conclusion concerning $\mK$ (that $\mK_{\infty}$ is continuous
  follows from the fact that Theorem~\ref{thm:mKconvCk} applies).

  Next, consider (\ref{eq:dtauellA}). Due to the arguments presented in connection with (\ref{eq:EbfIphiterms}) and (\ref{eq:dtauellA}), applying
  $E_{\bfI}$ to the first term on the right hand side of (\ref{eq:dtauellA}) leads to terms of the form
  \[
  \xi |E_{\bfI_{1}}\ln\hN|\cdots |E_{\bfI_{p}}\ln\hN|\cdot|\bD^{m_{1}}\mK|_{\bge_{\refer}}\cdots |\bD^{m_{j}}\mK|_{\bge_{\refer}}|\bD^{k}\hml_{U}\mK|_{\bge_{\refer}},
  \]
  where $\bfI_{r}\neq 0$, $m_{i}\neq 0$ and $|\bfI_{1}|+\dots+|\bfI_{p}|+m_{1}+\dots+m_{j}+k\leq |\bfI|$. The argument necessary to estimate this
  expression in weighted Sobolev spaces is quite similar to the estimate of the first term on the right hand side of (\ref{eq:dtaumKijform}).
  Keeping the arguments at the end of the proof of Theorem~\ref{thm:mKconvCk} in mind, the estimate for the second term on the right hand side of
  (\ref{eq:dtauellA}) is similar to the estimate of the second term on the right hand side of (\ref{eq:dtaumKijform}). Integrating the resulting
  estimate yields the desired conclusion concerning $\ell_{A}$ (that the $\ell_{A,\infty}$ are continuous again follows from the fact that
  Theorem~\ref{thm:mKconvCk} applies). 
\end{proof}

\section{Scalar field matter}\label{section:scalarfieldmatter}

So far, we have considered the quiescent vacuum setting. Due to the comments made in Subsection~\ref{ssection:quiescentregimes}, we then have to
have $n\geq 10$ or $\g^{A}_{BC}$ ($A<\min\{B,C\}$) either vanishing or tending to zero sufficiently quickly. However, it is also of interest to
consider situations with quiescent asymptotics for an open set of initial data and $n=3$. In order for this to occur, we need
to introduce matter.  In the present section, we discuss the case of a massless scalar field. The corresponding stress energy tensor is given by
\begin{equation}\label{eq:setnlsf}
  T_{\a\b}=\nabla_{\a}\phi\nabla_{\b}\phi-\textstyle{\frac{1}{2}}(\nabla_{\g}\phi\nabla^{\g}\phi)g_{\a\b}
\end{equation}
(even though we focus on the massless case here, similar results should hold in the presence of a potential $V$, assuming $V(s)$ does not
grow too quickly as $s\rightarrow\pm\infty$). Next, we wish to estimate the matter quantities appearing in (\ref{eq:qmainformula}) and
(\ref{eq:mlUmKwithEinstein}) in this particular setting; i.e., $\rho-\bp$ and $\mcP$. 

\begin{lemma}
  Let $(M,g)$ be a spacetime. Assume that it has an expanding partial pointed foliation. Assume, moreover, $\mK$ to be non-degenerate on $I$ and to
  have a global frame. Assume, finally, that Einstein's equations (\ref{eq:EE}) with a cosmological constant $\Lambda$ are satisfied where the
  stress energy tensor is of the form (\ref{eq:setnlsf}). Then, using the notation introduced in (\ref{eq:mfpmcPdef}), 
  \begin{align}
    \theta^{-2}\mfp^{A}_{\phantom{A}B}  = & e^{-2\mu_{A}}X_{A}(\phi)X_{B}(\phi)+\textstyle{\frac{1}{2}}(\hU\phi)^{2}\de^{A}_{B}
    -\textstyle{\frac{1}{2}\sum}_{C}e^{-2\mu_{C}}|X_{C}(\phi)|^{2}\de^{A}_{B},\label{eq:thetamtwomfpAB}\\
    \theta^{-2}\bp  = & \textstyle{\frac{1}{2}}(\hU\phi)^{2}-\textstyle{\frac{n-2}{2n}\sum}_{A}e^{-2\mu_{A}}|X_{A}(\phi)|^{2},\label{eq:bpformula}\\
    \theta^{-2}\rho  = & \textstyle{\frac{1}{2}}(\hU\phi)^{2}+\textstyle{\frac{1}{2}\sum}_{A}e^{-2\mu_{A}}|X_{A}(\phi)|^{2},\label{eq:rhoformula}
  \end{align}
  where there is no summation on $A$ in the first term on the right hand side of (\ref{eq:thetamtwomfpAB}). 
\end{lemma}
\begin{remark}\label{remark:formulaerhominusbpmcP}
  Note, in particular, that
  \begin{align*}
    \theta^{-2}(\rho-\bp) = & \textstyle{\frac{n-1}{n}\sum}_{A}e^{-2\mu_{A}}|X_{A}(\phi)|^{2},\\
    \theta^{-2}\mcP^{A}_{\phantom{A}B} = & e^{-2\mu_{A}}X_{A}(\phi)X_{B}(\phi)-\textstyle{\frac{1}{n}}\sum_{C}e^{-2\mu_{C}}|X_{C}(\phi)|^{2}\de^{A}_{B}.
  \end{align*}
\end{remark}
\begin{proof}
  Compute
  \begin{equation*}
    \begin{split}
      \mfp^{A}_{\phantom{A}B}  = & e^{-2\bmu_{A}}X_{A}(\phi)X_{B}(\phi)+\textstyle{\frac{1}{2}}(U\phi)^{2}\de^{A}_{B}
      -\textstyle{\frac{1}{2}\sum}_{C}e^{-2\bmu_{C}}|X_{C}(\phi)|^{2}\de^{A}_{B}.
    \end{split}
  \end{equation*}
  The lemma follows. 
\end{proof}

In order to estimate the scalar field, it is sufficient to appeal to the fact that it satisfies the wave equation $\Box_{g}\phi=0$ and to the results
of \cite{RinWave}. In fact, we have the following proposition. 

\begin{prop}\label{prop:asvelocityKGlikeeq}
  Let $0\leq\cweight\in\ro$, $n\geq 3$, $1\leq l\in\zo$ and assume the standard assumptions; see Definition~\ref{def:standardassumptions}; and the
  $(\cweight,l)$-Sobolev assumptions to be fulfilled; see Definition~\ref{def:sobklassumptions}. In particular, the spacetime is $n+1$-dimensional.
  Let $\kappa_{1}$ be the smallest integer strictly larger than $n/2+1$ and assume the $(\cweight,\kappa_{1})$-supremum assumptions to be satisfied;
  see Definition~\ref{def:supmfulassumptions}. Assume, moreover, that there is a constant $d_{q}$ such that
  \begin{equation}\label{eq:qminusnminusoneest}
    \|\ldr{\varrho(\cdot,t)}^{3/2}[q(\cdot,t)-(n-1)]\|_{C^{0}(\bM)}\leq d_{q}
  \end{equation}
  for all $t\leq t_{0}$. Finally, let $\phi$ be a solution to $\Box_{g}\phi=0$. Then, if $l\geq \kappa_{1}$,
  \begin{align}
    \hGe_{l}(\tau) \leq & C_{a}\ldr{\tau}^{2\a_{l,n}\cweight+2\b_{l,n}}\hGe_{l}(0),\label{eq:TheEnergyEstimateKG}\\
    \|\me_{1}(\cdot,\tau)\|_{C^{0}(\bM)} \leq & C_{b}\ldr{\tau}^{\a_{n}\cweight+\b_{n}}\hGe_{\kappa_{1}}(0)\label{eq:meksupestimateKG}
  \end{align}
  for all $\tau\leq 0$. Here $\a_{l,n}$ and $\b_{l,n}$ only depend on $n$ and $l$; and $C_{a}$ only depends on $s_{\cweight,l}$, 
  $c_{\cweight,\kappa_{1}}$, $d_{q}$, $(\bM,\bge_{\refer})$ and a lower bound on $\theta_{0,-}$. Moreover, $\a_{n}$ and
  $\b_{n}$ only depend on $n$; and $C_{b}$ only depends on $c_{\cweight,\kappa_{1}}$, $d_{q}$, $(\bM,\bge_{\refer})$
  and a lower bound on $\theta_{0,-}$. Finally,
  \begin{align}
    \me_{k} := & \tfrac{1}{2}\textstyle{\sum}_{|\bfI|\leq k}\left(|\hU(E_{\bfI}\phi)|^{2}+\textstyle{\sum}_{A}e^{-2\mu_{A}}|X_{A}(E_{\bfI}\phi)|^{2}
    +\ldr{\tau}^{-3}|E_{\bfI}\phi|^{2}\right),\label{eq:mekdef}\\
    \hGe_{l}(\tau) := & \textstyle{\int}_{\bM_{\tau}}\me_{l}\mu_{\bge_{\refer}}.\label{eq:hGeldef}
  \end{align}  
\end{prop}
\begin{proof}
  The statement is an immediate consequence of \cite[Proposition~14.24]{RinWave}.
\end{proof}

Given these estimates, we can derive results similar to Theorems~\ref{thm:SobhmlUmKestimates} and \ref{thm:SobconvofmK}.

\begin{thm}\label{thm:SobconvofmKdecofhmlUmKmatter}
  Let $0\leq\cweight\in\ro$, $n\geq 3$, $\kappa_{1}$ be the smallest integer strictly larger than $n/2+1$, $\kappa_{1}\leq l\in\zo$ and assume the
  standard assumptions; see Definition~\ref{def:standardassumptions}; the $(\cweight,l)$-Sobolev assumptions; and the $(\cweight,\kappa_{1})$-supremum
  assumptions to be fulfilled. In particular, the spacetime is $n+1$-dimensional. Assume, moreover, that there are constants $K_{q}$ and $0<\e_{\que}<1$
  such that (\ref{eq:qexpconvergence}) holds and an $\e_{p}\in (0,1)$ such that (\ref{eq:quiescentregime}) holds. Assume, finally, that the
  Einstein-scalar field equations with a cosmological constant $\Lambda$ are satisfied; i.e., (\ref{eq:EE}) and $\Box_g\phi=0$ hold, where $T$ is given
  by (\ref{eq:setnlsf}); and that $\varrho$ diverges uniformly to $-\infty$ in the direction of the singularity. Then there is a continuous
  $(1,1)$-tensor field $\mK_{\infty}$ and continuous functions $\ell_{A,\infty}$ on $\bM$, which also belong to $H^{l-1}(\bM)$, and constants $a_{l,n}$,
  $b_{l,n}$, $\sfK_{\mK,l}$ and $\sfK_{\infty,l}$ such that
  \begin{align}
    \|\hml_{U}\mK(\cdot,\tau)\|_{H^{l-1}_{\weight_{1}}(\bM)} \leq & \sfK_{\mK,l}\ldr{\tau}^{a_{l,n}\cweight+b_{l,n}}e^{2\vare_{R}\tau},\label{eq:SobestofhmlUmKqugeomatter}\\
    \|[q-(n-1)](\cdot,\tau)\|_{H^{l-1}_{\weight_{1}}(\bM)} \leq & \sfK_{q,l}\ldr{\tau}^{a_{l,n}\cweight+b_{l,n}}e^{2\vare_{R}\tau},\label{eq:Sobestofqminusnminonequgeomatter}\\
    \|\mK(\cdot,\tau)-\mK_{\infty}\|_{H^{l-1}(\bM)} \leq & \sfK_{\infty,l}\ldr{\tau}^{a_{l,n}\cweight+b_{l,n}}e^{\vare_{\phi}\tau},\label{eq:mKconvtomKinfHlmatter}\\
    \textstyle{\sum}_{A}\|\ell_{A}(\cdot,\tau)-\ell_{A,\infty}\|_{H^{l-1}(\bM)} \leq & \sfK_{\infty,l}\ldr{\tau}^{a_{l,n}\cweight+b_{l,n}}e^{\vare_{\phi}\tau}
    \label{eq:ellAconvtomellAinfHlmatter}
  \end{align}
  for all $\tau\leq 0$, where $\sfK_{\mK,l}$, $\sfK_{q,l}$ and $\sfK_{\infty,l}$ only depend on $s_{\cweight,l}$, $c_{\cweight,\kappa_{1}}$, $\e_{p}$, $\e_{\que}$, $K_{q}$,
  $\Lambda$, $\hGe_{l}(0)$, $(\bM,\bge_{\refer})$ and a lower bound on $\theta_{0,-}$. Moreover, $\vare_{\phi}:=\min\{2\vare_{R},\vare_{\Spe}\}$. Finally,
  $a_{l,n}$ and $b_{l,n}$ only depend on $n$ and $l$. 
\end{thm}
\begin{remark}
  In case $\chi=0$ (i.e., the shift vector field vanishes), the estimates (\ref{eq:mKconvtomKinfHlmatter}) and (\ref{eq:ellAconvtomellAinfHlmatter}) can
  be improved in that $\vare_{\phi}$ can be replaced by $2\vare_{R}$. 
\end{remark}
\begin{remark}\label{remark:SobconvofmKdecofhmlUmKmattersymmetry}
  A similar conclusion holds in case we replace the requirement that (\ref{eq:qexpconvergence}) and (\ref{eq:quiescentregime}) hold with the requirements
  that $\g^{A}_{BC}=0$ for all $A<\min\{B,C\}$ and that (\ref{eq:qminusnminusoneest}) hold for $t\leq t_{0}$. In this case, $\vare_{R}$ can be
  replaced by $\vare_{\Spe}$ in (\ref{eq:SobestofhmlUmKqugeomatter}) and (\ref{eq:Sobestofqminusnminonequgeomatter}); and $\vare_{\phi}$ can be replaced
  by $\vare_{\Spe}$ in (\ref{eq:mKconvtomKinfHlmatter}) and (\ref{eq:ellAconvtomellAinfHlmatter}). Moreover, the constants depend on $d_{q}$ but neither
  on $K_{q}$, $\e_{\que}$ nor $\e_{p}$. Note also that, in this case, (\ref{eq:Sobestofqminusnminonequgeomatter}) represents not only a quantitative, but a
  qualitative, improvement of the assumptions. 
\end{remark}
\begin{proof}
  The proof is analogous to the proofs of Theorems~\ref{thm:CkestofhmlUmK}, \ref{thm:mKconvCk}, \ref{thm:SobhmlUmKestimates} and
  \ref{thm:SobconvofmK} (even though we only wish to prove the analogues of the last two theorems, we still need $C^{0}$-estimates; see the
  proofs of Theorems~\ref{thm:SobhmlUmKestimates} and \ref{thm:SobconvofmK}). In the proofs of Theorems~\ref{thm:CkestofhmlUmK} and
  \ref{thm:SobhmlUmKestimates}, we have already estimated all the terms on the right hand side of (\ref{eq:hmlUvacuumsetting}). Considering
  (\ref{eq:mlUmKwithEinstein}), it is thus clear that what remains to be estimated is $\theta^{-2}(\rho-\bp)$ and
  $\theta^{-2}\mcP^{A}_{\phantom{A}B}$. Let us begin by deriving $C^{0}$-estimates. Note that (\ref{eq:qexpconvergence}) implies that
  (\ref{eq:qminusnminusoneest}) holds (with $d_{q}$ depending only on $K_{q}$ and $\e_{\que}$). This means that the conditions of
  Proposition~\ref{prop:asvelocityKGlikeeq} are satisfied, and we use the conclusions of this proposition without further comment in what follows
  (this means, in paticular, that we use the fact that $\Box_g\phi=0$). Since $|E_{i}(\phi)|$ does not grow faster than polynomially; see
  (\ref{eq:meksupestimateKG}); it follows from (\ref{eq:muminmainlowerbound}), (\ref{eq:eSpevarrhoeelowtaurelEi}) and
  Remark~\ref{remark:formulaerhominusbpmcP} that
  \begin{equation}\label{eq:rhominusbpmcPABestimatescfield}
    \theta^{-2}|\rho-\bp|+\theta^{-2}\textstyle{\sum}_{A,B}|\mcP^{A}_{\phantom{A}B}|\leq C_{a}\ldr{\tau}^{\a\cweight+\b}e^{2\vare_{\Spe}\tau}
  \end{equation}
  for all $\tau\leq 0$, where $C_{a}$ only depends on $c_{\cweight,\kappa_{1}}$, $K_{q}$, $\e_{\que}$, $(\bM,\bge_{\refer})$ and a lower bound on $\theta_{0,-}$.
  Moreover, $\a$ and $\b$ only depend on $n$. Combining this estimate with (\ref{eq:qmainformula appendix}), (\ref{eq:mlUmKwithEinstein}) and the proof of
  Theorem~\ref{thm:CkestofhmlUmK}, it follows that
  \begin{equation}\label{eq:hmlUmkqminusnminusoneCzestimatescalarfield}
    \|\hml_{U}\mK\|_{C^{0}(\bM)}+\|q-(n-1)\|_{C^{0}(\bM)}\leq C_{a}\ldr{\tau}^{\a\cweight+\b}e^{2\vare_{R}\tau}
  \end{equation}
  for all $\tau\leq 0$, where $C_{a}$ only depends on $c_{\cweight,\kappa_{1}}$, $K_{q}$, $\e_{\que}$, $\Lambda$, $(\bM,\bge_{\refer})$ and a lower bound on
  $\theta_{0,-}$. Combining this estimate with the arguments of Theorem~\ref{thm:mKconvCk} yields the conclusion that $\mK$ and the $\ell_{A}$ converge
  in $C^{0}$ and that the rate is the one indicated in (\ref{eq:mKconvtomKinfHlmatter}) and (\ref{eq:ellAconvtomellAinfHlmatter}).

  The next step is to derive Sobolev estimates. Due to Remark~\ref{remark:formulaerhominusbpmcP}, the estimates we need to derive are similar to
  the ones derived for terms referred to as ``type two'' in the proof of Lemma~\ref{lemma:CkRicciestimates}. Moreover, due to
  (\ref{eq:TheEnergyEstimateKG}), $\phi$ satisfies estimates similar to those satisfied by $\bmu_{A}$. For that reason, the derivation of the estimates
  of the matter contribution to the right hand side of (\ref{eq:mlUmKwithEinstein}) is quite similar to the arguments of
  Lemma~\ref{lemma:SobRicciestimates}. This leads to
  (\ref{eq:SobestofhmlUmKqugeomatter}). The argument to prove (\ref{eq:Sobestofqminusnminonequgeomatter}) is essentially the same. The estimates
  (\ref{eq:mKconvtomKinfHlmatter}) and (\ref{eq:ellAconvtomellAinfHlmatter}) follow by arguments similar to the proof of
  Theorem~\ref{thm:SobconvofmK}. 
\end{proof}

\subsection{Asymptotics of the scalar field}

Next, we derive asymptotics for the scalar field.

\begin{thm}\label{thm:Sobasymptoticsscalarfieldmatter}
  Let $0\leq\cweight\in\ro$, $n\geq 3$, $\kappa_{1}$ be the smallest integer strictly larger than $n/2+1$, $\kappa_{1}\leq l\in\zo$ and assume the
  standard assumptions; see Definition~\ref{def:standardassumptions}; the $(\cweight,l)$-Sobolev assumptions; and the $(\cweight,\kappa_{1})$-supremum
  assumptions to be fulfilled. In particular, the spacetime is $n+1$-dimensional. Assume, moreover, that there are constants $K_{q}$ and $0<\e_{\que}<1$
  such that (\ref{eq:qexpconvergence}) holds and an $\e_{p}\in (0,1)$ such that (\ref{eq:quiescentregime}) holds. Assume, finally, that the
  Einstein-scalar field equations with a cosmological constant $\Lambda$ are satisfied; i.e., (\ref{eq:EE}) and $\Box_g\phi=0$ hold, where $T$ is given
  by (\ref{eq:setnlsf}); and that $\varrho$ diverges uniformly to $-\infty$ in the direction of the singularity. Then there is a function
  $\Psi_{\infty}\in C^{0}(\bM)\cap H^{l-1}(\bM)$ and constants $\sfK_{\phi,l}$, $a_{l,n}$ and $b_{l,n}$ such that
  \begin{equation}\label{eq:asofhUphi}
    \|[\hU(\phi)](\cdot,\tau)-\Psi_{\infty}\|_{H^{l-1}(\bM)}\leq \sfK_{\phi,l}\ldr{\tau}^{\a_{l,n}\cweight+\b_{l,n}}e^{\vare_{\phi}\tau}
  \end{equation}  
  for all $\tau\leq 0$, where $\a_{l,n}$ and $\b_{l,n}$ only depend on $n$ and $l$ and $\sfK_{\phi,l}$ only depends on $s_{\cweight,l}$,
  $c_{\cweight,\kappa_{1}}$, $\hGe_{l}(0)$, $K_{q}$, $\e_{\que}$, $\e_{p}$, $\Lambda$, $(\bM,\bge_{\refer})$ and a lower bound on $\theta_{0,-}$. Moreover,
  $\vare_{\phi}:=\min\{2\vare_{R},\vare_{\Spe}\}$. If, in addition, $l\geq \kappa_{1}+1$, there is a function $\Phi_{\infty}\in C^{0}(\bM)\cap H^{l-2}(\bM)$
  and constants $\sfK_{\phi,2,l}$, $a_{l,n}$ and $b_{l,n}$ such that
  \begin{equation}\label{eq:asofhphi}
    \|\phi(\cdot,\tau)-\Psi_{\infty}\varrho(\cdot,\tau)-\Phi_{\infty}\|_{H^{l-2}(\bM)}\leq \sfK_{\phi,2,l}\ldr{\tau}^{\a_{l,n}\cweight+\b_{l,n}}e^{\vare_{\phi}\tau}
  \end{equation}  
  for all $\tau\leq 0$, where $\a_{l,n}$, $\b_{l,n}$ and $\sfK_{\phi,2,l}$ have the same dependence as the corresponding constants in (\ref{eq:asofhUphi}).
\end{thm}
\begin{remark}\label{remark:asHamConSF}
  Combining (\ref{eq:OmegaplussumellAsqbdpnegsc}) with Remark~\ref{remark:Ckscalarcurvatureestimate}; Theorem~\ref{thm:SobconvofmKdecofhmlUmKmatter};
  (\ref{eq:rhoformula}); the fact that the second term on the right hand side of (\ref{eq:rhoformula}) can be estimated by the right hand side
  of (\ref{eq:rhominusbpmcPABestimatescfield}); and (\ref{eq:asofhUphi}) yields the conclusion that
  \[
  \Psi_{\infty}^{2}+\textstyle{\sum}_{A}\ell_{A,\infty}^{2}=1.
  \]
\end{remark}
\begin{remark}\label{remark:Sobasymptoticsscalarfieldmattersymmetry}
  A similar conclusion holds in case we replace the requirement that (\ref{eq:qexpconvergence}) and (\ref{eq:quiescentregime}) hold with the requirements
  that $\g^{A}_{BC}=0$ for $A<\min\{B,C\}$ and that (\ref{eq:qminusnminusoneest}) hold for $t\leq t_{0}$. In this case, $\vare_{\phi}$ is replaced
  by $\vare_{\Spe}$. Moreover, the constants depend on $d_{q}$ but neither on $K_{q}$, $\e_{\que}$ nor $\e_{p}$. 
\end{remark}
\begin{proof}
  The equation $\Box_{g}\phi=0$ can be written
  \begin{equation}\label{eq:waveequforphi}
    -\hU^{2}(\phi)+\textstyle{\sum}_{A}e^{-2\mu_{A}}X_{A}^{2}(\phi)+\tfrac{1}{n}[q-(n-1)]\hU(\phi)+\hmcY^{A}X_{A}(\phi)=0,
  \end{equation}
  where
  \begin{equation*}
    \begin{split}
      \hmcY^{A} := & e^{-2\mu_{A}}X_{A}(\ln \hN)-2e^{-2\mu_{A}}X_{A}(\mu_{A})+e^{-2\mu_{A}}X_{A}(\mu_{\rotot})\\
      & -2e^{-2\mu_{A}}\alpha_{A}-(n-1)e^{-2\mu_{A}}X_{A}(\ln\theta),
    \end{split}
  \end{equation*}
  where $\alpha_A:=\g_{AB}^B/2$; see (\ref{eq:waveequforphiapp}) and (\ref{eq:hmcYA def}).

  The first goal is to prove that $\hU^{2}(\phi)$ decays exponentially in $H^{l-1}(\bM)$. We do so by proving that all the terms but the first one on
  the left hand side of (\ref{eq:waveequforphi}) decay to zero exponentially. The last term on the left hand side of (\ref{eq:waveequforphi}) consists
  of expressions similar to terms of ``type two'' in the proof of Lemma~\ref{lemma:CkRicciestimates}; see also
  the proofs of Lemma~\ref{lemma:SobbmNestimates} and Theorem~\ref{thm:SobconvofmKdecofhmlUmKmatter}. Similar arguments thus yield
  \[
  \|\hmcY^{A}X_{A}(\phi)\|_{H^{l-1}(\bM)} \leq  C_{a}\ldr{\tau}^{a_{l,n}\cweight+b_{l,n}}e^{2\vare_{\Spe}\tau},
  \]
  where $C_{a}$ only depends on $s_{\cweight,l}$, $c_{\cweight,\kappa_{1}}$, $\e_{\que}$, $K_{q}$, $\hGe_{l}(0)$, $(\bM,\bge_{\refer})$ and a lower bound on
  $\theta_{0,-}$. In order to estimate the second and third terms on the left hand side of (\ref{eq:waveequforphi}), it is necessary to commute
  $E_{\bfI}$ with $\hU$ and $e^{-2\mu_A}X_A^2$. Such arguments are carried out in \cite{RinWave}. Before stating the conclusions, let us begin by recalling the
  terminology of \cite{RinWave}; see also Appendix~\ref{section:summary of RinWave}. To begin with,
  \[
  \mu_{\tg;c}:=\tvarphi_{c}^{-1}\theta^{-(n-1)}\mu_{\chg}=\exp[\ln\tvarphi-\ln\tvarphi_{c}]\mu_{\bge_{\refer}};
  \]
  see (\ref{eq:mutgdef}) and (\ref{eq:mutgreformulation}). Here $\tvarphi=\theta\varphi$, where $\varphi$ is the volume density;
  $\tvarphi_{c}(\bx,t)=\tvarphi(\bx,t_{c})$ for a fixed choice of $t_{c}\leq t_{0}$; $\chg:=\theta^{2}\bge$; and $\mu_{\chg}$ is the volume form
  associated with $\chg$. However, we here assume (\ref{eq:qexpconvergence}) to hold. For this reason, we can appeal to
  Lemma~\ref{lemma:thetavarrhorelqconvtonmotwo}. In fact, chosing $t_{c}=t_{0}$, (\ref{eq:lntvarphimlntvarphicimp}) yields the conclusion that
  \begin{equation}\label{eq:tvarphiminustvarphiclnest}
    |\ln\tvarphi-\ln\tvarphi_{c}|\leq C_{a}
  \end{equation}
  on $M_{-}$, where $C_{a}$ only depends on $c_{\cweight,1}$, $K_{q}$, $\e_{\que}$ and $(\bM,\bge_{\refer})$. For the purposes
  of the present discussion, $\mu_{\tg;c}$ and $\mu_{\bge_{\refer}}$ are thus equivalent. The main energies used in \cite{RinWave} are
  \[
  \hE_{k}[u](\tau;\tau_{c}):=\textstyle{\int}_{\bM_{\tau}}\me_{k}[u]\mu_{\tg;c},
  \]
  where $\me_{k}$ is introduced in (\ref{eq:mekdef}); see \cite[(13.2)]{RinWave}. However, considering (\ref{eq:hGeldef}) and keeping the above
  observations in mind, it is clear that this energy is equivalent to $\hGe_{k}$.

  Given the above observations, we are in a position to estimate the third term on the left hand side of (\ref{eq:waveequforphi}). Applying
  $E_{\bfI}$ to this term results in a linear combination of terms of the form
  \[
  E_{\bfI_{1}}[q-(n-1)]E_{\bfI_{2}}[\hU(\phi)],
  \]
  where $|\bfI_{1}|+|\bfI_{2}|=|\bfI|$. Estimating this expression in $L^{2}$ by appealing to Moser-type estimates; see
  Corollary~\ref{cor:mixedmoserestweight};  yields the conclusion that if $|\bfI|\leq l-1$, it is sufficient to estimate
  \[
  \|q-(n-1)\|_{C^{0}(\bM)}\|\hU(\phi)\|_{H^{l-1}(\bM)}+\|q-(n-1)\|_{H^{l-1}(\bM)}\|\hU(\phi)\|_{C^{0}(\bM)}.
  \]
  Combining this observation with (\ref{eq:meksupestimateKG}), (\ref{eq:Sobestofqminusnminonequgeomatter}) and
  (\ref{eq:hmlUmkqminusnminusoneCzestimatescalarfield}) yields the conclusion that it is sufficient to estimate
  \[
  C_{a}\ldr{\tau}^{\a_{l,n}\cweight+\b_{l,n}}e^{2\vare_{R}\tau}[1+\|\hU(\phi)\|_{H^{l-1}(\bM)}],
  \]
  where $\a_{l,n}$ and $\b_{l,n}$ only depend on $n$ and $l$, and $C_{a}$ only depends on $s_{\cweight,l}$, $c_{\cweight,\kappa_{1}}$, $K_{q}$, $\e_{\que}$, $\e_{p}$,
  $\Lambda$, $\hGe_{l}(0)$,  $(\bM,\bge_{\refer})$ and a lower bound on $\theta_{0,-}$. On the other hand, combining the above observations concerning the
  equivalence of the measures with (\ref{eq:HlmHlequiv}), (\ref{eq:EbfIhUubfIorderone}), (\ref{eq:EbfIhUuestbfIgeneralorder}),
  (\ref{eq:TheEnergyEstimateKG}) and
  (\ref{eq:meksupestimateKG}), it follows that
  \begin{equation}\label{eq:hUphiHlestimatescalarfield}
    \|\hU(\phi)\|_{H^{l}(\bM)}\leq C_{a}\ldr{\tau}^{\a_{l,n}\cweight+\b_{l,n}}
  \end{equation}
  for all $\tau\leq 0$, where $\a_{l,n}$ and $\b_{l,n}$ only depend on $n$ and $l$, and $C_{a}$ only depends on $s_{\cweight,l}$, $c_{\cweight,\kappa_{1}}$, $K_{q}$,
  $\e_{\que}$, $\e_{p}$, $\Lambda$, $\hGe_{l}(0)$, $(\bM,\bge_{\refer})$ and a lower bound on $\theta_{0,-}$. Note, when appealing to (\ref{eq:EbfIhUubfIorderone})
  and (\ref{eq:EbfIhUuestbfIgeneralorder}), that the notation $\|\cdot\|_{2,w}$ is introduced in (\ref{eq:Lpwweightest}) and that $w$ is defined by
  (\ref{eq:wdef}); i.e., $w^{2}=\tvarphi_{c}^{-1}\tvarphi$. Due to (\ref{eq:tvarphiminustvarphiclnest}), $e^{-C_a/2}\leq w\leq e^{C_a/2}$, so that $w$,
  as far as estimates are concerned, can be replaced by $1$. This means that $\|\cdot\|_{2,w}$ is equivalent to the standard $L^{2}$-norm with respect
  to the measure induced by $\bge_{\refer}$. Summarising,
  \[
  \|[q-(n-1)]\hU(\phi)\|_{H^{l-1}(\bM)}\leq C_{a}\ldr{\tau}^{\a_{l,n}\cweight+\b_{l,n}}e^{2\vare_{R}\tau}
  \]
  for all $\tau\leq 0$, where $\a_{l,n}$, $\b_{l,n}$ and $C_{a}$ have the same dependence as in (\ref{eq:hUphiHlestimatescalarfield}). 

  Next, we wish to estimate the second term on the left hand side of (\ref{eq:waveequforphi}). In practice, it is sufficient to estimate
  \begin{equation}\label{eq:highestorderderiphieq}
    \|E_{\bfI}[e^{-2\mu_{A}}X_{A}^{2}(\phi)]\|_{2}\leq \|[E_{\bfI},e^{-2\mu_{A}}X_{A}^{2}](\phi)\|_{2}+\|e^{-2\mu_{A}}X_{A}^{2}[E_{\bfI}(\phi)]\|_{2}
  \end{equation}
  for all $|\bfI|\leq l-1$. Consider the second term on the right hand side of (\ref{eq:highestorderderiphieq}). Note that
  \begin{equation*}
    \begin{split}
      e^{-2\mu_{A}}X_{A}^{2}[E_{\bfI}(\phi)] = & e^{-2\mu_{A}}[(\bD_{X_{A}}\omega^{i})(X_{A})+\omega^{i}(\bD_{X_{A}}X_{A})]E_{i}E_{\bfI}(\phi)\\
      & +e^{-\mu_{A}}\omega^{i}(X_{A})\cdot e^{-\mu_{A}}X_{A}E_{i}E_{\bfI}(\phi).
    \end{split}
  \end{equation*}
  Due to this observation, the fact that $|X_{A}|_{\bge_{\refer}}=1$, (\ref{eq:bDbfAellAetcpteststmtEi}), (\ref{eq:muminmainlowerbound}),
  (\ref{eq:eSpevarrhoeelowtaurelEi}) and
  (\ref{eq:TheEnergyEstimateKG}),
  \[
  \|e^{-2\mu_{A}}X_{A}^{2}[E_{\bfI}(\phi)]\|_{2}\leq C_{a}\ldr{\tau}^{\a_{l,n}\cweight+\b_{l,n}}e^{\vare_{\Spe}\tau}
  \]
  for all $\tau\leq 0$, where $\a_{l,n}$, $\b_{l,n}$ and $C_{a}$ have the same dependence as in (\ref{eq:hUphiHlestimatescalarfield}). What remains to
  be estimated is the first term on the right hand side of (\ref{eq:highestorderderiphieq}). However, this term can be estimated by appealing to
  Lemma~\ref{lemma: 14.8}. Keeping in mind that $\|\cdot\|_{2,w}$ is equivalent to the standard $L^{2}$-norm with respect to the measure induced by
  $\bge_{\refer}$, Lemma~\ref{lemma: 14.8} and arguments similar to the above yield the conclusion that
  \[
  \|[E_{\bfI},e^{-2\mu_{A}}X_{A}^{2}](\phi)\|_{2}\leq C_{a}\ldr{\tau}^{\a_{l,n}\cweight+\b_{l,n}}e^{\vare_{\Spe}\tau}
  \]
  for $|\bfI|\leq l-1$ and $\tau\leq 0$, where $\a_{l,n}$, $\b_{l,n}$ and $C_{a}$ have the same dependence as in (\ref{eq:hUphiHlestimatescalarfield}).
  Combining the above estimates yields the conclusion that
  \begin{equation}\label{eq:hUsqphiHlminusonestimate}
    \|\hU^{2}(\phi)\|_{H^{l-1}(\bM)}\leq C_{a}\ldr{\tau}^{\a_{l,n}\cweight+\b_{l,n}}e^{\vare_{\phi}\tau}
  \end{equation}
  for all $\tau\leq 0$, where $\a_{l,n}$, $\b_{l,n}$ and $C_{a}$ have the same dependence as in (\ref{eq:hUphiHlestimatescalarfield}) and $\vare_{\phi}$
  is defined in the statement of the theorem. 

  Next, in analogy with (\ref{eq:dtauellA}) (recall that $\xi=(\d_{t}\tau)^{-1}\hN$),
  \[
  \d_{\tau}\hU(\phi)=\xi\hU^{2}(\phi)+\xi\hN^{-1}\chi[\hU(\phi)].
  \]
  Combining this identity with (\ref{eq:hUphiHlestimatescalarfield}) and (\ref{eq:hUsqphiHlminusonestimate}); the fact that
  (\ref{eq:hUphiHlestimatescalarfield}) and (\ref{eq:hUsqphiHlminusonestimate}) hold with $H^{l}(\bM)$ and $H^{l-1}(\bM)$, respectively, replaced with
  $C^{0}(\bM)$ (this is an immediate consequence of the fact that $l-1\geq \kappa_{1}-1>n/2$ and Sobolev embedding); the observation concerning $\xi$ made
  in connection with (\ref{eq:EbfIphiterms}); (\ref{eq:EbfIphiterms}); and arguments similar to the ones presented in the proof of
  Theorem~\ref{thm:SobconvofmK} yields
  \begin{equation}\label{eq:dtauhUphi}
    \|\d_{\tau}\hU(\phi)\|_{H^{l-1}(\bM)}\leq C_{a}\ldr{\tau}^{\a_{l,n}\cweight+\b_{l,n}}e^{\vare_{\phi}\tau}
  \end{equation}
  for $\tau\leq 0$, where $\a_{l,n}$, $\b_{l,n}$ and $C_{a}$ have the same dependence as in (\ref{eq:hUphiHlestimatescalarfield}). One particular consequence
  of this estimate is that there is a $\Psi_{\infty}\in H^{l-1}(\bM)\cap C^{0}(\bM)$ (where we appealed to the fact that $l-1>n/2$ and Sobolev embedding to
  conclude that $\Psi_{\infty}\in C^{0}(\bM)$) such that (\ref{eq:asofhUphi}) holds. Note also that (\ref{eq:asofhUphi}) yields
  \begin{equation}\label{eq:PsiinfinfHlminusoneest}
    \|\Psi_{\infty}\|_{\infty}+\|\Psi_{\infty}\|_{H^{l-1}(\bM)}\leq C_{a},
  \end{equation}
  where $C_{a}$ has the same dependence as in (\ref{eq:hUphiHlestimatescalarfield}). 

  Next, note that
  \begin{equation*}
    \begin{split}
      \hU(\phi-\Psi_{\infty}\varrho) = & \hU(\phi)-\Psi_{\infty}\hU(\varrho)-\hU(\Psi_{\infty})\varrho\\
      = & \hU(\phi)-\Psi_{\infty}(1+\hN^{-1}\rodiv_{\bge_{\refer}}\chi)-\hU(\Psi_{\infty})\varrho,
    \end{split}
  \end{equation*}
  where we appealed to (\ref{eq:hUvarrhoident}). Combining this identity with (\ref{eq:asofhUphi}) and the fact that $\Psi_{\infty}$ does
  not depend on the time variable yields
  \begin{equation}\label{eq:hUphiminusPsiinfrho}
    \begin{split}
      \|\hU(\phi-\Psi_{\infty}\varrho)\|_{H^{l-2}(\bM)} \leq & C_{a}\ldr{\tau}^{\a_{l,n}\cweight+\b_{l,n}}e^{\vare_{\phi}\tau}
      +\|\Psi_{\infty}\hN^{-1}\rodiv_{\bge_{\refer}}\chi\|_{H^{l-2}(\bM)}\\
      & +\|\hN^{-1}\chi(\Psi_{\infty})\varrho\|_{H^{l-2}(\bM)}
    \end{split}
  \end{equation}
  for $\tau\leq 0$, where $\a_{l,n}$, $\b_{l,n}$ and $C_{a}$ have the same dependence as in (\ref{eq:hUphiHlestimatescalarfield}). Due to
  Corollary~\ref{cor:mixedmoserestweight}, the second term on the right hand side can be estimated by
  \[
  C_{b}\|\Psi_{\infty}\|_{\infty}\|\hN^{-1}\rodiv_{\bge_{\refer}}\chi\|_{H^{l-2}(\bM)}+
  C_{b}\|\hN^{-1}\rodiv_{\bge_{\refer}}\chi\|_{\infty}\|\Psi_{\infty}\|_{H^{l-2}(\bM)},
  \]
  where $C_{b}$ only depends on $l$ and $(\bM,\bge_{\refer})$. Due to (\ref{eq:PsiinfinfHlminusoneest}), (\ref{eq:rodivchiestimpr}) and 
  (\ref{eq:eSpevarrhoeelowtaurelEi}), this expression can be estimated by
  \[
  C_{a}e^{\vare_{\Spe}\tau}+C_{a}\|\hN^{-1}\rodiv_{\bge_{\refer}}\chi\|_{H^{l-2}(\bM)},
  \]
  where $C_{a}$ has the same dependence as in (\ref{eq:hUphiHlestimatescalarfield}). In order to estimate the last term, recall the notation
  $A^{i}_{k}$ introduced in (\ref{eq:Aialphadef}):
  \[
  A^{k}_{i}:=-\hN^{-1}\omega^{k}(\ml_{\chi}E_{i})=-\hN^{-1}\omega^{k}(\bD_{\chi}E_{i}-\bD_{E_{i}}\chi).
  \]
  However, since (no summation) $\omega^{i}(\bD_{\chi}E_{i})=\bge_{\refer}(\bD_{\chi}E_{i},E_{i})=0$, where we used the fact that $\{E_{i}\}$ is orthonormal
  with respect to $\bge_{\refer}$, we conclude that
  \[
  \textstyle{\sum}_{i}A^{i}_{i}=\hN^{-1}\textstyle{\sum}_{i}\omega^{i}(\bD_{E_{i}}\chi)=\hN^{-1}\rodiv_{\bge_{\refer}}\chi.
  \]
  Combining the above arguments with the fact that $A^{i}_{k}$ can be estimated in $H^{l-2}(\bM)$ by appealing to (\ref{eq:AikmHlestsobassumpEi}),  
  we conclude that
  \begin{equation}\label{eq:PsiinfhNinvrodivbgereferchiest}
    \|\Psi_{\infty}\hN^{-1}\rodiv_{\bge_{\refer}}\chi\|_{H^{l-2}(\bM)}\leq C_{a}\ldr{\tau}^{(l-1)\cweight}e^{\vare_{\Spe}\tau}
  \end{equation}
  for $\tau\leq 0$, where $C_{a}$ has the same dependence as in (\ref{eq:hUphiHlestimatescalarfield}). What remains is to estimate the third term on the
  right hand side of (\ref{eq:hUphiminusPsiinfrho}). However, by appealing to Corollary~\ref{cor:mixedmoserestweight}, it can be estimated by
  \begin{equation*}
    \begin{split}
      & C_{b}\|\hN^{-1}\omega^{i}(\chi)\|_{H^{l-1}(\bM)}\|\Psi_{\infty}\|_{\infty}\|\varrho\|_{\infty}
      +C_{b}\|\hN^{-1}\omega^{i}(\chi)\|_{\infty}\|\Psi_{\infty}\|_{H^{l-1}(\bM)}\|\varrho\|_{\infty}\\
      & +C_{b}\|\hN^{-1}\omega^{i}(\chi)\|_{\infty}\|\Psi_{\infty}\|_{\infty}\|\varrho\|_{H^{l-1}(\bM)},
    \end{split}
  \end{equation*}
  where $C_{b}$ only depends on $l$ and $(\bM,\bge_{\refer})$. Appealing to (\ref{eq:PsiinfinfHlminusoneest}), Remark~\ref{remark:chiclvarrhodecay}, the
  assumptions and the fact that $\ldr{\varrho}$ and $\ldr{\tau}$ are equivalent leads to the conclusion that this expression can be estimated
  by
  \begin{equation}\label{eq:last term phi est}
    \begin{split}
      & C_{a}\ldr{\tau}\|\hN^{-1}\omega^{i}(\chi)\|_{H^{l-1}(\bM)}
      +C_{a}\ldr{\tau}e^{\vare_{\Spe}\tau}+C_{a}e^{\vare_{\Spe}\tau}\|\varrho\|_{H^{l-1}(\bM)},
    \end{split}
  \end{equation}
  where $C_{a}$ has the same dependence as in (\ref{eq:hUphiHlestimatescalarfield}). In order to estimate the last term in (\ref{eq:last term phi est}),
  note that, due to (\ref{eq:varphidefXAver}), 
  \begin{equation}\label{eq:varrho alt form}
    \varrho=\textstyle{\sum}_{A}\bmu_{A}-\tfrac{1}{2}\ln \det\bGe_\refer,
  \end{equation}
  where the components of the matrix $\bGe_\refer$ are defined by (\ref{eq:bGerefer def}). Due to (\ref{eq:bGerefer def}) and
  Lemma~\ref{lemma:frameinvest}, there is a constant $c>1$, depending only on $c_\robas$, such that $c^{-1}\leq \det\bGe_\refer\leq c$.
  Applying $E_{\bfI}$ to the second term on the right hand side of (\ref{eq:varrho alt form}), we thus need to estimate expressions of the
  form
  \[
  |\bD_{\bfI_{1}}X_{A_1}|_{\bge_{\refer}}\cdots |\bD_{\bfI_{j}}X_{A_j}|_{\bge_{\refer}},
  \]
  where $|\bfI_1|+\dots+|\bfI_j|\leq |\bfI|$. However, expressions of this type can be estimated by appealing to (\ref{eq:bDbfAellAetcpteststmtEi})
  and Corollary~\ref{cor:mixedmoserestweight}. Combining this observation with (\ref{eq:varrho alt form}) and (\ref{eq:bmuAmHlestimate}) yields
  \[
  \|\hN^{-1}\chi(\Psi_{\infty})\varrho\|_{H^{l-2}(\bM)}\leq C_{a}\ldr{\tau}^{l\cweight+1}e^{\vare_{\Spe}\tau}
  +C_{a}\ldr{\tau}\|\hN^{-1}\omega^{i}(\chi)\|_{H^{l-1}(\bM)},
  \]
  where $C_{a}$ has the same dependence as in (\ref{eq:hUphiHlestimatescalarfield}). In order to estimate the last term on the right hand side, note that
  $|E_{\bfI}[\hN^{-1}\omega^{i}(\chi)]|$ can be estimated by a linear combination of terms of the form
  \[
  \hN^{-1}|E_{\bfI_{1}}\ln\hN|\cdots |E_{\bfI_{k}}\ln\hN|\cdot |\bD_{\bfJ}\chi|_{\bge_{\refer}},
  \]
  where $|\bfI_{1}|+\dots+|\bfI_{k}|+|\bfJ|\leq |\bfI|$ and $|\bfI_{j}|\neq 0$. Since $\hN$ is equivalent to $\d_{t}\tau$, we can replace
  $\hN^{-1}$ with $(\d_{t}\tau)^{-1}$ in this expression. Doing so and appealing to Corollary~\ref{cor:mixedmoserestweight} and the assumptions, in particular
  (\ref{eq:mKbDlnNchicombest}), yields the conclusion that it is sufficient to estimate terms of the form
  \[
  \|\hN^{-1}\bD_{\bfK}\chi\|_{2}+\|\hN^{-1}\chi\|_{\infty}\|\bD\ln\hN\|_{H^{l-2}(\bM)},
  \]
  where $|\bfK|\leq l-1$. Combining this observation with Lemma~\ref{lemma:chimclbbeWClhy} and the assumptions yields
  \[
  \|\hN^{-1}\chi(\Psi_{\infty})\varrho\|_{H^{l-2}(\bM)}\leq C_{a}\ldr{\tau}^{l\cweight+1}e^{\vare_{\Spe}\tau},
  \]
  where $C_{a}$ has the same dependence as in (\ref{eq:hUphiHlestimatescalarfield}). Combining this estimate with (\ref{eq:hUphiminusPsiinfrho}) and
  (\ref{eq:PsiinfhNinvrodivbgereferchiest}) yields
  \begin{equation}\label{eq:hUphiminusPsiinfrhofinal}
    \begin{split}
      \|\hU(\phi-\Psi_{\infty}\varrho)\|_{H^{l-2}(\bM)} \leq & C_{a}\ldr{\tau}^{\a_{l,n}\cweight+\b_{l,n}}e^{\vare_{\phi}\tau},
    \end{split}
  \end{equation}
  where $\a_{l,n}$, $\b_{l,n}$ and $C_{a}$ have the same dependence as in (\ref{eq:hUphiHlestimatescalarfield}). On the other hand,
  \begin{equation}\label{eq:dtauphiminusPsiinfvarrhoid}
    \d_{\tau}(\phi-\Psi_{\infty}\varrho)=\xi\hU(\phi-\Psi_{\infty}\varrho)+\xi\hN^{-1}\chi(\phi-\Psi_{\infty}\varrho).
  \end{equation}
  Moreover, assuming $l\geq \kappa_{1}+1$, the estimate (\ref{eq:hUphiminusPsiinfrhofinal}) holds with $H^{l-2}(\bM)$ replaced by
  $C^{0}(\bM)$. In that setting, we can estimate the right hand side of (\ref{eq:dtauphiminusPsiinfvarrhoid}) in $H^{l-2}(\bM)$
  by arguments similar to the above. This yields
  \[
  \|\d_{\tau}(\phi-\Psi_{\infty}\varrho)\|_{H^{l-2}(\bM)} \leq C_{a}\ldr{\tau}^{\a_{l,n}\cweight+\b_{l,n}}e^{\vare_{\phi}\tau},
  \]
  where $\a_{l,n}$, $\b_{l,n}$ and $C_{a}$ have the same dependence as in (\ref{eq:hUphiHlestimatescalarfield}). This estimates implies the existence of a
  $\Phi_{\infty}\in H^{l-2}(\bM)\cap C^{0}(\bM)$ such that (\ref{eq:asofhphi}) holds. 
\end{proof}

\appendix

\section{Summary of results from \cite{RinWave}}\label{section:summary of RinWave}

In the present section, we, for the benefit of the reader, recall the terminology and conclusions of \cite{RinWave} we use in this article.

\subsection{Terminology and basic estimates concerning frames}

Let $\{E_i\}$ be the global orthonormal frame on $(\bM,\bge_{\refer})$ introduced in Remark~\ref{remark:globalframe}. Using this frame, we recall
\cite[Definition~4.7]{RinWave}, fixing our conventions concerning multiindices. 
\begin{definition}[Definition~4.7, \cite{RinWave}]\label{def:multiindexnotation}
  Let $(M,g)$ be a time oriented Lorentz manifold. Assume that it has an expanding partial pointed foliation. Assume, moreover, $\mK$ to be
  non-degenerate on $I$ and to have a global frame. Then a \textit{vector field multiindex}
  is a vector, say $\bfI=(I_{1},\dots,I_{l})$,
  where $I_{j}\in \{1,\dots,n\}$. The number $l$ is said to be the \textit{order}
  of the vector field multiindex, and it is denoted by $|\bfI|$. The
  vector field multiindex corresponding to the empty set is denoted by $\bfz$. Moreover, $|\bfz|=0$. Given that the letter used for the vector
  field multiindex is $\bfI$, $\bfJ$ etc.,
  \begin{align*}
    \bfE_{\bfI} := & (E_{I_{1}},\dots,E_{I_{l}}),\ \ \
    \bD_{\bfI}:=\bD_{E_{I_{1}}}\cdots \bD_{E_{I_{l}}},\ \ \
    E_{\bfI} := E_{I_{1}}\cdots E_{I_{l}}
  \end{align*}
  etc. where $\bfI=(I_{1},\dots,I_{l})$, with the special convention that $\bD_{\bfz}$ and $E_{\bfz}$ are the identity operators,
  and $\bfE_{\bfz}$ is the empty argument.
\end{definition}
Next, we recall \cite[Lemma~5.5, Lemma~5.7 and Corollary~5.9]{RinWave}:
\begin{lemma}[Lemma~5.5, \cite{RinWave}]\label{lemma:frameinvest}
  Let $(M,g)$ be a time oriented Lorentz manifold. Assume that it has an expanding partial pointed foliation. Assume, moreover, $\mK$ to be
  non-degenerate on $I$, to have a global frame and to be $C^{0}$-uniformly bounded on $I_{-}$; i.e., (\ref{eq:mKsupbasest}) to hold. Then there
  is a constant $C_{Y}$, depending only on $n$, $\mKsup$ and $\e_{\rond}$, such that $|Y^{A}|_{\bge_{\refer}}\leq C_{Y}$ on $\bM_{t}$ for all $A$ and
  $t\in I_{-}$.
\end{lemma}
\begin{lemma}[Lemma~5.7, \cite{RinWave}]\label{lemma:bDbfAbDkequiv}
  Let $(M,g)$ be a time oriented Lorentz manifold. Assume that it has an expanding partial pointed foliation. Assume, moreover, $\mK$ to be
  non-degenerate on $I$ and to have a global frame. Let $\mt$ be a family of tensor fields on $\bM$ for $t\in I$. Then $\bD_{\bfI}\mt$ can be
  written as a linear combination of terms of the form
  \[
  (\bD^{k}\mt)(\bfE_{\bfI_{1}})\omega^{J_{1}}(\bD_{\bfJ_{1}}E_{K_{1}})\cdots \omega^{J_{l}}(\bD_{\bfJ_{l}}E_{K_{l}}),
  \]
  where $|\bfI|=k+|\bfJ_{1}|+\dots+|\bfJ_{l}|$ and $k\geq 1$ if $|\bfI|\geq 1$. Similarly, if $k=|\bfI|$, then $(\bD^{k}\mt)(\bfE_{\bfI})$ can be
  written as a linear combination of terms of the form
  \[
  (\bD_{\bfJ}\mt)\omega^{I_{1}}(\bD_{\bfJ_{1}}E_{K_{1}})\cdots \omega^{I_{l}}(\bD_{\bfJ_{l}}E_{K_{l}}),
  \]
  where $k=|\bfJ|+|\bfJ_{1}|+\dots+|\bfJ_{l}|$ and $|\bfJ|\geq 1$ if $k\geq 1$.
\end{lemma}
\begin{cor}[Corollary~5.9, \cite{RinWave}]\label{cor:covderofframe}
  Let $(M,g)$ be a time oriented Lorentz manifold. Assume that it has an expanding partial pointed foliation. Assume, moreover, $\mK$ to be
  non-degenerate on $I$, to have a global frame and to be $C^{0}$-uniformly bounded on $I_{-}$; i.e. (\ref{eq:mKsupbasest}) to hold. Let $\xi$
  be a vector field on $\bM\times I$ which is tangent to the constant-$t$ hypersurfaces. Then there is a constant $C_{1}$, depending only on $n$, $\mKsup$
  and $\e_{\rond}$ such that
  \begin{equation}\label{eq:covderofframe}
    |\bD_{\xi}\ell_{A}|+|\bD_{\xi} Y^{A}|_{\bge_{\refer}}+|\bD_{\xi} X_{A}|_{\bge_{\refer}}\leq C_{1} |\xi|_{\bge_{\refer}}|\bD\mK|_{\bge_{\refer}}
  \end{equation}
  on $\bM_{t}$ for all $A\in \{1,\dots,n\}$ and $t\in I_{-}$. Defining the structure constants, say $\g_{AB}^{C}$, of the $X_{A}$ by
  $[X_{A},X_{B}]=\g_{AB}^{C}X_{C}$, the estimate
  \begin{equation}\label{eq:gaCABbasest}
    |\g^{C}_{AB}|\leq C_{1}|\bD\mK|_{\bge_{\refer}}
  \end{equation}
  also holds on $\bM_{t}$ for all $A,B,C\in \{1,\dots,n\}$ and $t\in I_{-}$.
\end{cor}
It will be convenient to use the following notation:
\begin{definition}[Definition~5.10, \cite{RinWave}]\label{def:mfPmKhN}
  Let $(M,g)$ be a time oriented Lorentz manifold. Assume it to have an expanding partial pointed foliation. Given $0\leq m\in\zo$, let
  \begin{align*}
    \mfP_{\mK,m} := & \textstyle{\sum}_{m_{1}+\dots+m_{j}=m,m_{i}\geq 1}|\bD^{m_{1}}\mK|_{\bge_{\refer}}\cdots |\bD^{m_{j}}\mK|_{\bge_{\refer}},\\
    \mfP_{N,m} := & \textstyle{\sum}_{m_{1}+\dots+m_{j}=m,m_{i}\geq 1}|\bD^{m_{1}}\ln\hN|_{\bge_{\refer}}\cdots |\bD^{m_{j}}\ln\hN|_{\bge_{\refer}},\\
    \mfP_{\mK,N,m} := & \textstyle{\sum}_{m_{1}+m_{2}=m}\mfP_{\mK,m_{1}}\mfP_{N,m_{2}},
  \end{align*}
  with the convention that $\mfP_{\mK,0}=1$ and $\mfP_{N,0}=1$.
\end{definition}
Next, we need estimates for the higher order derivatives of $\ell_{A}$, $X_{A}$ and $Y^{A}$. 
\begin{lemma}[Lemma~5.11, \cite{RinWave}]\label{lemma:bDbfabDlmjchKest}
  Let $(M,g)$ be a time oriented Lorentz manifold. Assume it to have an expanding partial pointed foliation. Assume, moreover, $\mK$ to be
  non-degenerate on $I$, to have a global frame and to be $C^{0}$-uniformly bounded on $I_{-}$; i.e. (\ref{eq:mKsupbasest}) to hold. Then, for
  every pair of integers $j$ and $l$ satisfying $1\leq j\leq l$, and every multiindex $\bfI$ with $|\bfI|=j$, there is a constant $D_{\mK,l}$,
  depending only on $l$, $n$ and $(\bM,\bge_{\refer})$, such that
  \begin{equation}\label{eq:bDbfAbDlmjmKpteststmtEi}
    |\bD_{\bfI}\bD^{l-j}\mK|_{\bge_{\refer}}\leq D_{\mK,l}\textstyle{\sum}_{m=l-j+1}^{l}|\bD^{m}\mK|_{\bge_{\refer}}
  \end{equation}
  on $\bM\times I_{-}$. Similarly, there is a constant $\md_{\mK,j}$ depending only on $\mKsup$, $n$, $l$, $\e_{\rond}$ and $(\bM,\bge_{\refer})$ such that
  \begin{equation}\label{eq:bDbfAellAetcpteststmtEi}
    |\bD_{\bfI}\ell_{A}|+|\bD_{\bfI}X_{A}|_{\bge_{\refer}}+|\bD_{\bfI}Y^{A}|_{\bge_{\refer}}\leq \md_{\mK,j}\textstyle{\sum}_{m=1}^{j}\mfP_{\mK,m}
  \end{equation}
  on $\bM\times I_{-}$.
\end{lemma}

\subsection{Basic formulae}
In this article, we need the following three basic formulae derived in \cite{RinWave}:
\begin{align}
  \chth = & 1+\hU(n\ln\theta)=-q,\label{eq:chthexpressionhUlntheta}\\
  \hU(\ell_{A}) = & (\hml_{U}\mK)(Y^{A},X_{A}),\label{eq:hU ellA}\\
  \hU(\varrho) = & 1+\hN^{-1}\rodiv_{\bge_{\refer}}\chi,\label{eq:hUvarrhoident}
\end{align}
see \cite[(3.5)]{RinWave}, \cite[(6.6)]{RinWave} and \cite[(7.9)]{RinWave} respectively, where there is no summation
in the second equality. Moreover, defining $A_i^0$ and $A_i^k$ by
\begin{align}  
  A_{i}^{0} := & E_{i}(\ln\hN),\ \ \
  A^{k}_{i} := -\hN^{-1}\omega^{k}(\ml_{\chi}E_{i}),\label{eq:Aialphadef}
\end{align}
see \cite[(6.22)]{RinWave}, we have $[\hU,E_{i}]=A_{i}^{0}\hU+A_{i}^{k}E_{k}$. Next, 
\begin{align}  
  \exp\left(\textstyle{\sum}_{A}\bmu_{A}\right) = & \varphi\cdot (\det\bGe_{\refer})^{1/2},\label{eq:varphidefXAver}
\end{align}
see \cite[(7.12)]{RinWave}. Here $\bGe_{\refer}$ and $\bGe_\refer^{-1}$ are matrices with components
\begin{equation}\label{eq:bGerefer def}
  \bGe_{\refer,AB}=\bge_{\refer}(X_{A},X_{B})=\textstyle{\sum}_{i}X_{A}^{i}X_{B}^{i},\ \ \
  (\bGe_\refer^{-1})^{AB}=\textstyle{\sum}_j Y^{A}_jY^B_j
\end{equation}
and $X_{A}^{i}$ and $Y^{A}_{i}$ are the components of $X_{A}$ and $Y^{A}$ respectively with respect to $\{E_i\}$.

\subsection{Bounds on the lengths of the eigenvector fields}

An important result, giving a lower bound on the $\mu_A$ is \cite[Lemma~7.5]{RinWave}:
\begin{lemma}[Lemma~7.5, \cite{RinWave}]\label{lemma:lowerbdonmumin}
  Let $(M,g)$ be a time oriented Lorentz manifold with an expanding partial pointed foliation. Assume $\mK$ to be non-degenerate on $I$, to have a
  global frame and to be $C^{0}$-uniformly bounded on $I_{-}$; i.e. (\ref{eq:mKsupbasest}) to hold. Assume $\chK$ to have a silent upper bound on $I$.
  Assume, moreover, that $\hml_{U}\mK$ satisfies a weak off-diagonal exponential bound; see Definition~\ref{def:offdiagonalexpdec}. Let $\e_{\chi}$ be
  defined by
  \begin{equation}\label{eq:echibasicassumptionmod}
    \e_{\chi}:=\tfrac{1}{4}e^{-M_{\mu}}\min\{1,\e_{\Spe}\},
  \end{equation}
  where $M_{\mu}$ is defined by
  \begin{equation}\label{eq:Mmudef}
    M_{\mu}:=(n+1)M_{0}+C_{\det,\rond}+\tfrac{1}{2};
  \end{equation}
  $C_{\det,\rond}$ is a constant depending only on $n$, $\mKsup$ and $\e_\rond$; $M_{0}$ is defined by
  \begin{equation}\label{eq:Mzdef}
    M_{0}:=\tfrac{3(n-1)}{\e_{\rond}\e_{\mK}}(C_{\mK,\mrod}+3M_{\mK,\mrod})+\tfrac{1}{2};
  \end{equation}
  and $\e_{\rond}$ is the constant appearing in Definition~\ref{def:silenceandnondegeneracy}. Assume, finally, that
  \begin{align}
    n^{1/2}\theta_{0,-}^{-1}|\bD\chi|_{\rohy} \leq & \e_{\chi},\label{eq:chiestfirststepestnoff}
  \end{align}
  for all $t\in I_{-}$, where $\theta_{0,-}$ is defined by (\ref{eq:thetazdef}). Then
  \begin{align}
    \hN^{-1}|\rodiv_{\bge_{\refer}}\chi| \leq & \min\{1,\e_{\Spe}\}e^{\e_{\Spe}\varrho},\label{eq:rodivchiestimpr}\\
    (2\hN)^{-1}|(\ml_{\chi}\bge_{\refer})(X_{A},X_{A})| \leq & \min\{1,\e_{\Spe}\}e^{\e_{\Spe}\varrho},\label{eq:WAAchicontrestimpr}\\
    \mu_{\min} \geq & -\e_{\Spe}\varrho+\ln\theta_{0,-}-M_{\min}\label{eq:muminmainlowerbound}
  \end{align}
  (no summation on $A$ in the second estimate) on $M_{-}$, where $M_{\min}:=M_{\mu}+1$. Moreover, if $\g$ is an integral curve of
  $\hU$ with $\g(0)\in\bM\times \{t_{0}\}$, then
  \begin{align}
    [\hN^{-1}|\rodiv_{\bge_{\refer}}\chi|]\circ\g(s) \leq & \tfrac{1}{4}\min\{1,\e_{\Spe}\}e^{\e_{\Spe}s},\label{eq:rodivchiroughestimpr}\\
    [(2\hN)^{-1}|(\ml_{\chi}\bge_{\refer})(X_{A},X_{A})|]\circ\g(s)
    \leq & \tfrac{1}{4}\min\{1,\e_{\Spe}\}e^{\e_{\Spe}s},\label{eq:WAAchicontrroughestimpr}\\
    \mu_{\min}\circ\g(s) \geq & -\e_{\Spe}s+\ln\theta_{0,-}-M_{\mu}\label{eq:mumincircgammalowerbound}
  \end{align}
  for all $s\leq 0$ such that $\g(s)\in M_{-}$. Moreover,
  \begin{equation}\label{eq:varrhosequivalencestmt}
    s-1/2\leq \varrho\circ\g(s)\leq s+1/2
  \end{equation}
  for all $s\leq 0$ such that $\g(s)\in M_{-}$.
\end{lemma}
In the course of the proof of this statement, the following three estimates occur. First,
\begin{equation}\label{eq:int mWAA}
  \textstyle{\int}_{s}^{0}|\mW_{A}^{A}\circ\g(u)|du\leq M_0
\end{equation}
(no summation on $A$) for all $A>1$ and $s\in J_-$ (see below for the definition of $J_-$); see \cite[(7.39)]{RinWave}. Here $\mW_{A}^{B}$
is defined by the equality
\begin{equation}\label{eq:mWABdef}
  \hml_{U}X_A=\mW_A^B X_B+\overline{W}_A^0U;
\end{equation}
see \cite[(6.5)]{RinWave}. Second, 
\begin{equation}\label{eq:bmuAeqintellA}
\big|\bmu_{A}\circ\g(s)+\textstyle{\int}_{s}^{0}\ell_{A}\circ\g(u)du\big|\leq M_{0}
\end{equation}
for all $A>1$ and $s\in J_{-}$; see \cite[(7.41)]{RinWave}. Third, 
\begin{equation}\label{eq:bmuoneeqintellone}
\big|\textstyle{\int}_{s}^{0}\ell_{1}\circ\g(u)du+\bmu_{1}\circ\g(s)\big|\leq (n-1)M_{0}+C_{\det,\rond}+\tfrac{1}{2}
\end{equation}
for all $s\in J_{-}$, see \cite[(7.42)]{RinWave}. In the last three estimates, we have replaced $J_1$ (appearing in \cite{RinWave}) with $J_-$, since
it is demonstrated that $J_1=J_-$ at the end of the proof of \cite[Lemma~7.5]{RinWave}. Moreover, if $J$ is the domain of definition of $\g$, then
$J_-$ is the set of $s\in J$ such that $s\leq 0$ and $\g(s)\in M_-$. Recall also that $M_-$ is introduced in Lemma~\ref{lemma:smallnessshiftconsequences}.

Combining \cite[(7.32)]{RinWave} and \cite[(7.33)]{RinWave} yields the conclusion that, if $s_a,s_b\in J_-$ with $s_a\leq s_b$, 
\begin{equation}\label{eq:muA diff A gt one}
\big|\mu_A\circ\g(s_b)-\mu_A\circ\g(s_a)-\textstyle{\int}_{s_{a}}^{s_b}[\ell_A\circ\g(s)-n^{-1}(q\circ\g(s)+1)]ds\big|\leq
\textstyle{\int}_{s_a}^{s_b}|\mW_{A}^{A}\circ\g(s)|ds
\end{equation}
for $A>1$. Next, due to \cite[(7.11)]{RinWave}, 
\[
\big|\textstyle{\sum}_A\mu_A-\varrho-n\ln\theta\big|\leq C_{\mathrm{det},\rond}
\]
on $M_-$, where we used the fact that $\mu_A=\bmu_A+\ln\theta$. Combining this estimate with (\ref{eq:varrhosequivalencestmt}) yields
\begin{equation}\label{eq:sum muA est}
\big|\textstyle{\sum}_A\mu_A\circ\g(s_b)-\textstyle{\sum}_A\mu_A\circ\g(s_a)-(s_b-s_a)-n\ln\theta\circ\g(s_b)+n\ln\theta\circ\g(s_a)\big|
\leq 2C_{\mathrm{det},\rond}+1.
\end{equation}
Combining (\ref{eq:int mWAA}), (\ref{eq:muA diff A gt one}) and (\ref{eq:sum muA est}) with the fact that $\sum_A\ell_A=1$ yields
\begin{equation}\label{eq:mu one est}
  \begin{split}
    & \big|\mu_1\circ\g(s_b)-\mu_1\circ\g(s_a)-\textstyle{\int}_{s_{a}}^{s_b}[\ell_1\circ\g(s)-n^{-1}(q\circ\g(s)+1)]ds\big|\\
    \leq & (n-1)M_0+2C_{\mathrm{det},\rond}+1.
  \end{split}
\end{equation}
Combining (\ref{eq:int mWAA}), (\ref{eq:muA diff A gt one}) and (\ref{eq:mu one est}) yields
\begin{equation}\label{eq:mu123 est}
  \begin{split}
    & \big|\mu_{\rowo}\circ\g(s_b)-\mu_{\rowo}\circ\g(s_a)-\textstyle{\int}_{s_{a}}^{s_b}[\ell_{\rowo}\circ\g(s)+n^{-1}(q\circ\g(s)+1)]ds\big|\\
    \leq & (n+1)M_0+2C_{\mathrm{det},\rond}+1,
  \end{split}
\end{equation}
where we use the notation $\ell_{\rowo}:=\ell_{1}-\ell_{n-1}-\ell_n$ and $\mu_{\rowo}:=\mu_1-\mu_{n-1}-\mu_n$.

Next, under the assumptions of Lemma~\ref{lemma:lowerbdonmumin}, there is a constant $M_{\rodiff}$, depending only on $c_\robas$, such that, assuming $A>B$, 
\begin{align}
\bmu_{A}-\bmu_{B} \leq & (A-B)\e_{\rond}\varrho+M_{\rodiff},\label{eq:bmuAmbmuBlowbd}\\
\ln\theta \geq & -(n^{-1}+\e_{\Spe})\varrho+\ln\theta_{0,-}-2\label{eq:lnthetalowbd}
\end{align}
on $M_{-}$; see \cite[(7.50) and (7.51)]{RinWave}. Next, we recall \cite[Lemma~3.33]{RinWave}.

\begin{lemma}[Lemma~3.33, \cite{RinWave}]\label{lemma:smallnessshiftconsequences}
  Assume the conditions of Definition~\ref{def:basicassumptions} to be fulfilled; i.e., the basic assumptions to hold. Assume, moreover, that
  there is a constant $c_{\chi,2}$ such that
  \[
  \theta_{0,-}^{-1}\|\chi\|_{C^{2,\weight_{0}}_{\rohy}(\bM)}\leq c_{\chi,2}
  \]
  holds for all $t\in I_{-}$, where $\weight_{0}$ is the same as in Definition~\ref{def:basicassumptions}. Then there is an $\e_{\chi}>0$, depending only
  on $c_{\robas}$, and a $\de_{\chi}$, depending only on $c_{\robas}$, $c_{\chi,2}$ and $(\bM,\bge_{\refer})$, such that if
  \begin{align}
    n^{1/2}\theta_{0,-}^{-1}|\chi|_{\rohy} \leq & \de_{\chi},\label{eq:dechicontrol}\\
    n^{1/2}\theta_{0,-}^{-1}|\bD\chi|_{\rohy} \leq & \e_{\chi}\label{eq:echicontrol}
  \end{align}
  hold on $M_{-}:=\bM\times I_{-}$,
  then
  \begin{equation}\label{eq:muminmainlowerboundintro}
    \mu_{\min} \geq  -\e_{\Spe}\varrho+\ln\theta_{0,-}-M_{\min}
  \end{equation}
  on $M_{-}$, where $M_{\min}$ only depends on $c_{\robas}$. Here $\mu_{\min}:=\min_{A}\mu_{A}$.
  Moreover, there is a constant $C_{\varrho}$, depending only
  on $c_{\robas}$, $c_{\chi,2}$ and $(\bM,\bge_{\refer})$, such that $|\bD\varrho|_{\bge_{\refer}}\leq C_{\varrho}\ldr{\varrho}$. Next, there is a constant
  $K_{\rovar}$, depending only on $\bDlnhNsup$ and $(\bM,\bge_{\refer})$, such that if $\bx_{1},\bx_{2}\in\bM$ and $t_{1},t_{2}\in I_{-}$ are such that
  $t_{1}<t_{2}$, then 
  \begin{equation}\label{eq:DeltavarrhorelvariationEiintro}
    \tfrac{1}{3K_{\rovar}}\leq \tfrac{\varrho(\bx_{2},t_{2})-\varrho(\bx_{2},t_{1})}{\varrho(\bx_{1},t_{2})-\varrho(\bx_{1},t_{1})}\leq 3K_{\rovar}.
  \end{equation}
  Finally  
  \begin{equation}\label{eq:hNinvdtvarrhoestEiintro}
    1/2\leq \hN^{-1}\d_{t}\varrho\leq 3/2
  \end{equation}
  holds on $M_{-}$. 
\end{lemma}

Next, let us quote \cite[Lemma~7.17]{RinWave}.
\begin{lemma}[Lemma~7.17, \cite{RinWave}]\label{lemma:epsilonlowdefEi}
  With assumptions and notation as in Lemma~\ref{lemma:smallnessshiftconsequences}, let $\tau$ be defined by (\ref{eq:taudefinition}).
  Then
  \begin{equation}\label{eq:eSpevarrhoeelowtaurelEi}
    e^{\e_{\Spe}\varrho(\bx,t)}\leq e^{\eSpe\tau(t)}
  \end{equation}
  for all $(\bx,t)\in M_{-}$, where $\eSpe:=\e_{\Spe}/(3K_{\rovar})$. Similarly, if $t_{1}\leq t_{2}\leq t_{0}$ and $\bx\in\bM$,
  \begin{equation}\label{eq:emKbxtottEi}
    e^{\e_{\mK}[\varrho(\bx,t_{1})-\varrho(\bx,t_{2})]}\leq e^{\emK[\tau(t_{1})-\tau(t_{2})]}
  \end{equation}
  where $\emK:=\e_{\mK}/(3K_{\rovar})$.
  Finally,
  \begin{equation}\label{eq:hNtaudotequivEi}
    (2K_{\rovar})^{-1}\leq [\hN(\bx,t)]^{-1}\d_{t}\tau(t)\leq 2K_{\rovar}
  \end{equation}
  for all $t\in I_{-}$ and $\bx\in \bM$.
\end{lemma}
\begin{remark}\label{remark:assumptions 7.17}
  Note that in the statement of \cite[Lemma~7.17]{RinWave}, it is assumed that the conditions of \cite[Lemma~7.13]{RinWave} are satisfied. However,
  Lemma~\ref{lemma:smallnessshiftconsequences}, which is the same as \cite[Lemma~3.33]{RinWave}, is based on a combination of
  \cite[Lemmas~7.5, 7.12 and 7.13]{RinWave}; see the proof of \cite[Lemma~3.33]{RinWave}. In other words, if the assumptions of
  Lemma~\ref{lemma:smallnessshiftconsequences} are satisfied, the assumptions of \cite[Lemma~7.13]{RinWave} are satisfied.
\end{remark}
\begin{remark}
  The starting point for proving (\ref{eq:hNtaudotequivEi}) is the identity (\ref{eq:hUvarrhoident}), which can be written
  \[
  \hN^{-1}\d_t\varrho = 1+\hN^{-1}\chi(\varrho)+\hN^{-1}\rodiv_{\bge_{\refer}}\chi.
  \]
  Due to the smallness assumptions we make concerning the shift vector field, this identity implies that (\ref{eq:hNinvdtvarrhoestEiintro})
  holds. Combining (\ref{eq:hNinvdtvarrhoestEiintro}) with (\ref{eq:taudefinition}) and
  (\ref{eq:mKbDlnNchicombest}) yields (\ref{eq:hNtaudotequivEi}). We refer the reader interested in the details to the proofs of
  \cite[Lemmas~7.13 and 7.17]{RinWave} for details. 
\end{remark}

We also need \cite[Lemma~7.19]{RinWave}.
\begin{lemma}[Lemma~7.19, \cite{RinWave}]\label{lemma:thetavarrhorelqconvtonmotwo}
  Assume the conditions of Definition~\ref{def:basicassumptions} and of Lemma~\ref{lemma:smallnessshiftconsequences} to be satisfied. Assume, moreover,
  \begin{equation}\label{eq:Kthetaoneestimate}
    \|\ln\theta\|_{C^{\bfl_{0}}_{\weight_{0}}(\bM)}\leq c_{\theta,1}
  \end{equation}
  to be satisfied for all $t\in I_{-}$, where $\bfl_{0}:=(1,1)$. Let $t_{c}\in I_{-}$ and $\tvarphi:=\theta\varphi$, where $\varphi$ is defined by
  (\ref{eq:varphidefitobmubgedp}). Define $\tvarphi_{c}$ by $\tvarphi_{c}(\bx,t):=\tvarphi(\bx,t_{c})$. Finally, let
  \begin{equation}\label{eq:tetaonedefprel}
    \teta_{1} := \textstyle{\frac{1}{n}}|q-(n-1)|.
  \end{equation}
  Then
  \begin{equation}\label{eq:lntvarphimlntvarphic}
    |\ln\tvarphi-\ln\tvarphi_{c}|\leq C_{a}\ldr{\tau_{c}}^{\bcweight}e^{\eSpe\tau_{c}}+2K_{\rovar}\textstyle{\int}_{\tau}^{\tau_{c}}\teta_{1}(\cdot,s)ds
  \end{equation}
  on $M_{c}:=\{(\bx,t)\in\bM\times I:t\leq t_{c}\}$, where $\tau_{c}:=\tau(t_{c})$, $\bcweight:=\max\{1,\cweight\}$ and $C_{a}$ only depends on
  $c_{\robas}$, $c_{\chi,2}$, $c_{\theta,1}$ and $(\bM,\bge_{\refer})$. Assuming, in addition to the above, that there is a constant $d_{q}$ such that 
  \begin{equation}\label{eq:qconvergence}
    \|\ldr{\varrho(\cdot,t)}^{3/2}[q(\cdot,t)-(n-1)]\|_{C^{0}(\bM)} \leq d_{q}
  \end{equation}
  for all $t\in I_{-}$,
  \begin{equation}\label{eq:lntvarphimlntvarphicimp}
    |\ln\tvarphi-\ln\tvarphi_{c}|\leq C_{a}\ldr{\tau_{c}}^{\bcweight}e^{\eSpe\tau_{c}}+C_{b}\ldr{\tau_{c}}^{-1/2}
  \end{equation}
  on $M_{c}$, where $C_{a}$ has the same dependence as in the case of (\ref{eq:lntvarphimlntvarphic}) and $C_{b}$ only depends on $K_{\rovar}$
  and $d_{q}$. 
\end{lemma}
Note that Remark~\ref{remark:assumptions 7.17} is equally relevant in this setting. In the context of Lemma~\ref{lemma:thetavarrhorelqconvtonmotwo},
it is convenient to recall \cite[(11.40)]{RinWave}:
\begin{equation}\label{eq:mutgreformulation}
  \begin{split}
    \tvarphi_{c}^{-1}\theta^{-(n-1)}\mu_{\chg} = & \tvarphi_{c}^{-1}\theta^{-(n-1)}\theta^{n}\mu_{\bge}=\tvarphi_{c}^{-1}\theta \varphi\mu_{\bge_{\refer}}\\
    = & \tvarphi_{c}^{-1}\tvarphi\mu_{\bge_{\refer}}=\exp[\ln\tvarphi-\ln\tvarphi_{c}]\mu_{\bge_{\refer}}.
  \end{split}    
\end{equation}
In what follows, we also use $w$, defined by
\begin{equation}\label{eq:wdef}
  w : = \tvarphi_{c}^{-1/2}\tvarphi^{1/2},
\end{equation}
as a weight in our estimates; see \cite[(14.2)]{RinWave}. More specifically, we use the notation
\begin{align}
  \|\mt(\cdot,t)\|_{p,w} := & \left(\textstyle{\int}_{\bM}|\mt(\cdot,t)|_{\bge_{\refer}}^{p}w^{p}(\cdot,t)\mu_{\bge_{\refer}}\right)^{1/p},\label{eq:Lpwweightest}\\
  \|\mt(\cdot,t)\|_{\infty,w} := & \sup_{\bx\in\bM}|\mt(\bx,t)|_{\bge_{\refer}}w(\bx,t),\label{eq:Linftyweightest}
\end{align}
see \cite[(14.5) and (14.6)]{RinWave}. Moreover, we use the notation 
\[
\|\bD^{1}_{\bbE}u\|_{\infty,w}:=\textstyle{\sum}_{i=1}^{n}\|E_{i}u\|_{\infty,w}.
\]

\subsection{Function spaces}

Next, we recall some function spaces from \cite[Chapter~8]{RinWave}. We begin by recalling \cite[(8.1) and (8.2)]{RinWave}.
Let $(\weight_{a},\weight_{b})=\weight\in\Weight$ and $(l_{0},l_{1})=\bfl\in\Index$. Define, using the notation introduced in
Definition~\ref{def:multiindexnotation},
\begin{subequations}\label{seq:bbE norms}
  \begin{align}
    \|\mt(\cdot,t)\|_{\mc^{\bfl}_{\bbE,\weight}(\bM)}
    := & \textstyle{\sup}_{\bx\in\bM}
         \big(\textstyle{\sum}_{j=l_{0}}^{l_{1}}\sum_{|\bfI|=j}
         \ldr{\varrho(\bx,t)}^{-2\weight_{a}-2j\weight_{b}}|\bD_{\bfI}\mt(\bx,t)|_{\bge_{\refer}}^{2}\big)^{1/2},\label{eq:mtmClbbEbS}\\
    \|\mt(\cdot,t)\|_{\mH^{\bfl}_{\bbE,\weight}(\bM)}
    := & \big(\textstyle{\int}_{\bM}\textstyle{\sum}_{j=l_{0}}^{l_{1}}\sum_{|\bfI|=j}
         \ldr{\varrho(\cdot,t)}^{-2\weight_{a}-2j\weight_{b}}|\bD_{\bfI}\mt(\cdot,t)|_{\bge_{\refer}}^{2}\mu_{\bge_{\refer}}\big)^{1/2}.\label{eq:mtmHlbbEbS}
  \end{align}
\end{subequations}
If $l_{0}=0$, then we replace $\bfl$ in (\ref{seq:bbE norms}) with $l:=l_{1}$. Next, assuming $l_{0}\leq 1$, there are constants
$C_{\sup,\bfl},C_{\rosob,\bfl}\geq 1$, depending only on $\bfl$, $n$, $(\bM,\bge_{\refer})$ and the type of the tensor field, such that
\begin{subequations}
  \begin{align}
    C_{\sup,\bfl}^{-1}\|\mt(\cdot,t)\|_{C^{\bfl}_{\weight}(\bM)} \leq & \|\mt(\cdot,t)\|_{\mc^{\bfl}_{\bbE,\weight}(\bM)}
                                                                        \leq C_{\sup,\bfl}\|\mt(\cdot,t)\|_{C^{\bfl}_{\weight}(\bM)},\label{eq:Clmclequiv}\\
    C_{\rosob,\bfl}^{-1}\|\mt(\cdot,t)\|_{H^{\bfl}_{\weight}(\bM)} \leq & \|\mt(\cdot,t)\|_{\mH^{\bfl}_{\bbE,\weight}(\bM)}
                                                                          \leq C_{\rosob,\bfl}\|\mt(\cdot,t)\|_{H^{\bfl}_{\weight}(\bM)},\label{eq:HlmHlequiv}
  \end{align}
\end{subequations}
see \cite[(8.5) and (8.6)]{RinWave}. Next, let $(l_{0},l_{1})=\bfl\in\Index$, $(\weight_{a},\weight_{b})=\weight\in\Weight$ and recall
\cite[(8.9) and (8.10)]{RinWave}:
\begin{subequations}
  \begin{align}
    \|\mt(\cdot,t)\|_{\mH^{\bfl,\weight}_{\bbE,\rocon}(\bM)}
    := & \big(\textstyle{\int}_{\bM}\textstyle{\sum}_{l_{0}\leq|\bfI|\leq l_{1}}\hN^{-2}(\cdot,t)
         \ldr{\varrho(\cdot,t)}^{-2\weight_{a}-2|\bfI|\weight_{b}}|\bD_{\bfI}\mt(\cdot,t)|_{\bge_{\refer}}^{2}\bmu_{\bge_{\refer}}\big)^{1/2},\label{eq:mHlbbeWnorm}\\
    \|\mt(\cdot,t)\|_{\mc^{\bfl,\weight}_{\bbE,\rocon}(\bM)}
    := & \textstyle{\sup}_{\bx\in\bM}\big(\textstyle{\sum}_{l_{0}\leq |\bfI|\leq l_{1}}\hN^{-2}(\bx,t)
         \ldr{\varrho(\bx,t)}^{-2\weight_{a}-2|\bfI|\weight_{b}}|\bD_{\bfI}\mt(\bx,t)|_{\bge_{\refer}}^{2}\big)^{1/2}.\label{eq:mClbbeWnorm}
  \end{align}
\end{subequations}
Next, we quote the first half of \cite[Lemma~8.4]{RinWave}.
\begin{lemma}[Lemma~8.4, \cite{RinWave}]\label{lemma:chimclbbeWClhy}
  Given that the assumptions of Lemma~\ref{lemma:smallnessshiftconsequences} hold, let $\tau$ be defined by (\ref{eq:taudefinition}). Let $\xi$ be
  a vector field on $\bM$, $(l_{0},l_{1})=\bfl\in\Index$ and $(\weight_{a},\weight_{b})=\weight\in\Weight$. Then, assuming $l_{0}\leq 1$,
  \begin{subequations}
  \begin{align}
    \|\xi(\cdot,t)\|_{\mH^{\bfl,\weight}_{\bbE,\rocon}(\bM)} \leq & Ce^{M_{\min}}e^{\eSpe\tau}\theta_{0,-}^{-1}
    \|\xi(\cdot,t)\|_{H^{\bfl,\weight}_{\rohy}(\bM)},\label{eq:chimHlHlest}\\
    \|\xi(\cdot,t)\|_{\mc^{\bfl,\weight}_{\bbE,\rocon}(\bM)} \leq & Ce^{M_{\min}}e^{\eSpe\tau}\theta_{0,-}^{-1}
    \|\xi(\cdot,t)\|_{C^{\bfl,\weight}_{\rohy}(\bM)},\label{eq:chimClClest}
  \end{align}
  \end{subequations}
  where $C$ only depends on $n$, $\bfl$, $\weight$ and $(\bM,\bge_{\refer})$; $M_{\min}$ is defined in the text adjacent to
  (\ref{eq:muminmainlowerbound}); and $\eSpe$ is defined in the text adjacent to (\ref{eq:eSpevarrhoeelowtaurelEi}).
\end{lemma}
We also recall the first half of \cite[Remark~8.5]{RinWave}.
\begin{remark}[Remark~8.5, \cite{RinWave}]\label{remark:chiclvarrhodecay}
  Arguments similar to the proof give the following conclusion. Given that the conditions of Lemma~\ref{lemma:smallnessshiftconsequences} are fulfilled
  and that $\bfl$ and $\weight$ are as in the statement of the lemma, 
  \[
  \ldr{\varrho}^{-\weight_{a}-|\bfI|\weight_{b}}\hN^{-1}|\bD_{\bfI}\xi|_{\bge_{\refer}}\leq Ce^{M_{\min}}e^{\e_{\Spe}\varrho}\theta_{0,-}^{-1}
  \|\xi(\cdot,t)\|_{C^{\bfl,\weight}_{\rohy}(\bM)}
  \]
  for all $(\bx,t)\in M_{-}$ and $l_{0}\leq |\bfI|\leq l_{1}$, where $C$ only depends on $n$, $\bfl$, $\weight$ and $(\bM,\bge_{\refer})$.
\end{remark}

\subsection{Estimates}

Next, we state some consequences of the higher order Sobolev assumptions, namely the first half of \cite[Lemma~9.11]{RinWave}:
\begin{lemma}[Lemma~9.11, \cite{RinWave}]\label{lemma:mWAhUASobestimates}
  Fix $l$, $\bfl_{0}$, $\bfl$, $\bfl_{1}$, $\cweight$, $\weight_{0}$ and $\weight$ as in Definition~\ref{def:sobklassumptions}.
  Let $\weight_{1}:=(2\cweight,\cweight)$. Then, given that the assumptions of Lemma~\ref{lemma:smallnessshiftconsequences} as well as the
  $(\cweight,l)$-Sobolev assumptions are satisfied,
  \begin{align}
    \|\mW^{A}_{B}(\cdot,t)\|_{H^{l+1}_{\weight}(\bM)} \leq & C_{a},\label{eq:mWABmHlmfwbdbconstant}\\
    \|A^{k}_{i}(\cdot,t)\|_{H^{l+1}_{\weight}(\bM)} \leq & C_{a}e^{\eSpe\tau(t)},\label{eq:AikmHlestsobassumpEi}\\
    \|\hU(A^{k}_{i})(\cdot,t)\|_{H^{l-1}_{\weight_{1}}(\bM)} \leq & C_{a}e^{\eSpe\tau(t)}\label{eq:hUAikSoblestulsobass}
  \end{align}
  for all $t\in I_{-}$, all $A,B$ and all $i,k$, where $C_{a}$ only depends on $s_{\cweight,l}$ and $(\bM,\bge_{\refer})$. 
\end{lemma}
Here $\mW^A_B$ is defined by (\ref{eq:mWABdef}) and $A^k_i$ is defined by (\ref{eq:Aialphadef}). Next, we recall
\cite[Lemma~10.3 and Remark~10.4]{RinWave}. 
\begin{lemma}[Lemma~10.3, \cite{RinWave}]\label{lemma:energyestimatesbmuA}
  Fix $l$, $\bfl_{0}$, $\bfl$, $\bfl_{1}$, $\cweight$, $\weight_{0}$ and $\weight$ as in Definition~\ref{def:sobklassumptions}.
  Given that the the assumptions of Lemma~\ref{lemma:smallnessshiftconsequences} as well as the $(\cweight,l)$-Sobolev assumptions are satisfied, there 
  is a constant $C_{\bmu,l}$ such that
  \begin{equation}\label{eq:bmuAmHlestimate}
    \|\bmu_{A}(\cdot,\tau)\|_{\mH^{l+1}_{\bbE,\weight}(\bM)}\leq C_{\bmu,l}\ldr{\tau}
  \end{equation}
  on $I_{-}$ for all $A$, where $C_{\bmu,l}$ only depends on $s_{\cweight,l}$ and $(\bM,\bge_{\refer})$.
\end{lemma}
\begin{remark}[Remark~10.4, \cite{RinWave}]\label{remark:muAmhlestimate}
  Combining (\ref{eq:bmuAmHlestimate}) with the assumptions and the fact that $\mu_{A}=\bmu_{A}+\ln\theta$ yields the conclusion that 
  \begin{equation}\label{eq:muAmHlestimatermk}
    \|\mu_{A}(\cdot,\tau)\|_{\mH^{\bfl_{1}}_{\bbE,\weight}(\bM)}\leq C_{\mu,l}\ldr{\tau}
  \end{equation}
  on $I_{-}$ for all $A$, where $C_{\mu,l}$ only depends on $s_{\cweight,l}$ and $(\bM,\bge_{\refer})$. 
\end{remark}
Next, we quote the $C^k$-estimates contained in \cite[Lemma~10.5 and Remark~10.6]{RinWave}.
\begin{lemma}[Lemma~10.5, \cite{RinWave}]\label{lemma:CkestofbmuAEi}
  Fix $l$, $\bfl_{1}$, $\cweight$, $\weight_{0}$ and $\weight$ as in Definition~\ref{def:supmfulassumptions}.
  Then, given that the assumptions of Lemma~\ref{lemma:smallnessshiftconsequences} as well as the $(\cweight,l)$-supremum assumptions are satisfied,
  there is a constant $C_{\bmu,l}$ such that
  \begin{equation}\label{eq:bmuAmclmfwestEi}
    \|\bmu_{A}(\cdot,t)\|_{\mc^{l+1}_{\bbE,\weight}(\bM)}\leq C_{\bmu,l}\ldr{\tau}
  \end{equation}
  for all $t\in I_{-}$, where $C_{\bmu,l}$ only depends on $c_{\cweight,l}$ and $(\bM,\bge_{\refer})$. 
\end{lemma}
\begin{remark}[Remark~10.6, \cite{RinWave}]\label{remark:CkestofmuAEi}
  Similarly to Remark~\ref{remark:muAmhlestimate}, combining (\ref{eq:bmuAmclmfwestEi}) with the assumptions and the fact that 
  $\mu_{A}=\bmu_{A}+\ln\theta$ yields the conclusion that 
  \begin{equation}\label{eq:muAClestimatermk app}
    \|\mu_{A}(\cdot,\tau)\|_{\mc^{\bfl_{1}}_{\bbE,\weight}(\bM)}\leq C_{\mu,l}\ldr{\tau}
  \end{equation}
  on $I_{-}$ for all $A$, where $C_{\mu,l}$ only depends on $c_{\cweight,l}$ and $(\bM,\bge_{\refer})$. 
\end{remark}

\subsection{Estimates of solutions to the wave equation}

Next, note that combining \cite[(11.6), (11.43) and (11.44)]{RinWave} with $\chth=-q$ (see Remark~\ref{remark:chitoregvar}) yields
the conclusion that $\Box_g\phi=0$ is equivalent to
\begin{equation}\label{eq:waveequforphiapp}
  -\hU^{2}(\phi)+\textstyle{\sum}_{A}e^{-2\mu_{A}}X_{A}^{2}(\phi)+\tfrac{1}{n}[q-(n-1)]\hU(\phi)+\hmcY^{A}X_{A}(\phi)=0,
\end{equation}
where
\begin{equation}\label{eq:hmcYA def}
  \begin{split}
    \hmcY^{A} := & e^{-2\mu_{A}}X_{A}(\ln \hN)-2e^{-2\mu_{A}}X_{A}(\mu_{A})+e^{-2\mu_{A}}X_{A}(\mu_{\rotot})\\
    & -2e^{-2\mu_{A}}\alpha_{A}-(n-1)e^{-2\mu_{A}}X_{A}(\ln\theta),
  \end{split}
\end{equation}
Here $\alpha_A:=\g_{AB}^B/2$, where $\g_{BC}^A$ is defined by (\ref{eq:structure coeffients}).

Next we recall the energies that will be used to control the scalar field. In particular, recall \cite[(13.1)--(13.3)]{RinWave}:
\begin{align}
  \me_{k}[u] := & \tfrac{1}{2}\textstyle{\sum}_{|\bfI|\leq k}\left(|\hU(E_{\bfI}u)|^{2}+\textstyle{\sum}_{A}e^{-2\mu_{A}}|X_{A}(E_{\bfI}u)|^{2}
  +\ldr{\tau-\tau_{c}}^{-3}|E_{\bfI}u|^{2}\right),\label{eq:mektaudefEi}\\
  \hE_{k}[u](\tau;\tau_{c}) := & \textstyle{\int}_{\bM_{\tau}}\me_{k}[u]\mutgc\label{eq:hEktaudefEi}
\end{align}
for all $\tau\leq\tau_{c}$ and some fixed $\tau_c\leq 0$, where (recalling the terminology introduced in
Lemma~\ref{lemma:thetavarrhorelqconvtonmotwo})
\begin{equation}\label{eq:mutgdef}
  \mutgc:=\tvarphi_{c}^{-1}\theta^{-(n-1)}\mu_{\chg}.
\end{equation}
Note that in our case, the $\iota_a$ and $\iota_b$ appearing in \cite[(13.1)]{RinWave} satisfy $\iota_a=0$ and $\iota_b=1$. The reason for
this is that in our case the $\hat{\alpha}$ appearing in \cite[(11.7)]{RinWave} vanishes. 

Next, keeping (\ref{eq:Lpwweightest}) and Remark~\ref{remark:assumptions 7.17} in mind, we recall
\cite[Lemmas~14.6 and 14.8]{RinWave}
\begin{lemma}[Lemma~14.6, \cite{RinWave}]\label{lemma:EbfIhUuestbfIgeneralorder}
  Let $0\leq \cweight\in\ro$, $\weight_{0}=(0,\cweight)$ and $\weight=(\cweight,\cweight)$. Assume that the conditions of
  Lemma~\ref{lemma:smallnessshiftconsequences} as well as the $(\cweight,1)$-supremum assumptions are satisfied. Let $c_{\chi,2}$ be defined 
  as in the statement of Lemma~\ref{lemma:smallnessshiftconsequences}. Then, if $|\bfI|=1$, 
  \begin{equation}\label{eq:EbfIhUubfIorderone}
    \|E_{\bfI}\hU u\|_{2,w}\leq C\hE_{1}^{1/2}
  \end{equation}
  for all $\tau\leq\tau_{c}$, where $C$ only depends on $c_{\robas}$, $c_{\chi,2}$ and $(\bM,\bge_{\refer})$. Fix $l$
  as in Definition~\ref{def:sobklassumptions} and assume that the $(\cweight,l)$-Sobolev assumptions are satisfied. Then, if $2\leq m\leq l$ and 
  $|\bfI|=m$, 
  \begin{equation}\label{eq:EbfIhUuestbfIgeneralorder}
    \begin{split}
      \|E_{\bfI}\hU u\|_{2,w}
      \leq & C_{a}\hE_{m}^{1/2}+C_{a}\ldr{\tau}^{\a_{m}\cweight}\ldr{\tau-\tau_{c}}^{\b_{m}}\hE_{m-1}^{1/2}\\
      & +C_{b}\ldr{\tau}^{\a_{m}\cweight}\ldr{\tau-\tau_{c}}^{\b_{m}}[\|\hU u\|_{\infty,w}+e^{\eSpe\tau}\|\bD^{1}_{\bbE}u\|_{\infty,w}]
    \end{split}
  \end{equation}
  for all $\tau\leq\tau_{c}$. Here $\a_{m}$ and $\b_{m}$ are constants depending only on $m$. Moreover, $C_{a}$ only depends on $c_{\cweight,1}$, $m$, 
  and $(\bM,\bge_{\refer})$; and $C_{b}$ only depends on $c_{\cweight,1}$, $s_{\cweight,m}$ and $(\bM,\bge_{\refer})$.
\end{lemma}
\begin{lemma}[Lemma~14.8, \cite{RinWave}]\label{lemma: 14.8}
  Let $0\leq \cweight\in\ro$, $\weight_{0}=(0,\cweight)$ and $\weight=(\cweight,\cweight)$. Assume that the conditions of
  Lemma~\ref{lemma:smallnessshiftconsequences} as well as the $(\cweight,1)$-supremum assumptions are satisfied. Fix $l$
  as in Definition~\ref{def:sobklassumptions} and assume that the $(\cweight,l)$-Sobolev assumptions are satisfied. Then, if $1\leq m\leq l$ and 
  $|\bfI|=m$, 
  \begin{equation}\label{eq:emtmuAXAsqcommtbesttotal}
    \begin{split}
      \|[E_{\bfI},e^{-2\mu_{A}}X_{A}^{2}]u\|_{2,w} \leq & \ldr{\tau}^{\a_{m}\cweight+\b_{m}}e^{\eSpe\tau}\left(C_{a}\hE_{m}^{1/2}
      +C_{b}\textstyle{\sum}_{i}\|e^{-\mu_{A}}X_{A}E_{i}u\|_{\infty,w}\right)\\
      & +C_{b}\ldr{\tau}^{\a_{m}\cweight+\b_{m}}e^{2\eSpe\tau}\|\bD^{1}_{\bbE}u\|_{\infty,w}
    \end{split}
  \end{equation}  
  for all $\tau\leq\tau_{c}$, where $C_{a}$ only depends on $c_{\cweight,1}$, $m$, $(\bM,\bge_{\refer})$ and a lower bound on $\theta_{0,-}$; and $C_{b}$ only
  depends on $c_{\cweight,1}$, $s_{\cweight,m}$, $(\bM,\bge_{\refer})$ and a lower bound on $\theta_{0,-}$. Here $\a_{m}$ and $\b_{m}$ only depend on $m$. 
\end{lemma}

\subsection{Global frames}

Next, we recall \cite[Lemma~A.1]{RinWave}:
\begin{lemma}[Lemma~A.1, \cite{RinWave}]\label{lemma:Lemma A.1}
  Let $(M,g)$ be a time oriented Lorentz manifold. Assume it to have an expanding partial pointed foliation and $\mK$ to be non-degenerate on $I$.
  Assuming $\bM$ to be connected, there is a connected finite covering space $\tM$ of $\bM$ with covering map $\pi_{a}:\tM\rightarrow \bM$. Letting
  $\pi_{b}:\tM\times I\rightarrow\bM\times I$ be defined by $\pi_{b}(\tx,t)=[\pi_{a}(\tx),t]$, then $\pi_{b}$ is also a covering map. Letting
  $\tg=\pi_{b}^{*}g$, $\pi_{b}$ is a local isometry. Moreover, the expansion normalised Weingarten map associated with $\tg$ and the foliation
  $\tM\times I$ has a global frame. 
\end{lemma}

\subsection{Normal derivatives of families of $(1,1)$-tensor fields} 
It is useful to recall \cite[(A.3) and (A.4)]{RinWave}. If $\mt$ is a family of $(1,1)$-tensor fields on $\bM$ (for $t\in I$), 
\begin{equation}\label{eq:mlUmKincoordinates}
\begin{split}
(\ml_{U}\mt)^{i}_{\phantom{i}j}  = & \tfrac{1}{N}\d_{t}(\mt^{i}_{\phantom{i}j})-\tfrac{1}{N}(\ml_{\chi}\mt)^{i}_{\phantom{i}j}
\end{split}
\end{equation}
and
\begin{equation}\label{eq:hmlUmtinfixedspatialcoord}
  \hml_{U}\mt:=\theta^{-1}\ml_{U}\mt=\hN^{-1}[\d_{t}(\mt^{i}_{\phantom{i}j})-(\ml_{\chi}\mt)^{i}_{\phantom{i}j}] E_{i}\otimes\omega^{j}.
\end{equation}

\subsection{Moser estimates}
Next, we quote a special case of \cite[Corollary~B.9]{RinWave}, namely the case that $q=0$ (note that the $q$ appearing in the statement
of \cite[Corollary~B.9]{RinWave} is unrelated to the deceleration parameter $q$ appearing in the present article). 
\begin{cor}[Corollary~B.9, \cite{RinWave}]\label{cor:mixedmoserestweight}
  Assume $(M,g)$ to be a time oriented Lorentz manifold. Assume that it has an expanding partial pointed foliation. Assume, moreover, $\mK$ to be
  non-degenerate on $I$, to have a global frame and to be $C^{0}$-uniformly bounded on $I_{-}$; i.e., (\ref{eq:mKsupbasest}) to hold. Let
  $0\leq r,s\in\zo$. For $1\leq j\leq r$ and $1\leq m\leq s$, let: $u_{j}$ and $v_{m}$
  be smooth strictly positive functions on $\bM\times I$; $g_{j}$ and $h_{m}$ be strictly positive functions on $I$; $k_{j}$ and $p_{m}$ be
  non-negative integers; and $\mt_{j}$ and $\mU_{m}$ be families of smooth tensor fields on $\bM$ for $t\in I$. Let $l$ be the sum of the
  $k_{j}$ and the $p_{m}$. Then, assuming $g_{j}\leq 1$ and $h_{m}\leq 1$, 
  \begin{equation}\label{eq:mixedmosergagliardonirenbergweight}
    \begin{split}
      & \big\|\textstyle{\prod}_{j=1}^{r}u_{j}g_{j}^{k_{j}}|\bD^{k_{j}}\mt_{j}|_{\bge_{\refer}}
      \textstyle{\prod}_{m=1}^{s}v_{m}h_{m}^{p_{m}}|\bD^{p_{m}}_{\bbE}\mU_{m}|_{\bge_{\refer}}\big\|_{2}\\
      \leq & C_{b}\textstyle{\sum}_{j}\textstyle{\sum}_{k\leq l}\b_{j}^{l-k}\|u_{j}g_{j}^{k}\bD^{k}\mt_{j}\|_{2}
      \textstyle{\prod}_{o\neq j}\|\mt_{o}\|_{\infty,u_{o}}
      \textstyle{\prod}_{m}\|\mU_{m}\|_{\infty,v_{m}}\\
      & +C_{b}\textstyle{\sum}_{m}\textstyle{\sum}_{k\leq l}\g_{m}^{l-k}\|v_{m}h_{m}^{k}\bD^{k}_{\bbE}\mU_{m}\|_{2}
      \textstyle{\prod}_{o\neq m}\|\mU_{o}\|_{\infty,v_{o}}
      \textstyle{\prod}_{j}\|\mt_{j}\|_{\infty,u_{j}}
    \end{split}
  \end{equation}
  on $I_{-}$, where the constant $C_{b}$ only depends on $l$, $n$ and $(\bM,\bge_{\refer})$; and 
  \begin{align*}    
    \b_{j}(t) := & 1+g_{j}(t)\textstyle{\sup}_{\bx\in\bM}|(\bD\ln u_{j})(\bx,t)|_{\bge_{\refer}},\\
    \g_{m}(t) := & 1+h_{m}(t)\textstyle{\sup}_{\bx\in\bM}|(\bD\ln v_{m})(\bx,t)|_{\bge_{\refer}}.
  \end{align*}
\end{cor}

\section{Geometric identities}\label{section:geometricidentities}

The purpose of the present section is to record some basic geometric identites we need in this article. We also combine them with the Einstein
equations in order to draw conclusions concerning the deceleration parameter and $\hml_{U}\mK$. In the calculations, we use the frame $\{E_{i}\}$
(with dual frame $\{\omega^{i}\}$) introduced in Remark~\ref{remark:globalframe}. 

\subsection{Notation and first order calculations}

Using the notation $\bk_{ij}:=\bk(E_{i},E_{j})$ and $\bge_{ij}:=\bge(E_{i},E_{j})$, the following lemma holds. 

\begin{lemma}
  Let $(M,g)$ be a spacetime. Assume that it has an expanding partial pointed foliation. Assume, moreover, $\mK$ to be non-degenerate on $I$ and to
  have a global frame. Then
  \begin{align}
    \d_{t}\bge_{ij} = & 2N\bk_{ij}+(\ml_{\chi}\bge)_{ij},\label{eq:dtbgeijintermsofNbkijandmlchibge}\\
    \d_{t}\bge^{ij} = & -2N\bk^{ij}-(\ml_{\chi}\bge)^{ij}.\label{eq:dtbgeinvijintermsofNbkijandmlchibge}
  \end{align}
  Moreover,
  \begin{align}
    \nabla_{E_{i}}U = & \bk_{i}^{\phantom{i}j}E_{j},\label{eq:nablaEiUform}\\
    \nabla_{U}E_{i} = & \bk_{i}^{\phantom{i}j}E_{j}+E_{i}(\ln N)U-N^{-1}\ml_{\chi}E_{i},\label{eq:nablaUEiform}\\
    \nabla_{U}U = & E_{i}(\ln N)\bge^{ij}E_{j},\label{eq:nablaUUform}\\
    \nabla_{E_{i}}E_{j} = & \bk_{ij}U+\bG^{k}_{ij}E_{k},\label{eq:nablaEiEjform}
  \end{align}
  where $\bG^{k}_{ij}$ are the connection coefficients associated with the Levi-Civita connection on $(\bM,\bge)$ and the frame $\{E_{i}\}$; i.e.,
  $\bnabla_{E_{i}}E_{j}=\bG^{k}_{ij}E_{k}$. 
\end{lemma}
\begin{proof}
  Since $\d_{t}=NU+\chi$,
  \begin{equation*}
    \begin{split}
      \bk_{ij} = & \bk(E_{i},E_{j})=g(\nabla_{E_{i}}U,E_{j})=(2N)^{-1}[g(\nabla_{E_{i}}(\d_{t}-\chi),E_{j})+g(\nabla_{E_{j}}(\d_{t}-\chi),E_{i})]\\
      = & (2N)^{-1}[\d_{t}\bge_{ij}-(\ml_{\chi}\bge)_{ij}].
    \end{split}
  \end{equation*}
  In other words, (\ref{eq:dtbgeijintermsofNbkijandmlchibge}) and (\ref{eq:dtbgeinvijintermsofNbkijandmlchibge}) hold.
  Note also that (\ref{eq:nablaEiUform}) holds and that
  \begin{equation}\label{eq:UEicommutator}
    [U,E_{i}]=E_{i}(\ln N)U-N^{-1}\ml_{\chi}E_{i}.
  \end{equation}
  Thus
  \[
  \nabla_{U}E_{i}=\nabla_{E_{i}}U+[U,E_{i}]=\bk_{i}^{\phantom{i}j}E_{j}+E_{i}(\ln N)U-N^{-1}\ml_{\chi}E_{i}
  \]
  so that (\ref{eq:nablaUEiform}) holds. Next, note that $\ldr{\nabla_{U}U,U}=0$ and that (with $\ldr{\cdot,\cdot}:=g$)
  \begin{equation}\label{eq:ldrnablaUUdj}
    \ldr{\nabla_{U}U,E_{j}}=-\ldr{U,\nabla_{U}E_{j}}=-\ldr{U,[U,E_{j}]+\nabla_{E_{j}}U}=E_{j}\ln N.
  \end{equation}
  In particular, (\ref{eq:nablaUUform}) holds. Finally,
  \[
  g(\nabla_{E_{i}}E_{j},U)=-\bk_{ij},\ \ \
  g(\nabla_{E_{i}}E_{j},E_{k})=\bge(\bnabla_{E_{i}}E_{j},E_{k}).
  \]
  The lemma follows. 
\end{proof}

Next, we calculate the time derivative of $\bK$.
\begin{lemma}\label{lemma:dtbK}
  Let $(M,g)$ be a spacetime. Assume that it has an expanding partial pointed foliation. Assume, moreover, $\mK$ to be
  non-degenerate on $I$ and to have a global frame. Then
  \begin{equation}\label{eq:dtbKuidj}
    \d_{t}\bK^{i}_{\phantom{i}j} = \bnabla^{i}\bnabla_{j}N-N\theta\bK^{i}_{\phantom{i}j}
    +N\bge^{il}R_{lj}-N\bR^{i}_{\phantom{i}j}+(\ml_{\chi}\bK)^{i}_{\phantom{i}j}.
  \end{equation}
\end{lemma}
\begin{remark}
  Keeping (\ref{eq:mlUmKincoordinates}) in mind, (\ref{eq:dtbKuidj}) can be written
  \begin{equation}\label{eq:mlUbKwithoutEinstein}
    (\ml_{U}\bK)^{i}_{\phantom{i}j} = N^{-1}\bnabla^{i}\bnabla_{j}N-\theta\bK^{i}_{\phantom{i}j}+\bge^{il}R_{lj}-\bR^{i}_{\phantom{i}j}.
  \end{equation}
  Taking the trace of this equality yields
  \begin{equation}\label{eq:preRaychaudhuri}
    U(\theta)=N^{-1}\Delta_{\bge}N-\theta^{2}+\bge^{ij}R_{ij}-\bS,
  \end{equation}
  where $\bS$ is the spatial curvature of the constant-$t$ hypersurfaces. Comparing this relation with (\ref{eq:hUnlnthetamomqbas})
  yields an expression for the deceleration parameter $q$. 
\end{remark}
\begin{proof}
  Note, to begin with, that 
  \[
  U\bk_{ij}=U\ldr{\nabla_{E_{i}}U,E_{j}}=\ldr{\nabla_{U}\nabla_{E_{i}}U,E_{j}}+\ldr{\nabla_{E_{i}}U,\nabla_{U}E_{j}}.
  \]
  Appealing to (\ref{eq:nablaEiUform}) and (\ref{eq:nablaUEiform}) yields  
  \[
  \ldr{\nabla_{E_{i}}U,\nabla_{U}E_{j}}=\bk_{i}^{\phantom{i}m}\bk_{mj}-N^{-1}\bk_{i}^{\phantom{i}m}\ldr{E_{m},\ml_{\chi}E_{j}}
  =\bk_{i}^{\phantom{i}m}\bk_{mj}-N^{-1}\bk(E_{i},\ml_{\chi}E_{j}).
  \]
  Next, consider
  \[
  \ldr{\nabla_{U}\nabla_{E_{i}}U,E_{j}}=\ldr{R_{UE_{i}}U,E_{j}}+\ldr{\nabla_{E_{i}}\nabla_{U}U,E_{j}}+\ldr{\nabla_{[U,E_{i}]}U,E_{j}}.
  \]
  Appealing to (\ref{eq:nablaUUform}) yields
  \begin{align*}
    \ldr{\nabla_{E_{i}}\nabla_{U}U,E_{j}} = & E_{i}[E_{l}(\ln N)\bge^{lm}]\bge_{mj}+E_{l}(\ln N)\bge^{lm}\ldr{\nabla_{E_{i}}E_{m},E_{j}}\\
    = & E_{i}E_{j}\ln N-E_{l}(\ln N)\bge^{lm}E_{i}\bge_{mj}+E_{l}(\ln N)\bge^{lm}\ldr{\nabla_{E_{i}}E_{m},E_{j}}\\
    = & E_{i}E_{j}\ln N-\bG^{l}_{ij}E_{l}\ln N=(\bnabla^{2}\ln N)(E_{i},E_{j}).
  \end{align*}
  Next, due to (\ref{eq:nablaUUform}) and (\ref{eq:UEicommutator}),
  \begin{align*}
    \ldr{\nabla_{[U,E_{i}]}U,E_{j}} = & (E_{i}\ln N)\ldr{\nabla_{U}U,E_{j}}-N^{-1}\ldr{\nabla_{\ml_{\chi}E_{i}}U,E_{j}}\\
    = & (E_{i}\ln N)E_{j}\ln N-N^{-1}\bk(\ml_{\chi}E_{i},E_{j}).
  \end{align*}
  Adding up the above,
  \begin{align*}
    U\bk_{ij} = & \ldr{\nabla_{U}\nabla_{E_{i}}U,E_{j}}+\ldr{\nabla_{E_{i}}U,\nabla_{U}E_{j}}\\
    = & \ldr{R_{UE_{i}}U,E_{j}}+\ldr{\nabla_{E_{i}}\nabla_{U}U,E_{j}}+\ldr{\nabla_{[U,E_{i}]}U,E_{j}}
    +\bk_{i}^{\phantom{i}m}\bk_{mj}-N^{-1}\bk(E_{i},\ml_{\chi}E_{j})\\
    = & \ldr{R_{UE_{i}}U,E_{j}}+N^{-1}\bnabla_{i}\bnabla_{j}N
    +\bk_{i}^{\phantom{i}m}\bk_{mj}-N^{-1}\chi\bk_{ij}+N^{-1}(\ml_{\chi}\bk)_{ij}.
  \end{align*}
  This equality can be rewritten
  \begin{equation}\label{eq:dtbkij}
    \d_{t}\bk_{ij}=N\ldr{R_{UE_{i}}U,E_{j}}+\bnabla_{i}\bnabla_{j}N
    +N\bk_{i}^{\phantom{i}m}\bk_{mj}+(\ml_{\chi}\bk)_{ij}.
  \end{equation}
  Next, note that
  \begin{equation}\label{eq:Rijzerlegung}
    R_{ij}=\mathrm{Ric}(E_{i},E_{j})=-\ldr{R_{UE_{i}}E_{j},U}+\bge^{mp}\ldr{R_{E_{m}E_{i}}E_{j},E_{p}}.
  \end{equation}
  On the other hand, due to the Gauss equation,
  \[
  \ldr{R_{XY}Z,W}=\ldr{\bR_{XY}Z,W}-\bk(X,Z)\bk(Y,W)+\bk(Y,Z)\bk(X,W),
  \]
  assuming $X,Y,Z,W$ to be vector fields tangent to the constant-$t$ hypersurfaces. In particular,
  \begin{equation}\label{eq:Gaussequationsinaframe}
    \ldr{R_{E_{m}E_{i}}E_{j},E_{p}}=\ldr{\bR_{E_{m}E_{i}}E_{j},E_{p}}-\bk_{mj}\bk_{ip}+\bk_{ij}\bk_{mp},
  \end{equation}
  so that
  \begin{equation}\label{eq:tracedGaussequation}
    \bge^{mp}\ldr{R_{E_{m}E_{i}}E_{j},E_{p}}=\bR_{ij}-\bk_{i}^{\phantom{i}m}\bk_{mj}+(\tr_{\bge}\bk)\bk_{ij}.
  \end{equation}
  Combining this observation with (\ref{eq:Rijzerlegung}) yields
  \begin{equation}\label{eq:RUEiEjUform}
    N\ldr{R_{UE_{i}}E_{j},U}=-NR_{ij}+N\bR_{ij}-N\bk_{i}^{\phantom{i}m}\bk_{mj}+N(\tr_{\bge}\bk)\bk_{ij}.
  \end{equation}
  Combining this calculation with (\ref{eq:dtbkij}) yields
  \[
  \d_{t}\bk_{ij}=\bnabla_{i}\bnabla_{j}N+2N\bk_{i}^{\phantom{i}m}\bk_{mj}-N(\tr_{\bge}\bk)\bk_{ij}+NR_{ij}-N\bR_{ij}+(\ml_{\chi}\bk)_{ij}.
  \]
  Next, let us calculate
  \begin{equation}\label{eq:dtmixedkijfirsttry}
    \begin{split}
      \d_{t}\bk^{i}_{\phantom{i}j} = & -\bge^{im}\bge^{lp}\d_{t}\bge_{mp}\bk_{lj}+\bge^{il}\d_{t}\bk_{lj}\\
      = & -2N\bk^{il}\bk_{lj}-(\ml_{\chi}\bge)^{il}\bk_{lj}
      +\bnabla^{i}\bnabla_{j}N+2N\bk^{im}\bk_{mj}\\
      & -N(\tr_{\bge}\bk)\bk^{i}_{\phantom{i}j}+N\bge^{il}R_{lj}-N\bR^{i}_{\phantom{i}j}+(\ml_{\chi}\bk)^{i}_{\phantom{i}j}\\
      = & \bnabla^{i}\bnabla_{j}N-N(\tr_{\bge}\bk)\bk^{i}_{\phantom{i}j}
      +N\bge^{il}R_{lj}-N\bR^{i}_{\phantom{i}j}+(\ml_{\chi}\bk)^{i}_{\phantom{i}j}-(\ml_{\chi}\bge)^{il}\bk_{lj}.
    \end{split}
  \end{equation}
  It should here be noted that, e.g., $(\ml_{\chi}\bk)^{i}_{\phantom{i}j}$ is the Lie derivative of $\bk$ (considered as a symmetric
  covariant $2$-tensor field), one of whose indices is then raised. On the other hand,
  \[
  (\ml_{\chi}\bK)^{i}_{\phantom{i}j}=(\ml_{\chi}\bk)^{i}_{\phantom{i}j}-(\ml_{\chi}\bge)^{il}\bk_{lj}.
  \]
  In particular, (\ref{eq:dtmixedkijfirsttry}) can thus be simplified to (\ref{eq:dtbKuidj}). The lemma follows. 
\end{proof}

Next, let us specialise the above expressions to the case that Einstein's equations are fulfilled. We begin by deriving Raychaudhuri's 
equation. 
\begin{lemma}\label{lemma:hUthetaicoEE}
  Let $(M,g)$ be a spacetime. Assume that it has an expanding partial pointed foliation. Assume, moreover, $\mK$ to be non-degenerate on $I$ and
  to have a global frame. Finally, assume (\ref{eq:EE}) to hold. Then
  \begin{equation}\label{eq:Raychaydhurione}
    U(\theta)=N^{-1}\Delta_{\bge}N-\theta^{2}+\tfrac{n}{n-1}(\rho-\bp)+\tfrac{2n}{n-1}\Lambda-\bS,
  \end{equation}
  where $\rho:=T(U,U)$ and $\bp:=\bge^{ij}T_{ij}/n$.
\end{lemma}
\begin{remark}
  One immediate consequence of (\ref{eq:Raychaydhurione}) is that
  \[
  \hU(n\ln\theta)=n\tfrac{\Delta_{\bge}N}{\theta^{2}N}-n+\tfrac{n^{2}}{n-1}\tfrac{\rho-\bp}{\theta^{2}}+\tfrac{2n^{2}}{n-1}\tfrac{\Lambda}{\theta^{2}}
  -n\tfrac{\bS}{\theta^{2}}.
  \]
  Combining this equality with (\ref{eq:hUnlnthetamomqbas}) yields
  \begin{equation}\label{eq:qmainformula appendix}
    q=n-1-n\tfrac{\Delta_{\bge}N}{\theta^{2}N}-\tfrac{n^{2}}{n-1}\tfrac{\rho-\bp}{\theta^{2}}-\tfrac{2n^{2}}{n-1}\tfrac{\Lambda}{\theta^{2}}
    +n\tfrac{\bS}{\theta^{2}}.
  \end{equation}
\end{remark}
\begin{remark}
  The Hamiltonian constraint associated with (\ref{eq:EE}) reads (\ref{eq:Hamiltonianconstraintbasicversion}). 
  Introducing $\sigma_{ij}:=\bk_{ij}-\theta\bge_{ij}/n$, (\ref{eq:Hamiltonianconstraintbasicversion}) can be written
  \begin{equation}\label{eq:bSviaHamconstraint}
    \bS=\sigma^{ij}\sigma_{ij}-\tfrac{n-1}{n}\theta^{2}+2\rho+2\Lambda.
  \end{equation}
  Combining this equality with (\ref{eq:Raychaydhurione}) yields
  \begin{equation}\label{eq:Raychaydhuritwo}
    U(\theta)=N^{-1}\Delta_{\bge}N-\tfrac{1}{n}\theta^{2}-\sigma^{ij}\sigma_{ij}+\tfrac{n}{n-1}(\rho-\bp)-2\rho+\tfrac{2}{n-1}\Lambda.
  \end{equation}
\end{remark}
\begin{remark}\label{remark:CMCtransportedspacoord}
  In the case of CMC transported spatial coordinates, see \cite[Section~3]{rasq} (with a change of sign convention concerning the second fundamental
  form), $\chi=0$ and $\theta=1/t$. In that setting, (\ref{eq:Raychaydhuritwo}) yields an elliptic equation for $N$. 
\end{remark}
\begin{proof}
  Taking the trace of (\ref{eq:EE}) yields
  \[
  -\tfrac{n-1}{2}S+(n+1)\Lambda=\tr_{g}T=-T(U,U)+\bge^{ij}T_{ij}=-\rho+n\bp.
  \]
  In particular,
  \[
  S=\tfrac{2}{n-1}\rho-\tfrac{2n}{n-1}\bp+\tfrac{2(n+1)}{n-1}\Lambda,
  \]
  where $S$ is the scalar curvature of $g$. On the other hand, $\bge^{ij}R_{ij}=R_{UU}+S$. Moreover, taking the $UU$ component of (\ref{eq:EE}) yields
  $R_{UU}+S/2-\Lambda=\rho$. Thus
  \[
  \bge^{ij}R_{ij}=R_{UU}+S=\rho+\Lambda+\tfrac{1}{2}S=\tfrac{n}{n-1}(\rho-\bp)+\tfrac{2n}{n-1}\Lambda.
  \]
  Combining this calculation with (\ref{eq:preRaychaudhuri}) yields (\ref{eq:Raychaydhurione}). The lemma follows. 
\end{proof}

Next, let us turn to the evolution of $\bK$ and $\mK$. 
\begin{lemma}\label{lemma:mlUWeingarten}
  Let $(M,g)$ be a spacetime. Assume that it has an expanding partial pointed foliation. Assume, moreover, $\mK$ to be non-degenerate on $I$ and to
  have a global frame. Finally, assume (\ref{eq:EE}) to hold. Then
  \begin{equation}\label{eq:mlUbKwithEinstein}
    (\ml_{U}\bK)^{i}_{\phantom{i}j} = N^{-1}\bnabla^{i}\bnabla_{j}N-\theta\bK^{i}_{\phantom{i}j}
    +\mcP^{i}_{\phantom{i}j}+\tfrac{1}{n-1}(\rho-\bp)\de^{i}_{j}+\tfrac{2}{n-1}\Lambda\de^{i}_{j}-\bR^{i}_{\phantom{i}j},
  \end{equation}
  using the notation introduced in (\ref{eq:mfpmcPdef}). Moreover,
  \begin{equation}\label{eq:mlUmKwithEinstein appendix}
    \begin{split}
      (\hml_{U}\mK)^{i}_{\phantom{i}j} =  & -\left(\tfrac{\Delta_{\bge}N}{\theta^{2}N}+\tfrac{n}{n-1}\tfrac{\rho-\bp}{\theta^{2}}
      +\tfrac{2n}{n-1}\tfrac{\Lambda}{\theta^{2}}-\tfrac{\bS}{\theta^{2}}\right)\mK^{i}_{\phantom{i}j}+\tfrac{1}{n-1}\tfrac{\rho-\bp}{\theta^{2}}\de^{i}_{j}\\
      & +\tfrac{2}{n-1}\tfrac{\Lambda}{\theta^{2}}\de^{i}_{j}
      +\tfrac{1}{N\theta^{2}}\bnabla^{i}\bnabla_{j}N+\tfrac{\mcP^{i}_{\phantom{i}j}}{\theta^{2}}-\tfrac{\bR^{i}_{\phantom{i}j}}{\theta^{2}}.
    \end{split}
  \end{equation}
\end{lemma}
\begin{remark}
  In the special case of CMC transported spatial coordinates, see Remark~\ref{remark:CMCtransportedspacoord}, (\ref{eq:mlUmKwithEinstein appendix}) can be
  written (keeping (\ref{eq:Raychaydhurione}) in mind)
  \begin{equation}\label{eq:mlUmKwithEinsteinandCMC}
    \begin{split}
      (\hml_{U}\mK)^{i}_{\phantom{i}j} = & -\tfrac{1}{N}(N-1)\mK^{i}_{\phantom{i}j}+\tfrac{1}{N\theta^{2}}\bnabla^{i}\bnabla_{j}N
      +\theta^{-2}\mcP^{i}_{\phantom{i}j}+\tfrac{1}{n-1}\tfrac{\rho-\bp}{\theta^{2}}\de^{i}_{j}+\tfrac{2}{n-1}\tfrac{\Lambda}{\theta^{2}}\de^{i}_{j}
      -\tfrac{\bR^{i}_{\phantom{i}j}}{\theta^{2}}.
    \end{split}
  \end{equation}
\end{remark}
\begin{remark}
  It is also of interest to note that (keeping (\ref{eq:qmainformula appendix}) in mind)
  \begin{equation}\label{eq:mlUmKwithEinsteinandq}
    \begin{split}
      (\hml_{U}\mK)^{i}_{\phantom{i}j} =  & \tfrac{q-(n-1)}{n}\mK^{i}_{\phantom{i}j}-\theta^{-2}\bR^{i}_{\phantom{i}j}
      +\theta^{-1}\hN^{-1}\bnabla^{i}\bnabla_{j}N+\theta^{-2}\mcP^{i}_{\phantom{i}j}\\
      & +\tfrac{1}{n-1}\tfrac{\rho-\bp}{\theta^{2}}\de^{i}_{j}+\tfrac{2}{n-1}\tfrac{\Lambda}{\theta^{2}}\de^{i}_{j}.
    \end{split}
  \end{equation}
\end{remark}
\begin{proof}
  Considering (\ref{eq:mlUbKwithoutEinstein}), it is of interest to calculate $\bge^{il}R_{lj}$. Note, to this end, that the Einstein
  equations can be reformulated to
  \[
  \mathrm{Ric}=T-\tfrac{1}{n-1}(\tr_{g}T)g+\tfrac{2}{n-1}\Lambda g.
  \]
  Moreover,
  $\tr_{g}T=-T(U,U)+\bge^{ij}T_{ij}=-\rho+n\bp$. Combining the last two observations yields
  \[
  \bge^{il}R_{lj}=\bge^{il}T_{lj}+\tfrac{\rho}{n-1}\de^{i}_{j}-\tfrac{n}{n-1}\bp\de^{i}_{j}+\tfrac{2}{n-1}\Lambda\de^{i}_{j}.
  \]
  Introducing the notation (\ref{eq:mfpmcPdef}), it is clear that $\mcP^{i}_{\phantom{i}i}=0$ and that
  \[
  \bge^{il}R_{lj}=\mcP^{i}_{\phantom{i}j}+\tfrac{1}{n-1}(\rho-\bp)\de^{i}_{j}+\tfrac{2}{n-1}\Lambda\de^{i}_{j}.
  \]
  Combining this equality with (\ref{eq:mlUbKwithoutEinstein}) yields (\ref{eq:mlUbKwithEinstein}). Next, note that
  \[
  (\hml_{U}\mK)^{i}_{\phantom{i}j} =\theta^{-2}(\ml_{U}\bK)^{i}_{\phantom{i}j}-\theta^{-2}(\ml_{U}\theta)\mK^{i}_{\phantom{i}j}.
  \]
  Combining this equality with (\ref{eq:Raychaydhurione}) and (\ref{eq:mlUbKwithEinstein}) yields
  \begin{equation*}
    \begin{split}
      (\hml_{U}\mK)^{i}_{\phantom{i}j} = & \tfrac{1}{N\theta^{2}}\bnabla^{i}\bnabla_{j}N-\mK^{i}_{\phantom{i}j}
      +\theta^{-2}\mcP^{i}_{\phantom{i}j}+\tfrac{1}{n-1}\tfrac{\rho-\bp}{\theta^{2}}\de^{i}_{j}+\tfrac{2}{n-1}\tfrac{\Lambda}{\theta^{2}}\de^{i}_{j}
      -\tfrac{\bR^{i}_{\phantom{i}j}}{\theta^{2}}\\
      & +\left(1-\tfrac{\Delta_{\bge}N}{\theta^{2}N}-\tfrac{n}{n-1}\tfrac{\rho-\bp}{\theta^{2}}-\tfrac{2n}{n-1}\tfrac{\Lambda}{\theta^{2}}
      +\tfrac{\bS}{\theta^{2}}\right)\mK^{i}_{\phantom{i}j}.
    \end{split}
  \end{equation*}
  The lemma follows. 
\end{proof}

\section{Conditions ensuring quiescence}\label{section:condquiescentregimes}

The purpose of the present section is to prove Lemma~\ref{lemma:vacuumquiescent}. 

\begin{proof}[Lemma~\ref{lemma:vacuumquiescent}]
Let $3\leq n\in\zo$ and let $\ell\in K_{o}$ be such that $F|_{K_{o}}$ is maximised. To begin with, we prove that there is no loss of generality in 
assuming that the set $E:=\{\ell_{1},\dots,\ell_{n}\}$ consists of exactly two elements. In order to prove this, note first that the assumption that
$E$ consists of exactly one element is not consistent with (\ref{eq:bothKasnerrelations}). Assume, therefore, that $E$ contains three distinct 
numbers. We can then choose three distinct elements of $E$, say $\ell_{a}$, $\ell_{b}$ and $\ell_{c}$, such that $\ell_{a}$ equals $\ell_{1}$, $\ell_{c}$ 
equals $\ell_{n}$ and $\ell_{b}$ equals the second largest element in $E$. Then $\ell_{a}<\ell_{b}<\ell_{c}$. Assume that there are $k$ elements $\ell_{A}$
equalling $\ell_{a}$, $m$ elements $\ell_{A}$ equalling $\ell_{b}$ and $p$ elements $\ell_{A}$ equalling $\ell_{c}$. Define $\a$ and $\b$ by the relations
\[
k\ell_{a}+m\ell_{b}+p\ell_{c}=\a,\ \ \
k\ell_{a}^{2}+m\ell_{b}^{2}+p\ell_{c}^{2}=\b.
\]
In what follows, we are interested in the set $S$ of elements $(l_{a},l_{b},l_{c})\in\rn{3}$ such that 
\[
kl_{a}+ml_{b}+pl_{c}=\a,\ \ \
kl_{a}^{2}+ml_{b}^{2}+pl_{c}^{2}=\b
\]
and such that $l_{a}<l_{b}<l_{c}$. By assumption, $S$ is non-empty. Let $q:=k+m+p$, $(l_{a},l_{b},l_{c})\in S$ and define $v_{0}$ and $\msL$ by 
\[
v_{0}:=(\sqrt{k},\sqrt{m},\sqrt{p}),\ \ \
\msL:=(\sqrt{k}l_{a},\sqrt{m}l_{b},\sqrt{p}l_{c}).
\]
Then
\[
|\a|=|v_{0}\cdot \msL|\leq \sqrt{q}\sqrt{\b}.
\]
In fact, we have strict inequality, since equality would imply that $v_{0}$ is parallel with $\msL$, so that $l_{a}=l_{b}=l_{c}$, contradicting 
the assumptions. Thus $\a^{2}<q\b$. Next, let 
\[
v_{+}:=\tfrac{1}{\sqrt{(m+p)q}}(m+p,-\sqrt{mk},-\sqrt{pk}),\ \ \
v_{-}:=\tfrac{1}{\sqrt{p+m}}(0,\sqrt{p},-\sqrt{m}).
\]
Then $\{v_{0},v_{+},v_{-}\}$ is an orthogonal frame. Let 
\[
L:=[\sqrt{k}(l_{a}-\a/q),\sqrt{m}(l_{b}-\a/q),\sqrt{p}(l_{c}-\a/q)].
\]
Then $L\cdot v_{0}=0$, so that there are real numbers $L_{\pm}$ such that 
\[
L_{+}v_{+}+L_{-}v_{-}=L. 
\]
Since $v_{+}$ and $v_{-}$ are orthonormal,
\begin{align*}
L_{+} = & \left(\tfrac{kq}{m+p}\right)^{1/2}\left(l_{a}-\tfrac{\a}{q}\right)
=\left(\tfrac{q(m+p)}{k}\right)^{1/2}\left(\tfrac{\a}{q}-\tfrac{1}{m+p}(ml_{b}+pl_{c})\right),\\
L_{-} = & \left(\tfrac{pm}{p+m}\right)^{1/2}(l_{b}-l_{c}). 
\end{align*}
Moreover, since $v_{+}$ and $v_{-}$ are orthonormal,
\[
L_{+}^{2}+L_{-}^{2}=\b-\tfrac{\a^{2}}{q}=r^{2},
\]
where the last equality defines $r>0$; recall that $q\b-\a^{2}>0$. Let $t\in (-\pi,\pi]$ be such that 
\[
L_{+}=-r\sin(t),\ \ \
L_{-}=-r\cos(t).
\]
Since $l_{b}<l_{c}$, we know that $L_{-}<0$, so that $t\in (-\pi/2,\pi/2)$. Moreover, 
\begin{align*}
l_{a} = & -\left(\tfrac{m+p}{kq}\right)^{1/2}r\sin(t)+\tfrac{\a}{q},\\
l_{b} = & \left(\tfrac{k}{q(m+p)}\right)^{1/2}r\sin(t)-\left(\tfrac{p}{m(m+p)}\right)^{1/2}r\cos(t)+\tfrac{\a}{q},\\
l_{c} = & \left(\tfrac{k}{q(m+p)}\right)^{1/2}r\sin(t)+\left(\tfrac{m}{p(m+p)}\right)^{1/2}r\cos(t)+\tfrac{\a}{q}.
\end{align*}
From now on, we consider $l_{a}$, $l_{b}$ and $l_{c}$ to be functions of $t$, and we write $l_{a}(t)$ etc. The condition that $l_{a}<l_{b}$ translates to 
\[
\tan(t)>\left(\tfrac{pk}{mq}\right)^{1/2}.
\]
This condition can be written $t>\theta_{0}$ for some $\theta_{0}\in (0,\pi/2)$. Summarising, $t\in (\theta_{0},\pi/2)$. In particular, 
$(\ell_{a},\ell_{b},\ell_{c})$ corresponds to some $t_{0}\in (\theta_{0},\pi/2)$. Let $l(t)$ be the function obtained from $\ell$ by replacing
all the $\ell_{A}$ equalling $\ell_{a}$ by $l_{a}(t)$; replacing all the $\ell_{A}$ equalling $\ell_{b}$ by $l_{b}(t)$; and replacing
all the $\ell_{A}$ equalling $\ell_{c}$ by $l_{c}(t)$. Then $\ell=l(t_{0})$. By construction, (\ref{eq:bothKasnerrelations}) is satisfied for 
all $t\in (-\pi,\pi]$. Moreover, (\ref{eq:ellorder}) is satisfied for $t$ close to $t_{0}$. In particular, $l(t)\in K_{o}$ for $t$ close to 
$t_{0}$. If $t_{0}$ is not a local maximum of $F\circ l$, we get a contradiction to the assumption that $\ell$ is a global maximum of $F$ on $K_{o}$. 
In case $p\geq 2$, maximising $F\circ l$ is the same as maximising the function $h$ defined by $h(t):=l_{a}(t)-2l_{c}(t)$. However, $h(t)$ can, up to
a constant, be written as a sum $a\sin(t)+b\cos(t)$ where $a,b<0$. Thus $h$ does not have a local maximum in $(\theta_{0},\pi/2)$ (since the second
derivative is strictly positive). This yields a 
contradiction to the fact that $t_{0}$ is a maximum for $F\circ l$. Next, assume that $p=1$. In that setting we wish to maximize $f$
defined by $f(t):=l_{a}(t)-l_{b}(t)-l_{c}(t)$. Again, $f(t)$ can be written as a sum $a\sin(t)+b\cos(t)$ where $a<0$ and $b\leq 0$ (if 
$m=1$, $b=0$ and if $m\geq 2$, $b<0$). Thus $f$ does not have a local maximum in $(\theta_{0},\pi/2)$ and we obtain a contradiction. 
To conclude, the assumption that $S$ contains three distinct elements leads to a contradiction. 

Since, as already noted, $\{\ell_{1},\dots,\ell_{n}\}$ cannot consist of only one element, we conclude that this set consists of two elements, say 
$\ell_{a}$ and $\ell_{b}$. Assume that there are $k$ elements of the form $\ell_{a}$ and $m$ elements of the form $\ell_{b}$. Then $\ell_{a}$ and $\ell_{b}$
have to satisfy
\begin{equation}\label{eq:Kasnertwoelements}
k\ell_{a}+m\ell_{b}=1,\ \ \
k\ell_{a}^{2}+m\ell_{b}^{2}=1,\ \ \
\ell_{a}\leq \ell_{b}.
\end{equation}
Note also that in order for $F(\ell)$ to be strictly positive, we have to have $\ell_{a}<0$ and $m\geq 3$. The reason why $\ell_{a}<0$ is that if 
$\ell_{a}\geq 0$, then $\ell_{n}=1$ and $\ell_{i}=0$ if $i\leq n-1$, so that $F(\ell)=0$; note that $\ell_{a}\geq 0$ and $\ell_{b}<1$ yields
\[
1=k\ell_{a}^{2}+m\ell_{b}^{2}<k\ell_{a}+m\ell_{b}=1,
\]
a contradiction. Moreover, (\ref{eq:bothKasnerrelations}) can be used to deduce that 
\[
F(\ell)=2\ell_{1}+\textstyle{\sum}_{j=2}^{n-2}\ell_{j}.
\]
Thus $m\geq 3$, since we would otherwise have $F(\ell)\leq 0$. In the present setting, (\ref{eq:Kasnertwoelements}) combined with the condition
that $\ell_{a}<0$ yields the conclusion that 
\[
\ell_{a} = \tfrac{1}{n}\left[1-\left(\tfrac{m}{k}\right)^{1/2}(n-1)^{1/2}\right],\ \ \
\ell_{b} = \tfrac{1}{n}\left[1+\left(\tfrac{k}{m}\right)^{1/2}(n-1)^{1/2}\right].
\]
The corresponding $\ell$ (assuming $m\geq 3$) is such that 
\[
F(\ell)=\tfrac{n-1}{n}\left(1-\left(\tfrac{m}{k(n-1)}\right)^{1/2}-2\left(\tfrac{k}{m(n-1)}\right)^{1/2}\right). 
\]
We want to maximise this expression. This means that we want to minimise
\[
v(s):=s+2s^{-1}, 
\]
where $s$ should be thought of as equalling $(m/k)^{1/2}$. This function has a global minimum when $s=\sqrt{2}$; i.e., when $m=2k$. However, this 
only happens when $n=3k$ for some $1\leq k\in\zo$. If $n$ cannot be written in this way, we need to make a division such that $m/k$ is as close
to $2$ as possible. Note now that in case $n=9$, the optimal choice is $k=3$ and $m=6$. This yields $F(\ell)=0$ (recall that we are
assuming $m\geq 3$; choosing $\ell=(0,\dots,0,1)$ yields $F(\ell)=0$). Since there is no $\ell\in K_{o}$ such that $F(\ell)>0$ in case $n=9$, there is,
by the above, also no $\ell\in K_{o}$ such that $F(\ell)>0$ in case $3\leq n\leq 9$. However, in the case of $n=10$, letting $m=7$ and $k=3$ yields an
$\ell$ such that $F(\ell)>0$. Since the solution in the case of $n=10$ is also a solution for $n\geq 10$, the lemma follows. 
\end{proof}

\section{The spatial scalar curvature}\label{section:sp sc curv}

In most of this article, the arguments concerning the asymptotics are based on assumptions concerning $\hml_{U}^{l}\mK$ for $l=0,1$. However, 
assuming Einstein's equations to be satisfied, $\hml_{U}\mK$ can be calculated in terms of the spatial curvatures, the lapse function, 
the mean curvature, the cosmological constant and components of the stress energy tensor; see (\ref{eq:mlUmKwithEinstein appendix}). For this reason, 
it is of interest to calculate the curvatures of the hypersurfaces $\bM_{t}$ in terms of the structure coefficients $\g^{A}_{BC}$, the $\mu_{A}$, the
$\bmu_A$ and derivatives of these quantities with respect to the frame $\{X_A\}$. 
Moreover, some of the bounds on $\mK$ arise from the constraint equations. For that reason, it is of interest to express the constraint 
equations with respect to the same quantities. 

Due to (\ref{eq:muAbmuAdefinition}), $\{\ONSF_{A}\}$, defined by 
\[
\ONSF_{A}:=\theta e^{-\mu_{A}}X_{A}=e^{-\bmu_{A}}X_{A}
\]
(no summation on $A$), is an orthonormal frame with respect to $\bge$. We begin by calculating the Ricci and scalar curvature. 

\subsection{Structure coefficients and connection coefficients}
In what follows, we use the notation
\begin{equation}\label{eq:bgABCdef}
  \bnabla_{\ONSF_{A}}\ONSF_{B}=\ONCS_{AB}^{C}\ONSF_{C},\ \ \
  \ONCS^{C}:=\de^{AB}\ONCS_{AB}^{C}. 
\end{equation}
Defining $\bga^{A}_{BC}$ by $[\ONSF_{A},\ONSF_{B}]=\bga_{AB}^{C}\ONSF_{C}$, the Koszul formula yields 
\begin{equation}\label{eq:bGCABdef}
\ONCS^{C}_{AB}=\tfrac{1}{2}(-\bga^{A}_{BC}+\bga^{B}_{CA}+\bga^{C}_{AB}). 
\end{equation}
It is of interest to express the connection coefficients with respect to $\g^{A}_{BC}$, $\mu_{A}$ and $\theta$. Note, to that end, that
\[
[\ONSF_{A},\ONSF_{B}]=-\theta e^{-\mu_{A}}X_{A}(\bmu_{B})\ONSF_{B}+\theta e^{-\mu_{B}}X_{B}(\bmu_{A})\ONSF_{A}+\theta^{2}e^{-\mu_{A}-\mu_{B}}[X_{A},X_{B}]
\]
(no summation on $A$ or $B$). Thus
\begin{equation}\label{eq:bgaABCitobasqu}
\bga_{AB}^{C}=-\theta e^{-\mu_{A}}X_{A}(\bmu_{C})\de_{BC}+\theta e^{-\mu_{B}}X_{B}(\bmu_{C})\de_{AC}+\theta e^{\mu_{C}-\mu_{A}-\mu_{B}}\g^{C}_{AB}
\end{equation}
(no summation on any index). Next, note that 
\[
\ONCS^{C}=\tfrac{1}{2}\textstyle{\sum}_{A}(-\bga^{A}_{AC}+\bga^{A}_{CA}+\bga^{C}_{AA})=\bga^{A}_{CA}=:\ba_{C}.
\]
On the other hand, (\ref{eq:bgaABCitobasqu}) yields
\begin{equation}\label{eq:aAformula}
\ba_{A}=\bga_{AB}^{B}
=\theta e^{-\mu_{A}}\left(X_{A}\left(\bmu_{A}-\textstyle{\sum}_{B}\bmu_{B}\right)+\g^{B}_{AB}\right)
\end{equation}
(no summation on $A$). 

\subsection{Spatial scalar curvature}

Next, let us calculate the spatial Ricci and scalar curvature. To begin with
\begin{equation}\label{eq:bRicbasicformula}
\begin{split}
\bRic(\ONSF_{A},\ONSF_{B}) = & \textstyle{\sum}_{C}\ldr{\bR_{\ONSF_{A}\ONSF_{C}}\ONSF_{C},\ONSF_{B}}\\
= & \textstyle{\sum}_{C}\ldr{\bnabla_{\ONSF_{A}}\bnabla_{\ONSF_{C}}\ONSF_{C}-\bnabla_{\ONSF_{C}}\bnabla_{\ONSF_{A}}\ONSF_{C}
  -\bnabla_{[\ONSF_{A},\ONSF_{C}]}\ONSF_{C},\ONSF_{B}}\\
= & \textstyle{\sum}_{C}\ldr{\bnabla_{\ONSF_{A}}(\ONCS^{D}_{CC}\ONSF_{D})-\bnabla_{\ONSF_{C}}(\ONCS_{AC}^{D}\ONSF_{D})
  -\bga^{D}_{AC}\bnabla_{\ONSF_{D}}\ONSF_{C},\ONSF_{B}}\\
= & \textstyle{\sum}_{C}\left[\ONSF_{A}(\ONCS^{B}_{CC})+\ONCS^{D}_{CC}\ONCS_{AD}^{B}-\ONSF_{C}(\ONCS_{AC}^{B})
  -\ONCS_{AC}^{D}\ONCS_{CD}^{B}-\bga^{D}_{AC}\ONCS_{DC}^{B}\right].
\end{split}
\end{equation}
Let us focus on the scalar curvature
\begin{equation*}
\begin{split}
  \bS = & \textstyle{\sum}_{C}\left[\ONSF_{A}(\ONCS^{A}_{CC})+\ONCS^{B}_{CC}\ONCS_{AB}^{A}-\ONSF_{C}(\ONCS_{AC}^{A})
    -\ONCS_{AC}^{B}\ONCS_{CB}^{A}-\bga^{B}_{AC}\ONCS_{BC}^{A}\right].
\end{split}
\end{equation*}

\begin{lemma}\label{lemma:scalarcurvature}
Let $(M,g)$ be a spacetime. Assume that it has an expanding partial pointed foliation. Assume, moreover, $\mK$ to be 
non-degenerate on $I$ and to have a global frame. Then 
\begin{equation}\label{eq:bSexpression}
\begin{split}
\theta^{-2}\bS = &  -\textstyle{\frac{1}{4}\sum}_{A,B,C}e^{2\mu_{B}-2\mu_{A}-2\mu_{C}}(\g^{B}_{AC})^{2}
-\tfrac{1}{2}\sum_{A,B,C}e^{-2\mu_{A}}\g^{B}_{AC}\g^{C}_{AB}\\
 & -\textstyle{\sum}_{A,B}e^{-2\mu_{A}}|X_{A}(\bmu_{B})|^{2}+2\textstyle{\sum}_{A,B}e^{-2\mu_{A}}X_{A}(\bmu_{B})\g^{B}_{AB}-\sum_{A}e^{-2\mu_{A}}\cha_{A}^{2}\\
 & +\textstyle{\sum}_{A}e^{-2\mu_{A}}|X_{A}(\bmu_{A})|^{2}
 +2\textstyle{\sum}_{A}e^{-2\mu_{A}}[X_{A}(\cha_{A})-X_{A}(\bmu_{A})\cha_{A}],
\end{split}
\end{equation}
where 
\begin{align}
\cha_{A} := & -X_{A}(\hmu_{A})+a_{A},\label{eq:chAdef}\\
\hmu_{A} := & \textstyle{\sum}_{B}\bmu_{B}-\bmu_{A},\label{eq:hmuAdef}\\
a_{A} := & \g^{B}_{AB}.\label{eq:aAdefaux}
\end{align}
\end{lemma}
\begin{proof}
Note that
\begin{equation}\label{eq:bGACCsumCaA}
\textstyle{\sum}_{C}\ONCS^{A}_{CC}=\bga^{C}_{AC}=\ba_{A},\ \ \
\ONCS_{AC}^{A}=-\ba_{C}.
\end{equation}
Thus
\[
\textstyle{\sum}_{C}[\ONSF_{A}(\ONCS^{A}_{CC})-\ONSF_{C}(\ONCS_{AC}^{A})]=\textstyle{\sum}_{A}2\ONSF_{A}(\ba_{A}).
\]
Considering (\ref{eq:aAformula}), it is clear that 
\begin{equation}\label{eq:chAdefiningrelation}
\ba_{A}=\theta e^{-\mu_{A}}\cha_{A}=e^{-\bmu_{A}}\cha_{A},
\end{equation}
where $\cha_{A}$ is given by (\ref{eq:chAdef}). Thus
\[
\ONSF_{A}(\ba_{A})=\theta^{2} e^{-2\mu_{A}}[X_{A}(\cha_{A})-X_{A}(\bmu_{A})\cha_{A}].
\]
Next, note that 
\[
\textstyle{\sum}_{C}\ONCS^{B}_{CC}\ONCS_{AB}^{A}=-\textstyle{\sum}_{B}\ba_{B}^{2}
\]
and that 
\begin{equation*}
\begin{split}
\textstyle{\sum}_{A,B,C}\ONCS_{AC}^{B}\ONCS_{CB}^{A}
= & \textstyle{\frac{1}{4}\sum}_{A,B,C}(-\bga^{A}_{CB}+\bga^{C}_{BA}+\bga^{B}_{AC})(-\bga^{C}_{BA}+\bga^{B}_{AC}+\bga^{A}_{CB})\\
= & -\textstyle{\frac{1}{4}\sum}_{A,B,C}(\bga^{C}_{BA})^{2}
+\tfrac{1}{2}\textstyle{\sum}_{A,B,C}\bga^{C}_{BA}\bga^{B}_{AC}.
\end{split}
\end{equation*}
In addition, 
\begin{equation*}
\begin{split}
\textstyle{\sum}_{A,B,C}\bga^{B}_{AC}\ONCS_{BC}^{A} = & \textstyle{\frac{1}{2}\sum}_{A,B,C}\bga^{B}_{AC}(-\bga^{B}_{CA}+\bga^{C}_{AB}+\bga^{A}_{BC})\\
 = & \textstyle{\frac{1}{2}\sum}_{A,B,C}(\bga^{B}_{AC})^{2}+\sum_{A,B,C}\bga^{B}_{AC}\bga^{C}_{AB}. 
\end{split}
\end{equation*}
Summing up yields 
\begin{equation}\label{eq:bSpreliminaryexpression}
\begin{split}
\theta^{-2}\bS = & 2\textstyle{\sum}_{A}e^{-2\mu_{A}}[X_{A}(\cha_{A})-X_{A}(\bmu_{A})\cha_{A}]-\sum_{A}e^{-2\mu_{A}}\cha_{A}^{2}\\
& -\textstyle{\frac{1}{4}\sum}_{A,B,C}\theta^{-2}(\bga^{B}_{AC})^{2}-\tfrac{1}{2}\sum_{A,B,C}\theta^{-2}\bga^{B}_{AC}\bga^{C}_{AB},
\end{split}
\end{equation}
where $\cha_{A}$ is given by (\ref{eq:chAdef}). In order to proceed, it is of interest to expand the last two terms on the right 
hand side of (\ref{eq:bSpreliminaryexpression}). Due to (\ref{eq:bgaABCitobasqu}),
\[
\theta^{-1}\bga_{AC}^{B}=-e^{-\mu_{A}}X_{A}(\bmu_{B})\de_{BC}+e^{-\mu_{C}}X_{C}(\bmu_{B})\de_{AB}+e^{\mu_{B}-\mu_{A}-\mu_{C}}\g^{B}_{AC}.
\]
Thus
\begin{equation*}
\begin{split}
-\textstyle{\frac{1}{4}\sum}_{A,B,C}\theta^{-2}(\bga^{B}_{AC})^{2} = & -\textstyle{\frac{1}{4}\sum}_{A,B,C}e^{2\mu_{B}-2\mu_{A}-2\mu_{C}}(\g^{B}_{AC})^{2}\\
 & -\textstyle{\frac{1}{2}\sum}_{A,B}e^{-2\mu_{A}}|X_{A}(\bmu_{B})|^{2}+\tfrac{1}{2}\textstyle{\sum}_{A}e^{-2\mu_{A}}|X_{A}(\bmu_{A})|^{2}\\
 & +\textstyle{\sum}_{A,B}e^{-2\mu_{A}}X_{A}(\bmu_{B})\g^{B}_{AB}.
\end{split}
\end{equation*}
It can also be calculated that 
\begin{equation*}
\begin{split}
-\textstyle{\frac{1}{2}\sum}_{A,B,C}\theta^{-2}\bga^{B}_{AC}\bga^{C}_{AB} = & 
-\textstyle{\frac{1}{2}\sum}_{A,B}e^{-2\mu_{A}}|X_{A}(\bmu_{B})|^{2}+\tfrac{1}{2}\textstyle{\sum}_{A}e^{-2\mu_{A}}|X_{A}(\bmu_{A})|^{2}\\
 & +\textstyle{\sum}_{A,B}e^{-2\mu_{A}}X_{A}(\bmu_{B})\g^{B}_{AB}-\tfrac{1}{2}\sum_{A,B,C}e^{-2\mu_{A}}\g^{B}_{AC}\g^{C}_{AB}.
\end{split}
\end{equation*}
Summing up yields (\ref{eq:bSexpression}). 
\end{proof}

\section{The spatial Ricci curvature}\label{section:sp ricci cu}

Above, we calculated the scalar curvature of $\bge$. However, it is also of interest to calculate the Ricci curvature. Note, to this end, that 
\[
\bRic(X_{A},X_{B})=\theta^{-2}e^{\mu_{A}+\mu_{B}}\bRic(\ONSF_{A},\ONSF_{B}).
\]
Thus
\[
\bR^{A}_{\phantom{A}B}:=\bge^{AC}\bRic(X_{C},X_{B})=e^{\mu_{B}-\mu_{A}}\bRic(\ONSF_{A},\ONSF_{B}).
\]
In what follows, we wish to calculate $\theta^{-2}\bR^{A}_{\phantom{A}B}$. It will be convenient to use the notation
\begin{equation}\label{eq:Upsilondef}
\Upsilon_{AD}^{B}:=\theta^{-1}e^{\mu_{B}-\mu_{A}-\mu_{D}}\ONCS_{AD}^{B}.
\end{equation}

\begin{lemma}\label{lemma:therenormalisedspatialriccicurvature}
  Let $(M,g)$ be a spacetime. Assume that it has an expanding partial pointed foliation. Assume, moreover, $\mK$ to be non-degenerate on $I$ and to
  have a global frame. Then
\begin{equation*}
\begin{split}
\theta^{-2}\bR^{A}_{\phantom{A}B} = & \mR_{\mrI,B}^{A}+\mR_{\mrII,B}^{A}+\mR_{\mrIII,B}^{A}+\mR_{\mrIV,B}^{A}
\end{split}
\end{equation*}
where
\begin{align}
\mR_{\mrI,B}^{A} := & e^{-2\mu_{A}}X_{A}(\cha_{B})-\textstyle{\sum}_{C}X_{C}(\Upsilon^{B}_{AC}),\label{eq:mRIABdef}\\
\mR_{\mrII,B}^{A} := & \textstyle{\frac{1}{2}\sum}_{C}e^{2\mu_{B}-2\mu_{A}-2\mu_{C}}\left[a_{C}-X_{C}\left(2\ln\theta+\bmu_{\rotot}\right)\right]\g^{B}_{AC}
\label{eq:mRIIABdef}\\
 & -\textstyle{\frac{1}{2}\sum}_{C,D}e^{2\mu_{D}-2\mu_{A}-2\mu_{C}}\g^{D}_{AC}\g^{D}_{BC}+\textstyle{\frac{1}{4}\sum}_{C,D}e^{2\mu_{B}-2\mu_{C}-2\mu_{D}}\g^{A}_{CD}\g^{B}_{CD},\nonumber\\
\mR_{\mrIII,B}^{A}:= & \textstyle{\frac{1}{2}}e^{-2\mu_{A}}X_{A}(\bmu_{B})X_{B}(2\bmu_{\rotot}-\bmu_{A})\label{eq:mRIIIABdef}\\
 & +\textstyle{\frac{1}{2}}e^{-2\mu_{A}}X_{A}(2\bmu_{\rotot}+4\ln\theta+4\bmu_{A}-3\bmu_{B})X_{B}(\bmu_{A})\nonumber\\
 & -\textstyle{\sum}_{C}e^{-2\mu_{C}}X_{C}(\bmu_{\rotot}+2\ln\theta)X_{C}(\bmu_{A})\de_{AB}-\textstyle{\sum}_{C}e^{-2\mu_{A}}X_{A}(\bmu_{C})X_{B}(\bmu_{C}),\nonumber\\
\mR_{\mrIV,B}^{A} := & -e^{-2\mu_{A}}\left(X_{A}(\bmu_{B})a_{B}+X_{B}(\bmu_{A})a_{A}\right)+\textstyle{\sum}_{C}e^{-2\mu_{C}}X_{C}(\bmu_{A})a_{C}\de_{AB}\label{eq:mRIVABdef}\\
 & +\textstyle{\sum}_{C}e^{-2\mu_{A}}[X_{B}(\bmu_{C})\g^{C}_{AC}+X_{A}(\bmu_{C})\g^{C}_{BC}]\nonumber\\
 & -\textstyle{\frac{1}{2}}\textstyle{\sum}_{C}e^{-2\mu_{C}}X_{C}(\bmu_{\rotot}+2\ln\theta+2\bmu_{A}-2\bmu_{B})\g^{A}_{BC}\nonumber\\
 & +\textstyle{\frac{1}{2}}\textstyle{\sum}_{C}e^{-2\mu_{A}}X_{C}(\bmu_{\rotot}+2\ln\theta+2\bmu_{A}-2\bmu_{B})\g^{C}_{AB},\nonumber\\
\mR_{\mrV,B}^{A} := & -\textstyle{\frac{1}{2}}\textstyle{\sum}_{C}e^{-2\mu_{C}}a_{C}\g^{A}_{CB}+\textstyle{\frac{1}{2}}\textstyle{\sum}_{C}e^{-2\mu_{A}}a_{C}\g^{C}_{BA} 
-\textstyle{\frac{1}{2}}\textstyle{\sum}_{C,D}e^{-2\mu_{A}}\g^{D}_{AC}\g^{C}_{BD}\label{eq:mRVABdef}
\end{align}
where $\bmu_{\rotot}:=\textstyle{\sum}_{D}\bmu_{D}$ and 
\begin{equation}\label{eq:renormalizedChristoffelsymbol}
\begin{split}
\Upsilon_{AD}^{B} = & e^{-2\mu_{D}}X_{D}(\bmu_{A})\de_{AB}-e^{-2\mu_{A}}X_{B}(\bmu_{A})\de_{AD}-\textstyle{\frac{1}{2}}e^{-2\mu_{D}}\g^{A}_{DB}\\
 & +\textstyle{\frac{1}{2}}e^{-2\mu_{A}}\g^{D}_{BA}+\tfrac{1}{2}e^{2\mu_{B}-2\mu_{A}-2\mu_{D}}\g^{B}_{AD}. 
\end{split}
\end{equation}
\end{lemma}
\begin{remark}
The purpose of the division into $\mR_{\mrI,B}^{A}$ etc. is that the first term contains all the second derivatives, the second term contains all the 
lower order terms that could potentially grow, and the third to fifth terms correspond to a division into different types of decaying lower order terms. 
\end{remark}
\begin{proof}
Due to (\ref{eq:bRicbasicformula}), we need to calculate
\begin{equation*}
\begin{split}
\textstyle{\sum}_{C}\theta^{-2}e^{\mu_{B}-\mu_{A}}\ONSF_{A}(\ONCS^{B}_{CC}) = & \theta^{-1}e^{\mu_{B}-2\mu_{A}}X_{A}(\ba_{B})=\theta^{-1}e^{\mu_{B}-2\mu_{A}}X_{A}(e^{-\bmu_{B}}\cha_{B})\\
 = & e^{-2\mu_{A}}[X_{A}(\cha_{B})-X_{A}(\bmu_{B})\cha_{B}]\\
 = & e^{-2\mu_{A}}[-X_{A}X_{B}(\hmu_{B})+X_{A}(a_{B})+X_{A}(\bmu_{B})X_{B}(\hmu_{B})-X_{A}(\bmu_{B})a_{B}],
\end{split}
\end{equation*}
where we used (\ref{eq:bGACCsumCaA}). Next, consider
\begin{equation}\label{eq:secondricciterm}
\theta^{-2}e^{\mu_{B}-\mu_{A}}\textstyle{\sum}_{C}\ONCS^{D}_{CC}\ONCS_{AD}^{B}=
\textstyle{\sum}_{D}\theta^{-1}e^{\mu_{B}-\mu_{A}-\mu_{D}}\cha_{D}\ONCS_{AD}^{B}=
\textstyle{\sum}_{D}\cha_{D}\Upsilon_{AD}^{B},
\end{equation}
where we used (\ref{eq:bGACCsumCaA}), (\ref{eq:chAdefiningrelation}) and (\ref{eq:Upsilondef}). In what follows, it is of interest to 
calculate $\Upsilon_{AD}^{B}$. Note, to this end, that 
\begin{equation}\label{eq:renormalizedchristoffelsymbols}
\begin{split}
\Upsilon_{AD}^{B} = & -\textstyle{\frac{1}{2}}\theta^{-1}e^{\mu_{B}-\mu_{A}-\mu_{D}}\bga^{A}_{DB}
+\tfrac{1}{2}\theta^{-1}e^{\mu_{B}-\mu_{A}-\mu_{D}}\bga^{D}_{BA}+\textstyle{\frac{1}{2}}\theta^{-1}e^{\mu_{B}-\mu_{A}-\mu_{D}}\bga^{B}_{AD}.
\end{split}
\end{equation}
Keeping (\ref{eq:bgaABCitobasqu}) in mind,
\begin{align*}
\theta^{-1}e^{\mu_{B}-\mu_{A}-\mu_{D}}\bga^{A}_{DB} = & -e^{-2\mu_{D}}X_{D}(\bmu_{A})\de_{BA}+e^{-2\mu_{D}}X_{B}(\bmu_{A})\de_{DA}+e^{-2\mu_{D}}\g^{A}_{DB},\\
\theta^{-1}e^{\mu_{B}-\mu_{A}-\mu_{D}}\bga^{D}_{BA} = & -e^{-2\mu_{A}}X_{B}(\bmu_{D})\de_{AD}+e^{-2\mu_{A}}X_{A}(\bmu_{D})\de_{BD}+e^{-2\mu_{A}}\g^{D}_{BA},\\
\theta^{-1}e^{\mu_{B}-\mu_{A}-\mu_{D}}\bga^{B}_{AD} = & -e^{-2\mu_{A}}X_{A}(\bmu_{B})\de_{BD}+e^{-2\mu_{D}}X_{D}(\bmu_{B})\de_{AB}
+e^{2\mu_{B}-2\mu_{A}-2\mu_{D}}\g^{B}_{AD}.
\end{align*}
Thus (\ref{eq:renormalizedChristoffelsymbol}) holds. Summing up yields
\begin{equation*}
\begin{split}
\theta^{-2}e^{\mu_{B}-\mu_{A}}\textstyle{\sum}_{C}\ONCS^{D}_{CC}\ONCS_{AD}^{B} = & \textstyle{\sum}_{C}e^{-2\mu_{C}}\cha_{C}X_{C}(\bmu_{A})\de_{AB}
-e^{-2\mu_{A}}\cha_{A}X_{B}(\bmu_{A})\\
 & -\textstyle{\frac{1}{2}\sum}_{C}e^{-2\mu_{C}}\cha_{C}\g^{A}_{CB}+\textstyle{\frac{1}{2}\sum}_{C}e^{-2\mu_{A}}\cha_{C}\g^{C}_{BA}\\
 & +\textstyle{\frac{1}{2}\sum}_{C}e^{2\mu_{B}-2\mu_{A}-2\mu_{C}}\cha_{C}\g^{B}_{AC}.
\end{split}
\end{equation*}
Using (\ref{eq:chAdef}), this expression can be rewritten
\begin{equation*}
\begin{split}
 & \theta^{-2}e^{\mu_{B}-\mu_{A}}\textstyle{\sum}_{C}\ONCS^{D}_{CC}\ONCS_{AD}^{B}\\
 = & -\textstyle{\sum}_{C}e^{-2\mu_{C}}X_{C}(\hmu_{C})X_{C}(\bmu_{A})\de_{AB}+\textstyle{\sum}_{C}e^{-2\mu_{C}}X_{C}(\bmu_{A})a_{C}\de_{AB}\\
 & +e^{-2\mu_{A}}X_{A}(\hmu_{A})X_{B}(\bmu_{A})-e^{-2\mu_{A}}X_{B}(\bmu_{A})a_{A}+\textstyle{\frac{1}{2}\sum}_{C}e^{-2\mu_{C}}X_{C}(\hmu_{C})\g^{A}_{CB}\\
 & -\textstyle{\frac{1}{2}\sum}_{C}e^{-2\mu_{C}}a_{C}\g^{A}_{CB}-\textstyle{\frac{1}{2}\sum}_{C}e^{-2\mu_{A}}X_{C}(\hmu_{C})\g^{C}_{BA}
+\textstyle{\frac{1}{2}\sum}_{C}e^{-2\mu_{A}}a_{C}\g^{C}_{BA}\\
 & -\textstyle{\frac{1}{2}\sum}_{C}e^{2\mu_{B}-2\mu_{A}-2\mu_{C}}X_{C}(\hmu_{C})\g^{B}_{AC}+\textstyle{\frac{1}{2}\sum}_{C}e^{2\mu_{B}-2\mu_{A}-2\mu_{C}}a_{C}\g^{B}_{AC}.
\end{split}
\end{equation*}
Returning to (\ref{eq:bRicbasicformula}), note that 
\begin{equation}\label{eq:thirdricciterm}
\begin{split}
\theta^{-2}e^{\mu_{B}-\mu_{A}}\textstyle{\sum}_{C}\ONSF_{C}(\ONCS^{B}_{AC}) = & \theta^{-1}\textstyle{\sum}_{C}e^{\mu_{B}-\mu_{A}-\mu_{C}}X_{C}(\ONCS^{B}_{AC})\\
 = & \textstyle{\sum}_{C}X_{C}(\Upsilon^{B}_{AC})+\textstyle{\sum}_{C}X_{C}(\mu_{A}+\mu_{C}+\ln\theta-\mu_{B})\Upsilon^{B}_{AC}. 
\end{split}
\end{equation}
It is of interest to expand the second term on the right hand side using (\ref{eq:renormalizedChristoffelsymbol}). This yields
\begin{equation*}
\begin{split}
 & \textstyle{\sum}_{C}X_{C}(\bmu_{A}+\bmu_{C}+2\ln\theta-\bmu_{B})\Upsilon^{B}_{AC}\\
 = & \textstyle{\sum}_{C}e^{-2\mu_{C}}X_{C}(\bmu_{C}+2\ln\theta)X_{C}(\bmu_{A})\de_{AB}\\
 & -e^{-2\mu_{A}}X_{A}(2\bmu_{A}+2\ln\theta-\bmu_{B})X_{B}(\bmu_{A})\\
 & -\textstyle{\frac{1}{2}}\textstyle{\sum}_{C}e^{-2\mu_{C}}X_{C}(\bmu_{A}+\bmu_{C}+2\ln\theta-\bmu_{B})\g^{A}_{CB}\\
 & +\textstyle{\frac{1}{2}}\textstyle{\sum}_{C}e^{-2\mu_{A}}X_{C}(\bmu_{A}+\bmu_{C}+2\ln\theta-\bmu_{B})\g^{C}_{BA}\\
 & +\textstyle{\frac{1}{2}\sum}_{C}e^{2\mu_{B}-2\mu_{A}-2\mu_{C}}X_{C}(\bmu_{A}+\bmu_{C}+2\ln\theta-\bmu_{B})\g^{B}_{AC}.
\end{split}
\end{equation*}
At this stage we can collect all the higher order terms, and this yields (\ref{eq:mRIABdef}). We also collect the first order terms that 
have arisen so far into
\begin{align}
\mQ_{\mrII,B}^{A} := & \textstyle{\frac{1}{2}\sum}_{C}e^{2\mu_{B}-2\mu_{A}-2\mu_{C}}\left[\cha_{C}-X_{C}(\bmu_{A}+\bmu_{C}+2\ln\theta-\bmu_{B})\right]\g^{B}_{AC},\label{eq:mQIIAB}\\
\mQ_{\mrIII,B}^{A} := & e^{-2\mu_{A}}X_{A}(\bmu_{B})X_{B}(\hmu_{B})+e^{-2\mu_{A}}X_{A}(\bmu_{\rotot}+2\ln\theta+\bmu_{A}-\bmu_{B})X_{B}(\bmu_{A})\label{eq:mQIIIAB}\\
 & -\textstyle{\sum}_{C}e^{-2\mu_{C}}X_{C}(\bmu_{\rotot}+2\ln\theta)X_{C}(\bmu_{A})\de_{AB},\nonumber\\
\mQ_{\mrIV,B}^{A} := & -e^{-2\mu_{A}}\left(X_{A}(\bmu_{B})a_{B}+X_{B}(\bmu_{A})a_{A}\right)+\textstyle{\sum}_{C}e^{-2\mu_{C}}X_{C}(\bmu_{A})a_{C}\de_{AB}\label{mQIVAB}\\
 & +\textstyle{\frac{1}{2}}\textstyle{\sum}_{C}e^{-2\mu_{C}}X_{C}(\bmu_{\rotot}+2\ln\theta+\bmu_{A}-\bmu_{B})\g^{A}_{CB}\nonumber\\
 & -\textstyle{\frac{1}{2}}\textstyle{\sum}_{C}e^{-2\mu_{A}}X_{C}(\bmu_{\rotot}+2\ln\theta+\bmu_{A}-\bmu_{B})\g^{C}_{BA},\nonumber\\
\mQ_{\mrV,B}^{A} := & -\textstyle{\frac{1}{2}}\textstyle{\sum}_{C}e^{-2\mu_{C}}a_{C}\g^{A}_{CB}+\textstyle{\frac{1}{2}}\textstyle{\sum}_{C}e^{-2\mu_{A}}a_{C}\g^{C}_{BA}.
\end{align}
Next, consider
\begin{equation*}
\begin{split}
\textstyle{\sum}_{C,D}\ONCS_{AC}^{D}\ONCS_{CD}^{B} = & 
\textstyle{\frac{1}{4}\sum}_{C,D}(-\bga^{A}_{CD}+\bga^{C}_{DA}-\bga^{D}_{CA})(-\bga^{C}_{DB}-\bga^{D}_{CB}+\bga^{B}_{CD})\\
 = & \textstyle{\frac{1}{4}\sum}_{C,D}(-\bga^{A}_{CD}\bga^{B}_{CD}+2\bga^{C}_{DA}\bga^{B}_{CD}).
\end{split}
\end{equation*}
Moreover, 
\begin{equation*}
\begin{split}
\textstyle{\sum}_{C,D}\bga_{AC}^{D}\ONCS_{DC}^{B} 
 = & \textstyle{\frac{1}{2}\sum}_{C,D}(-\bga_{AC}^{D}\bga^{D}_{CB}+\bga_{AC}^{D}\bga^{C}_{BD}+\bga_{AC}^{D}\bga^{B}_{DC})\\
 = & \textstyle{\frac{1}{2}\sum}_{C,D}(-\bga_{AC}^{D}\bga^{D}_{CB}+\bga_{AC}^{D}\bga^{C}_{BD}-\bga_{DA}^{C}\bga^{B}_{CD}).
\end{split}
\end{equation*}
Adding yields
\begin{equation*}
\begin{split}
\textstyle{\sum}_{C,D}(\ONCS_{AC}^{D}\ONCS_{CD}^{B}+\bga_{AC}^{D}\ONCS_{DC}^{B}) = & 
\textstyle{\frac{1}{4}\sum}_{C,D}(-\bga^{A}_{CD}\bga^{B}_{CD}+2\bga_{AC}^{D}\bga^{D}_{BC}+2\bga_{AC}^{D}\bga^{C}_{BD}). 
\end{split}
\end{equation*}
Due to this equality, it is of interest to calculate
\begin{equation}\label{eq:firstofthreenontrivial}
\begin{split}
-\theta^{-2}e^{\mu_{B}-\mu_{A}}\textstyle{\sum}_{C,D}\bga^{A}_{CD}\bga^{B}_{CD} = & -2\textstyle{\sum}_{C}e^{-2\mu_{C}}|X_{C}(\bmu_{A})|^{2}\de_{AB}
+2e^{-2\mu_{A}}X_{A}(\bmu_{B})X_{B}(\bmu_{A})\\
 & -2\textstyle{\sum}_{C}e^{2\mu_{B}-2\mu_{A}-2\mu_{C}}X_{C}(\bmu_{A})\g^{B}_{AC}\\
 & -2\textstyle{\sum}_{C}e^{-2\mu_{C}}X_{C}(\bmu_{B})\g^{A}_{BC}\\
 & -\textstyle{\sum}_{C,D}e^{2\mu_{B}-2\mu_{C}-2\mu_{D}}\g^{A}_{CD}\g^{B}_{CD}.
\end{split}
\end{equation}
Moreover, 
\begin{equation*}
\begin{split}
2\theta^{-2}e^{\mu_{B}-\mu_{A}}\textstyle{\sum}_{C,D}\bga_{AC}^{D}\bga^{D}_{BC} = & 2\textstyle{\sum}_{C}e^{-2\mu_{A}}X_{A}(\bmu_{C})X_{B}(\bmu_{C})
-2e^{-2\mu_{A}}X_{A}(\bmu_{B})X_{B}(\bmu_{B})\\
 & -2\textstyle{\sum}_{C}e^{-2\mu_{A}}X_{A}(\bmu_{C})\g^{C}_{BC}-2e^{-2\mu_{A}}X_{A}(\bmu_{A})X_{B}(\bmu_{A})\\
 & +2\textstyle{\sum}_{C}e^{-2\mu_{C}}|X_{C}(\bmu_{A})|^{2}\de_{AB}+2\textstyle{\sum}_{C}e^{-2\mu_{C}}X_{C}(\bmu_{A})\g^{A}_{BC}\\
 & -2\textstyle{\sum}_{C}e^{-2\mu_{A}}X_{B}(\bmu_{C})\g^{C}_{AC}+2\textstyle{\sum}_{C}e^{2\mu_{B}-2\mu_{A}-2\mu_{C}}X_{C}(\bmu_{B})\g^{B}_{AC}\\
 & +2\textstyle{\sum}_{C,D}e^{2\mu_{D}-2\mu_{A}-2\mu_{C}}\g^{D}_{AC}\g^{D}_{BC}.
\end{split}
\end{equation*}
Adding these expressions yields
\begin{equation*}
\begin{split}
 & \theta^{-2}e^{\mu_{B}-\mu_{A}}\textstyle{\sum}_{C,D}(-\bga^{A}_{CD}\bga^{B}_{CD}+2\bga_{AC}^{D}\bga^{D}_{BC})\\
= & 2e^{-2\mu_{A}}X_{A}(\bmu_{B})X_{B}(\bmu_{A}-\bmu_{B})-2e^{-2\mu_{A}}X_{A}(\bmu_{A})X_{B}(\bmu_{A})\\
 & +2\textstyle{\sum}_{C}e^{-2\mu_{A}}X_{A}(\bmu_{C})X_{B}(\bmu_{C})-2\textstyle{\sum}_{C}e^{2\mu_{B}-2\mu_{A}-2\mu_{C}}X_{C}(\bmu_{A})\g^{B}_{AC}\\
 & -2\textstyle{\sum}_{C}e^{-2\mu_{A}}X_{A}(\bmu_{C})\g^{C}_{BC}+2\textstyle{\sum}_{C}e^{-2\mu_{C}}X_{C}(\bmu_{A}-\bmu_{B})\g^{A}_{BC}\\
 & -2\textstyle{\sum}_{C}e^{-2\mu_{A}}X_{B}(\bmu_{C})\g^{C}_{AC}+2\textstyle{\sum}_{C}e^{2\mu_{B}-2\mu_{A}-2\mu_{C}}X_{C}(\bmu_{B})\g^{B}_{AC}\\
 & +2\textstyle{\sum}_{C,D}e^{2\mu_{D}-2\mu_{A}-2\mu_{C}}\g^{D}_{AC}\g^{D}_{BC}-\textstyle{\sum}_{C,D}e^{2\mu_{B}-2\mu_{C}-2\mu_{D}}\g^{A}_{CD}\g^{B}_{CD}.
\end{split}
\end{equation*}
Next, compute
\begin{equation*}
\begin{split}
 & \theta^{-2}e^{\mu_{B}-\mu_{A}}\textstyle{\sum}_{C,D}2\bga_{AC}^{D}\bga^{C}_{BD}\\ 
 = & 2\textstyle{\sum}_{C}e^{-2\mu_{A}}X_{A}(\bmu_{C})X_{B}(\bmu_{C})-2e^{-2\mu_{A}}X_{A}(\bmu_{B})X_{B}(\bmu_{B})\\
 & +2e^{-2\mu_{A}}X_{A}(\bmu_{B}-\bmu_{A})X_{B}(\bmu_{A})-2\textstyle{\sum}_{C}e^{-2\mu_{A}}X_{A}(\bmu_{C})\g^{C}_{BC}\\
 & -2\textstyle{\sum}_{C}e^{-2\mu_{A}}X_{B}(\bmu_{C})\g^{C}_{AC}+2\textstyle{\sum}_{C}e^{-2\mu_{A}}X_{C}(\bmu_{B}-\bmu_{A})\g^{C}_{AB}
 +2\textstyle{\sum}_{C,D}e^{-2\mu_{A}}\g^{D}_{AC}\g^{C}_{BD}.
\end{split}
\end{equation*}
Summing up, 
\begin{equation*}
\begin{split}
 & \theta^{-2}e^{\mu_{B}-\mu_{A}}\textstyle{\sum}_{C,D}(-\bga^{A}_{CD}\bga^{B}_{CD}+2\bga_{AC}^{D}\bga^{D}_{BC}+2\bga_{AC}^{D}\bga^{C}_{BD})\\
= & 2e^{-2\mu_{A}}X_{A}(\bmu_{B})X_{B}(\bmu_{A}-2\bmu_{B})+2e^{-2\mu_{A}}X_{A}(\bmu_{B}-2\bmu_{A})X_{B}(\bmu_{A})\\
 & +4\textstyle{\sum}_{C}e^{-2\mu_{A}}X_{A}(\bmu_{C})X_{B}(\bmu_{C})-4\textstyle{\sum}_{C}e^{-2\mu_{A}}X_{B}(\bmu_{C})\g^{C}_{AC}\\
 & -4\textstyle{\sum}_{C}e^{-2\mu_{A}}X_{A}(\bmu_{C})\g^{C}_{BC}-2\textstyle{\sum}_{C}e^{2\mu_{B}-2\mu_{A}-2\mu_{C}}X_{C}(\bmu_{A}-\bmu_{B})\g^{B}_{AC}\\
 & +2\textstyle{\sum}_{C}e^{-2\mu_{C}}X_{C}(\bmu_{A}-\bmu_{B})\g^{A}_{BC}+2\textstyle{\sum}_{C}e^{-2\mu_{A}}X_{C}(\bmu_{B}-\bmu_{A})\g^{C}_{AB}\\
 & +2\textstyle{\sum}_{C,D}e^{-2\mu_{A}}\g^{D}_{AC}\g^{C}_{BD}+2\textstyle{\sum}_{C,D}e^{2\mu_{D}-2\mu_{A}-2\mu_{C}}\g^{D}_{AC}\g^{D}_{BC}\\
 & -\textstyle{\sum}_{C,D}e^{2\mu_{B}-2\mu_{C}-2\mu_{D}}\g^{A}_{CD}\g^{B}_{CD}.
\end{split}
\end{equation*}
Collecting the terms that are potentially large, keeping (\ref{eq:mQIIAB}) in mind, yields
\begin{equation*}
\begin{split}
\mR_{\mrII,B}^{A}:= & \textstyle{\frac{1}{2}\sum}_{C}e^{2\mu_{B}-2\mu_{A}-2\mu_{C}}\left[\cha_{C}-X_{C}(\bmu_{A}+\bmu_{C}+2\ln\theta-\bmu_{B})\right]\g^{B}_{AC}\\
 & +\textstyle{\frac{1}{2}\sum}_{C}e^{2\mu_{B}-2\mu_{A}-2\mu_{C}}X_{C}(\bmu_{A}-\bmu_{B})\g^{B}_{AC}\\
 & -\textstyle{\frac{1}{2}\sum}_{C,D}e^{2\mu_{D}-2\mu_{A}-2\mu_{C}}\g^{D}_{AC}\g^{D}_{BC}+\textstyle{\frac{1}{4}\sum}_{C,D}e^{2\mu_{B}-2\mu_{C}-2\mu_{D}}\g^{A}_{CD}\g^{B}_{CD}\\
= & \textstyle{\frac{1}{2}\sum}_{C}e^{2\mu_{B}-2\mu_{A}-2\mu_{C}}\left[a_{C}-X_{C}(\bmu_{\rotot}+2\ln\theta)\right]\g^{B}_{AC}\\
 & -\textstyle{\frac{1}{2}\sum}_{C,D}e^{2\mu_{D}-2\mu_{A}-2\mu_{C}}\g^{D}_{AC}\g^{D}_{BC}+\textstyle{\frac{1}{4}\sum}_{C,D}e^{2\mu_{B}-2\mu_{C}-2\mu_{D}}\g^{A}_{CD}\g^{B}_{CD}.
\end{split}
\end{equation*}
This leads to the definition (\ref{eq:mRIIABdef}). Next, we can collect part of the remaining terms into $\mR_{\mrIII,B}^{A}$, $\mR_{\mrIV,B}^{A}$ and 
$\mR_{\mrV,B}^{A}$ according to (\ref{eq:mRIIIABdef}), (\ref{eq:mRIVABdef}) and (\ref{eq:mRVABdef}). The lemma follows.
\end{proof}

It is sometimes of interest to divide the terms in $\bmR$ somewhat differently. 

\begin{cor}\label{cor:rescaledRiccicurvatureformwoderoflntheta}
  Let $(M,g)$ be a spacetime. Assume that it has an expanding partial pointed foliation. Assume, moreover, $\mK$ to be non-degenerate on $I$ and to have
  a global frame. Then
  \begin{equation}\label{eq:rescaledRiccicurvatureformwoderoflntheta}
    \theta^{-2}\bR^{A}_{\phantom{A}B} = \mS_{\mrI,B}^{A}+\mS_{\mrII,B}^{A}+\mS_{\mrIII,B}^{A}+\mS_{\mrIV,B}^{A}+\mS_{\mrV,B}^{A}
  \end{equation}
  where
  \begin{align}
    \mS_{\mrI,B}^{A} := & e^{-2\mu_{A}}X_{A}(\cha_{B})-\textstyle{\sum}_{C}e^{-2\mu_{C}}X_{C}^{2}(\bmu_{A})\de_{AB}
    +e^{-2\mu_{A}}X_{A}X_{B}(\bmu_{A})\label{eq:mSIABdef}\\
    & +\textstyle{\frac{1}{2}}\textstyle{\sum}_{C}e^{-2\mu_{C}}X_{C}(\g^{A}_{CB})-\textstyle{\frac{1}{2}}\textstyle{\sum}_{C}e^{-2\mu_{A}}X_{C}(\g^{C}_{BA})
    -\tfrac{1}{2}\textstyle{\sum}_{C}e^{2\mu_{B}-2\mu_{A}-2\mu_{C}}X_{C}(\g^{B}_{AC}),\nonumber\\
    \mS_{\mrII,B}^{A} := & \textstyle{\frac{1}{2}\sum}_{C}e^{2\mu_{B}-2\mu_{A}-2\mu_{C}}\left[a_{C}
      -X_{C}\left(2\bmu_{B}-2\bmu_{A}-2\bmu_{C}+\bmu_{\rotot}\right)\right]\g^{B}_{AC}\label{eq:mSIIABdef}\\
    & -\textstyle{\frac{1}{2}\sum}_{C,D}e^{2\mu_{D}-2\mu_{A}-2\mu_{C}}\g^{D}_{AC}\g^{D}_{BC}
    +\textstyle{\frac{1}{4}\sum}_{C,D}e^{2\mu_{B}-2\mu_{C}-2\mu_{D}}\g^{A}_{CD}\g^{B}_{CD},\nonumber\\
    \mS_{\mrIII,B}^{A}:= & \textstyle{\frac{1}{2}}e^{-2\mu_{A}}X_{A}(\bmu_{B})X_{B}(2\bmu_{\rotot}-\bmu_{A})
    +\textstyle{\frac{1}{2}}e^{-2\mu_{A}}X_{A}(2\bmu_{\rotot}-3\bmu_{B})X_{B}(\bmu_{A})\label{eq:mSIIIABdef}\\
    & +\textstyle{\sum}_{C}e^{-2\mu_{C}}X_{C}(2\bmu_{C}-\bmu_{\rotot})X_{C}(\bmu_{A})\de_{AB}-\textstyle{\sum}_{C}e^{-2\mu_{A}}X_{A}(\bmu_{C})X_{B}(\bmu_{C}),\nonumber\\
    \mS_{\mrIV,B}^{A} := & -e^{-2\mu_{A}}\left(X_{A}(\bmu_{B})a_{B}+X_{B}(\bmu_{A})a_{A}\right)
    +\textstyle{\sum}_{C}e^{-2\mu_{C}}X_{C}(\bmu_{A})a_{C}\de_{AB}\label{eq:mSIVABdef}\\
    & +\textstyle{\sum}_{C}e^{-2\mu_{A}}[X_{B}(\bmu_{C})\g^{C}_{AC}+X_{A}(\bmu_{C})\g^{C}_{BC}]
    +\textstyle{\frac{1}{2}}\textstyle{\sum}_{C}e^{-2\mu_{A}}X_{C}(\bmu_{\rotot}-2\bmu_{B})\g^{C}_{AB}\nonumber\\
    & -\textstyle{\frac{1}{2}}\textstyle{\sum}_{C}e^{-2\mu_{C}}X_{C}(\bmu_{\rotot}+2\bmu_{A}-2\bmu_{B}-2\bmu_{C})\g^{A}_{BC},\nonumber\\
    \mS_{\mrV,B}^{A} := & -\textstyle{\frac{1}{2}}\textstyle{\sum}_{C}e^{-2\mu_{C}}a_{C}\g^{A}_{CB}+\textstyle{\frac{1}{2}}\textstyle{\sum}_{C}e^{-2\mu_{A}}a_{C}\g^{C}_{BA}
    -\textstyle{\frac{1}{2}}\textstyle{\sum}_{C,D}e^{-2\mu_{A}}\g^{D}_{AC}\g^{C}_{BD}\label{eq:mSVABdef}
  \end{align}
  and $\bmu_{\rotot}:=\textstyle{\sum}_{D}\bmu_{D}$.
\end{cor}
\begin{remark}
  One advantage of (\ref{eq:mSIABdef})--(\ref{eq:mSVABdef}) in comparison with (\ref{eq:mRIABdef})--(\ref{eq:mRVABdef}) is that it is easier to identify
  the terms that are potentially growing exponentially. 
\end{remark}
\begin{proof}
Note, to begin with, that 
\begin{equation}\label{eq:divUpsilonexpand}
\begin{split}
-\textstyle{\sum}_{C}X_{C}(\Upsilon^{B}_{AC}) = & -\textstyle{\sum}_{C}e^{-2\mu_{C}}X_{C}^{2}(\bmu_{A})\de_{AB}
+e^{-2\mu_{A}}X_{A}X_{B}(\bmu_{A})+\textstyle{\frac{1}{2}}\textstyle{\sum}_{C}e^{-2\mu_{C}}X_{C}(\g^{A}_{CB})\\
 & -\textstyle{\frac{1}{2}}\textstyle{\sum}_{C}e^{-2\mu_{A}}X_{C}(\g^{C}_{BA})-\tfrac{1}{2}\textstyle{\sum}_{C}e^{2\mu_{B}-2\mu_{A}-2\mu_{C}}X_{C}(\g^{B}_{AC})\\
 & +2\textstyle{\sum}_{C}e^{-2\mu_{C}}X_{C}(\mu_{C})X_{C}(\bmu_{A})\de_{AB}-2e^{-2\mu_{A}}X_{A}(\mu_{A})X_{B}(\bmu_{A})\\
 & -\textstyle{\sum}_{C}e^{-2\mu_{C}}X_{C}(\mu_{C})\g^{A}_{CB}+\textstyle{\sum}_{C}e^{-2\mu_{A}}X_{C}(\mu_{A})\g^{C}_{BA}\\
 & -\textstyle{\sum}_{C}e^{2\mu_{B}-2\mu_{A}-2\mu_{C}}X_{C}(\mu_{B}-\mu_{A}-\mu_{C})\g^{B}_{AC}.
\end{split}
\end{equation}
Combining the first five terms on the right hand side with the first term on the right hand side of (\ref{eq:mRIABdef}) yields (\ref{eq:mSIABdef}). 
The last five terms on the right hand side of (\ref{eq:divUpsilonexpand}) can then be combined with (\ref{eq:mRIIABdef})--(\ref{eq:mRVABdef})
in order to yield (\ref{eq:mSIIABdef})--(\ref{eq:mSVABdef}). 
\end{proof}

\section{Contribution from the lapse function}

Next, we consider the contribution of the lapse function to the right hand side of (\ref{eq:mlUbKwithEinstein}).

\begin{lemma}\label{lemma:contributionsfromlapse}
  Let $(M,g)$ be a spacetime. Assume that it has an expanding partial pointed foliation. Assume, moreover, $\mK$ to be non-degenerate on $I$ and to
  have a global frame. Then
\begin{equation}\label{eq:bnablaABnormN}
\theta^{-2}N^{-1}\bnabla^{A}\bnabla_{B}N = \mN_{\mrI,B}^{A}+\mN_{\mrII,B}^{A}+\mN_{\mrIII,B}^{A},
\end{equation}
where
\begin{align}
\mN_{\mrI,B}^{A} = & e^{-2\mu_{A}}X_{A}X_{B}(\ln N),\label{eq:bnablaABnormNrI}\\
\mN_{\mrII,B}^{A} = & -\textstyle{\frac{1}{2}\sum}_{C}e^{2\mu_{B}-2\mu_{C}-2\mu_{A}}\g^{B}_{CA}X_{C}(\ln N),\label{eq:bnablaABnormNrII}\\
\mN_{\mrIII,B}^{A} = & e^{-2\mu_{A}}X_{A}(\ln N)X_{B}(\ln N)-e^{-2\mu_{A}}X_{A}(\bmu_{B})X_{B}(\ln N)\label{eq:bnablaABnormNrIII}\\
 & -e^{-2\mu_{A}}X_{B}(\bmu_{A})X_{A}(\ln N)+\textstyle{\sum}_{C}e^{-2\mu_{C}}X_{C}(\bmu_{B})\de_{BA}X_{C}(\ln N)\nonumber\\
 & +\textstyle{\frac{1}{2}\sum}_{C}e^{-2\mu_{C}}\g^{A}_{BC}X_{C}(\ln N)-\textstyle{\frac{1}{2}\sum}_{C}e^{-2\mu_{A}}\g^{C}_{AB}X_{C}(\ln N).\nonumber
\end{align}
In particular,
\begin{equation}\label{eq:normalisedDeltaNthroughN}
\begin{split}
\theta^{-2}N^{-1}\Delta_{\bge}N = & \textstyle{\sum}_{A}e^{-2\mu_{A}}X_{A}^{2}(\ln N)-\textstyle{\sum}_{C}e^{-2\mu_{C}}a_{C}X_{C}(\ln N)\\
 & +\textstyle{\sum}_{A}e^{-2\mu_{A}}|X_{A}(\ln N)|^{2}-2\textstyle{\sum}_{A}e^{-2\mu_{A}}X_{A}(\bmu_{A})X_{A}(\ln N)\\
 & +\textstyle{\sum}_{A,C}e^{-2\mu_{C}}X_{C}(\bmu_{A})X_{C}(\ln N). 
\end{split}
\end{equation}
\end{lemma}
\begin{proof}
Compute
\[
(\bnabla^{2}N)(\ONSF_{A},\ONSF_{B})=(\bnabla_{\ONSF_{A}}\bnabla N)(\ONSF_{B})=\ONSF_{A}[\ONSF_{B}(N)]-(\bnabla_{\ONSF_{A}}\ONSF_{B})N.
\]
Thus
\begin{equation*}
\begin{split}
(\bnabla^{2}N)(X_{A},X_{B}) = & e^{\bmu_{B}}X_{A}[e^{-\bmu_{B}}X_{B}(N)]-\textstyle{\sum}_{C}e^{\bmu_{A}+\bmu_{B}-\bmu_{C}}\ONCS^{C}_{AB}X_{C}(N)\\
 = & X_{A}X_{B}(N)-X_{A}(\bmu_{B})X_{B}(N)-\textstyle{\sum}_{C}\theta^{-1}e^{\mu_{A}+\mu_{B}-\mu_{C}}\ONCS^{C}_{AB}X_{C}(N).
\end{split}
\end{equation*}
In particular,
\begin{equation}\label{eq:bnablaupdownrenormform}
\begin{split}
\theta^{-2}\bnabla^{D}\bnabla_{B}N = & \textstyle{\sum}_{A}e^{-2\mu_{D}}\de^{DA}(\bnabla^{2}N)(X_{A},X_{B}) \\
 = & e^{-2\mu_{D}}X_{D}X_{B}(N)-e^{-2\mu_{D}}X_{D}(\bmu_{B})X_{B}(N)\\
 & -\textstyle{\sum}_{C}\theta^{-1}e^{\mu_{B}-\mu_{C}-\mu_{D}}\ONCS^{C}_{DB}X_{C}(N).
\end{split}
\end{equation}
It is therefore of interest to calculate, using (\ref{eq:Upsilondef}) and (\ref{eq:renormalizedChristoffelsymbol}), that 
\begin{equation*}
\begin{split}
\theta^{-1}e^{\mu_{B}-\mu_{C}-\mu_{A}}\ONCS^{C}_{AB} = & e^{2\mu_{B}-2\mu_{C}}\Upsilon^{C}_{AB}\\
 = & e^{-2\mu_{C}}X_{B}(\bmu_{C})\de_{CA}-e^{-2\mu_{C}}X_{C}(\bmu_{B})\de_{BA}-\textstyle{\frac{1}{2}}e^{-2\mu_{C}}\g^{A}_{BC}\\
& +\textstyle{\frac{1}{2}}e^{-2\mu_{A}}\g^{C}_{AB}+\textstyle{\frac{1}{2}}e^{2\mu_{B}-2\mu_{C}-2\mu_{A}}\g^{B}_{CA}.
\end{split}
\end{equation*}
Summing up, 
\begin{equation*}
\begin{split}
\theta^{-2}N^{-1}\bnabla^{A}\bnabla_{B}N = & e^{-2\mu_{A}}X_{A}X_{B}(\ln N)+e^{-2\mu_{A}}X_{A}(\ln N)X_{B}(\ln N)\\
 & -e^{-2\mu_{A}}X_{A}(\bmu_{B})X_{B}(\ln N)-e^{-2\mu_{A}}X_{B}(\bmu_{A})X_{A}(\ln N)\\
 & +\textstyle{\sum}_{C}e^{-2\mu_{C}}X_{C}(\bmu_{B})\de_{BA}X_{C}(\ln N)+\textstyle{\frac{1}{2}\sum}_{C}e^{-2\mu_{C}}\g^{A}_{BC}X_{C}(\ln N)\\
& -\textstyle{\frac{1}{2}\sum}_{C}e^{-2\mu_{A}}\g^{C}_{AB}X_{C}(\ln N)-\textstyle{\frac{1}{2}\sum}_{C}e^{2\mu_{B}-2\mu_{C}-2\mu_{A}}\g^{B}_{CA}X_{C}(\ln N).
\end{split}
\end{equation*}
The lemma follows. 
\end{proof}

\section{Bianchi class B}\label{section:Bianchi class B}

As is clear from the discussion in Subsection~\ref{ssection:revisitsphom}, the unimodular setting is very special in that it allows variables, due to
Wainwright
and Hsu, that summarise, at the same time and in a very natural and concise way, the algebraic structure associated with the Lie algebra and the geometric
structure associated with the expansion normalised Weingarten map. The connection is much less natural in the case of Bianchi class B. It is therefore of
interest to see if the perspective developed in this article is useful in that setting as well. 

Let $G$ be a $3$-dimensional non-unimodular Lie group with corresponding Lie algebra $\mfg$. Define the one-form
$\xi_{G}:\mfg\rightarrow\mathbb{R}$ by $\xi_{G}(x):=\tr(\mathrm{ad}_{x})/2$, where $\mathrm{ad}_{x}(y)=[x,y]$. Note that $\xi_{G}\neq 0$ due to the fact
that $G$ is non-unimodular. Due to \cite[Lemma~11.2, p.~754]{RadermacherNonStiff}, we know that $\mfg_{2}:=\mathrm{ker}\xi_{G}$ is a $2$-dimensional Abelian
subalgebra. Let $v_{2}\in\mfg$ be such that $\xi_{G}(v_{2})=1$. Then $v_{2}\notin\mfg_{2}$. Define $A_{2}:\mfg_{2}\rightarrow\mfg_{2}$ by
$A_{2}(x):=\mathrm{ad}_{v_{2}}(x)$; that $A_{2}$ takes its values in $\mfg_{2}$ is justified by the proof of \cite[Lemma~11.5, p.~755]{RadermacherNonStiff}.
It can be verified that $A_{2}$ does not depend on the choice of $v_{2}$. 

The non-unimodular (or Bianchi class B) Lie groups can now be classified as follows. If $A_{2}=\mathrm{Id}$, the Lie group is said to be of Bianchi type V;
if $A_{2}\neq \mathrm{Id}$ and $4\mathrm{det}A_{2}=(\tr A_{2})^{2}$, the Lie group is said to be of Bianchi type IV; if $4\mathrm{det}A_{2}<(\tr A_{2})^{2}$,
the Lie group is said to be of Bianchi type VI${}_{\eta}$; and if $4\mathrm{det}A_{2}>(\tr A_{2})^{2}$, the Lie group is said to be of Bianchi type
VII${}_{\eta}$. Here $\eta$ is determined by the following relation
\begin{equation}\label{eq:etadef}
  \eta:=\tfrac{(\tr A_{2})^{2}}{4\mathrm{det}A_{2}-(\tr A_{2})^{2}}.
\end{equation}
It is convenient to evaluate the above with respect to an appropriately chosen basis. Let $\{ e_{i}\}$, $i=1,2,3$, be a basis of $\mfg$. Define the structure
constants associated with this basis by $[e_{i},e_{j}]=\g_{ij}^{k}e_{k}$. Due to, e.g., \cite[Subsection~E.1.1, pp.~695--696]{stab}, there is a uniquely determined
symmetric matrix $\nu$ and a uniquely determined vector $a$ such that $\g^{i}_{jk}=\e_{jkl}\nu^{li}+a_{j}\de^{i}_{k}-a_{k}\de^{i}_{j}$. Here
\begin{equation}\label{eq:nuaformulae}
  \nu^{ij}=\tfrac{1}{2}\g^{(i}_{kl}\e^{j)kl},\ \ \
  a_{k}=\tfrac{1}{2}\g^{i}_{ki},
\end{equation}
where the parenthesis denotes symmetrisation over $i$ and $j$ in the first equality. 
Note, moreover, that $a_{i}=\xi_{G}(e_{i})$, so that $a\neq 0$. In addition, due to the Jacobi identity, $\nu a=0$; see, e.g.,
\cite[Subsection~E.1.2, p.~696]{stab}. Due to \cite[(E.5), p.~697]{stab}, a basis $\{e_{i}\}$, $i=1,2,3$, of $\mfg$ can be chosen such that $a_{1}\neq 0$;
$a_{2}=a_{3}=0$; $\nu^{ij}=0$, $i\neq j$; $\nu^{11}=0$. Then $\{e_{2},e_{3}\}$ span $\mfg_{2}$ and $\xi_{G}(v_{2})=1$, if we let $v_{2}:=a_{1}^{-1}e_{1}$.
Compute
\begin{align}
  A_{2}(e_{2}) = & a_{1}^{-1}[e_{1},e_{2}]=a_{1}^{-1}\g_{12}^{i}e_{i}=a_{1}^{-1}\nu^{33}e_{3}+e_{2},\label{eq:Atwoetwo}\\
  A_{2}(e_{3}) = & a_{1}^{-1}[e_{1},e_{3}]=-a_{1}^{-1}\nu^{22}e_{2}+e_{3}.\label{eq:Atwoethree}
\end{align}
With this information in mind, the above classification of Bianchi class B Lie groups is summarised by \cite[Table~5, p.~753]{RadermacherNonStiff}. 
It can also be computed that $a_{1}^{-2}\nu^{22}\nu^{33}=\eta^{-1}$. 

Next, assume that a left invariant metric $\bge$ and a left invariant symmetric covariant $2$-tensor field $\bk$ have been specified on $G$ (they should
be thought of as parts of initial data for Einstein's equations; the remaining part consisting of matter fields). Then the above frame can be assumed to
be such that it is orthonormal with respect to
$\bge$; see \cite[Section~E.1.4, p.~697]{stab}. Let $\bK$ be the left invariant $(1,1)$-tensor field obtained from $\bk$ by raising one of the indices of
$\bk$ using $\bge$. Let $\theta=\tr\bK$, assume $\theta>0$ and let $\mK:=\bK/\theta$. Then we use the notation $\mK e_{i}=\mK_{i}^{\phantom{i}j}e_{j}$. Compute
\begin{equation*}
  \begin{split}
    \xi_{G}\circ\mK\circ A_{2}(e_{2}) = & \xi_{G}\circ\mK(a_{1}^{-1}\nu^{33}e_{3}+e_{2})=\xi_{G}(a_{1}^{-1}\nu^{33}\mK_{3}^{\phantom{3}1}e_{1}+\mK_{2}^{\phantom{2}1}e_{1})\\
    = & \nu^{33}\mK_{3}^{\phantom{3}1}+a_{1}\mK_{2}^{\phantom{2}1}.
  \end{split}
\end{equation*}
Similarly, $\xi_{G}\circ\mK\circ A_{2}(e_{3})=-\nu^{22}\mK_{2}^{\phantom{2}1}+a_{1}\mK_{3}^{\phantom{3}1}$. Comparing these equalities with
\cite[(74), p.~764]{RadermacherNonStiff} yields the conclusion that two of the components of the momentum constraint (in the orthogonal perfect fluid
setting) read $\xi_{G}\circ\mK\circ B_{2}=0$, where $B_{2}:=A_{2}+2\mathrm{Id}$. Note that this is equivalent to saying that $\mK\circ B_{2}$ maps $\mfg_{2}$
into itself. Next, 
\[
\mK\circ B_{2}(e_{2}) =  \mK(a_{1}^{-1}\nu^{33}e_{3}+3e_{2}),\ \ \
\mK\circ B_{2}(e_{3}) =  \mK(-a_{1}^{-1}\nu^{22}e_{2}+3e_{3}).
\]
In particular,
\[
\tr (\mK\circ B_{2})=a_{1}^{-1}\nu^{33}\mK_{3}^{\phantom{3}2}+3\mK_{2}^{\phantom{2}2}
-a_{1}^{-1}\nu^{22}\mK_{2}^{\phantom{2}3}+3\mK_{3}^{\phantom{3}3}.
\]
Since $\mK_{3}^{\phantom{3}2}=\mK_{2}^{\phantom{2}3}$, we conclude that
\[
\tr (\mK\circ B_{2})-2\tr\mK=-2\mK_{1}^{\phantom{1}1}+a_{1}^{-1}(\nu^{33}-\nu^{22})\mK_{3}^{\phantom{3}2}+\mK_{2}^{\phantom{2}2}
+\mK_{3}^{\phantom{3}3}.
\]
Comparing this equality with \cite[(73), p.~764]{RadermacherNonStiff} yields the conclusion that $\tr (\mK\circ B_{2})=2$. Thus the momentum constraint
can be summarised as saying that $\mK\circ B_{2}$ maps $\mfg_{2}$ into itself and that $\tr (\mK\circ B_{2})=2$. Next, note that the operator $B_{2}$ can be
represented by the matrix
\[
\left(\begin{array}{cc} 3 & -a_{1}^{-1}\nu^{22}\\ a_{1}^{-1}\nu^{33} & 3\end{array}\right).
\]
In particular $\mathrm{det}B_{2}=9+a_{1}^{-2}\nu^{22}\nu^{33}=9+\eta^{-1}$. One consequence of this observation is that $\eta=-1/9$ is very special, since this
is the only $\eta$-value for which $B_{2}$ is not a bijective linear operator. If $\eta\neq -1/9$, the momentum constraint implies that $\mK$ maps $\mfg_{2}$
to itself. The general Bianchi VI${}_{-1/9}$ case can be expected to be quite complicated. In the stiff fluid setting, there are recent results
due to Hans Oude Groeniger, see \cite{GBVIII}. However, in the vacuum setting there are, to the best of our knowledge, no mathematical
results; see, however, \cite{hhw}. For that reason, we here mainly restrict our attention to the case $\eta\neq -1/9$, for which there are more complete
results; see, e.g., \cite{haw,RadermacherNonStiff,RadermacherStiff}. It is referred to as the non-exceptional Bianchi class B subcase.

Assuming $\eta\neq -1/9$, $\mK$ maps $\mfg_{2}$ into itself. This means that $\mfg_{2}^{\perp}$, spanned by $e_{1}$, has to be an eigenspace of $\mK$.
In particular, $e_{1}$ is an eigenvector of $\mK$. If $E_{2}$ and $E_{3}$ are orthogonal left invariant eigenvector fields of $\mK$ which are orthogonal to
$e_{1}$, then $E_{A}\in\mfg_{2}$, $A=2,3$. Thus $[E_{2},E_{3}]=0$, so that if $\bga^{i}_{jk}$ are the structure constants associated with the basis
$\{e_{1},E_{2},E_{3}\}$, then
$\bga^{1}_{23}=0$. In other words, the situation that $e_{1}$ is the eigenvector corresponding to the smallest eigenvalue of $\mK$ is favourable; see
the assumptions of Theorem~\ref{thm:SobestimatesQuiescentVacuumNG}. Note also that, for the frame $\{e_{1},E_{2},E_{3}\}$, one would in general expect
the corresponding $\nu^{22}$ and $\nu^{33}$ to be non-zero, which would mean that $\bga^{2}_{13}\neq 0$ and $\bga^{3}_{12}\neq 0$. 

\begin{figure}
  \begin{center}
    \includegraphics{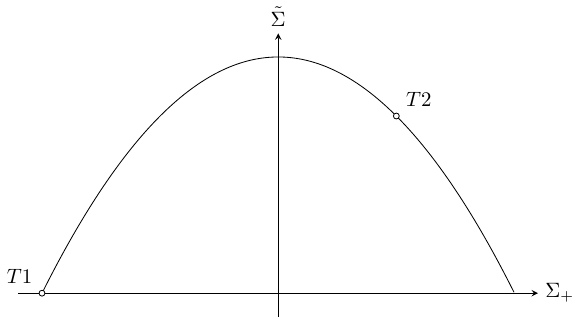}
  \end{center}
  \caption{The Kasner parabola (representing the Kasner solutions) with the two special points $T1$ and $T2$ indicated (these points correspond to flat
    Kasner solutions). In the vacuum
    setting, generic Bianchi type VI${}_{\eta}$ and VII${}_{\eta}$ solutions converge to a point on the Kasner parabola to the right
    of $T2$.}\label{fig:KasnerParabola}
\end{figure}

When formulating Einstein's equations in the non-exceptional Bianchi class B setting, it is convenient to use expansion normalised variables; see
\cite{haw}. The relation between these variables, initial data and maximal globally hyperbolic developments is discussed at length in
\cite[Section~11]{RadermacherNonStiff}. Let, to begin with, $\theta_{ij}$ be the components of the second fundamental form with respect to an
appropriate orthonormal frame. Due to arguments given in \cite[Section~11.6]{RadermacherNonStiff}, we can assume the frame to be such that
$\theta_{1A}=0$ for $A\in\{2,3\}$; see also the above discussion. For this reason, the components of the second fundamental form can be summarised
by $\theta$ (the mean curvature) and $\sigma_{AB}=\theta_{AB}-\theta\de_{AB}/3$, $A,B\in \{2,3\}$. Moreover, $\sigma_{AB}$ can be represented by
$\sigma_{+}=3\sigma^{A}_{\phantom{A}A}/2$ and $\tsigma_{AB}=\sigma_{AB}-\sigma_{+}\de_{AB}/3$. The expansion normalised variables introduced in \cite{haw}
include $\Sigma_{+}:=\sigma_{+}/\theta$ and $\tSigma:=\tsigma/\theta^{2}$, where $\tsigma:=3\tsigma^{AB}\tsigma_{AB}/2$. It is of interest to relate these
quantities with the eigenvalues of $\mK$. By the above, $e_{1}$ is an eigenvector of $\mK$ and the corresponding eigenvalue is given by $(1-2\Sigma_{+})/3$.
It can also be calculated that the eigenvalues of $\mK$ corresponding to eigenvectors in $\mfg_{2}$ are given by $(\Sigma_{+}+1)/3\pm \tSigma^{1/2}/\sqrt{3}$.
Note that the sum of these eigenvalues is $1$. In the vacuum quiescent setting, the limits of the eigenvalues should be such that the sum of their
squares equals $1$; see Remark~\ref{remark:eigenvaluesQuiescentVacuumNG}. On the other hand, demanding that the sum of the squares equals $1$ yields the
relation $\tSigma+\Sigma_{+}^{2}=1$, a set referred to as the \textit{Kasner parabola}; see Figure~\ref{fig:KasnerParabola}. As mentioned above, we expect
the situation that the eigenvalue
corresponding to $e_{1}$ is the smallest to be favourable for quiescent behaviour. In order for $(1-2\Sigma_{+})/2$ to be the strictly smallest eigenvalue
on the Kasner parabola, it has to be strictly negative. In other words, we have to have $\Sigma_{+}>1/2$. To conclude, in the case of non-exceptional
Bianchi class B vacuum solutions, we expect the generic behaviour to be convergence to a point on the Kasner parabola with $\Sigma_{+}>1/2$. In fact,
generic Bianchi type VI${}_{\eta}$ and VII${}_{\eta}$ solutions have the property that they converge to a point on the Kasner parabola with $\Sigma_{+}>1/2$;
see \cite[Theorem~1.18, p.~700]{RadermacherNonStiff}. 

In the stiff fluid setting, it is of interest to calculate the union of the region in which the eigenvalue corresponding to $e_{1}$ is smallest with
the region in which $\mK$ is positive definite. The region where $\mK$ is positive definite is defined by $\tSigma\geq 0$, $\Sigma_{+}<1/2$ and
$\tSigma<(\Sigma_{+}+1)^{2}/3$. Adding the region $\Sigma_{+}\geq 1/2$ to this (in case there is matter present) results in the shaded region
$\mathcal{J}_{+}$ in Figure~\ref{fig:Jacobsset}. Moreover, all Bianchi type VII${}_{\eta}$ solutions converge to a point in $\mathcal{J}_{+}$ and generic
Bianchi type VI${}_{\eta}$ solutions converge to a point in the closure of $\mathcal{J}_{+}$; see \cite[Theorem~1.6, p.~6]{RadermacherStiff}. In
the case of Bianchi type VI${}_{-1/9}$ orthogonal stiff fluids, we generically obtain convergence to to a point at which all the eigenvalues of the
expansion normalised Weingarten map are strictly positive, see \cite{GBVIII}. 

\begin{figure}
  \begin{center}
    \includegraphics{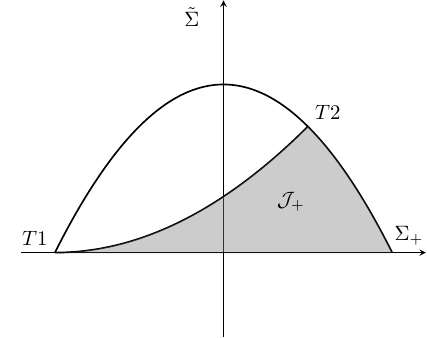}
  \end{center}
  \caption{The shaded region, denoted $\mathcal{J}_{+}$, indicates the set to which generic Bianchi class B stiff fluid solutions are
    expected to converge.}\label{fig:Jacobsset}
\end{figure}

\textbf{Kantowski-Sachs.} The Kantowski-Sachs metrics take the form
\[
g=-dt\otimes dt+a^{2}(t)dx\otimes dx+b^{2}(t)\bge_{\sn{2}}
\]
on $\ro\times\sn{2}\times I$ (or $\so\times\sn{2}\times I$), where $I$ is an open interval and $a,b$ are smooth positive functions on $I$. In this
case, one eigenvector of $\mK$, say $e_{1}$, is parallel to $\d_{x}$ and two, say $e_{2}$, $e_{3}$, are parallel to $\sn{2}$. Moreover, $e_{1}$ commutes
with $e_{A}$, $A=2,3$, and if the eigenvalue associated with $e_{i}$ is denoted $\ell_{i}$, then $\ell_{2}=\ell_{3}$. In other words, the situation
is, due to the symmetry class, automatically degenerate, and $\g^{i}_{jk}=0$ if $\{i,j,k\}=\{1,2,3\}$. In particular, oscillations are suppressed in
the vacuum setting.

\end{document}